\documentclass[aps,pre,reprint,superscriptaddress,showpacs,amsmath,floatfix]{revtex4-2} 
\usepackage{mathtools}
\usepackage{amsmath,amssymb,mathrsfs,amsthm}
\usepackage{mathrsfs}
\usepackage{bm}
\usepackage{bbm}
\usepackage{graphicx}
\usepackage{multirow}
\usepackage{color}
\usepackage{booktabs}
\usepackage{enumitem}
\usepackage{comment}
\usepackage{nameref}
\usepackage[colorlinks=true,linkcolor=blue,urlcolor=blue,citecolor=blue,anchorcolor=blue]{hyperref}
\usepackage{subcaption}  
\usepackage[font=small, justification=raggedright, singlelinecheck=off]{caption}
\definecolor{monbbleu}{RGB}{76, 114, 176}
\usepackage{algorithm}
\makeatletter
\renewcommand{\ALG@name}{Procedure}
\makeatother
\usepackage{algpseudocode}

\usepackage{tikz-cd}

\makeatletter
\g@addto@macro\normalsize{%
  \setlength\abovedisplayskip{5pt}
  \setlength\belowdisplayskip{5pt}
  \setlength\abovedisplayshortskip{5pt}
  \setlength\belowdisplayshortskip{5pt}
}
\makeatother

\usepackage{prettyref}
\newrefformat{SIsec}{SI~\ref{#1}}
\newrefformat{SIsubsec}{SI~\ref{#1}}
\setlist[itemize]{noitemsep, topsep=0pt}
\setlist[enumerate]{noitemsep, topsep=0pt}

   \newtheorem{theorem}{Theorem}
   \newtheorem{corollary}[theorem]{Corollary}
   \newtheorem{lemma}[theorem]{Lemma}
   \newtheorem{proposition}[theorem]{Proposition}
   \newtheorem{assumptions}[theorem]{Assumptions}
\theoremstyle{definition}
   \newtheorem{definition}[theorem]{Definition}
   \newtheorem{example}[theorem]{Example}
\theoremstyle{remark}    
  \newtheorem{remark}[theorem]{Remark}

\DeclareMathOperator{\tr}{tr}
\DeclareMathOperator{\rank}{rank}

\DeclareMathOperator{\diag}{diag}

\DeclareMathOperator*{\argmax}{arg\,max}
\DeclareMathOperator*{\argmin}{arg\,min}

\makeatletter
\newcommand{\labitem}[2]{%
\def\@itemlabel{\textcolor{black}{#1}}
\item
\def\@currentlabel{#1}\label{#2}}
\makeatother

\begin{document}

\title{The low-rank hypothesis of complex systems}

\author{Vincent Thibeault}\email[]{vincent.thibeault.1@ulaval.ca}
\affiliation{D\'epartement de physique, de g\'enie physique et d'optique, Universit\'e Laval, Qu\'ebec (Qc), Canada}
\affiliation{Centre interdisciplinaire en mod\'elisation math\'ematique de l'Universit\'e Laval, Qu\'ebec (Qc), Canada} 

\author{Antoine Allard}
\affiliation{D\'epartement de physique, de g\'enie physique et d'optique, Universit\'e Laval, Qu\'ebec (Qc), Canada}
\affiliation{Centre interdisciplinaire en mod\'elisation math\'ematique de l'Universit\'e Laval, Qu\'ebec (Qc), Canada} 

\author{Patrick Desrosiers}\email[]{patrick.desrosiers@phy.ulaval.ca}
\affiliation{D\'epartement de physique, de g\'enie physique et d'optique, Universit\'e Laval, Qu\'ebec (Qc), Canada} 
\affiliation{Centre interdisciplinaire en mod\'elisation math\'ematique de l'Universit\'e Laval, Qu\'ebec (Qc), Canada}
\affiliation{Centre de recherche CERVO, Qu\'ebec (Qc), Canada}

\maketitle

\let\oldaddcontentsline\addcontentsline
\renewcommand{\addcontentsline}[3]{}

\onecolumngrid


\vspace{-1.5\baselineskip}
\noindent\textbf{%
Complex systems are high-dimensional nonlinear dynamical systems with heterogeneous interactions among their constituents. To make interpretable predictions about their large-scale behavior, it is typically assumed that these dynamics can be reduced to a few equations involving a low-rank matrix describing the network of interactions. Our paper sheds light on this low-rank hypothesis and questions its validity. Using fundamental theorems on singular value decomposition, we probe the hypothesis for various random graphs, either by making explicit their low-rank formulation or by demonstrating the exponential decrease of their singular values. We verify the hypothesis for real networks by revealing the rapid decrease of their singular values, which has major consequences on their effective ranks. We then evaluate the impact of the low-rank hypothesis for general dynamical systems on networks through an optimal dimension reduction. This allows us to prove that recurrent neural networks can be exactly reduced, and to connect the rapidly decreasing singular values of real networks to the dimension reduction error of the nonlinear dynamics they support. Finally, we prove that higher-order interactions naturally emerge from the dimension reduction, thus providing insights into the origin of higher-order interactions in complex systems.
}
\vspace{\baselineskip}

\twocolumngrid

Unraveling the emergent phenomena that drive the functions of complex systems requires to rally the microscopic mechanisms with the macroscopic ones.
Rather than decomposing complex systems in as many components as possible, dimension reduction seeks a reduced system of macrostates or observables with a small enough dimension to get an insightful description, but large enough to preserve the phenomena of interest. 
Yet, complex systems are characterized by extremely high dimensions---perhaps some sort of curse of dimensionality~\cite{Bellman1957, Ganguli2012, Abbott2020}---and finding such reduced system remains a challenge in several scientific disciplines. 

In the paradigm ``More is different''~\cite{Anderson1972, Strogatz2022}, it could appear contradictory to look for simple representations of complex systems. But ``simple model'' does not mean ``simple behavior'': the logistic equation~\cite{May1976}, cellular automata~\cite{vonNeumann1963, Wolfram1984}, or spin glasses~\cite{Parisi1993, Stein2013} exhibit complex behaviors such as chaos, and 
recurrent neural networks can approximate any finite trajectory of $N$-dimensional dynamical systems~\cite{Funahashi1993}.

\begin{figure*}[t]
    \centering
    \includegraphics[width=\textwidth]{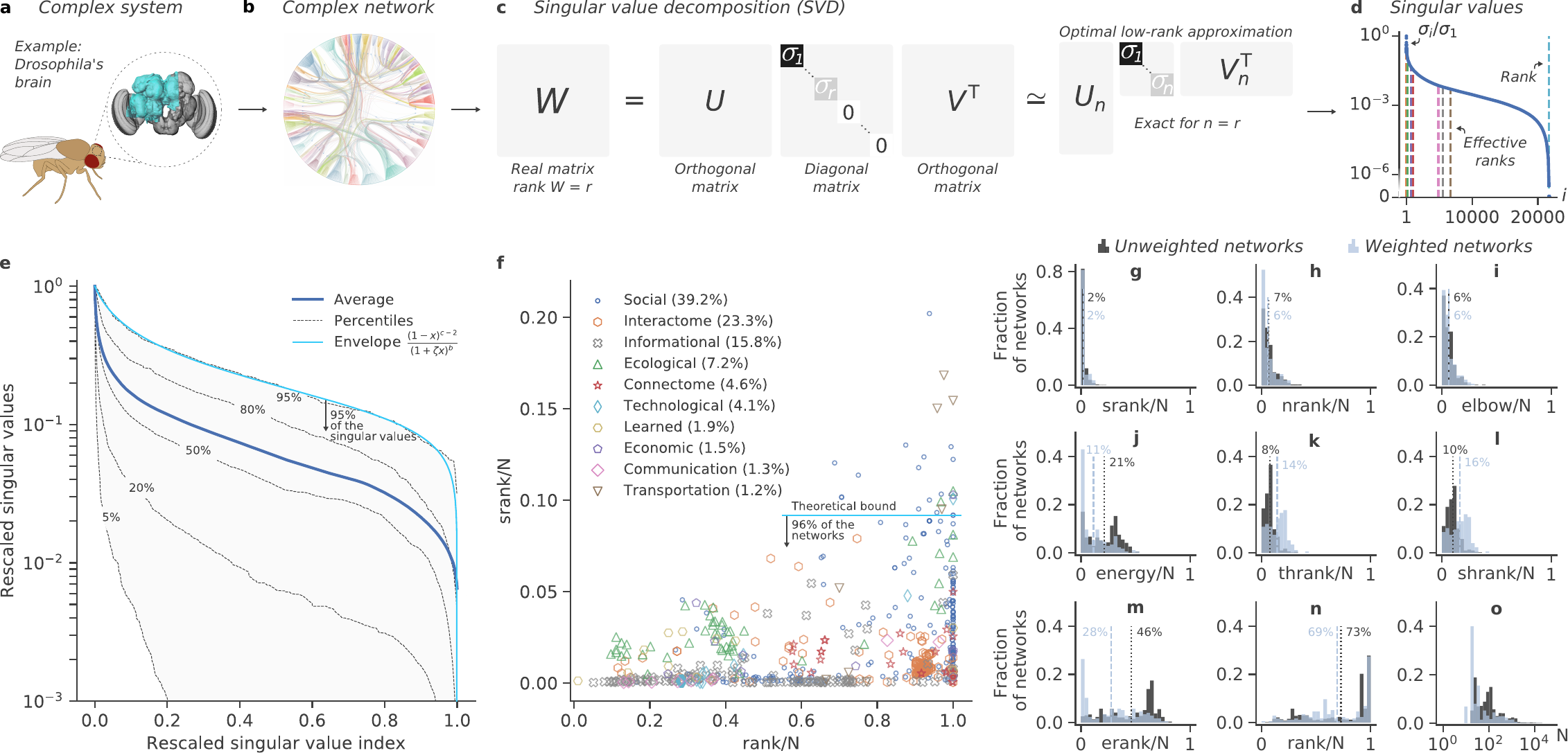}
    \vspace{-0.5cm}
    \caption{\textbf{Experimental verification of the low-rank hypothesis for real networks}.
    \textbf{a},~\textit{Drosophila melanogaster}'s hemibrain as an example of complex system. The open-source image of the hemibrain is from Ref.~\cite{Scheffer2020}. \textbf{b},~A complex network illustration of \textit{Drosophila melanogaster}'s connectome~\cite{Scheffer2020} where only 5\% of the $21 733$ vertices were randomly selected for the sake of visualization. \textbf{c},~The singular value decomposition of a real matrix of rank $r$. The truncated SVD is the optimal low-rank approximation of a matrix, as guaranteed by the Schmidt-Eckart-Young-Mirsky theorem (Theorem~\ref{thm:SEYM}).\\ \textbf{d}, Rapid decrease of the singular values of the matrix describing the \textit{Drosophila melanogaster}'s connectome with the ordinates in logarithmic scale. The vertical dashed lines indicate the rank of the matrix as well as seven measures of effective rank 
    (Methods, Table~\ref{tab:effective_ranks}). \textbf{e}, The average and the percentiles of the singular value distribution of 679 real networks of different origins rescaled by their respective largest singular value (Methods). The shaded background is the region between the 5th and the 95th percentiles. The parameters of the singular-value (hypergeometric) envelope above 95\% of all the singular values are $b \approx 0.54$, $c \approx 2.3$, and $\zeta \approx 25$. \textbf{f}, The stable rank to dimension ratio vs. the rank to dimension ratio for real networks. The theoretical bound, above 96\% of the networks' stable ranks,  is obtained from the singular-value envelope in \textbf{e} and Theorem~\ref{thm:bounds_effective_ranks_hypergeometric_methods}. The approximate proportion of networks is in the parentheses beside the name of each category. Fraction of 679 real networks (502 unweighted networks and 177 weighted networks) vs. \textbf{g-m}, different effective ranks divided by $N$, \textbf{n}, the rank divided by $N$, and \textbf{o}, the number of vertices $N$ with the abscissa shown in log scale. The vertical dashed lines with their corresponding percentage are the averages of the distributions.}
    \label{fig:low_rank_hypothesis}
\end{figure*}

In network science, the topology of the interactions among the constituents of complex systems is typically simplified to a graph, defined by a set of vertices and a set of edges (Figs.~\ref{fig:low_rank_hypothesis}a-b).
Such representation allows the extraction of dominant properties of complex networks, such as their organization into modules~\cite{Fortunato2022}. An ongoing change of paradigm is to use hypergraphs or simplicial complexes rather than graphs to take into account the significant higher-order interactions observed in some real-world systems~\cite{Bianconi2021, Battiston2021}. In addition to finding an appropriate dimension to describe a complex system, one has to uncover the orders of its interactions. As shown later, both problems are intertwined. 

A graph can always be described as a matrix. This simple, yet essential, possibility unlocks several tools from linear algebra to characterize networks. Among them, spectral theory allows identifying the fundamental components of matrices through matrix decomposition. Eigenvalue decomposition has long been used to extract key properties of graphs, such as their invariants~\cite{Wilf1967}, their modular structure~\cite{Donath1973}, the centrality of their vertices \cite{Bonacich1972}, or the bifurcations of dynamical systems taking place on these networks~\cite{Restrepo2005}. 

One pressing challenge in network science is to efficiently adapt the tools of spectral theory to directed, weighted, and signed (e.g., excitatory-inhibitory) networks and hence, to general real matrices. Indeed, eigenvalue decomposition yields complex eigenvalues and complex-valued eigenvectors in general, potentially causing methodological problems (\prettyref{SIsubsec:centrality} and \prettyref{SIsubsec:svd_dynamical_systems}). Worse still, it is not even guaranteed that the matrix representation of the network is diagonalizable. For instance, the trivial directed graph with two vertices connected by one directed edge or any network whose (real) matrix representation, $W$, is rectangular are not diagonalizable (e.g., incidence matrix, interlayer matrix in multilayer networks).

Yet, the matrices $WW^\top$ and $W^\top W$ are always square, symmetric, and thus diagonalizable, which lays the foundations of singular value decomposition (SVD, see Fig.~\ref{fig:low_rank_hypothesis}c and Theorem~\ref{thm:SVD}). 
Interestingly, the decomposition exists for any matrix, the singular vectors are real-valued, and the singular values $\sigma_1,...,\sigma_N$ are nonnegative real numbers. Notably, the number of nonzero singular values equals the rank of $W$. Moreover, SVD inherits various theorems from eigenvalue decomposition~\cite{Horn2013}, such as Weyl's theorem~\cite{Weyl1912, Fan1951}, but it also produces new fundamental results. In particular, SVD is a central tool for dimension reduction in general: the Schmidt-Eckart-Young-Mirsky theorem guarantees that the truncated SVD yields the best low-rank approximation of a matrix (Fig.~\ref{fig:low_rank_hypothesis}c and Theorem~\ref{thm:SEYM}).

The salient properties of SVD and its close relationship with the (effective) rank of a matrix have not yet been completely recognized in network science and spectral graph theory, if we compare to its ubiquity in data science (e.g., matrix completion~\cite{Cai2010}, dynamic mode decomposition~\cite{Kutz2016}, and optimal singular value shrinkage~\cite{Gavish2017}), control theory (e.g., Kalman criterion~\cite{Kalman1960proceeding,*Kalman1960, Yan2017}), random matrix theory (e.g., Mar{\v{c}}enko-Pastur's law~\cite{Marcenko1967}), and linear algebra (e.g., matrix norms~\cite{Horn2013}). SVD is not even mentioned in many of the main introductory textbooks of network science or spectral graph theory (\prettyref{SIsubsec:svd_random_graph}).

Throughout the paper, we leverage the key attributes of SVD to define and evaluate the impact of the low-rank hypothesis of complex systems. Before tackling the case of complex systems as high-dimensional nonlinear dynamical systems, we first expose theoretical evidence of the hypothesis for random graphs followed by an empirical verification of the hypothesis for real networks.

\vspace{\baselineskip} 
\noindent\textbf{Evidence of the hypothesis for network models}\\
It is first instructive to consider random graphs, i.e., sets of graphs equipped with a probability measure that depends on some properties, such as the degrees, the modules, or the distance between vertices in some metric space (\prettyref{SIsubsec:svd_random_graph} and~\prettyref{SIsubsec:exponential_decrease}). Mathematically, they can always be written as random matrices $W = \langle W \rangle + R$, where $\langle W \rangle$ is the expected weight matrix and $R$ is a random matrix with mean 0. 

By examining many widely used random graphs, we observed that their expected matrices involve low-rank matrices. Indeed, we highlight the---usually implicit---assumption that $\langle W \rangle$ is equal to a function $\Phi$ of a low-rank matrix $L$ (Fig.~\ref{fig:exposing_low_rank_hypothesis}a, Table~\ref{tab:random_graphs} in Methods, \prettyref{SIsubsec:svd_random_graph}). In many cases, $\Phi(L) = L$ and it is straightforward to see the low rank of $\langle W \rangle$ since it can be written into its rank-factorized form. A particular Weyl inequality already establishes an expected, but important, outcome of the hypothesis: a small random part $R$ ensures that each singular value of $W$ are close to those of $\langle W \rangle$, i.e.,
\begin{align}\label{eq:weyl}
    \Delta_i = \left|\sigma_i(W) - \sigma_i(\langle W \rangle)\right| \leq \|R\|_2 
\end{align}
for all $i\in\{1,...,N\}$, where $\sigma_i(A)$ denotes the $i$-th singular value of $A$ and $\|\cdot\|_2$ denotes the spectral matrix norm~(\prettyref{thm:weyl_vas} and Corollary~\ref{cor:weyl_vas}). Viewing $W = \langle W \rangle + R$ with $\langle W \rangle = L$ and $\rank(L) = r$ as a spiked random matrix \cite{feral2007largest,capitaine2009largest, benaych2011eigenvalues, Benaych-Georges2012, pizzo2013finite} offers an even more precise perspective. For such matrices, the singular values have a ``bulk'' related to the singular values of $R$ and the creation or annihilation of outlying singular values is asymptotically characterized by the Baik-Ben Arous-Péché (BBP) phase transition \cite{Baik2005}. Notably, the presence of $p \leq r$ singular values outliers in $W$ only depends upon a threshold on the dominant singular values of $\langle W \rangle$, namely $\sigma_1(\langle W \rangle), ..., \sigma_r(\langle W \rangle)$~\cite{Benaych-Georges2012} (\prettyref{SIsubsec:svd_random_graph}). Therefore, a low rank $r$ for $\langle W \rangle$ together with mild threshold conditions imply that the largest singular values of $W$ are located in the vicinity of $\sigma_1(\langle W \rangle), ..., \sigma_r(\langle W \rangle)$, which is a first indicator of the low-rank hypothesis.

However, the low rank of $\langle W \rangle$ is not always obvious, such as in the cases of the directed soft configuration model and its weighted version. Indeed, their expected weight matrices are nonlinear functions of rank-one matrices (Methods). Leveraging Weyl's inequalities, we demonstrated for both models that the singular values of $\langle W \rangle$ are bounded above by an exponentially decreasing term (Theorem~\ref{thm:upper_bound_singvals_scm_simplified} in Methods, Figs.~\ref{fig:exposing_low_rank_hypothesis}e~and~\ref{fig:exposing_low_rank_hypothesis}i). Figs.~\ref{fig:exposing_low_rank_hypothesis}b--\ref{fig:exposing_low_rank_hypothesis}i illustrate how the singular values of $W$ in four different weighted random graphs and two noise regimes inherit the decreasing trend of the dominant singular values of $\langle W \rangle$, while the subdominant ones are related to $R$. 
The rapid decrease of the dominant singular values of $W$ hints at the approximate low rank of a network and thus constitutes a second crucial indicator of the low-rank hypothesis. 

%
%
\begin{figure*}[t]
    \centering
    \includegraphics[width=0.9\textwidth]{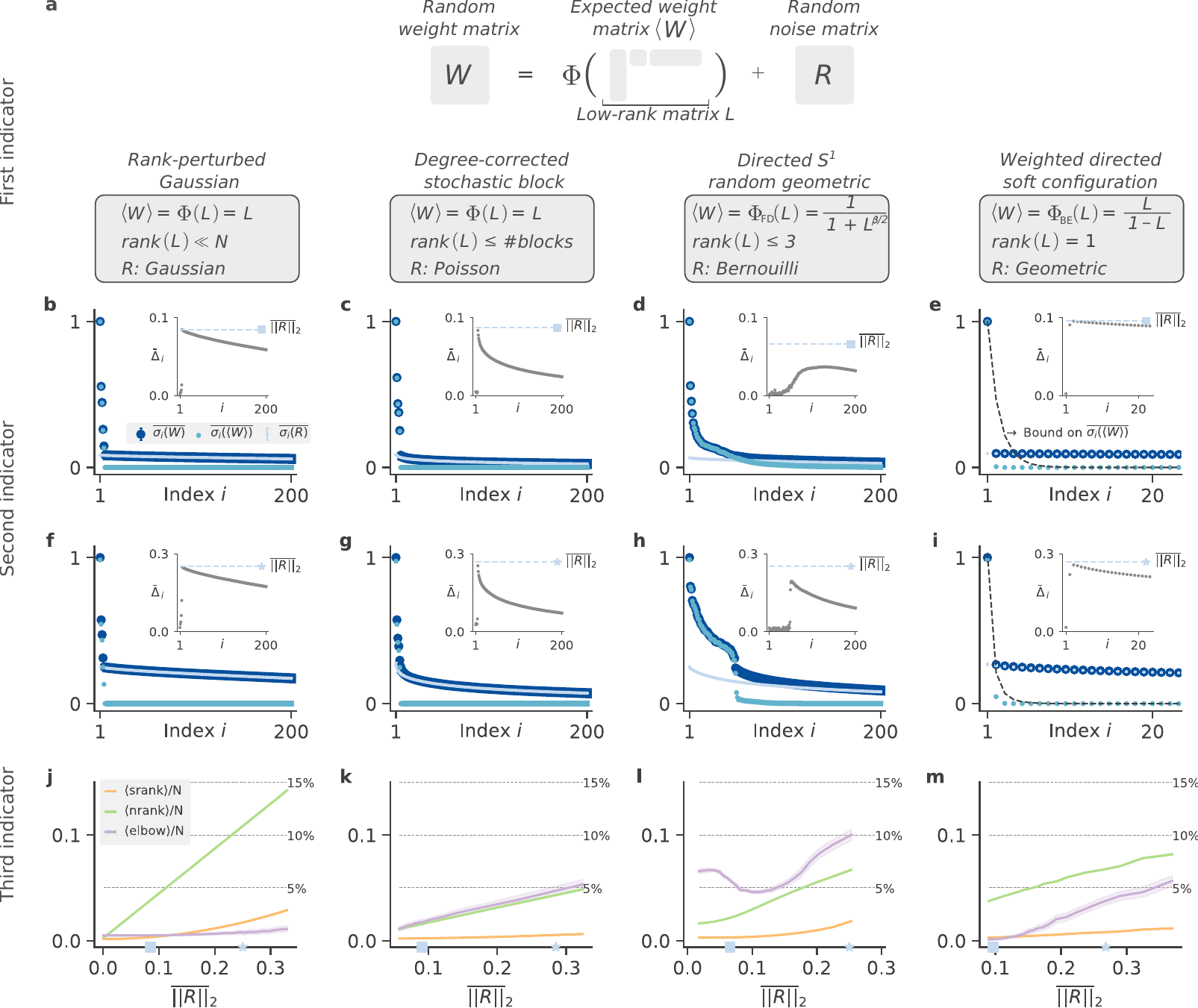}
    \caption{\textbf{Three indicators of the low-rank hypothesis for random graphs}.
    \textbf{a}, Many random graphs have a random matrix representation where the expected weight matrix $\langle W \rangle$ is a matrix-valued function $\Phi$ of a low-rank matrix $L$ plus a centered random part $R$. Four examples of random matrices with different weight distributions and functions $\Phi$ are illustrated and aligned with their subfigures below. The functions $\Phi_{\mathrm{FD}}$ and $\Phi_{\mathrm{BE}}$ respectively stand for a Fermi-Dirac distribution with inverse temperature $\beta$ and a Bose-Einstein distribution where the division is element-wise, e.g., the element $(i,j)$ of $L/(1 - L)$ is $L_{ij}/(1 - L_{ij})$. \textbf{b}-\textbf{i}, The rescaled and averaged singular values of the random weight matrix, its expected part, and its random part for each random graphs are shown in two noise regimes (square markers for $\overline{\|R\|}_2$ near 0.1 [\textbf{b}-\textbf{e}] and star markers for $\overline{\|R\|}_2$ near 0.3 [\textbf{f}-\textbf{i}]). The singular values are respectively denoted $\overline{\sigma_i(W)}$, $\overline{\sigma_i(\langle W \rangle)}$, and $\overline{\sigma_i(R)}$ (from darker to lighter blue markers) where $\overline{x} = \langle x \rangle/\langle\,\|W\|_2\,\rangle$ and $\langle\,\rangle$ denotes the average over the ensemble of graphs. Error bars indicate the standard deviation of the singular values, but are too small to be seen. The random graphs have $N = 10^3$ vertices and only the first 200 (or 20 in \textbf{e} and \textbf{i}) singular values are shown for the sake of visualization. The dashed black lines in \textbf{e} and \textbf{i} are the rescaled upper bounds on the singular values of $\langle W \rangle$ in Theorem~\ref{thm:upper_bound_singvals_scm_simplified} (Methods) with root-mean-square errors over all $i\in\{1,...,N\}$ of 0.02 in \textbf{e} and 0.006 in \textbf{i}. The insets show the rescaled and averaged $\Delta_i$ and its upper bound defined in Eq.~\eqref{eq:weyl}. \textbf{j}-\textbf{m}, The evolution of three effective ranks (averaged over the ensemble of graphs and rescaled by $N$) according to the strength of the noise $\|R\|_2$ is shown. The shaded areas are the standard deviations of the effective ranks. 
    The parameters used for each random graphs can be found in Methods.
    }
    \label{fig:exposing_low_rank_hypothesis}
\end{figure*}

The attributes ``rapid decrease'' and ``approximate low rank'' remain to be quantified, however. To do so, we invoke the notion of \textit{effective ranks}. For instance, the stable rank measures the relative importance of the squared singular values with respect to $\sigma_1^2$ (Methods, Table~\ref{tab:effective_ranks}). In Figs.~\ref{fig:exposing_low_rank_hypothesis}j--\ref{fig:exposing_low_rank_hypothesis}m, we depict its persistence with the increase of the noise level in four random graphs. How ``low'' is an effective rank of a random graph is better understood through its asymptotic behavior as $N\to \infty$ (Methods).
Different singular value decreases lead to different asymptotic behaviors for the effective ranks, from constant $O(1)$ and sub-linear growth $O(N^{1-\epsilon})$ with $\epsilon \in (0, 1]$ to linear growth $O(N)$ (\prettyref{SIsubsec:impact_singvals_effrank}). Notably, sub-linear growth implies that the effective ranks to dimension ratio fall to zero asymptotically as $O(N^{-\epsilon})$~:~we will thus say that an effective rank is low if it grows at most sub-linearly.
For example, we demonstrate that any growing network model with exponentially decreasing singular values (e.g., soft configuration models) imply the lowest asymptotic behavior $O(1)$ for the stable rank and two other effective ranks (Methods, Corollary~\ref{cor:expo_o1_methods}). However, when dealing with a single instance of a random graph or with a real network, $N$ should be kept fixed and the above asymptotic perspective is not applicable. Yet, we can give a more subtle, graded, response to the question ``how low ?'' with effective rank to dimension ratios: values much smaller than 1 indicate that few singular values contribute significantly in the SVD, meaning that $W$ can be well approximated by a low-rank matrix. Having small effective rank to dimension ratios is thus a third indicator, this time quantitative, of the low-rank hypothesis.

Recapitulating, the low-rank hypothesis has been described with three indicators for random graphs. The second one, the rapid decrease of the singular values, is the central indicator of the hypothesis: the first indicator being a theoretical cause for the decrease and the third indicator being a consequence. The second and third indicators are not tied to any theoretical model and can be applied to any type of networked data. We hence adopt the following general, yet workable, definition of the low-rank hypothesis: it is the assumption that the singular values of the network's weight matrix decrease rapidly, implying low effective ranks. We now put this hypothesis to the test. 

\vspace{\baselineskip} 
\noindent\textbf{Verification of the hypothesis for real networks}\\
Despite its frequent use---often implicit, but sometimes very explicit~\cite{Valdano2019, Beiran2021}---the low-rank hypothesis has yet to be verified experimentally for real networks in all their diversity. 

%
%
%
%
%
%
\begin{figure*}[t]
    \centering
    \includegraphics[width=0.7\linewidth]{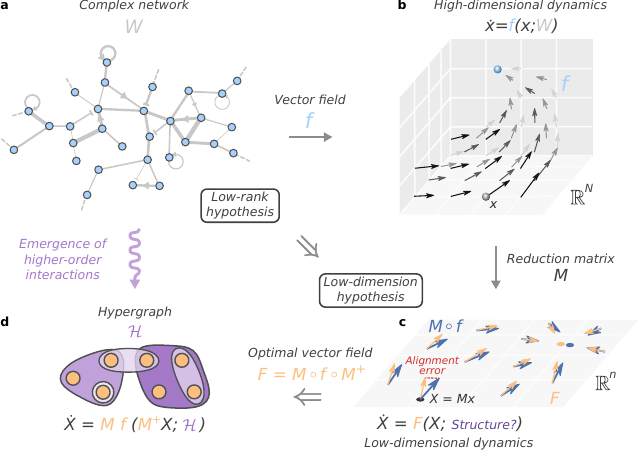}
    \caption{\textbf{The low-rank hypothesis of complex systems and the emergence of higher-order interactions.}\\ \textbf{a}, A complex network represented as a weighted (edges' width), signed, and directed (edges with arrows or a perpendicular line for inhibition) graph with weight matrix $W$. \textbf{b}, A vector field $f$ of a $N$-dimensional dynamical system on a network converging to an equilibrium point. \textbf{c}, Dimension reduction of a dynamical system through the reduction matrix---a linear transformation $M: \mathbb{R}^N \to \mathbb{R}^n$; $x \mapsto X = Mx$. The blue arrows illustrate the exact vector field $M\circ f$ in $\mathbb{R}^n$ (where $\circ$ is the function composition) while the orange arrows represent an approximate vector field $F$. Dimension reduction is about aligning the vector fields, i.e., minimizing alignment errors. \textbf{d}, The least-square optimal vector field $M\circ f \circ M^+$ yields higher-order interactions between the observables $X_1,...,X_n$ represented by some general hypergraph $\mathcal{H}$ with $n$ vertices. The hyperedges are represented by the shaded regions, their weight and their orientation (\prettyref{SIsubsec:emergence}) are not illustrated to avoid cluttering the figure. Note that we make a slight abuse of notation by considering $x$ (resp. $X$) as a function of time and also as a point in $\mathbb{R}^N$ (resp. $\mathbb{R}^n$).} 
    \label{fig:low_dimension_hypothesis}
\end{figure*}

Our experiments revealed that the rapid decay of the singular values in real networks is the norm. As an example, we illustrate the singular value profile of the connectome of \textit{Drosophila melanogaster} in Fig.~\ref{fig:low_rank_hypothesis}d. Figure~\ref{fig:low_rank_hypothesis}e presents a coalesced view of the singular value profiles for 679 real networks from 10 different origins. 
As a guide to appreciate the decreases, we trace a general singular-value envelope below which 95\% of the singular values of all the networks belong. 

Having an explicit form for the singular-value envelope allows interpreting the stable rank as the area under a curve (\prettyref{SIsubsec:impact_singvals_effrank}) and then to find a theoretical bound below which most of the networks' stable ranks lie (Methods, Theorem~\ref{thm:bounds_effective_ranks_hypergeometric_methods}). In Fig.~\ref{fig:low_rank_hypothesis}f, we illustrate the stable rank of the real networks along with the theoretical bound above 96\% of the networks, which indicates that the stable rank is generally expected to be less than 10\% of the number of vertices $N$. 

To ensure that this observation is not limited to the stable rank, we report in Figs.~\ref{fig:low_rank_hypothesis}g--\ref{fig:low_rank_hypothesis}m similar observations for other effective ranks (Methods). Having larger values than srank is not surprising for nrank and erank. In fact, it is easily shown that $\mathrm{srank} \leq \mathrm{nrank} \leq \mathrm{erank} \leq \mathrm{rank}$ (Methods). Contrarily to the effective ranks, the rank of real networks is often comparable to their dimension (Fig.~\ref{fig:low_rank_hypothesis}n). This observation is expected, especially for weighted networks with real weights, since non-invertible matrices form a set of measure 0.

The datasets considered consist in real networks with fixed $N$, but the asymptotic behaviors of their effective ranks can still be evaluated as if there was a related growing graph whose singular values remain within experimental singular-value envelopes as $N$ grows. Using this approach, we prove that singular-value envelopes such as the one in Fig.~\ref{fig:low_rank_hypothesis}e admits constant and sublinear growth for srank, nrank, and erank (Methods). 

All in all, we show that many real networks have rapidly decreasing singular values, leading to low effective ranks. Interestingly, such observation seems to be widespread for big data matrices~\cite{Gao2015a, Beckermann2017, Udell2019}, but it remains a puzzling phenomenon. In particular, the consequences of these observations for high-dimensional nonlinear dynamics on networks are still to be untangled, which is the subject of the next section.

\begin{figure*}[t]
    \centering
    \includegraphics[width=1\linewidth]{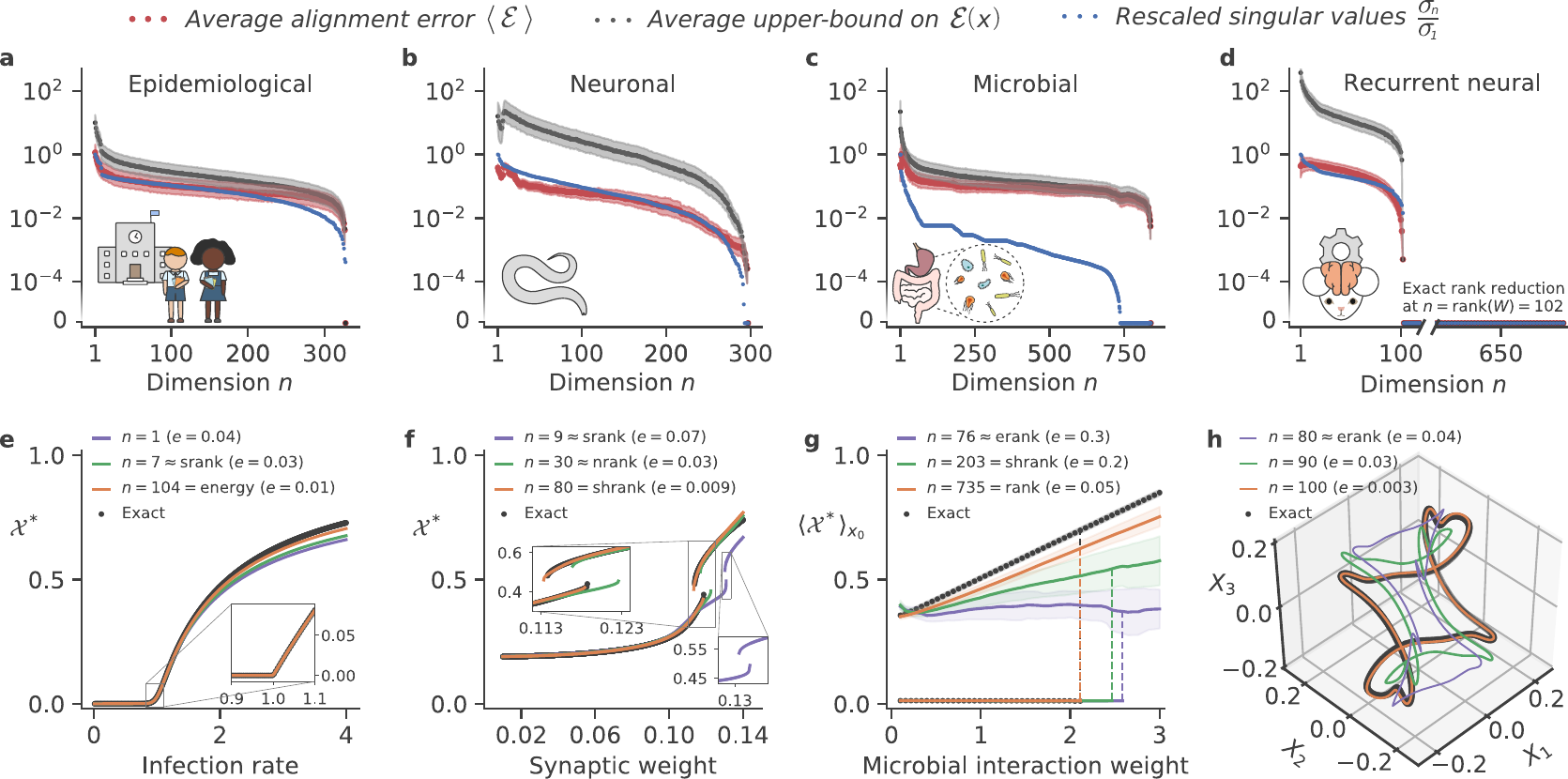}
    \caption{\textbf{Dimension reduction errors for nonlinear dynamics on real complex networks in relation with their singular values and effective ranks.} \textbf{a}-\textbf{d}, The decrease of the alignment error $\mathcal{E}(x)$ (red markers) is in accordance with the rapid decrease of singular values (blue markers) as expected by the analytical upper bound in Eq.~\eqref{eq:upper_bound_rapid_decrease_letter} (solid black line). The shaded regions in gray and light red represent the standard deviation of the upper bound and the error respectively. We have $10^3$ different samples for $x$ and the parameters for each $n$ and the upper bounds are computed exactly in \textbf{a}~and~\textbf{c}, while they are approximated in \textbf{b}~and~\textbf{d} (details in \prettyref{SIsubsec:evaluate_bound}).
    \textbf{e}-\textbf{h}, Comparison of the bifurcation diagrams (resp. trajectories in \textbf{h}) for the global observable, denoted $\mathcal{X}^* = w\cdot X^*$ at equilibrium where $w$ is a $n\times 1$ real vector specific to the dynamics, of the complete dynamics (black markers) vs. the reduced dynamics (solid colored lines) at different dimensions $n$ with root-mean-square errors $e$ shown in parentheses (Methods).
    \textbf{a} and \textbf{e}, Epidemiological dynamics (quenched mean-field SIS) on a high-school contact network ($N = 327$, undirected, binary) rescaled by the largest singular value.
    \textbf{b} and \textbf{f}, Neuronal (Wilson-Cowan) dynamics on the \textit{C. elegans} connectome ($N = 297$, signed, weighted, directed).
    \textbf{c} and \textbf{g}, Microbial population dynamics on a human gut microbiome network ($N = 838$, signed, weighted, directed). Note that there are multiple stable upper branches depending on the initial condition $x_0$ (Methods). Here we show an average on $x_0$ of the upper branches (black markers and solid colored lines) with the standard deviation (shaded regions) and we show one lower branch. The loss of stability of the lower branch is indicated by a dashed vertical line that connects it, for visualization purpose, to the average of the upper branches.
    \textbf{d} and \textbf{h}, Recurrent neural network (RNN) dynamics on a learned network ($N = 669$, signed, weighted, directed) for which we have shrunk its singular values using optimal shrinkage with the Frobenius norm~\cite{Gavish2017} to emphasize the fact that dimension reduction for the RNN dynamics is exact when $n$ is the rank of the network~(Methods).}
    \label{fig:error_vector_fields}
\end{figure*}

\vspace{\baselineskip}
\noindent\textbf{Induced low-dimension hypothesis}\\
\noindent 
Intuitively, we expect that having low (effective) rank networks gives grounds to dimension reduction of dynamics on these networks. Consider the complete dynamics $\dot{x} = f(x\,;\,W)$, where $x(t)\in\mathbb{R}^N$ is the system's state at time $t$, $f:~\mathbb{R}^N~\to~\mathbb{R}^N$ is a continuously differentiable vector field, and $W$ is a $N\times N$ weight matrix describing the network (Figs.~\ref{fig:low_dimension_hypothesis}a-~\ref{fig:low_dimension_hypothesis}b). More specifically, given $g:\mathbb{R}^N\times \mathbb{R}^N \to \mathbb{R}^N$ and $W$ ($x(t)$ is unknown), we examine the subclass of dynamics $\dot{x} = g(x, y)$ where $y=Wx$. 

Considering this subclass of dynamics already highlight an important implication of the low-rank hypothesis. The linear function $x \mapsto y = Wx$ in $g$ has a very special role: even if $x$ is part of a $N$-dimensional manifold, when $W$ has a low rank, the vector in the image of $W$ will be part of a low-dimension submanifold. Even if $W$ has full rank, our experimental observations in Fig.~\ref{fig:low_rank_hypothesis} show that it is likely to have a low effective rank. We can hence say that $Wx$ will be part of an effectively low-dimension submanifold.

Just as some random graph models are crafted from a nonlinear function $\Phi$ of a low-rank matrix $L$ (see Fig.~\ref{fig:exposing_low_rank_hypothesis}a), the vector field $g$ depends nonlinearly on $Wx$, making it challenging to assess the low dimensionality of $g(x,y)$. Despite recent developments~\cite{Gao2016, *Tu2017, *Jiang2018, *Laurence2019, *Vegue2023, *Kundu2022, Thibeault2020}, it remains unclear how to choose a dimension for the reduced dynamics and how to quantify the reduction error for nonlinear dynamics on complex networks.

Dimension reduction of dynamical systems can be imagined as the problem of aligning a low-dimensional vector field with its high-dimensional counterpart (Fig.~\ref{fig:low_dimension_hypothesis}c and \prettyref{SIsubsec:generalities}). 
This involves selecting a $n\times N$ reduction matrix $M$ that maps the elements of the complete system to the reduced system, as well as a vector field $F$ describing the evolution of a set of observables $(X_{\mu})_{\mu=1}^n$ in $\mathbb{R}^n$. The alignment error in $\mathbb{R}^n$ at $x\in\mathbb{R}^N$, denoted $\mathcal{E}(x)$, can then be defined as the error between the vector fields $M\circ f$ and $F\circ M$ (Methods).



Minimizing the alignment error to find the optimal pair $(M, F)$ is challenging in general (\prettyref{SIsubsec:generalities}) and the best choice hinges on the modeler's objective. For instance, selecting $M$ to ensure that the temporal evolution of $X$ remains interpretable throughout time (e.g., synchronization observables~\cite{Thibeault2020}), might further complicate the optimization problem. 

Let us concentrate on identifying $F$ without taking into account $M$ for now.
Using least squares, we proved that $M\circ f \circ M^+$ minimizes an alignment error in $\mathbb{R}^N$, where ${}^+$ denotes pseudoinversion (Methods). Doing so allowed us to show, for $\dot{x} = g(x, y)$, that the alignment error $\mathcal{E}(x)$ caused by the least-square vector field satisfies
 \begin{multline}\label{eq:upper_bound_unsimplified}
     \sqrt{n}\,\mathcal{E}(x) \leq  \|MJ_x'(I-M^+M)x\| \\
     + \|W(I-M^+M)\|_2\|MJ_y'\|_2\|x\|\,, 
 \end{multline}
where $J_x'$ and $J_y'$ are Jacobian matrices (Methods). 

Interestingly, the previous inequality suggests a non-arbitrary way of selecting the reduction matrix. Indeed, \begin{equation}\label{eq:choiceM}
    M=V_n^\top
\end{equation}
minimizes the factor $\|W(I-M^+M)\|_2$ related to the interactions in the system, generally making each observable $X_\mu$ global, i.e., containing information on most vertices (Methods).

The choice made in Eq.~\eqref{eq:choiceM} prompted us to derive another inequality revealing the contribution of the network singular values to the alignment error (Methods, Theorem~\ref{thm:upper_bound_rapid_decrease_simplified}):
\begin{align}\label{eq:upper_bound_rapid_decrease_letter}
    \sqrt{n}\,\mathcal{E}(x)& \leq\|V_n^{\top}J_x'(I-P)x\| +  {\sigma_{n+1}} \|V_n^{\top}J_y'\|_2\|x\|,
\end{align}
where $P=V_nV_n^\top$. Notably, the inequality provides a criterion for exact dimension reduction: if $J_x' = dI$ for $d\in \mathbb{R}$ and $n = \rank(W)$, the upper bound vanishes to zero and the dimension reduction is exact (Methods). Consequently, a general class of dynamics, including recurrent neural networks and the Wilson-Cowan neuronal dynamics, can be exactly reduced (Methods). 
The upper bound~\eqref{eq:upper_bound_rapid_decrease_letter} is meant to be intuitive (not necessarily tight): it connects the swift decay of singular values of a network with the dimension reduction error.
As a basic example, the relative alignment error $\mathcal{E}(x)/\|x\|$ for the linear system $\dot{x} = Wx$ is simply upper-bounded by $\sigma_{n+1}/\sqrt{n}$, meaning that a rapid decrease of the singular values of $W$, be it related to an arbitrarily weighted network, directly induces a rapid decrease of the alignment error.

Figure~\ref{fig:error_vector_fields}a-d illustrates the decrease of the alignment error with $n$ ---the latter being in accordance with the rapid decay of the upper bound and of the singular values---in four dynamics on real networks. We show how $n$ can be tuned to predict an epidemic in an epidemiological dynamics (Fig.~\ref{fig:error_vector_fields}e), a hysteresis in a neuronal dynamics (Fig.~\ref{fig:error_vector_fields}f), stable branches in a microbial dynamics (Fig.~\ref{fig:error_vector_fields}g), or a limit cycle in a recurrent neural network (Fig.~\ref{fig:error_vector_fields}h). While effective ranks can help select a suitable dimension $n$ to describe a collective phenomenon, we use them only as an indication: $n$ should be chosen according to the modeler's tolerance to qualitative (e.g., is the hysteresis preserved?) or quantitative (e.g., is the predicted transition accurate?) errors. It thus becomes clear that having low (effective) rank matrices describing complex networks gives ground to dimension reduction of nonlinear dynamics on these networks.

The reduced system is akin to a low-dimensional dynamics taking place on a smaller structure, whose nature remains to be specified (Fig.~\ref{fig:low_dimension_hypothesis}c).
We show in the next section that dimension reduction ultimately leads to the emergence of higher-order interactions, as illustrated in Fig.~\ref{fig:low_dimension_hypothesis}d.

\vspace{\baselineskip}
\noindent\textbf{Emergence of higher-order interactions}\hfill \break
Theoretical and experimental evidence for the existence of higher-order interactions in various complex systems has been reported and its consequences---e.g., on explosive transitions~\cite{Kuehn2021} or mesoscopic localization~\cite{St-Onge2021_PRLmeso}---have been extensively studied~\cite{Battiston2020a}. However, their origin remains under active investigation, notably for oscillatory systems~\cite{Matheny2019, Nijholt2022} (\prettyref{SIsubsec:emergence}).

Using our framework, a simple example readily provides insights over the emergence of higher-order interactions. Consider the epidemiological dynamics $\dot{x}_i = -d_ix_i + \gamma\,(1 - x_i)\, y_i$ with $i\in\{1,...,N\}$, where $x_i$ is the probability for the vertex $i$ to be infected, $y = Wx$ while $d_i$ and $\gamma$ denote the recovery rate of vertex $i$ and the infection rate respectively. The reduced dynamics is then given by
\begin{align}\label{eq:reduced_qmf_sis}
    \dot{X}_{\mu} = &\sum_{\nu=1}^n \left(\mathcal{D}_{\mu\nu} + \mathcal{W}_{\mu\nu}\right)X_{\nu}  \\- \gamma &\sum_{i=1}^N M_{\mu i} \left(\sum_{\nu=1}^n M_{i\nu}^+X_{\nu}\right)\left( \sum_{j=1}^N\sum_{\kappa = 1}^n W_{ij} M_{j\kappa}^+X_{\kappa}\right)\nonumber
\end{align}
for all $\mu\in\{1,...,n\}$, where $\mathcal{D} = -MD M^{+}$ is a reduced $n\times n$ recovery rate matrix with $D = \diag(d_1,...,d_N)$, and $\mathcal{W} = \gamma MWM^+$ is a reduced $n\times n$ weight matrix.

Let us inspect the last term in Eq.~\eqref{eq:reduced_qmf_sis} more carefully. For simplicity, consider that $M^+ = M^\top$, i.e., $M$ has orthogonal rows. Then, $M_{\mu i}$ quantifies the influence of vertex $i$ on the $\mu$-th observable, $M_{i\nu}^\top X_{\nu}$ is the influence of the $\nu$-th observable weighted by its dependence over vertex $i$, and $W_{ij}M_{j\kappa}^\top X_{\kappa}$ is the influence of the $\kappa$-th observable weighted by its dependence over vertex $j$ that connects to vertex $i$. Altogether, these factors form a third-order interaction between the observables $X_{\mu}$, $X_{\nu}$, and $X_{\kappa}$, an observation that becomes more explicit by rearranging Eq.~\eqref{eq:reduced_qmf_sis} as
\begin{equation}
    \dot{X}_{\mu} = \sum_{\nu=1}^n \left(\mathcal{D}_{\mu\nu} + \mathcal{W}_{\mu\nu}\right)X_{\nu} +\sum_{\nu, \kappa = 1}^n \mathcal{T}_{\mu\nu\kappa}X_{\nu}X_{\kappa}\,,
\end{equation}
where the third-order interactions are encoded in a third-order tensor $\mathcal{T}$ with elements
\begin{align}\label{eq:third_order_interactions}
    \mathcal{T}_{\mu \nu \kappa} = - \gamma\sum_{i,j=1}^NM_{\mu i}M_{i\nu}^+ W_{ij}  M_{j\kappa}^+
\end{align}
for all $\mu,\nu,\kappa \in\{1,...,n\}$. Hence, the resulting structure of the reduced system is a hypergraph $\mathcal{H}$ with $n$ vertices (Fig.~\ref{fig:low_dimension_hypothesis}c-d; see \prettyref{SIsubsec:emergence}), which is generally directed~\cite{Gallo1993}, weighted, signed, and formed from $\mathcal{D}$, $\mathcal{W}$, and $\mathcal{T}$.

Beyond the influence of dynamical parameters like the weight matrix $W$, Eq.~\eqref{eq:third_order_interactions} highlights the crucial role of the reduction matrix $M$ in shaping higher-order interactions. Indeed, $M$ partially determines the directed, weighted, and signed nature of the hypergraph. Moreover, if the observables respectively depend on disjoint groups of vertices, i.e., $M_{\mu i} \propto \delta_{\mu\,s(i)}$, where $\delta$ is the Kronecker delta and $s$ maps each vertex $i$ to its group, then the tensor with elements in Eq.~\eqref{eq:third_order_interactions} can be exactly mapped to a matrix. In other words, in the epidemiological dynamics, the higher-order interactions emerge from observables depending on overlapping groups of vertices (e.g., $M = V_n^\top$ in general). Interestingly, such overlapping is a very common characteristic of complex networks such as social networks~\cite{Palla2005}.

These observations encouraged us to seek generic conditions for such emergence. For $\dot{x}_i = h_i(x_i, y_i)$, where $h_i~:~\mathbb{R}^2~\to~\mathbb{R}$ is an analytical scalar field for all $i\in\{1,...,N\}$, we proved that the least-square optimal vector field depends upon higher-order interactions between the observables $X_1,...,X_n$ (Methods, Proposition~\ref{prop:emergence_simplified}). We then deduced two insightful consequences. First, if the scalar field is a polynomial of total degree $\delta$ in $x_i$ and $y_i$ for all $i$, the hypergraph of the reduced system has interactions of maximal order $\delta + 1$ (Methods, Corollary~\ref{cor:emergence_polynomial}). Second, having observables depending on disjoint groups of vertices is not sufficient to avoid higher-order interactions in general: the nonlinearity in $y_i$ also plays its part (Methods, Corollary~\ref{cor:sufficient_condition_emergence}). Other worked-out examples for a microbial and an oscillator dynamics are given in Extended Data Table~1 to complement the previous observations on the epidemiological dynamics.  

All in all, our results suggest that many instances of higher-order interactions could be a byproduct of the low-dimensional (macroscopic) representation chosen to model a wide variety of complex systems. They clarify the essential role of the description dimension and of the nonlinearity of the original system in shaping the interactions of the ensuing reduced system.

\vspace{\baselineskip}
\noindent\textbf{Conclusions and outlook}\hfill\break
In this paper, we established the ubiquity of the low-rank hypothesis in complex systems and its consequences, from the dimension reduction of high-dimensional nonlinear dynamics on networks to the emergence of higher-order interactions. 

Our experimental results suggest that the low-rank hypothesis is perhaps not only a hypothesis, but something intrinsic to many real complex systems. Our findings hint at the possibility that some emergent collective phenomena are consequences of much fewer variables than what would be expected a priori, thanks to the low-rank nature of their complex network. However, the low-rank hypothesis should be used very carefully: the effective ranks of real networks are often at a non-negligible fraction of $N$ and adopting the low-rank hypothesis unknowingly can lead to an oversimplified model of a given complex system. It thus seems relevant to design new random graphs based on the observed singular values of real networks. Networks' singular values are not a mere abstraction from spectral theory: like the degree, the clustering or the reciprocity, they have an intuitive interpretation as indicators of the effective dimension of complex networks/systems.

Our theoretical framework also suggests that inferring the connections in complex systems from time series observed at a relatively coarse-grained resolution (e.g., local field potentials in the brain~\cite{Yu2011} or abundances in plant communities~\cite{Mayfield2017}) is likely to reveal significant higher-order interactions. We conjecture that monitoring complex systems at different scales experimentally will clarify the role of the dimension at which the measurements are done on the emergence of higher-order interactions.
Dimension reduction of dynamics on higher-order networks~\cite{FerrazdeArruda2021, Bianconi2021} is also to be pursued, perhaps through Tucker decomposition~\cite{Qi2017}. 

Nevertheless, determining the precise form of the dominant observables that drive the behavior of complex systems remains an open problem. While we focused on linear observables, there might exist a small set of nonlinear observables well suited for a given high-dimensional dynamics~\cite{Watanabe1994}. However, finding appropriate, intuitive, nonlinear observables is much harder~\cite{Brunton2022}. Our observations on the effective ranks of real networks also motivate further research on the inference of interpretable low-rank models from time series~\cite{Valente2022a}.

Finally, one defining property of complex systems that we have not addressed is their capacity for adaptation~\cite{Holland1995}. Our preliminary results suggest that the low effective rank of complex networks plays a central role for controlling~\cite{Montanari2022, Sanhedrai2022} and assessing the resilience of complex adaptive systems~\cite{Desrosiers2022}. This, alongside indications that maturation or learning could reduce network's effective ranks (\prettyref{SIsubsec:adaptive_networks} and Ref.~\cite{Martin2021a}), will be the topic of an upcoming publication.


\vspace{0.3cm}
 
\noindent\textbf{Acknowledgments}.
We are grateful to Gabriel Eilerstein for sharing the code to extract the weight matrices from the repository NWS, Gáspár Jékely for sharing the neuronal and desmosomal connectomes of \textit{Platynereis dumerilii}, Charles Murphy for useful discussions on artificial neural networks, Guillaume St-Onge for his comments on the preprint, and Xavier Roy-Pomerleau for helping to explore the microbial dynamics numerically. We thank Émile Boran for his fundamental contribution to linear algebra. This work was supported by the Fonds de recherche du Qu\'ebec -- Nature et technologies (V.T., P.D.), the Natural Sciences and Engineering Research Council of Canada (V.T., A.A., P.D.), and the Sentinelle Nord program of Universit\'e Laval, funded by the Canada First Research Excellence Fund (V.T., A.A., P.D.).

\vspace{0.3cm}

\noindent\textbf{Author contributions}.
All authors contributed to the formulation of the study, the interpretation of the results, and the edition of the paper. V.T. and P.D. obtained the mathematical results and conceived the conceptual basis of the project. V.T. led the writing of the manuscript, wrote the supplementary information with P.D., designed the figures, wrote the code, and performed the numerical experiments to generate the results. V.T., A.A., and P.D. contributed to the code and analyzed the data to generate Fig.~\ref{fig:low_rank_hypothesis}.

\vspace{0.3cm}

\noindent\textbf{Competing interests}. The authors declare no competing interests.

\clearpage

\noindent\textbf{Methods}
\vspace{0.2cm}

\footnotesize

\noindent\textbf{Random graphs.} 
A random graph can be described by a random matrix
\begin{equation}\label{eq:random_matrix_methods}
    W = \langle W \rangle + R\,,
\end{equation}
where $\langle W \rangle$ is the expected weight matrix and $R$ is a zero-mean random matrix. Even if one instance in a typical model is generally of full rank $N$, the expected weight matrix $\langle W \rangle$ is often defined as an element-wise function of a low-rank matrix $L$, i.e., 
\begin{equation}
    \langle W \rangle = \big(\phi(L_{ij})\big)_{i,j=1}^N\,,
\end{equation}
where $\phi$ is a real-valued function of a real variable. This is an alternative, but equivalent, way to write $\langle W \rangle = \Phi(L)$ as in the main text. In Table~\ref{tab:random_graphs}, we list some classical examples of random graphs and the corresponding low-rank matrices.

\begin{table}[ht]
\caption{\label{tab:random_graphs}
\footnotesize Low-rank matrix $L$ characterizing the expected adjacency matrix for different random graphs of $N$ vertices. 
SBM:~Stochastic Block Model, CL:~Chung-Lu, MD:~Metadegree, DSCM:~Directed Soft Configuration Model, RDPG:~Random Dot Product Graph, RGM:~Random Geometric Model, RPG:~Rank-Perturbed Gaussian, DCSBM:~Degree-Corrected Stochastic Block Model, ``W'' in front of an acronym stands for ``weighted''. For the $S^D$ RGM, the rank of $L$ is, more precisely, $D$, $D+1$, or $D+2$ which is a consequence of Ref.~\cite[Theorem 7]{Gower1985} and the inequality $\rank(A\circ B) \leq \rank(A)\rank(B)$. The parameters $q$, $r$, $d$ and $D$ are usually assumed to be small compared to $N$. More details about these random graphs and others are given in \prettyref{SIsubsec:svd_random_graph}.} 
\begin{ruledtabular}
\begin{tabular}{l c c c c c}
&\footnotesize\textbf{Model} & \footnotesize\textbf{Low-rank matrix $L$} & \footnotesize{rank($L$)}& \footnotesize{$\phi(L_{ij})$} \\
\colrule
&\footnotesize $\mathcal{G}(N, p)$ & \parbox{3.6cm}{$Np\,\hat{\bm{1}}\hat{\bm{1}}^\top$}& \parbox{1cm}{1} & \parbox{2cm}{$L_{ij}$}\\
\parbox[t]{2mm}{\multirow{3}{*}{\rotatebox[origin=c]{90}{Unweighted}}}&\footnotesize CL &  \parbox{4cm}{$\frac{\|\kappa\|^2}{2M}\bm{\hat{\kappa}}\bm{\hat{\kappa}}^\top$} & \parbox{1cm}{1} & \parbox{2cm}{$L_{ij}$}\\
&\footnotesize DSCM &  \parbox{4cm}{$\|\bm\alpha\|\|\bm\beta\|\,\hat{\bm{\alpha}}\,\hat{\bm{\beta}}^\top$} & \parbox{1cm}{1} & \parbox{2cm}{$\frac{L_{ij}}{1 + L_{ij}}$}\\
&\footnotesize MD &  \parbox{4cm}{$\sum_{\mu,\nu=1}^r\Delta_{\mu\nu}\,\bm{v}_{\mu}\bm{v}_{\nu}^\top$} & \parbox{1cm}{$r$}
& \parbox{2cm}{$L_{ij}$}\\ 
&\footnotesize SBM & \parbox{4cm}{$\sum_{\mu,\nu=1}^q\sqrt{n_{\mu}n_{\nu}}\,p_{\mu\nu}\,\bm{b}_{\mu}\bm{b}_{\nu}^\top$} &  \parbox{1cm}{$\leq q$} &  \parbox{2cm}{$L_{ij}$}\\
&\footnotesize $S^D$ RGM &  \parbox{4cm}{$\frac{R^2}{\mu^2}\left(\bar{\bm{\kappa}}_{\mathrm{in}} \,\bar{\bm{\kappa}}_{\mathrm{out}}^\top\right)\circ \bar\theta\,\,$} & \parbox{1cm}{\hspace{-0.5cm}$\leq D+2$}  & \parbox{2cm}{$\frac{1}{1 + L_{ij}^{\beta/2}}$} \vspace*{0.05cm}\\
\hline \vspace*{0.05cm}
&\footnotesize $\mathcal{G}(N, p, w)$ & \parbox{4cm}{$Npw\,\hat{\bm{1}}\hat{\bm{1}}^\top$} &  \parbox{1cm}{1} &  \parbox{2cm}{$L_{ij}$}\\
&\footnotesize WCL & \parbox{4cm}{$\bm{y}\bm{y}^\top$} &  \parbox{1cm}{1} &  \parbox{2cm}{$L_{ij}$}\\
\parbox[t]{2mm}{\multirow{3}{*}{\rotatebox[origin=c]{90}{Weighted}}}&\footnotesize WDSCM & \parbox{4cm}{$\bm{y}\bar{\bm{y}}^\top$} &  \parbox{1cm}{1} &  \parbox{2cm}{$\frac{L_{ij}}{1 - L_{ij}}$}\\
&\footnotesize RPG &  \parbox{3.6cm}{\phantom{$I^{I^I}$}$\sum_{\mu =1}^r \bm{m}_\mu \bm{n}_\mu^\top$\phantom{$I^{I^I}$}} & \parbox{1cm}{$r$} & \parbox{2cm}{$L_{ij}$}\\
&\footnotesize WSBM & \parbox{4cm}{$\sum_{\kappa,\nu=1}^q\sqrt{n_{\kappa}n_{\nu}}\,\mu_{\kappa\nu}\,\bm{b}_{\kappa}\bm{b}_{\nu}^\top$} &  \parbox{1cm}{$\leq q$} &  \parbox{2cm}{$L_{ij}$}\\
&\footnotesize DCSBM & \parbox{4cm}{$\Lambda \circ (\hat{\bm{\kappa}}_{\mathrm{in}}\hat{\bm{\kappa}}_{\mathrm{out}}^\top)$} &  \parbox{1cm}{$\leq q$} &  \parbox{2cm}{$L_{ij}$}\\
&\footnotesize RDPG &  \parbox{3.6cm}{$\sum_{\mu=1}^d \bm{X}_{\mu}\bm{X}_{\mu}^\top$} & \parbox{1cm}{$\leq d$} & \parbox{2cm}{$L_{ij}$}
\end{tabular}
\end{ruledtabular}
\end{table}
In \prettyref{SIsubsec:svd_random_graph}, we also report random network models involving two low-rank matrices, such as the general weighted soft configuration model, the general weighted directed soft configuration model, and the $S^1$ weighted random geometric model, along with other examples (and counter-examples) from network science (e.g., Watts-Strogatz model), random matrix theory, spin glasses, machine learning, and neuroscience. Based on these observations and those of Ref.~\cite{Valdano2019}, one can create many new random graphs with matrices of different ranks.

It is straightforward to assess the low rank of $L$, but it is harder to assess the low rank of $\langle W \rangle$ when $\phi$ is nonlinear. For example, in the directed soft configuration model (DSCM), $\phi = \phi_{\mathrm{FD}}$, a Fermi-Dirac distribution and in its weighted version (WDSCM), $\phi = \phi_{\mathrm{BE}}$, a Bose-Einstein distribution. For both models, the following theorem demonstrates that the singular values of their expected weight matrix are bounded above by an exponentially decreasing term.
\begin{theorem}[Simplified version of Theorems~\ref{thm:upper_bound_singvals_scm} and \ref{thm:upper_bound_singvals_wdscm}]\label{thm:upper_bound_singvals_scm_simplified}
\phantom{henri}
\newline
Let $\sigma_1\geq ...\geq \sigma_N$ be the singular values of $\langle W\rangle$. If $\langle W_{ij}\rangle = \phi_{\mathrm{FD}}(L_{ij})<1/2$ or $\langle W_{ij}\rangle = \phi_{\mathrm{BE}}(L_{ij})$ for all $i,j\in\{1,...,N\}$, where $L$ is a rank-one matrix, then
\begin{equation}\label{eq:upper_bound_infinite_sum_dscm_methods}
    \sigma_i \leq \sum_{k=i}^\infty \ell_k  \leq \frac{N\gamma^i}{1-\gamma}\qquad \forall\,\,i \in\{1,...,N\},
\end{equation}
 where $\ell_k = \sqrt{\sum_{i,j=1}^N  L_{ij}^{2k}}$ and $\gamma = \max_{i,j}L_{ij}$.
\end{theorem}
The proof is based on Weyl's inequalities~(Theorem~\ref{thm:weyl_vas} in \prettyref{SIsubsec:weyl_facto}) and the truncated geometric series. The bound for $\langle W_{ij}\rangle = \phi_{\mathrm{FD}}(L_{ij})>1/2$ is also given in Theorem~~\ref{thm:upper_bound_singvals_scm}. The upper bounds in Theorem~\ref{thm:upper_bound_singvals_scm_simplified} expose the low-rank formulation of soft configuration models and paves the way for new bounds on the singular values of other random graphs, such as random geometric models.

In Fig.~\ref{fig:exposing_low_rank_hypothesis}, the singular values of $W$, $\langle W\rangle$, and $R$ are shown for the RPG, DCSBM, $S^1$ RGM, and WDSCM. The upper bounds shown in Fig.~\ref{fig:exposing_low_rank_hypothesis}e and i are given by Eq.~\eqref{eq:upper_bound_infinite_sum_dscm_methods} which is computed by summing the constants $n_i > n_{i+1} > ...$ until $n_k$ is smaller than $10^{-12}$. For RPG, the vectors $\bm{m}_{\mu}$ and $\bm{n}_{\mu}$ are instances of different Gaussian distributions and $r = 5$. Instances of truncated Pareto distributions were used to generate the expected degrees (DCSBM and $S^1$ RGM) and $\bm{y}$, $\bar{\bm{y}}$ (WDSCM). The number of blocks $q$ is set to 5 for the DCSBM and the expected number of edges block matrix $\Lambda$ is defined such that there are more edges expected within the blocks than between them. To obtain the norm of the random part $R$ of the random weight matrices (except RPG, where $R$ is already set to be a Gaussian of mean 0), we have generated 100 instances of $W$, we have computed $R = W - \langle W \rangle$, and then its norm for each instance. The spectral norm of $R$ is increased by changing the variance of each Gaussian element in $R$ for RPG, the expected number of edges in DCSBM, the temperature $1/\beta$ in $S^1$ RGM, and the minimum value of $\bm{y}$ and $\bar{\bm{y}}$ in WDSCM. The detailed parameters are given in \prettyref{SIsubsec:svd_random_graph}.

\begin{table}[ht]
\caption{\label{tab:effective_ranks}
\footnotesize Different effective ranks of a matrix of dimension $N\times N$ and of rank $r$ expressed in terms of its singular values $\sigma_1\geq ...\geq \sigma_N$. For \textit{energy}, the constant $\tau$ is a threshold to be set between 0 and 1. For \textit{thrank}, $\sigma_{\mathrm{med}}$ is the median singular value and  $\mu_{\mathrm{med}}$ is the median of a Mar{\v{c}}enko-Pastur probability density function~\cite{Gavish2014}. For \textit{shrank}, $s^*$ denotes an optimal singular value shrinkage function~\cite{Gavish2017, Donoho2018}. The complete names and the details about each of the effective ranks are given in \prettyref{SIsubsec:effective_ranks}.}
\begin{ruledtabular}
\begin{tabular}{l c c}
\parbox{1.8cm}{\footnotesize\textbf{Abbreviation}} & \footnotesize\textbf{Expression} \\
\colrule
\footnotesize srank & \parbox{6.5cm}{$\phantom{\min\left[\left(\frac{\sum_{i=1}^\ell}{\sum_{j=1}^r} \right)\right]}\sum_{i=1}^r\sigma_i^2/\sigma_1^2\phantom{\min\left[\left(\frac{\sum_{i=1}^\ell}{\sum_{j=1}^r} \right)\right]}$} \\
\footnotesize nrank & \parbox{6.5cm}{$\sum_{i=1}^r\sigma_i/\sigma_1$} \\
\footnotesize energy &  \parbox{6.5cm}{$\min\left[\argmax_{\ell\in\{1,...,N\}}\left(\frac{\sum_{i=1}^\ell \sigma_i^2}{\sum_{j=1}^r \sigma_j^2} > \tau\right)\right]$}\\
\footnotesize elbow &  \parbox{6.5cm}{$\frac{1}{\sqrt{2}}\argmax_{i\in\{1,...,N\}}\,\left|\frac{i-1}{N-1} + \frac{\sigma_i - \sigma_N}{\sigma_1 - \sigma_N} - 1\right| - 1$}\\
\footnotesize erank & \parbox{6.5cm}{$\exp\left[-\sum_{i=1}^r \frac{\sigma_i}{\sum_{j=1}^r\,\sigma_j} \log \frac{\sigma_i}{\sum_{j=1}^r\,\sigma_j}\right]$}\\
\footnotesize thrank & \parbox{6.5cm}{$\#\left\{\sigma_i\,\Big|\,i\in\{1,...,N\}\,\mathrm{and}\,\sigma_i > \frac{4\sigma_{\mathrm{med}}}{\sqrt{\,3\,\mu_{\mathrm{med}}}}\right\}$}\\
\footnotesize shrank & \parbox{6.5cm}{$\#\{s^*(\sigma_i)\,|\,i\in\{1,...,N\}\,\text{and}\,\,s^*(\sigma_i)>0\,\}$}\\
\end{tabular}
\end{ruledtabular}
\end{table}
\noindent\textbf{Effective ranks.}
The idea of extracting the number of significant components in a matrix decomposition is an old theme (e.g., in factor analysis \cite{Malinowski1977, Sanchez1986} or PCA~\cite[How Many Components~?]{Abdi2010}), but is still subject to new interesting developments in random matrix theory, data science~\cite{Gavish2014, Gavish2017}, and in network science where hyperbolic geometry~\cite{Almagro2022} and information theory~\cite{Lynn2021} are used. Because of the close relationship of SVD with the rank, many effective ranks are defined using the singular values. Intuitively, these effective ranks are numbers that indicate how many singular values are significant when decomposing a matrix.  
Table~\ref{tab:effective_ranks} presents the list of different effective ranks that we have inventoried. The effective ranks thrank and shrank are defined from matrix denoising techniques such as the ones introduced by Refs.~\cite{Perry2009, Gavish2014, Gavish2017}, which rely on the spectral theory of infinite random matrices~\cite{Benaych-Georges2012} to determine optimal ways of shrinking the singular values~(see \prettyref{SIsubsec:effective_ranks}). In Fig.~\ref{fig:low_rank_hypothesis}l, the Frobenius norm is used to obtain shrank and a threshold of 0.9 is used for the energy ratio in Fig.~\ref{fig:low_rank_hypothesis}j.

As shown in Lemma~\ref{lem:ordering_ens_ranks}, the following ordering of the effective ranks holds: $\mathrm{srank} \leq \mathrm{nrank} \leq \mathrm{erank} \leq \mathrm{rank}$. Because of their simple forms, $\mathrm{srank}$, $\mathrm{nrank}$, and $\mathrm{erank}$ are amenable to analytic calculations.
In particular, we prove that these effective ranks are of order $O(1)$ for singular values with exponentially decreasing envelopes (only stated for srank below).
\begin{corollary}[Simplified version of Corollary~\ref{cor:expo_o1}]\label{cor:expo_o1_methods}
Let $(\,W_N\,)_{N\in \mathbb{Z}_+}$ be an infinite sequence of matrices in which $W_N$ has size $N\times N$. Suppose that there are parameters $\alpha$ and $\omega$ such that $0<\alpha\leq \omega<1$ and for each $N$,  the singular values $\sigma_1\geq \sigma_2\geq \cdots \sigma_N\geq 0$ of $W_N$ satisfy the inequalities 
\begin{equation}
    \alpha^{i-1}\,\leq \,\frac{\sigma_i}{\sigma_1}\,\leq \,\omega^{i-1},\qquad i\in\{1,\ldots, N\}\,. 
\end{equation}
 Then, as $N\to \infty$,
\begin{align}
   \frac{1}{1-\alpha^2}+O(\alpha^{2N})\,\leq \, &\, \mathrm{srank}(W_N)\,\leq   \frac{1}{1-\omega^2}+O(\omega^{2N})\,,
\end{align} 
\end{corollary}
Combined with Theorem~\ref{thm:upper_bound_singvals_scm_simplified}, the latter theorem implies that the expected weight matrices for the directed soft configuration model and its weighted version have $O(1)$ effective ranks.

Moreover, we show in Lemma~\ref{lem:effrank_area} that $\mathrm{srank}$, $\mathrm{nrank}$, and $\mathrm{erank}$ all have an interpretation in terms of area under the normalized singular value scree plots. This point of view allows considering a more general family of singular-value envelopes, such as the one in Fig.~\ref{fig:low_rank_hypothesis}e, to bound the effective ranks. Interestingly, the bounds are related to Gaussian hypergeometric functions, as shown in the next theorem (only stated for srank below, for simplicity).
\begin{theorem}[Simplified version of Theorem~\ref{thm:bounds_effective_ranks_hypergeometric}]\label{thm:bounds_effective_ranks_hypergeometric_methods}
Suppose that the singular values of matrix $W$, $\sigma_1\geq \sigma_2\geq \cdots \sigma_N\geq 0$, satisfy the inequality 
\begin{equation}\label{eq:ratio_bounds}
     \frac{\left(1-x_i\right)^{c^*-2}}{\left(1+\zeta^*x_i\right)^{b^*}} \,\leq \,\frac{\sigma_i}{\sigma_1}\,\leq \,  \frac{\left(1-x_i\right)^{c_*-2}}{\left(1+\zeta_*x_i\right)^{b_*}}
\end{equation}
where $x_i = (i-1)/(N-1)$ and for some $0\leq b_*\leq b^* $, $2\leq c_*\leq c^*$, $0<\zeta_*\leq \zeta^*$, and for all $i\in\{1,\ldots, N\}$. Then,
\begin{align}\label{eq:srank_hypergeo_bound} 
     \frac{N-1}{2c^*-3}\; H(b^*, c^*, \zeta^*)\,\leq \, \, \mathrm{srank}(W)\,\leq  1+  \frac{N-1}{2c_*-3}\; H(b_*, c_*, \zeta_*)\,,
\end{align}
where $H(b, c, \zeta) := {}_2F_1(1,2b;2(c-1);-\zeta)$ and ${}_2F_1$ being the Gaussian hypergeometric function.
\end{theorem}
In Fig.~\ref{fig:low_rank_hypothesis}e, each singular value distribution of the real networks is interpolated linearly with 1000 points and the indices are then divided by 1000. The singular-value envelope is then obtained by fitting the upper bound in Eq.~\eqref{eq:ratio_bounds} to the 95th percentile of the singular values. The fit is done by minimizing the L2 norm for the parameters $b := b_* \in [0.01, 10]$, $c := c_* \in [2, 10]$, and $\zeta := \zeta_* \in [0.01, 1000]$ and the minimization gives $b \approx 0.54$, $c \approx 2.3$, and $\zeta \approx 25$. We then use those parameters to evaluate the upper bound in Eq.~\eqref{eq:srank_hypergeo_bound} divided by $N$ (where we neglect the terms $1/N$), which is shown in Fig.~\ref{fig:low_rank_hypothesis}f. 

Corollary~\ref{cor:hypergeometric} shows that if there is a growing graph whose singular values remain bounded within hypergeometric envelopes, then srank, nrank, and erank are of order $O(N^{1 - \epsilon})$ with $\epsilon \in (0, 1]$ in different asymptotic regimes for the parameters $b$ and $\zeta$, meaning that the effective rank to dimension ratios become negligible asymptotically. SI~\ref{SIsubsec:impact_singvals_effrank} clarifies how various singular-value envelopes can lead to very distinct asymptotic behaviors (see Fig.~\ref{fig:asymptotic-ranks-theory}).

When the asymptotic perspective is no longer applicable (e.g., for real networks), we cannot classify an effective rank as either ``low'' or ``high''. Yet, as explained in the main text, we can use effective ranks to dimension ratios, which are well defined for all $N$ and their values range from 0 ($W$ has rank 0) to 1 ($W$ has full rank).

\vspace{0.2cm}
\noindent\textbf{Dimension reduction of dynamical systems}. Dimension reduction of high-dimensional nonlinear dynamics is a fundamental approach to get analytical and numerical insights on complex systems. Low-dimensional dynamics can be obtained from an optimization problem, where some error is minimized under a set of constraints to preserve the salient properties of the original system. For dynamical systems, a natural optimization variable is the reduced vector field $F$ itself, which is chosen to represent approximately the complete vector field $f$. Yet, it is rather puzzling to find how the different vector field errors are related to each other and which one can be minimized analytically. In \prettyref{SIsubsec:least_square_vf}, we provide a useful diagram (see Diagram~\ref{diag:alignment_errors}) that sheds light on the links between the different ways to define \textit{alignment errors} between vector fields.

More precisely, let $f$ be a complete vector field in $\mathbb{R}^N$, $F$ be a reduced vector field in $\mathbb{R}^n$, and $M$ be the $n\times N$ reduction matrix. At $x \in \mathbb{R}^N$, the alignment error in $\mathbb{R}^N$ is the RMSE between the vector fields $f$ and $M^+\circ F\circ M$,
\begin{align}\label{eq:alignment_error_N_methods}
    \varepsilon(x) = \|f(x) - M^+F(Mx)\|/\sqrt{N}\,;
\end{align}
and the alignment error in $\mathbb{R}^n$ is the RMSE between the vector field $M\circ f$ and $F\circ M$,
\begin{align}\label{eq:alignment_error_n_methods}
    \mathcal{E}(x) = \|Mf(x) - F(Mx)\|/\sqrt{n}\,,
\end{align}
where $\|\,\|$ is the Euclidean vector norm. By applying the definition of alignment errors on the projected complete vector field $f\circ P$ instead of $f$ only, we also define the alignment errors
\begin{align}
    \varepsilon'(x) &= \|f(Px) - M^+F(Mx)\|/\sqrt{N}\,\label{eq:alignment_error_N_projected_methods}\\
    \mathcal{E}'(x) &= \|Mf(Px) - F(Mx)\|/\sqrt{n}\,\label{eq:alignment_error_n_projected_methods}
\end{align}
with $P = M^{+}M$ being a projector and $M^+$ being the Moore–Penrose pseudoinverse of $M$.
In principle, the alignment error $\mathcal{E}(x)$ in $\mathbb{R}^n$ is to be minimized in order to be as close as possible to an exact dimension reduction (see Definition~\ref{def:dimension_reduction}, Theorem~\ref{thm:dimension_reduction}, and Diagram~\ref{diag:dimension_reduction}), but this is far from a simple task. However, as shown in Theorem~\ref{thm:least_square_optimal_vector_field}, one can use least squares to show that the vector field of the reduced dynamics
\begin{equation}\label{eq:least_square_optimal_vector_field_methods}
    \dot{X} = Mf(M^+X)
\end{equation}
is optimal in the sense that it minimizes the alignment error $\varepsilon'(x)$ in $\mathbb{R}^N$. As a consequence, the alignment error $\mathcal{E}'(x)$ is exactly 0. 

In Extended Data Table~1, we carry out the optimal dimension reduction on five dynamics from different fields of application. For the RNN and the neuronal dynamics, we have $\mathcal{D}^{(2)} = -MDM^+$ where $D = \diag(d_1,...,d_N)$ and $\mathscr{W}_{j\nu} = \sum_{k=1}^N W_{jk}M_{k\nu}^+$ and we discuss about the other dynamics in the next part of the Methods. With the optimal vector field in Eq.~\eqref{eq:least_square_optimal_vector_field_methods} and for dynamics of the general form $\dot{x} = g(x, y)$ (see Assumptions~\ref{ass:error_bound}), we find an upper bound on the alignment error $\mathcal{E}(x)$ related to the singular values of $W$. 
\begin{theorem}[Simplified version of Theorem~\ref{thm:upper_bound_rapid_decrease}]\label{thm:upper_bound_rapid_decrease_simplified}
The alignment error $\mathcal{E}(x)$ in $\mathbb{R}^n$ at $x \in \mathbb{R}^N$ is upper-bounded as
\begin{equation}\label{eq:upper_bound_rapid_decrease_simplified}
    \sqrt{n}\,\mathcal{E}(x) \leq\|V_n^{\top}J_x'(I-V_nV_n^\top)x\| +  {\sigma_{n+1}} \|V_n^{\top}J_y'\|_2\|x\|,
\end{equation}
where $y' = Wx'$ with $x'$ being some point between $x$ and $V_nV_n^\top x$, $\sigma_i$ is the $i$-th singular value of $W$, and $J_x' = J_x(x',y')$, $J_y' = J_y(x',y')$ are the Jacobian matrices of $f$ with derivatives according to the vectors $x$ and $y$ respectively. Moreover, for any $x$ not at the origin of $\mathbb{R}^N$, the following upper bound holds:
\begin{equation}\label{eq:upper_bound_rapid_decrease_relative_theo_methods}
\frac{\mathcal{E}(x)}{\|x\|} \leq \frac{1}{\sqrt{n}}\Big[\alpha(x',y')+ \sigma_{n+1}\beta(x',y')\Big]\,,
\end{equation}
where $\alpha(x',y')=\sigma_1(J_x(x',y'))$ and $\beta(x',y')=\sigma_1(J_y(x',y'))$.
\end{theorem}
As a bonus, the proof of the theorem suggests choosing $M$ as the truncated right singular vectors $V_n^\top$, since it allows minimizing a part of the bound.
This is a consequence of the Schmidt-Eckart-Young-Mirsky theorem and more specifically, Theorem~\ref{thm:optimal_projection_general}. This choice for $M$ also has a notable consequence: each observable $X_{\mu}$ generally becomes a global observable in that it contains information on most vertices. This characteristic, alongside that it is a finite-size dimension reduction, make our approach stands out from many mean-field modeling approaches used in network science in which vertices are coarse-grained according to their degree (local property) or to some other mesoscopic property of the network.

Theorem~\ref{thm:upper_bound_rapid_decrease_simplified} also provides a criterion for exact dimension reduction: if $J_x(x', y') = dI$ for some real constant $d$ and $n$ is the rank of $W$, then $\mathcal{E}(x)= 0$ (see Corollary~\ref{cor:exact_rank} in \prettyref{SIsubsec:error_bound}). For example, we find that the class of dynamics of matrix form
\begin{equation}
    \dot{x} = d\,x + s(Wx)\,,
\end{equation}
where $s$ is a vector of $N$ functions $s_i:\mathbb{R} \to \mathbb{R}$ and $W$ has rank $r$ and compact SVD $U_r\Sigma_rV_r^\top$, can be exactly reduced to the $r$-dimensional reduced dynamics
\begin{align}\label{eq:rank_reduced_methods}
    \dot{X} &= d\,X + V_r^\top s(U_r\Sigma_r X),
\end{align}
where $X = V_r^\top x$. For any $n$ and $X = V_n^\top x$, the vector field in Eq.~\eqref{eq:rank_reduced_methods} is the least-square optimal one in the sense described in Theorem~\ref{thm:least_square_optimal_vector_field} of the \prettyref{SIsubsec:least_square_vf}. This result implies that any RNN or any neuronal dynamics (with $a = 0$) having the forms given in Extended Data Table~1 can be exactly reduced (see Examples~\ref{ex:exact_rnn}-\ref{ex:exact_wilson_cowan} in \prettyref{SIsubsec:error_bound}). 

A simple corollary of the latter theorem (Corollary~\ref{cor:upper_bound_linear_system}) shows that if the dynamics is a linear system, the relative alignment error in $\mathbb{R}^n$ at $x\in\mathbb{R}^N$ is 
\begin{equation}\label{eq:upper_bound_lineear_system_method}
    \frac{\mathcal{E}(x)}{\|x\|} \leq \frac{\sigma_{n+1}}{\sqrt{n}},
\end{equation}
implying that a rapid decrease of the singular values of $W$ directly induces a rapid decrease of the alignment error.

\vspace{0.2cm}



\noindent\textbf{Emergence of higher-order interactions.}
All the $N$-dimensional (complete) dynamics on a network in Extended Data Table~1 (and many more, see \prettyref{SIsubsec:emergence}) have the general form $\dot{x}_i = h_i(x_i, y_i)$ for all $i\in\{1,...,N\}$, where $x_i:[0,\infty) \to \mathbb{R}$, $y_i = \sum_{j=1}^N W_{ij}x_j$, and $h_i: \mathbb{R}^2 \to \mathbb{R}$ is an analytic function. 

\begin{proposition}[Simplified version of Proposition~\ref{prop:emergence}]\label{prop:emergence_simplified}
The least-square reduced dynamics can be expressed in terms of higher-order interactions between the observables as
 \begin{align*}
      \dot{X}_{\mu} = \mathcal{C}_{\mu} &+ \textstyle{\sum_{d_x=1}^\infty\sum_{\bm{\alpha}}}\mathcal{D}_{\mu\bm{\alpha}}^{(d_x + 1)} X_{\bm{\alpha}} + \textstyle{\sum_{d_y=1}^\infty\sum_{\bm{\beta}}} \mathcal{W}_{\mu\bm\beta}^{(d_y + 1)} X_{\bm\beta} \\ &+ \textstyle{\sum_{d_x, d_y=1}^\infty\sum_{\bm\alpha, \bm\beta}} \mathcal{T}_{\mu \bm\alpha \bm\beta}^{(d_x+d_y+1)} X_{\bm\alpha\bm\beta},
\end{align*}
where we have introduced the multi-indices $\bm{\alpha} = (\alpha_1,...,\alpha_{d_x})$ and $\bm{\beta} = (\beta_1,...,\beta_{d_y})$ with $\alpha_p,\beta_q \in \{1,...,n\}$, the compact notation for  products $X_{\bm{\gamma}} = X_{\gamma_1}...X_{\gamma_d}$, while $\mathcal{C}_{\mu}$ denotes a real constant and $\mu\in\{1,\ldots, n\}$. The higher-order interactions are described by three tensors of respective order $d_x+1$, $d_y+1$, $d_x+d_y+1$, and whose elements are 
\begin{align*}
    \mathcal{D}_{\mu\bm\alpha}^{(d_x + 1)} &= \textstyle{\sum_{i=1}^N} c_{i d_x 0}M_{\mu i}M_{i\bm{\alpha}}^+,\\
    \mathcal{W}_{\mu\bm\beta}^{(d_y + 1)} &= \textstyle{\sum_{i=1}^N \sum_{\bm{j}}} c_{i 0 d_y}M_{\mu i} W_{i\bm{j}}M_{\bm{j}\bm{\beta}}^+,\\
    \mathcal{T}_{\mu \bm\alpha \bm\beta}^{(d_x+d_y+1)} &= \textstyle{\sum_{i=1}^N \sum_{\bm{j}}} c_{id_xd_y}M_{\mu i} M_{i\bm{\alpha}}^+W_{i\bm{j}}M_{\bm{j}\bm{\beta}}^+,
\end{align*}
for some real coefficients $c_{id_xd_y}$ with $i \in \{1,...,N\}$, $d_x,d_y\in\mathbb{Z}_+$ and the multi-index $\bm{j}$ in the sums is in $\{1,...,N\}^{d_y}$.
\end{proposition}
This proposition led us to two corollaries. First, if $h_i(x_i, y_i)$ is a polynomial of total degree $\delta$ in $x_i$ and $y_i$, then the reduced dynamics has a polynomial vector field of total degree $\delta$ with interactions of maximal order $\delta + 1$ (Corollary~\ref{cor:emergence_polynomial}). Second, if $M$ is block diagonal and $h_i$ linearly depends on $y_i$, then there are solely pairwise interactions in the reduced system, which doesn't hold in general for nonlinear dependencies of $h_i$ over $y_i$ (Corollary~\ref{cor:sufficient_condition_emergence}).

In Extended Data Table 1, we apply Proposition~\ref{prop:emergence_simplified} and Corollary~\ref{cor:emergence_polynomial} to the QMF SIS dynamics, the microbial dynamics, and the Kuramoto-Sakaguchi dynamics, which illustrates concretely the emergence of higher-order interactions through dimension reduction.
More details are given in \prettyref{SIsubsec:emergence}.

\vspace{0.2cm}

\noindent\textbf{Integration and properties of the dynamics}.
The trajectories of the dynamics on the real networks presented in Fig.~\ref{fig:error_vector_fields} were obtained with \href{https://docs.scipy.org/doc/scipy/reference/generated/scipy.integrate.solve_ivp.html}{\texttt{solve\_ivp}} from \texttt{scipy.integrate}.  We used the backward differentiation formula (BDF), an implicit method with variable step length and order, which is known to be well suited for stiff problems, such as the microbial dynamics on the gut microbiome. We observed that a relative tolerance $\mathrm{rtol} = 10^{-8}$ and an absolute tolerance of $\mathrm{atol} = 10^{-12}$ for the complete microbial dynamics ($\mathrm{rtol} = 10^{-6}$ and $\mathrm{atol} = 10^{-10}$ for the reduced dynamics) gave reliable results with decent integration time while being in line with the recent benchmarks of Ref.~\cite{Stadter2021}. Moreover, we have provided the Jacobian matrices of the complete and reduced dynamics to the integrator as recommended in the documentation of solve\_ivp for the BDF method. We also integrated the other dynamics with the BDF method with a relative tolerance of $10^{-8}$ and an absolute tolerance of $10^{-12}$.  

For the epidemiological dynamics, the phenomenon of critical slowing down appears, but it is easily dealt with by increasing the number of time steps near the transcritical bifurcation (at the infection rate of 1, that is, the largest singular value of the rescaled network) as we have done in the inset of Fig.~\ref{fig:error_vector_fields}e. Note that increasing the dimension improves the prediction for higher infection rates. In Fig.~\ref{fig:error_vector_fields}f, we observe a hysteresis for the global observable of the neuronal dynamics vs. the synaptic weight. In Fig.~\ref{fig:error_vector_fields}e-f, the root-mean-square errors (RMSE) are simply computed between the global equilibrium points of the complete and the reduced dynamics at different $n$.

As illustrated in Fig.~\ref{fig:error_vector_fields}g, multiple branches of stable equilibrium points for the global observables of the microbial dynamics arise. 
We proceeded as follows to get a simplified picture involving only some equilibrium point branches. We focused on one forward branch obtained with initial conditions $x_0$ sampled from a uniform distribution between 0 and 1 and showed its loss of stability when incrementally increasing the microbial interaction weight in Fig.~\ref{fig:error_vector_fields}g. To obtain one backward branch, we sampled the initial condition $x_0$ from a uniform distribution between 0 and $z$ where $z$ is a random integer between 1 and 15, we integrated the dynamics to get the equilibrium point, we decreased the microbial interaction weight and used the last equilibrium point as the initial condition for the integration, and repeated these last two steps until the minimum coupling value (0.1 in Fig.~\ref{fig:error_vector_fields}g) is reached. We repeated all these steps 100 times (300 for $n=76$) to generate different initial conditions and stable branches. At each iteration, we ensured that the vector fields evaluated at the equilibrium points gave a vector with elements below the tolerance $10^{-7}$ and that the equilibrium points were positive (see \prettyref{SIsubsec:numerical_integration}). In this case, the RMSE is computed between the average upper and lower branches of the complete and reduced dynamics.

For the (finite-size) recurrent neural network, similar to the observations in the conclusion of Ref.~\cite{sompolinsky1988chaos}, there is a stable equilibrium point at zero for lower coupling and increasing the coupling eventually gives rise to limit cycles of increasing complexity such as the one in Fig.~\ref{fig:error_vector_fields}h. We illustrate a 3-dimensional projection of this high-dimensional limit cycle in the complete dynamics and the ones in the reduced dynamics as the dimension $n$ approaches the rank of the learned network. The RMSE is computed between the points of the limit cycle for the complete recurrent neural dynamics and the closest points on the limit cycles of the reduced dynamics. 

The choices of global observables used in Fig.~\ref{fig:error_vector_fields} are justified in \prettyref{SIsubsec:global_observable} and the parameters of the dynamics are in the Extended Data Table~1.

\vspace{0.2cm}

\noindent\textbf{Data availability}.
All the details about the real networks data used in the paper, mostly from the network repository \href{https://networks.skewed.de/}{Netzschleuder}, are given in \prettyref{SIsec:real_network_dataset}. The data to generate Fig.~1, 2 and 4 are available on Zenodo~(\url{https://doi.org/10.5281/zenodo.8342130}).

\vspace{0.2cm}

\noindent\textbf{Code availability}.
The \textit{Python} code used to generate the results of the paper is available on Zenodo~(\url{https://doi.org/10.5281/zenodo.8342130}). The code for the optimal shrinkage of singular values is a \textit{Python} implementation of the \textit{Matlab} codes \href{https://purl.stanford.edu/vg705qn9070}{optimal\_singval\_threshold}~\cite{Gavish2014} and \href{http://purl.stanford.edu/kv623gt2817}{optimal\_singval\_shrink}~\cite{Gavish2017}, which is partly based on the repository \href{https://github.com/erichson/optht}{optht} by B. Erichson. 


\clearpage


\setcounter{equation}{0}
\setcounter{figure}{0}
\setcounter{table}{0}
\setcounter{page}{1}
\renewcommand{\theequation}{S\arabic{equation}}
\renewcommand{\thefigure}{S\arabic{figure}}
\renewcommand{\thetable}{S\Roman{table}}
\renewcommand{\thetheorem}{S\arabic{theorem}}
\thispagestyle{empty}

\onecolumngrid

\normalsize
\let\addcontentsline\oldaddcontentsline

\centerline{\large\textbf{The low-rank hypothesis of complex systems}}
\centerline{\textbf{--- Supplementary information ---}}

\tableofcontents



\section{Preliminaries on singular value decomposition}
\label{SIsec:svd}
Singular Value Decomposition (SVD) goes back to Beltrami (1873) and Jordan (1874) 
and has become a central linear algebra tool in many areas of science, partly because of its fundamental role in dimension reduction \cite{Schmidt1907, Eckart1936, Stewart1993}\cite[Chapter 1]{Brunton2019}. Although one must be careful with the comparisons, which have led to abuses of language~\cite{Gerbrands1981}, SVD possesses some similarities with techniques such as Principal Component Analysis (PCA)~\cite{Hotelling1933a, Hotelling1933b, Wold1987,Ferre1995, Abdi2010, Johnstone2018, Cook2022}, Karhunen-Loève Transform (KLT)~\cite{Karhunen1947, Loeve1955, Everson1995}, Proper Orthogonal Decomposition (POD)~\cite{Kerschen2005, Volkwein2013, Kutz2016}, and Empirical Orthogonal Function (EOF) \cite{lorenz1956, Monahan2009}. In machine learning, some autoencoders have been shown to be at best equivalent to SVD~\cite{Bourlard1988, Bourlard2022}. Even if the subject is old in itself, there are still many interesting developments about SVD, notably in random matrix theory~\cite{Bai2010, Benaych-Georges2012, Tao2012, tao2012random, Gavish2014, bloemendal2016limits, Gavish2017, Beckermann2017, Donoho2018} where the singular value distribution is often called the eigenvalue distribution of the Wishart, chiral or Laguerre matrix ensembles \cite[Chap.~3]{Forrester2010} or of sample covariance matrices~\cite[Chap.~3]{Bai2010}. Because of its importance in our work and for the sake of completeness, we gather fundamental theorems related to SVD which will be useful to prove the main mathematical results of the paper. We begin this section by recalling the definition of SVD and its close relationship with the rank, i.e., the maximal number of linearly independent rows or columns of a matrix.

\subsection{Definition of SVD and its link to the rank}
\label{SIsubsec:svd_def}
First of all, any matrix admits a factorization based on its rank. Indeed, if $A$ is a matrix of dimension $m\times n$ and of rank $r$, then there exists a rank factorization of $A$, i.e., a decomposition of the form $A = LM$, where $L$ and $M$ are matrices of dimension $m\times r$ and $r \times n$, respectively. Moreover, the rank factorization $A = LM$ is not unique. One very popular rank factorization valid, in particular, for real symmetric matrices is the eigenvalue decomposition. Yet, an arbitrary matrix $A$ is not always diagonalizable by a similarity relation $A = PDP^{-1}$ (e.g., any rectangular matrix). Note, however, that the matrices $AA^{\dagger}$ and $A^{\dagger}A$ ($^\dagger$ denoting the Hermitian conjugation) are square and diagonalizable by a unitary matrix since they are Hermitian (hence, normal). Using this important remark, it can be shown that there always exists a unitary equivalence relation between a matrix and a diagonal matrix of nonnegative elements, the singular value decomposition.
\begin{theorem}\label{thm:SVD}
Let $A$ be a complex matrix of dimension $m\times n$ and rank $r$. Then, there exists a SVD of $A$, i.e., a factorization of the form
\begin{equation}\label{eq:svd}
    A = U\Sigma V^\dagger
\end{equation}
where $U = (u_1,...,u_m)$ and $V = (v_1,...,v_n)$ are unitary matrices of dimension $m\times m$ and $n\times n$, containing respectively the eigenvectors $u_i$ of $AA^{\dagger}$ and the eigenvectors $v_i$ of $A^{\dagger}A$. Moreover, the matrix $\Sigma$ is a rectangular diagonal matrix of size $m\times n$ defined as
\begin{equation}
    \Sigma = \begin{pmatrix}\sigma_1&0& ... \\0&\sigma_2& ... \\ \vdots&\vdots& \ddots\end{pmatrix}\quad \text{with} \quad \begin{array}{lll}
        \sigma_1 \geq \sigma_2 \geq ... \geq \sigma_r > 0  \\
         \sigma_{r+1} = ... = \sigma_q = 0
    \end{array} 
\end{equation}
where $q = \min(m,n)$ and $\sigma_i = \sqrt{\lambda_i}$ with $\lambda_i$ being the $i$-th eigenvalue of $A^{\dagger}A$ or $AA^{\dagger}$. If additionally all the elements of $A$ are real, then $U$ and $V$ are real orthogonal matrices. 
\end{theorem}
\begin{proof}
    See theorem 3.1.1 of Ref~\cite{Horn1991}, theorem 2.6.3 of Ref.~\cite{Horn2013}, or theorem 1.3.9 of Ref.~\cite{Tao2012}.
\end{proof}
\begin{remark}
The nonnegative numbers $\sigma_1, ..., \sigma_q$ in the previous theorem are called the \emph{singular values} of $A$ while the vectors $u_1, ..., u_m$ and $v_1,...,v_n$ are respectively called the \emph{left and right singular vectors} of $A$. 
For clarity, especially when the singular values of multiple matrices are involved, we will define $\sigma_i$ as a function of $A$ and write its values as $\sigma_i(A)$.
\end{remark}

\begin{remark}
In general, there is no obvious relationship between the eigenvalues and the singular values of a (square) matrix. However, for the family of normal matrices (including hermitian, anti-hermitian, unitary, and anti-unitary matrices), the singular values are given by the module of the eigenvalues.
%
To visualize the singular values, it is typical to plot them in a decreasing order, which is called a scree plot in the context of PCA \cite{Ferre1995, Abdi2010}, or illustrate them in a histogram.
\end{remark}

The SVD is thus closely related to the notion of rank, since the number of nonzero singular values of a matrix is \textit{equal} to its rank (while the number of its nonzero eigenvalues is lower or equal to its rank~\cite[p.151]{Horn2013}). Its relation to dimension reduction then becomes obvious: one can truncate the matrices $U$, $V$, and $\Sigma$ by removing their last columns (and rows for $\Sigma$) to get smaller matrices $U_r = (u_1\,...\,u_r)$, $V_r = (v_1\,...\,v_r)$, and $\Sigma_r = \diag(\sigma_1,...,\sigma_r)$ with $r = \rank{A}$, and obtain a rank factorization:
\begin{equation}\label{eq:compact_svd}
    A = U_r \Sigma_r V_r^\dagger,
\end{equation}
which is sometimes called the compact singular value decomposition. More importantly for dimension reduction, when the matrices $U$, $V$, and $\Sigma$ are truncated to $U_k$, $V_k$, $\Sigma_k$ with $k<n$, the truncated SVD is the optimal low-rank factorization as it will be seen in the next subsection.

\begin{remark}
It is often more convenient to rewrite the SVD in Eq.~\eqref{eq:svd} or equivalently in Eq.~\eqref{eq:compact_svd} as
\begin{equation}\label{eq:svd_sum}
    A = \sum_{i=1}^r \sigma_i \, u_i\,v_i^\dagger\,.
\end{equation}
This shows that any matrix of rank $r$ is equal to the sum of $r$ linearly independent unitary matrices, each being of rank 1 and having a (Frobenius or spectral) norm equal to 1. If all the singular values are distinct, then $\sigma_1 \, u_1\,v_1^\dagger$ and $\sigma_r \, u_r\,v_r^\dagger$ respectively constitute the most and the least important contributions to the matrix $A$. Moreover, Eq.~\eqref{eq:svd_sum} implies an explicit formula for the Moore-Penrose pseudo-inverse of $A$,
\begin{equation}\label{eq:pseudo_sum}
    A^+ = \sum_{i=1}^r \frac{1}{\sigma_i} \, v_i\,u_i^\dagger\,,
\end{equation}
proving that $A$ and $A^+$ share the same rank.
\end{remark}

\subsection{Weyl's theorem and optimal low-rank factorization}
\label{SIsubsec:weyl_facto}
The SVD shares many equivalent theorems with the eigenvalue decomposition \cite{Horn2013}, such as Rayleigh's theorem, the Courant-Fischer theorem, Cauchy's interlacing theorem, and, in particular, Weyl's theorem, which is of fundamental importance in the paper. The following result was obtained in 1951 by Fan~\cite[Theorem 2]{Fan1951}. 
\begin{theorem}\label{thm:weyl_vas}
  Let $A$ and $B$ be two matrices of dimension $m \times n$ and let $q = \min(m, n)$. Then, 
  \begin{equation}\label{eq:inegalite_weyl_vas}
      \sigma_{i+j-1}(A+B) \leq \sigma_i(A) + \sigma_j(B) \qquad \forall \; 1 \leq i,\,j,\, i+j - 1 \leq q,
  \end{equation}
  where $\sigma_k(X)$ is the $k$-th singular value of $X$ and the singular values are ordered in the usual decreasing order.
\end{theorem}
\begin{proof}
A detailed proof based on Weyl's theorem can be done by following the steps of Horn~\&~Johnson~\cite{Horn2013}. A proof that uses the Courant-Fisher theorem for singular values is also given in Ref.~\cite[Theorem 3.3.16]{Horn1991}.
\end{proof}

\begin{remark}\label{rem:weyl}
If $i = j = 1$, then the previous theorem implies that the dominant singular values satisfy
\begin{equation}
      \sigma_1(A+B) \leq \sigma_1(A) + \sigma_1(B).
\end{equation}
The latter inequality was known before the generalization by Ky Fan and it is often attributed \cite{Marshall2011} to Wittmeyer~\cite[Eq.~(VIII)]{Wittmeyer1936}, but Wittmeyer himself writes in a footnote that the equation is in Wintner,``Spektraltheorie der unendlicheri Matrizen'', Leipzig 1929, p.~130. Nowadays, the result is, perhaps, not surprising: it is the triangle inequality for the spectral matrix norm.
\end{remark}
A first key corollary~\cite[Exercise 1.3.22 (iv)]{Tao2012} allows us to analyze random graphs through perturbation theory of random matrices. Indeed, the following result establishes that the strength (norm) of a matrix perturbation bounds the difference between each singular value of a matrix and the ones of its perturbed version.
\begin{corollary}\label{cor:weyl_vas}
Let $A$ and $B$ be two matrices of dimension $m \times n$ and let $q = \min(m, n)$. 
\begin{align}\label{eq:weyl_vas_cor}
    \left|\sigma_i(A+B) - \sigma_i(A)\right| \leq \|B\|_2 \qquad \forall \; 1 \leq i \leq q,
\end{align}
where $\sigma_i(X)$ is the $i$-th singular value of $X$ and the singular values are ordered in the usual decreasing order.
\end{corollary}
The importance of the Weyl theorem in the paper also relies on what it implies for dimension reduction. In particular, it allows proving the Schmidt-Eckart-Young-Mirsky theorem \cite{Schmidt1907, Eckart1936, Mirsky1960, Stewart1993, Ben-Israel2003, Antoulas2005} (often called the Eckart-Young theorem \cite[Theorem 2.4.8]{Golub2012} or the Eckart-Young-Mirsky theorem \cite{Markovsky2019}) which shows that the truncated SVD is the optimal low-rank approximation of a matrix according to unitarily invariant norms. In Theorem~\ref{thm:SEYM}, we present our formulation of the result (illustrated in Fig.~\ref{fig:SVD_low_rank}) for the Frobenius norm and the spectral norm.

\begin{figure}[t]
    \centering
    \includegraphics[width=0.9\linewidth]{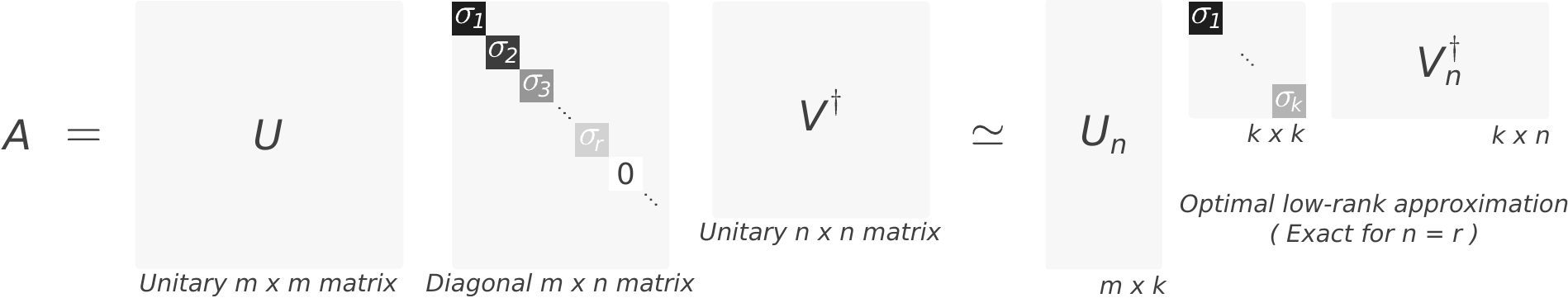}
    \caption{Truncated SVD is the optimal low-rank approximation of any matrix $A$ according to unitarily invariant norms.}
    \label{fig:SVD_low_rank}
\end{figure}
\begin{theorem}\label{thm:SEYM}
Let $A$ be a matrix of rank greater than or equal to $k$. Consider the optimization problem
  \begin{equation}
        \begin{array}{ll}
        \underset{B}{\operatorname{minimize} } & \|A - B\|^2 \tag{P0}\label{eq:P0}\\
        \operatorname{subject\;to} & \rank{B} \leq k \,,
        \end{array}
\end{equation}
where $\|\,\|$ denotes the spectral norm $\|\,\|_2$ or the Frobenius norm $\|\,\|_F$. Then, the minimum error of problem~\eqref{eq:P0} is
\begin{align}
   \min_{\substack{B, \rank{B} \leq k}}\|A - B\|_2^2 = \sigma_{k+1}^2 \quad \text{or} \quad \min_{\substack{B, \rank{B} \leq k}}\|A - B\|^2_F = \sum_{i=k+1}^q \sigma_i^2,
\end{align}
where $q=\min\{m,n\}$ and $\sigma_1\geq ...\geq \sigma_q$ are the singular values of $A$. Furthermore, in both cases, a solution to problem~\eqref{eq:P0} is provided by the $k$-truncated SVD of $A$, i.e.,
\begin{equation} 
    B^*  = \argmin_{_{\substack{B,\rank{B} \leq k}}} \|A - B\|_2^2 =  \argmin_{_{\substack{B,\rank{B} \leq k}}} \|A - B\|^2_F = \sum_{i=1}^{k}\sigma_{i}u_{i}v_{i}^{\dagger},
\end{equation}
where  $u_{i}$, $v_{i}$ are the $i$-th left and right singular vectors of $A$, respectively. The solution $B^*$ is unique if $\sigma_{k} > \sigma_{k+1}$.
\end{theorem}

For our paper, especially to find the upper bound on the alignment error [Theorem~\ref{thm:upper_bound_rapid_decrease}], Theorem~\ref{thm:SEYM} entails another important result: the projectors formed by the left and right singular vector matrices are optimal orthogonal projectors. This fact seems to be well known \cite[Fact 2]{Harvey2011} but, to the authors' knowledge, has not yet been presented in a comprehensive form accompanied by a detailed proof. We hence introduce the following theorem, which will be used later to prove Theorem~\ref{thm:upper_bound_rapid_decrease}.

\begin{theorem}\label{thm:optimal_projection_general}
Let $A$ be a $m \times n$ real matrix of rank $r$ with singular value decomposition $U\Sigma V^\top$ and $k$-truncated singular value decomposition $U_k \Sigma_k V_k^\top$. Let  $\|\,\|$ denote the spectral norm $\|\,\|_2$ or the Frobenius norm $\|\,\|_F$. Consider the optimization problem
 \begin{align}
    \text{\normalfont minimize} \,\,\|(I-M^+M)A\|^2
    \tag{P1}\label{eq:P1},
\end{align}
where the optimization variable $M$ is a $k\times m$ matrix such that
$k \leq m$.
\begin{enumerate}[wide, labelindent=0pt]
    \item If $k = n \leq m$, then $M = A^+$ solves the problem~\eqref{eq:P1} with error 0.
    \item If $k = m$, then any matrix $M$ with rank $m$ solves the problem~\eqref{eq:P1} with error 0.
    \item If $k \leq r < \min(m, n)$, then $M = U_k^\top$ minimizes problem~\eqref{eq:P1} with errors
    \begin{align}\label{eq:optimal_projection_errors}
        \min_{M, \rank{M} \leq k}\|(I-M^+M)A\|_2^2 = \sigma_{k+1}^2\quad\text{and}\quad\min_{M, \rank{M} \leq k}\|(I-M^+M)A\|^2_{F} = \sum_{i=k+1}^{\min(m,n)} \sigma_i^2,
    \end{align}
    which are equal to zero if $k = r$. 
\end{enumerate}
Similarly, let $B$ be a $\ell \times m$ real matrix of rank $r$ with singular value decomposition $LSR^\top$ and $k$-truncated singular value decomposition $L_k S_k R_k^\top$. Consider the optimization problem
 \begin{align}
    \text{\normalfont minimize} \,\,\|B(I-M^+M)\|^2
    \tag{P2}\label{eq:P2}
\end{align}
where, again, the optimization variable $M$ is a $k\times m$ matrix with
$k \leq m$. 
\begin{enumerate}[wide, labelindent=0pt]\addtocounter{enumi}{3}
    \item If $k = \ell \leq m$, then $M = B$ solves the problem~\eqref{eq:P2} with error 0.
    \item If $k = m$, then any matrix $M$ with rank $m$ solves the problem~\eqref{eq:P2} with error 0.
    \item If $k \leq r < \min(\ell, m)$, then $M = R_k^\top$ minimizes problem~\eqref{eq:P2} with errors
    \begin{align}
        \min_{M, \rank{M} \leq k}\|B(I-M^+M)\|_2^2 = \sigma_{k+1}^2 \quad\text{and}\quad\min_{M, \rank{M} \leq k}\|B(I-M^+M)\|^2_{F} = \sum_{i=k+1}^{\min(\ell, m)} \sigma_i^2,
    \end{align}
which are equal to zero if $k = r$. 
\end{enumerate}
\end{theorem}
\begin{proof}
We first consider problem~\eqref{eq:P1} and prove items 1--3.
\begin{enumerate}[wide, labelindent=0pt, itemsep=1em]

\item If $k = n \leq m$, then the dimensions of the matrices $M$ and $A^+$ coincide and one can choose $M = A^+$. 
Hence
\begin{equation*}
    \|(I-M^+M)A\|^2= \|A-AA^+A\|^2 = \|A - A\|^2= 0,
\end{equation*}
since $(A^+)^+ = A$ and $AA^+A = A$ by the defining properties of the Moore-Penrose pseudo-inverse \cite{Penrose1955}.

\item If $k = m$, then $M$ is square. Any rank $m$ matrix $M$ of dimension $m\times m$ is invertible, so $M^+ = M^{-1}$ and $I - M^{-1}M = 0$, which implies that $\|(I-M^+M)A\|^2=0$.

\item We first prove that 
\begin{equation*}
    \min_{M, \rank{M} \leq k}\|(I-M^+M)A\|^2 =   \min_{\substack{C\\ C=M^+MA, \rank{M} \leq k}}\|A-C\|^2.
\end{equation*}
Indeed, due to Sylvester's rank inequality \cite[Section 0.4.5 (c)]{Horn2013} and the inequality $\rank{M} \leq k$, 
\begin{equation*}
    \rank(M^+M)\leq \min\{\rank{M},\rank{M^+}\}\leq k,
\end{equation*}
which in turn implies that
\begin{equation*}
    \rank(M^+MA)\leq \min\{\rank(M^+M),r\} = \rank(M^+M) \leq k,
\end{equation*}
where the equality follows from $k\leq r$. Thus, 
\begin{equation}\label{eq:equiv_prob}
    \min_{M, \rank{M} \leq k}\|(I-M^+M)A\|^2 =   \min_{\substack{C\\ C=M^+MA,\, \rank{C} \leq k}}\|A-C\|^2.
\end{equation}
Let us now focus on the Frobenius norm. The new form of the problem in Eq.~\eqref{eq:equiv_prob} is compatible with Theorem~\ref{thm:SEYM}, but with the additional equality constraint that $C = M^+MA$, which directly implies the inequality
\begin{equation*}
    \min_{\substack{C\\ C=M^+MA,\, \rank{C} \leq k}}\|A-C\|^2 \geq \min_{\substack{C\\\rank{C} \leq k}}\|A-C\|^2 = \sum_{i=k+1}^{\min\{m,n\}} \sigma_i^2
\end{equation*}
or equivalently, from Eq.~\eqref{eq:equiv_prob},
\begin{equation}\label{eq:lowerbound_C}
    \min_{M, \rank{M} \leq k}\|(I-M^+M)A\|^2_F \geq \sum_{i=k+1}^{\min\{m,n\}} \sigma_i^2.
\end{equation}
Therefore, if we find a matrix $M$ that reaches the lower bound of inequality \eqref{eq:lowerbound_C}, then the minimization problem is solved. Below, we prove that $M=U_k^\top$ is such a solution.

The matrix $I-M^+M$ is an orthogonal projector (this is directly proven from the properties of the Moore-Penrose pseudoinverse) and therefore,
\begin{equation*}
    (I-M^+M)^\top(I-M^+M) = (I-M^+M)^2 = I-M^+M.
\end{equation*}
The cyclic property of the trace and the eigenvalue decomposition of $AA^\top$ from the SVD imply
\begin{align*}
    \|(I-M^+M)A\|^2_{F} = \tr\left[AA^\top(I-M^+M)\right] = \tr\left[U\Sigma^2U^\top(I-M^+M)\right].
\end{align*}
Let $M$ be equal to $U_k^\top$. Then, $M^+M = U_k U_k^\top$ and
\begin{align*}
    \|(I-M^+M)A\|^2_{F} = \tr\left[\sum_{i=1}^{\min(m,n)} \sigma_i^2 u_i u_i^\top - \sum_{i=1}^{\min(m,n)}\sum_{j=1}^k \sigma_i^2 u_i u_i^\top u_j u_j^\top\right].
\end{align*}
Since $u_i^\top u_j = \delta_{ij}$, we obtain
\begin{equation*}
    \sum_{i=1}^{\min(m,n)}\sum_{j=1}^k \sigma_i^2 u_i u_i^\top u_j u_j^\top = \sum_{i=1}^{\min(m,n)}\sum_{j=1}^k \sigma_i^2 u_i \delta_{ij} u_j^\top =  \sum_{j=1}^k \sigma_j^2 u_j u_j^\top
\end{equation*}
and thus,
\begin{align*}
    \|(I-M^+M)A\|^2_{F} &= \tr\left[\sum_{i=k+1}^{\min(m,n)} \sigma_i^2 u_i u_i^\top\right]
    = \sum_{i=k+1}^{\min(m,n)} \sigma_i^2 \tr\left[u_i u_i^\top\right]
    = \sum_{i=k+1}^{\min(m,n)} \sigma_i^2 .
\end{align*}
Hence, $B^* = U_kU_k^+A$ is a solution to the problem~\eqref{eq:P0}. If $k = r$, then $\|(I-M^+M)A\|^2_{F} = \sum_{i=r+1}^{\min(m,n)} \sigma_i^2 = 0$, because $\sigma_i = 0$ for all $i>r$. 

For the spectral norm, with $M = U_k^\top$, we have $(I-M^+M)A = (I-U_kU_k^\top) U\Sigma V^\top$ which is equal to
\begin{align*}
    \sum_{i=1}^{\min(m,n)}\sigma_i u_iv_i^\top -  \sum_{i=1}^k \sigma_i \sum_{j=1}^k u_ju_j^\top u_iv_i^\top -  \sum_{i=k+1}^{\min(m,n)}\sigma_i \sum_{j=1}^k u_ju_j^\top u_iv_i^\top = \sum_{i=k+1}^{\min(m,n)}\sigma_i u_iv_i^\top,
\end{align*}
where we have used $u_j^\top u_i = \delta_{ij}$ in the last two terms and the fact that $i$ is never equal $j$ in the last term. We conclude that $\|(I-M^+M)A\|_2 = \|\sum_{i=k+1}^{\min(m,n)}\sigma_i u_iv_i^\top\| = \sigma_{k+1}$ which shows that $M = U_k^\top$ minimizes the error in problem~\eqref{eq:P1}.

\smallskip

\end{enumerate}

The proofs of items 4--6 related to problem~\eqref{eq:P2} closely follow that of items 1--3.
\end{proof}
\noindent There is an interesting data science application for Theorem~\ref{thm:optimal_projection_general} as explained in the following example. 
\begin{example}\label{ex:optimal_projection_data_matrix_dynamical_system}
Let $X$ be a $m\times T$ data matrix where $m$ is the number of variables (features) and $T$ is the number of time steps (samples). Then, choosing $M = U_n^\top$ where $U_n = (u_1\,\,...\,\,u_n)$ with $u_\mu$ being the $\mu$-th left singular vector of the data matrix $X$ gives the minimal error to the optimization problem~\eqref{eq:P1} with $A = X$ and $d = T$. This particular example is related to the so-called proper orthogonal decomposition~\cite[p.278-279]{Antoulas2005}.
\end{example}

\subsection{Effective ranks}
\label{SIsubsec:effective_ranks}

In this section, we give more details about the different effective ranks presented in Table~\ref{tab:effective_ranks}.
\begin{itemize}
\item The \textit{stable rank}~\cite[Definition 7.6.7]{Vershynin2018}, also called numerical rank~\cite{Rudelson2007}, is defined as
\begin{equation}
    \mathrm{srank}(A) = \frac{\|A\|_F^2}{\|A\|_2^2} = \frac{\sum_{i=1}^r\sigma_i^2}{\sigma_1^2}\,.
\end{equation}
It thus measures the relative importance of the sum of the squared singular values with respect to the squared largest singular value. More colloquially, $\mathrm{srank}(A)$ compares the total energy of $A$ with the energy contained in the first component (first singular vectors) of $A$. Note that 
\begin{equation}\label{eq:norm_inequality_2F}
    \|A\|_2\leq\|A\|_F\leq\sqrt{r}\|A\|_2\,,
\end{equation}
where $r=\rank(A)$.
From the second inequality, we easily deduce the following upper bound: 
\begin{equation*}
    \mathrm{srank}(A) \leq r\,.
\end{equation*}
The stable rank is stable in the sense that it remains essentially unchanged under a small perturbation of the matrix $A$, contrarily to the rank \cite{Rudelson2007}. It is used in the design of fast (randomized) algorithms for low-rank approximations~\cite{Harvey2011, Cohen2016}. Because it also quantifies to what extent the elements of the matrix are gathered around the diagonal, the stable rank also measures the complexity of the connection patterns between the modules of a network~\cite{Desy2023}. 

\item The \textit{nuclear rank}~\cite[p.2183]{Kyrillidis2014} is defined as 
\begin{align}
    \mathrm{nrank}(A) = \frac{\|A\|_*}{\|A\|_2} = \frac{\sum_{i=1}^r\sigma_i}{\sigma_1}\,,
\end{align}
where $\|\,\|_*$ is the nuclear norm, also known as the trace norm or the Ky Fan norm. Similarly to the stable rank, it measures the relative importance of the sum of the singular values with respect to the largest singular value. The nuclear norm is upper-bounded such that
\begin{equation}
    \|A\|_*\leq \sqrt{r}\|A\|_F\leq r\|A\|_2\,,
\end{equation}
where we have used the first inequality of Eq.~\eqref{eq:norm_inequality_2F}. Therefore, we find that
\begin{equation}
    \mathrm{nrank}(A) \leq \sqrt{r}\,\,\mathrm{srank}^{1/2}(A) \leq r\,.
\end{equation}

\item The energy ratio, also called the cumulative explained variance, the reconstructed proportion, or the $R_v$ coefficient \cite{Abdi2010}, is 
\begin{equation}
    E(\ell) = \frac{\|A_\ell\|_F^2}{\|A\|_F^2} = \frac{\sum_{i=1}^\ell \sigma_i^2}{\sum_{j=1}^r \sigma_j^2}\,,
\end{equation}
where $A_\ell$ is the $\ell$-truncated SVD of $A$. The \textit{energy ratio} effective rank is
\begin{equation}
    \mathrm{energy}(A) = \min\Big(\argmax_{\ell\in\{1,...,N\}}\big(E(\ell) > \tau\big)\Big)\,,
\end{equation}
where $\tau \in (0, 1)$ is a threshold to be chosen. Note that this ``graph energy" differs (but is related) to the ones introduced in combinatorics by Gutman and Nikiforov that have applications in theoretical chemistry and spectral graph theory~\cite{Gutman2001, Nikiforov2007, Nica2018_SI}. 

\item Let the coordinate $(x_i, y_i)$ of the $i$-th singular values be given by
\begin{equation}
    x_i =\frac{i-1}{N-1}\quad \text{and} \quad y_i = \frac{\sigma_i - \sigma_N}{\sigma_1 - \sigma_N}
\end{equation}
for all $i\in\{1,...,N\}$, such that the largest singular value is at $(0,1)$ and the smallest singular value is at $(1,0)$. The distance between the line $L = \{(x, y)\,|\,x + y = 1\}$, passing through the largest and the smallest singular value, and the position $(x_i, y_i)$ of the $i$-th singular value is
\begin{equation}
    d_i = \frac{1}{\sqrt{2}}\,\left|x_i + y_i - 1\right|.
\end{equation}
The \textit{elbow position} is the largest distance between in $\{d_1,...,d_N\}$, i.e.,
\begin{equation}
    i_{\mathrm{elbow}} = \argmax_{i\in\{1,...,N\}}\,d_i.
\end{equation}
The \textit{elbow rank} is thus defined as the number of singular values above the position of the elbow, which is described by
\begin{equation}
    \mathrm{elbow}(A) = i_{\mathrm{elbow}} - 1 = \,\frac{1}{\sqrt{2}}\argmax_{i\in\{1,...,N\}}\,\left|\frac{i-1}{N-1} + \frac{\sigma_i - \sigma_N}{\sigma_1 - \sigma_N} - 1\right| - 1.
\end{equation}
This effective rank is often used as a rule of thumb to truncate the singular value distribution \cite{Shabalin2013, Gavish2014}. It is also named the ``scree'' or elbow test~\cite{Abdi2010} and may be computed in different ways than above~\cite{Ferre1995}.

\item Roy and Vetterli's effective rank \cite{Roy2007} or Cangelosi and Goriely's information dimension \cite{Cangelosi2007} is here called \textit{erank}. It is defined as
\begin{equation}
    \mathrm{erank}(A) = \exp\left[H(p_1,...,p_r)\right]
\end{equation}
where $H(p_1,...,p_r) = -\sum_{i=1}^r p_i \log p_i$ is the Shannon (spectral) entropy, measured in nat, as a function of the singular value mass function
\begin{equation}
    p_i =\frac{\sigma_i}{\|A\|_*} = \frac{\sigma_i}{\sum_{j=1}^r\,\sigma_j}, \quad \forall i \in \{1,...,r\}.
\end{equation}
Note that the square of the singular values could be used to define the singular value mass function, as in Ref.~\cite{Alter2000}. Among other interesting properties, the erank satisfies $1 \leq \mathrm{erank}(A) \leq r$ and it is naturally related to the minimum coefficient rate~\cite{Campbell1960} (see Ref.~\cite[Sec. 3]{Roy2007} for more details). Moreover, the maximum Shannon entropy is reached for a distribution of identical singular values ($\mathrm{erank}(A) = r \leq N$). Intuitively, this means that the erank measures the uniformity of the singular value distribution. For example, the erank of the singular values $(1,1,1,1,1,0,0,0,0,0)$ is 5, the erank of $(5, 1, 1, 1, 1, 1, 1, 1, 1, 1)$ is approximately 7.9, and the erank of $(30, 1, 1, 1, 1, 1, 1, 1, 1, 1)$ is approximately 2.8. 

\item 
Let $A = A_\ell + R$ where $A_\ell$ is a (deterministic) matrix of unknown rank $\ell$ and $R$ is some noise random matrix. Based on Ref.~\cite[Definition 4.2]{Perry2009} and Ref.~\cite{Gavish2014}, the \textit{optimal threshold} $\tau^*(A)$ is defined as
\begin{equation}
    \tau^*(A) = \argmin_\tau \|A_\ell - \hat{A}(\tau)\|,
\end{equation}
where $\hat{A}(\tau)$ is the $\tau$-truncated SVD of $A$ and $\|\cdot\|$ is some matrix norm (e.g., spectral norm, Frobenius norm). Intuitively, the problem of finding $\tau^*(A)$ is the problem of finding the singular values of the rank--$\ell$ matrix $A_\ell$ (signal matrix) by removing the ``noisy" singular values of $A$ due to $\gamma R$. When the level of noise is unknown, under some conditions on $R$, the optimal threshold
\begin{equation}\label{eq:gavish_donoho_unknown}
    \tau^*(A) = \frac{4\sigma_{\mathrm{med}}}{\sqrt{\,3\,\mu_{\mathrm{med}}}},
\end{equation}
minimizes the Frobenius norm $\|A_\ell - \hat{A}(\tau)\|_F$ in the limit of infinite matrices~\cite[Corollary 3 and Theorem 1]{Gavish2014}, where $\sigma_{\mathrm{med}}$ is the median of the observed singular value distribution of the weight matrix $A$ and $\mu_{\mathrm{med}}$ is the median of a Mar{\v{c}}enko-Pastur probability density function. 
The median $\mu_{\mathrm{med}}$ is generally unknown, but can be computed as explained in Ref.~\cite[p.5046]{Gavish2014}. These results, based on random matrix theory~\cite{Benaych-Georges2012}, are all rigorous in an asymptotic framework under specific conditions given in Ref.~\cite{Gavish2014}. We define $\mathrm{thrank}(A)$ has the number of singular values above the optimal singular value threshold $\tau^*(A)$, i.e.,
\begin{equation}
    \mathrm{thrank}(A) = \#\big\{\sigma_i\,\big|\,i\in\{1,...,N\}\text{ and }\sigma_i > \tau^*(A)\big\},
\end{equation}
where $\#$ is the cardinal of a set.

\item In a similar spirit as the optimal threshold, one can consider the optimal shrinkage of singular values \cite{Shabalin2013, Gavish2017, Leeb2022} to define an effective rank. Let $A = A_\ell + R$ where $A_\ell$ is a (deterministic) matrix of unknown rank $\ell$ and $R$ is some noise random matrix. 
Shortly, given the singular values of $A$, the scalar function $s: [0,\infty) \to [0,\infty), \,\, \sigma_i \mapsto s(\sigma_i)$ is called a \textit{shrinker} or a \textit{denoiser} of singular values. From Refs.~\cite{Shabalin2013, Gavish2017, Donoho2018, Leeb2022}, one can find analytically the optimal denoiser $s^*$ that minimizes different errors defined from the Frobenius norm, the spectral (operator) norm, or the nuclear norm. We define \textit{shrank} has the rank of the matrix with optimally shrinked singular values, i.e.,
\begin{equation}
    \mathrm{shrank}(A) = \#\big\{s^*(\sigma_i)\,\big|\,i\in\{1,...,N\}\text{ and }s^*(\sigma_i)>0\,\big\}.
\end{equation}
Note that this effective rank also depends on the median of the Mar{\v{c}}enko-Pastur distribution when the level of noise is unknown and estimated as in Ref.~\cite{Gavish2017}.

\end{itemize}

\begin{remark}
\phantom{henri}
\begin{itemize}
\item A simple criterion to determine whether a matrix is low rank can be formulated in terms of the minimal number of elements that are needed to fully describe the matrix by a rank decomposition. More precisely, a $N\times N$ matrix of rank $r$ can be defined to be of low rank if $2rN < N^2$ or identically, $r < N/2$.  
Similarly, a $N\times N$ matrix of effective rank $e$ can be defined to be of low effective rank if $e < N/2$. 
However, in the paper, we do not set one criterion to say that a matrix has a low (effective) rank: we rather compare different effective ranks of a graph with the actual rank and dimension of the corresponding matrix.
\item To compute the effective ranks exactly, the complete set of singular values is needed, which might not be possible to have for very large matrices (networks). Even if we did not use it in the paper, we acknowledge the fact that the singular values can be approximately obtained by using randomized SVD~\cite{Mahoney2011} (e.g., with sklearn.utils.extmath.randomized\_svd in \textit{Python}). Note also that the rank is computed from the singular values with a numerical tolerance of $10^{-13}$ in the paper. 

\item Among the all the effective ranks listed above, $\mathrm{erank}$, $\mathrm{nrank}$, and  $\mathrm{srank}$ distinguish themselves by their simple analytic formulation and their clear upper bounds. As we show below, they also enjoy a natural ordering.
\end{itemize}
\end{remark}

\begin{lemma}[Effective ranks ordering]\label{lem:ordering_ens_ranks}For any matrix $A$, 
    \begin{equation}
    \mathrm{srank}(A)\leq \mathrm{nrank}(A)\leq \mathrm{erank}(A)\leq \mathrm{rank}(A)\,.
\end{equation}
\end{lemma}
\begin{proof}
    Let $r=\mathrm{rank}(A)$. Then $A$ has exactly $r$ positive singular values: $\sigma_1\geq \sigma_2\geq \ldots\geq \sigma_r>0$. 

    To prove the first inequality, we recall that $\mathrm{srank}(A)=\sum_{i=1}^r(\sigma_i/\sigma_1)^2$ while $\mathrm{nrank}(A)=\sum_{i=1}^r\sigma_i/\sigma_1$. Now, $\sigma_i/\sigma_1\leq 1$ for all $i$, so $(\sigma_i/\sigma_1)^2\leq \sigma_i/\sigma_1$ for all $i$. Therefore, $\mathrm{srank}(A)\leq \mathrm{nrank}(A)$.

    The second and third inequalities both involve $\mathrm{erank}(A)$, which is defined as $e^{H}$, where $H$ is the entropy of the probability vector associated with the singular values of $A$, that is
    \begin{equation*}
        H=\sum_{i=1}^r\,p_i\ln\frac{1}{p_i},\qquad  p_i=\frac{\sigma_i}{\sum_{j=1}^r\sigma_j}.
    \end{equation*}
    Going back to the definition of  $\mathrm{nrank}$ allows us to write
    \begin{equation*}
        p_i=\frac{\sigma_i}{\sigma_1\sum_{j=1}^r\frac{\sigma_j}{\sigma_1}}=\frac{\sigma_i}{\sigma_1 \mathrm{nrank}(A)}.
    \end{equation*}
    Hence,
    \begin{align*}
        H=\sum_{i=1}^r \frac{\sigma_i}{\sigma_1 \mathrm{nrank}(A)}\ln\left(\frac{\sigma_1 \mathrm{nrank}(A)}{\sigma_i}\right)
     = \ln \big(\mathrm{nrank}(A)\big)+\frac{1}{\mathrm{nrank}(A)}\sum_{i=1}^r\frac{\sigma_i}{\sigma_1}\ln\frac{\sigma_1 }{\sigma_i}.
    \end{align*}
    Therefore, 
    \begin{equation}\label{eq:erank_expH}
        \mathrm{erank}(A)=e^{H}=\mathrm{nrank}(A)\cdot\exp\left(\frac{1}{\mathrm{nrank}(A)}\sum_{i=1}^r\frac{\sigma_i}{\sigma_1}\ln\frac{\sigma_1 }{\sigma_i}\right)\geq \mathrm{nrank}(A),
    \end{equation}
    as expected. Note that the equality holds only when $r=1$. Finally, using the concavity of the logarithm and Jensen's inequality, we deduce that
    \begin{equation*}
        H=\sum_{i=1}^r\,p_i\ln\frac{1}{p_i}\leq \ln\left(\sum_{i=1}^r\,p_i\frac{1}{p_i}\right)=\ln(r),
    \end{equation*}
    which readily implies the last inequality of the lemma. 
\end{proof}

\begin{remark}
   There is another type of effective rank, with the same form as the erank, that is related to the srank rather than the nrank (see the proof of the latter lemma). Indeed, we can define the ``stable erank'' as $\mathrm{serank}(A) = e^H$ with $p_i = \sigma_i^2/(\sum_{j=1}^r\sigma_j^2) = \sigma_i^2/(\sigma_1^2\mathrm{srank}(A))$. However, we do not use this effective rank in the paper.
\end{remark}

\noindent Regarding the optimal threshold and shrinkage, we also make the following remarks.
\begin{remark}
\phantom{henri}
\begin{itemize}
\item In the definition of the optimal threshold and shrinkage, it is assumed that the rank of the signal matrix is finite in the limit $N\to \infty$. 
\item If the type of noise is unknown, the assumptions of Refs.~\cite{Gavish2014, Gavish2017} do not necessarily hold and it is not guaranteed that the threshold and the shrinkage effective ranks are optimal.
\item In the GitHub repository  \href{https://github.com/VinceThi/low-rank-hypothesis-complex-systems}{low-rank-hypothesis-complex-systems}, module singular\_values/optimal\_shrinkage.py, we provide a Python translation (the first to our knowledge) of the Matlab script \href{http://purl.stanford.edu/kv623gt2817}{optimal\_shrinkage.m} from Ref.\cite{Gavish2017}. Moreover, we correct an error made in Ref.\cite{Gavish2017} concerning the optimal singular value shrinkage for operator norm loss with the Theorem 3.1 of W. Leeb \cite{Leeb2022}. We also merge and adapt for our purpose the Github repository \href{https://github.com/erichson/optht}{optht}, which is a Python implementation of the Matlab script \href{https://purl.stanford.edu/vg705qn9070}{optimal\_SVHT\_coef.m} \cite{Gavish2014}. Note that, when applied to a data matrix with unknown noise and a median smaller than the numerical zero (set to 1e-13), the optimal threshold and the optimal shrinkage effective ranks are computed for the singular values greater than 1e-13 only to ensure that the estimated noise is not zero.                
\end{itemize}
\end{remark}
We believe that the techniques that lead to the optimal threshold and shrinkage~\cite{Benaych-Georges2012, Shabalin2013, Gavish2014, Gavish2017, Leeb2022} will have a considerable impact on network science (besides, it already has an impact in neuroscience~\cite{Gao2015a}). Indeed, considering noisy networks is a long-standing challenge in network science (e.g., in sociology \cite{Killworth1976}) that has been addressed, for instance, with Bayesian inference~\cite{Peixoto2018, Newman2018nature, Young2020a, Young2021}. We think that there is still plenty of work to do to denoise or evaluate the level of noise of a network from its singular values. In particular, it would be interesting to find optimal singular value shrinkage functions~\cite{Gavish2017} with noise types that are more specific to real networks and random graphs.

All in all, we have gathered some important results on SVD. We will show how these results can be leveraged in network science, spectral graph theory, and dynamical systems.

\section{SVD in the study of complex systems}
\label{SIsec:apologiaSVD}

In this section, we present applications of SVD in the study of complex systems. 
First, we highlight the ubiquity of the low-rank hypothesis in random graph theory. Second, we present original theorems for the rapid decrease of the singular values in the directed soft configuration model and its weighted version. Third, inequalities and scaling behaviors for the effective ranks are deduced from different decreasing behaviors of the singular values. Fourth, we recall how SVD yields centrality measures for directed networks. Fifth, we discuss about preliminary results concerning the evolution of the effective rank in adaptive systems. Finally, we give a short overview of the use of SVD in dynamical systems. 

\subsection{SVD of random graphs}
\label{SIsubsec:svd_random_graph}
Random graphs and their eigenvalue spectrum have a long and rich history~\cite{Furedi1981, Bonacich1987, sompolinsky1988chaos, Chung1994, chung2003spectra, Dorogovtsev2003, VanMieghem2011, chung2011spectra, Nadakuditi2012, Peixoto2013, Castellano2017a, Nica2018_SI,  Newman2019, tang2022eigenvalues}, but less attention has been given to their singular value decomposition. Indeed, SVD is not mentioned in many of the main introductory textbooks of network science~\cite{Estrada2015, *Barabasi2016, *Latora2017, *Newman2018} or spectral graph theory~\cite{Cvetkovic1980, *Chung1994, *Nica2018_SI}. This phenomenon is somewhat expected, since both fields needed to develop their own set of tools, but we believe that SVD deserve much more attention. In the following, we present the low-rank formulation in a wide variety of models ranging from network science and random matrix theory to machine learning and neuroscience.

The adjacency or the weight matrix of a random network model can always be written as
\begin{equation}\label{eq:random_matrix}
    W = \langle W \rangle + R\,,
\end{equation}
where $R$ is a zero mean random matrix and $\langle W \rangle$ is the (deterministic) expected weight matrix. Typically, $\langle W \rangle$ depends upon a low-rank matrix $L$:
\begin{equation}\label{eq:Phi}
    \langle W \rangle = \Phi(L)\,,
\end{equation} 
where $\Phi$ a matrix-valued function of a matrix variable. In all the cases studied below, the $(i,j)$ element $\Phi(L)$ is equal to $\phi(L_{ij})$, with $\phi$ being a real scalar function of a real variable. To expose the low-rank formulation of $\langle W \rangle$, recall that there always exists a rank factorization 
\begin{equation}
    \langle W \rangle = LR^\top\,,
\end{equation}
where $L, R$ are $N\times r$ matrices and $r$ is the rank of $\langle W \rangle$. Another convenient form is the sum of rank one matrices 
\begin{equation}
    \langle W \rangle = \sum_{\mu, \nu = 1}^s \alpha_{\mu\nu}\bm{a}_{\mu}\bm{c}_{\nu}^{\top} = \sum_{\mu=1}^s \bm{a}_{\mu}\bm{b}_{\mu}^{\top}\,,
\end{equation}
where $\bm{a}_{\mu}$, $\bm{c}_{\mu}$ are $N\times 1$ vectors, $\alpha_{\mu\nu}$ is a real constant for all $\mu, \nu$, and $\bm{b}_{\mu} = \sum_{\nu=1}^s \alpha_{\mu\nu}\bm{c}_{\nu}$. Indeed, defining the $N\times s$ matrices $A = (\bm{a}_1,...,\bm{a}_s)^\top$ and $B = (\bm{b}_1,...,\bm{b}_s)^\top$ yields $\langle W \rangle = AB^\top$ and ensures that the rank of $W$ is at most $s$. In the next examples, we provide the details about the random graphs of Table~\ref{tab:random_graphs} in the Methods. 

\begin{example}[Network science---unweighted graphs]\label{ex:network_science} A large class of binary random graphs are described by Bernouilli random matrices, $R_{ij}$ being equal to either $-\langle W_{ij} \rangle$ or $1 -\langle W_{ij} \rangle$. The expected adjacency matrix for...
\begin{itemize}
    \item ...the \textit{$\mathcal{G}(N, p)$ model} \cite{Solomonoff1951, Gilbert1959, Erdos1960a} with self-loops is
\begin{equation}\label{eq:gnp}
    \langle W \rangle = L = Np\,\hat{\bm{1}}\hat{\bm{1}}^\top\,,
\end{equation}
where $Np\hat{\bm{1}}\hat{\bm{1}}^\top$ is the (exact) SVD of the mean adjacency matrix, which is a rank one matrix with singular value $Np$ and $N\times 1$ singular vectors $\hat{\bm{1}} = (1\,...\,1)^\top/\sqrt{N}$. The model is also called Poisson random graph, Erd\H{o}s-Rényi model \cite{Newman2003}, Bernouilli random matrix~\cite{Guionnet2021}, or spiked Wigner matrix~\cite{Perry2018}.

\item ...the \textit{stochastic block model} (SBM)~\cite{Holland1983, Young2018} with $q$ communities (generalization of $\mathcal{G}(N, p)$) is 
\begin{equation}
    \langle W \rangle = L = \sum_{\mu,\nu=1}^q \sqrt{n_{\mu}n_{\nu}}\,p_{\mu\nu}\,\bm{b}_{\mu}\bm{b}_{\nu}^\top\,,
\end{equation}
where $p_{\mu\nu}$ is the probability for a vertex in the $\mu$-th block of size $n_{\mu}$ to be connected to a vertex in the $\nu$-th block of size $n_\nu$ and $b_{\mu}$ is a block vector with $1/\sqrt{n_{\mu}}$ at the indices of the $\mu$-th block and zeros elsewhere.

\item ... the \textit{Chung-Lu model} \cite{Chung2002a, Chung2002pnas} is
\begin{equation}\label{eq:SI_ChungLu}
    \langle W \rangle = L = \frac{\|\bm\kappa\|^2}{2M}\hat{\bm{\kappa}}\hat{\bm{\kappa}}^\top\,,
\end{equation}
where $\kappa$ is a vector of expected degrees. Note that the annealed approximation, omnipresent in epidemiology \cite{Wang2017} (or for spin models \cite{Dorogovtsev2008a}), is thus a very strong low-rank hypothesis.

\item ... the \textit{metadegree model}~\cite{Valdano2019} is
\begin{equation}
    \langle W \rangle = L = \sum_{\mu,\nu=1}^r \Delta_{\mu\nu}\,\bm{v}_{\mu} \bm{v}_{\nu}^\top,
\end{equation}
where $(v_{\mu})_{\mu=1}^r$ are the $N$-dimensional vectors of metadegree and $\Delta$ is a $r \times r$ nonsingular matrix that contains the ``coefficients of mixing'' among metadegrees. In Ref.~\cite[p.2, 2nd column, 2nd paragraph]{Valdano2019}, a low-rank hypothesis is explicitly made as they assume that the rank of $V\Delta V^{\top}$ is much smaller than the size of the system. One must say, however, that the model is flexible and generalizes, in particular, the Chung-Lu model. The framework developed by the authors of Ref.~\cite{Valdano2019} offers a solid discussion and a strong theoretical ground to better understand, classify, and design random graphs in the future. Their work has inspired our preliminary thoughts on the low-rank hypothesis.

\item ... the directed $S^1$ model of \textit{random geometric networks} \cite{Krioukov2010, Allard2023} has elements
\begin{equation}
    \langle W_{ij} \rangle = \phi(L_{ij}) = \frac{1}{1 + L_{ij}^{\beta/2}} \,,
\end{equation}
where $\beta>0$ (inverse temperature of the Fermi-Dirac distribution). The elements of the matrix $L$ are defined as
\begin{equation}
    L_{ij} = \frac{R^2\theta_{ij}^2}{\mu^2 (\bm\kappa_{\mathrm{in}})_i^2(\bm\kappa_{\mathrm{out}})_j^2}\,,
\end{equation}
where $\mu$ and $R$ are positive constants, the latter representing the radius of the circle on which the vertices are distributed, $\theta_{ij}$ is the angular distance between the vertices $i$ and $j$ on the circle, and $(\bm\kappa_{\mathrm{in}})_i, (\bm\kappa_{\mathrm{out}})_i$ denote the $i$-th latent in- and out-degrees respectively (positive constants). To estimate the rank of the matrix $L$, it is more convenient to rewrite it using the Hadamard product:
\begin{equation}
    L = \frac{R^2}{\mu^2}\left(\bar{\bm\kappa}_{\mathrm{in}} \bar{\bm\kappa}_{\mathrm{out}}^{\top}\right)\circ \bar{\theta}\,,
\end{equation}
where 
\begin{equation}
     \bar{\bm\kappa}_{\mathrm{in}}=(\,1/(\bm\kappa_\mathrm{in})_i^2\,)_{i=1}^N,\qquad\bar{\bm\kappa}_{\mathrm{out}}=(\,1/(\bm\kappa_\mathrm{out})_i^2\,)_{i=1}^N,\qquad \bar\theta=(\,\theta_{ij}^2\,)_{i,j=1}^N\,.
\end{equation}
Clearly, $\bar{\bm\kappa}_{\mathrm{in}} {\bar{\bm\kappa}_{\mathrm{out}}}^\top$ is a rank-one matrix. Now, according to Ref.~\cite[Theorem 7]{Gower1985}, the rank of a distance matrix (with squared elements) is $D$, $D+1$, or $D+2$, where $D$ is the dimension of the manifold where the points are embedded. Here $D=1$, which means that the rank of $\bar\theta$ is at most $3$. Recalling the well-known inequality  $\rank(A\circ B)\leq \rank(A)\rank(B)$, we conclude that the matrix $L$ defining the expected adjacency matrix of the $S^1$ model has a rank of at most $3$. For the $S^D$ model~\cite{Desy2023}, one can proceed similarly to conclude that its expected adjacency matrix has a rank of at most $D + 2$.

\item ... the \textit{soft directed configuration model} also has elements following a Fermi-Dirac distribution such that
\begin{equation}\label{eq:SI_expectedW_SDCM}
    \langle W_{ij} \rangle = \phi(L_{ij}) = \frac{L_{ij}}{1 + L_{ij}}\,,\qquad L_{ij}=\alpha_i\beta_j
\end{equation}
for some positive parameters $\alpha_i,\beta_j$. Thus, $L = \bm{\alpha}\bm{\beta}^T$ is a rank-one matrix and $\bm{\alpha}$, $\bm{\beta}$ are positive vectors defined in subsection~\ref{SIsubsec:exponential_decrease}.

\item ... of the \textit{Barabási-Albert model} (BA) \cite{Barabasi1999}, a model of preferential attachment and a particular case of Price's model~\cite{Price1976}, does not possess an explicit formula of the form $ \langle W \rangle =\Phi(L)$. However, in Fig.~\ref{fig:svd_graphs_random_graphs}b, we show the singular values of the model for increasing values of $m$, the number of edges to which a new vertex is attached, and we observe rapid decreases.
\end{itemize}
\end{example}

\begin{example}[Network science---weighted graphs]\label{ex:network_science_weighted}
All the above-mentioned unweighted graph models can be generalized to include weights. 
\begin{itemize}
    \item The simplest procedure consists in posing
\begin{equation}\label{eq:get_weighted_graphs}
    W = A \circ \mathcal{W}\,,
\end{equation}
where $A$ and $\mathcal{W}$ are two independent $N\times N$ random matrices. The first matrix, $A$, corresponds to the Bernoulli matrices introduced in Example \ref{ex:network_science} that control the existence of edges, while $\mathcal{W}$ is a (possibility continuous) random matrix that only encodes the values of the weights. Due to the independence of $A$ and $\mathcal{W}$, the expected weight matrix factorizes as 
\begin{equation}\label{eq:expected_weighted_graphs}
    \langle \,W \,\rangle= \langle \,A \,\rangle\circ \langle \,\mathcal{W}\,\rangle\,.
\end{equation}
For instance, supposing that the elements of $A$ are i.i.d. such that $\mathrm{Prob}(A_{ij}=1)=p$, we get the \textit{weighted} $\mathcal{G}(N,p)$ \textit{model} satisfying 
\begin{equation}\label{eq:weighted_prob_w}
    \mathrm{Prob}(W_{ij}=0)=1-p,\qquad \mathrm{Prob}(W_{ij}\neq 0)=p,\qquad \langle W \rangle = p \langle \mathcal{W} \rangle, 
\end{equation}
meaning that the average matrix of weights completely determine the rank of $\langle W \rangle$.
If we additionally impose that all $W_{ij}$'s are i.d.d. with mean $w$, then we conclude that 
\begin{equation}\label{eq:gnp-weighted}
    \langle W \rangle =  Npw\,\hat{\bm{1}}\hat{\bm{1}}^\top\,,
\end{equation}
thus corresponding to a rank-one model that we denote $\mathcal{G}(N, p, w)$ in Table~\ref{tab:random_graphs}.

\item When $p=1$ in Eq.~\eqref{eq:weighted_prob_w}, all edges exist and one recovers the models of complete weighted graphs, such as the simplest form of the \textit{weighted stochastic block model} (WSBM) with $q$ communities (groups) \cite[Eq.~(2.3)]{aicher2015learning} (where $q = K$), for which the probability density function of the weights is 
\begin{equation}
    f_W(w)=\prod_{i,j=1}^N\frac{1}{\sqrt{2\pi \Sigma_{ij}^2}}e^{-\frac{(w_{ij}-M_{ij})^2}{2\Sigma_{i,j}^2}},
\end{equation}
where $M=(M_{ij})_{i,j=1}^N$ and $\Sigma=(\Sigma_{ij})_{i,j=1}^N$ are block matrices whose elements can only take a few different values encoded in the $q\times q$ matrices $\mu$ and $\sigma$ as
\begin{equation}
    M_{ij} = \mu_{g_ig_j}\qquad \text{and}\qquad\Sigma_{ij}=\sigma_{g_ig_j}.
\end{equation}
In the last equations, $g_i\in\{1,\ldots,q\}$ is the group label of vertex $i$. The expected weight matrix of this WSBM is simply
\begin{equation}
    \langle W \rangle =L = M = \sum_{\kappa,\nu=1}^q\sqrt{n_{\kappa}n_{\nu}}\,\mu_{\kappa\nu}\,\bm{b}_{\kappa}\bm{b}_{\nu}^\top\,,
\end{equation}
where $n_\kappa$ and $\bm{b}_{\kappa}$ are defined as in Example~\ref{ex:network_science} for the SBM. Since the vertices belonging to the same group produce the same rows in matrix $M$, we conclude that the rank of $L$ is at most $q$.

\item Random models as in Eq.~\eqref{eq:get_weighted_graphs} were also used to define weighted versions of the SBM \cite{ng2021weighted} and planted-partition model \cite{brandes2009structural}.
Moreover, they served to study synchronization in random weighted directed networks \cite{porfiri2008synchronization} as well as the spectral properties of neuronal networks with inhibition \cite{rajan2006eigenvalue} and the transitions to chaos of dilute random neuronal networks~\cite{kadmon2015transition}. Separating edges and their weights as in Eq.~\eqref{eq:get_weighted_graphs} is also common practice in random matrix theory when studying the spectral properties of random weighted directed graphs \cite{tao2008random, gotze2010circular, costello2010rank,wood2012universality,cook2017circular}.

\item Although variables $A$ and $\mathcal W$ are independent in Eq.~\eqref{eq:get_weighted_graphs}, variables $A$ and $W= A\circ \mathcal W$ depend on each other. Indeed, they have a non-factorizable joint probability density function (pdf) of the form
\begin{equation}\label{eq:pdf_AW}
    F_{A,W}(a,w) = \prod_{1\leq i<j\leq N}\Big(p_{ij}\delta(a_{ij}-1)f_{\mathcal{W}_{ij}}(w_{ij})+(1-p_{ij})\delta(a_{ij})\delta(w_{ij})\Big)\,,
\end{equation}
where $\delta$ denotes Dirac's delta distribution, $p_{ij}$ is the marginal probability for $A_{ij}=1$, and $f_{\mathcal{W}_{ij}}$ is the pdf of the independent variable $\mathcal{W}_{ij}$. Note that we have assumed that the graph is undirected for simplicity. We see from the previous equation that $\mathcal W_{ij}$ can interpreted as the random variable $W_{ij}$
given the existence of an edge from $j$ to $i$, i.e.,  $W_{ij}\,|\,A_{ij}=1$. Setting $\mu_{ij}=\langle \mathcal W_{ij}\rangle =\int wf_{\mathcal{W}_{ij}}(w)dw$ and returning to Eq.~\eqref{eq:expected_weighted_graphs}, we conclude that any model defined using Eq.~\eqref{eq:pdf_AW} also satisfies 
\begin{equation}\label{eq:expected_weighted_graphs_2nd}
    \langle \,W_{ij} \,\rangle= p_{ij}\mu_{ij}.
\end{equation}

A good example of such a model is the $S^1$ version of the \textit{weighted random geometric model} (WRGM)~\cite{Allard2017} with fixed hidden variables 
\begin{equation}
    \bm \theta\in [0,2\pi)^N,\qquad 
    \bm \kappa\in\mathbb{R}_+^N,\qquad 
    \bm \sigma\in\mathbb{R}_+^N,\qquad 
\end{equation}
and parameters 
\begin{equation} 
    \alpha\in[0,1),\qquad 
    \beta>1,\qquad 
    \mu>0,\qquad 
    \nu>0,\qquad
    R>0\,.
\end{equation}
One can prove that the pdf of this model is given by Eq.~\eqref{eq:pdf_AW} with 
\begin{equation}
    p_{ij}=\frac{1}{1+L_{ij}^\beta},\qquad L_{ij}=\frac{R\theta_{ij}}{\mu \kappa_i\kappa_j},\qquad f_{\mathcal{W}_{ij}}(w)=\frac{\nu\sigma_i\sigma_j}{ (\kappa_i\kappa_j)^{1-\alpha}(R \theta_{ij})^\alpha}f(\epsilon)\,,
\end{equation}
where $R$ stands for the radius of the circle on which the vertices are distributed, $\theta_{ij}$ is the angular distance between the vertices $i$ and $j$, respectively placed at angle $\theta_i$ and $\theta_j$ on the circle, and $\epsilon$ is an auxiliary random variable whose pdf is $f$ and whose mean is equal to 1.
The expected weighted matrix is thus
\begin{equation}\label{eq:weight_WRGM}
    \langle \,W_{ij} \,\rangle= \phi(L_{ij},M_{ij})=\frac{1}{L_{ij}^\alpha(1+L_{ij}^\beta)}\, M_{ij},\qquad M_{ij} = \frac{\nu\sigma_i\sigma_j}{\mu^\alpha\kappa_i\kappa_j}\,.
\end{equation}
The rank of the corresponding matrices $L$ and $M$ are at most 3 and equal to 1, respectively. 

\item All multigraphs can be interpreted as weighted graphs in which the weights can only take nonnegative integer values. One of the best known and most widely used examples of a random multigraph, and thus a weighted random graph, is the \textit{degree-corrected stochastic block model} (DCSBM)~\cite{karrer2011stochastic,peixoto2018nonparametric} with $q$ communities whose probability mass function is defined as 
\begin{equation}
    \mathrm{Prob}(W = w)=\prod_{1\leq i, j\leq N} \frac{ e^{-\alpha_{ij}}\alpha_{ij}^{w_{ij}} }{w_{ij}!}\,,
\end{equation}
that is, the random variable $W_{ij}$ follows a Poisson distribution with values 
$w_{ij}\in \mathbb{N}$ for all $i,j$ and parameter 
\begin{equation}
    \alpha_{ij} = \lambda_{g_ig_j}(\hat{\bm{\kappa}}_{\mathrm{in}})_i (\hat{\bm{\kappa}}_{\mathrm{out}})_j\,.
\end{equation}
In the last equation, $g_i\in\{1,\ldots,q\}$ denotes the group (community) label to which vertex $i$ belongs and $\lambda_{g_ig_j}$ is the expected number of edges from group $g_j$ to $g_i$. Moreover, $(\hat{\bm{\kappa}}_{\mathrm{in}})_i$ and $(\hat{\bm{\kappa}}_{\mathrm{out}})_i$ are group-normalized expected in- and out-degrees, i.e.,
\begin{align}\label{eq:group_normalization}
    \sum_{i=1}^N (\hat{\bm{\kappa}}_{\mathrm{in}})_i \delta_{g_i,\,\mu} = 1\,,\qquad
    \sum_{i=1}^N (\hat{\bm{\kappa}}_{\mathrm{out}})_i \delta_{g_i,\,\mu} = 1
\end{align}
for all $\mu \in\{1,\ldots,q\}$. The expected weight matrix is thus
\begin{equation}
    \langle W \rangle = L = \Lambda \circ (\hat{\bm{\kappa}}_{\mathrm{in}} \hat{\bm{\kappa}}_{\mathrm{out}}^\top),
\end{equation}
where $\Lambda=(\lambda_{g_ig_j})_{i,j=1}^N$ is a block matrix of rank at most $q$ while $\hat{\bm{\kappa}}$ is the vector of expected degrees. Now, according to the rank inequality for the Hadamard product,  
\begin{equation}
    \rank(L) = \rank\left(\Lambda \circ (\hat{\bm{\kappa}}_{\mathrm{in}} \hat{\bm{\kappa}}_{\mathrm{out}}^\top)\right)\leq \underbrace{\rank\left(\Lambda\right)}_{\leq q} \underbrace{\rank \left(\hat{\bm{\kappa}}_{\mathrm{in}} \hat{\bm{\kappa}}_{\mathrm{out}}^\top\right)}_{1}=q.
\end{equation}
Thus, including a ``rank-one correction" for better describing vertex degrees in SBM does not affect its low-rank property. 

\item ...the \textit{random dot product graph} \cite{Athreya2018} is
\begin{equation}
    \langle W \rangle = L = XX^\top\,,  
\end{equation}
where $X$ is a $N\times d$ matrix where the rows are the latent positions of each vertex of the graph. The rank of $\langle W \rangle$ is obviously less or equal to $d$ and from Ref.~\cite[p.14]{Athreya2018}: ``[...]in the RDPG case, where $\bm{P}$ [(i.e., $\langle W \rangle$)] is of low rank [...]''. The model generalizes the SBM, the degree-corrected SBM and the mixed-membership SBM as shown in Theorem~15 of Ref.~\cite{Athreya2018}, which is why we classify it in the weighted networks.

\item  The authors \cite{garlaschelli2009generalized,garlaschelli2009weighted} introduced yet another simple method for providing weights to many classical random network models. Using entropy maximization, they defined exponential families of random weighted graphs, such as those whose weights $w_{ij}$ are nonnegative integers and whose probability mass function is  \begin{equation}\label{eq:weighted_garlaschelli}
    \mathrm{Prob}(W = w)=\prod_{1\leq i<j\leq N}(1-y_iy_j)(y_iy_j)^{w_{ij}}
\end{equation}
where $y_i$ and $y_j$ satisfy $0 < y_i y_j < 1$ for all $i,j$ and are related in a nonlinear way to the expected strengths while controlling the probability of having an edge from $j$ to $i$. Indeed,
\begin{equation}
    \mathrm{Prob}(W_{i,j}=0)= 1-y_iy_j\qquad \text{and}\qquad\mathrm{Prob}(W_{i,j}\neq 0)= y_iy_j\,.
\end{equation}
We see that the random variable $W_{i,j}+1$ follows a geometric distribution of parameter $1-y_iy_j$. One easily shows that the elements of the expected weight matrix follow the Bose-Einstein distribution 
\begin{equation}
    \langle W_{ij} \rangle =\phi(L_{ij}) = \frac{L_{ij}}{1-L_{ij}}\,, \qquad L_{ij}=y_iy_j,
\end{equation}
meaning that the expected weight matrix $\langle W\rangle$ is a function of the rank-one matrix 
\begin{equation}
    L=\bm y\bm y^\top\,.
\end{equation}
The above random graph is somewhat analogous to the model presented in Eq.~\eqref{eq:SI_expectedW_SDCM}. For this reason, we call the weighted graph satisfying Eq.~\eqref{eq:weighted_garlaschelli} the \textit{weighted soft configuration model} (WSCM). 

\item The weighted random graph defined in Eq.~\eqref{eq:weighted_garlaschelli} can be related to the Chung-Lu model. Indeed,  supposing $y_iy_j\ll 1$, we get
\begin{equation}
    \langle W_{ij}\rangle = y_iy_j+ O\big((y_iy_j)^2\big)\,.
\end{equation}
Comparing with Eq.~\eqref{eq:SI_ChungLu}, we conclude that the limit case $\langle W \rangle = \bm y \bm y^\top$, where the rank of the expected weight matrix is exactly equal to one, corresponds to a \textit{weighed Chung-Lu model}, sometimes called weighted configuration model (cf. \cite[Eq.~5]{serrano2005weighted}). 

\item To define the \textit{weighted directed soft configuration model} (WDSCM), we need to modify Eq.~\eqref{eq:weighted_garlaschelli}. First, we introduce new positive parameters, say $\bar y_1,\ldots, \bar y_N$. Then we set 
\begin{equation}\label{eq:weighted_garlaschelli_directed}
    \mathrm{Prob}(W = w)=\prod_{1\leq i,j\leq N}(1-y_i\bar{y}_j)(y_i\bar{y}_j)^{w_{ij}},\qquad w_{ij}\in\mathbb{N},
\end{equation}
which means that $W_{ij}+1$ a random variable having a geometric distribution of parameter $1-y_i\bar{y}_j$. Therefore,
\begin{equation}\label{eq:expectedW_wdscm}
    \langle W_{ij} \rangle =\phi(L_{ij})=\frac{L_{ij}}{1-L_{ij}}\,, \qquad L_{ij}=y_i\bar{y}_j,\,,
\end{equation}
meaning that the edge directionality has no impact on the rank of the expected weight matrix.

\item The model defined by Eq.~\eqref{eq:weighted_garlaschelli} is a special case of the following random graph, also introduced in \cite{garlaschelli2009generalized}:
\begin{equation}\label{eq:weighted_garlaschelli_xy}
    \mathrm{Prob}(W = w)=\prod_{1\leq i<j\leq N}\frac{(x_ix_j)^{\Theta(w_{ij})}(1-y_iy_j)(y_iy_j)^{w_{ij}}}{1-y_iy_j+x_ix_jy_iy_j},\qquad w_{ij}\in\mathbb{N},
\end{equation}
where $\Theta$ denotes the Heaviside function and the $x_i$'s are positive parameters. We call the corresponding random graph the \textit{general weighted soft configuration model} (GWSCM). For this general model, the expected weight matrix depends upon two rank-one matrices, $L$ and $M$, and its elements follow a Bose-Fermi distribution~\cite{garlaschelli2009generalized}. Explicitly, 
\begin{equation}
    \langle \,W_{ij} \,\rangle =\phi(L_{ij}, M_{ij})\,,
\end{equation}
where 
\begin{equation}
    \phi(L_{ij}, M_{ij}) = \frac{L_{ij}M_{ij}}{(1-L_{ij}+L_{ij}M_{ij})(1-L_{ij})}\,, \qquad L_{ij}=y_iy_j,\qquad M_{ij}=x_ix_j\,.
\end{equation}

\item Eqs.~\eqref{eq:weighted_garlaschelli_xy} can be easily generalized to include directed edges as follows:
\begin{equation}\label{eq:weighted_garlaschelli_xy_directed}
    \mathrm{Prob}(W = w)=\prod_{1\leq i,j\leq N}\frac{(x_i\bar x_j)^{\Theta(w_{ij})}(1-y_i\bar y_j)(y_i\bar y_j)^{w_{ij}}}{1-y_i\bar y_j+x_i\bar x_jy_i\bar y_j},\qquad w_{ij}\in\mathbb{N},
\end{equation}
where all parameters are positive. We call the random graph whose weight matrix satisfies \eqref{eq:weighted_garlaschelli_xy_directed} the \textit{general weighted directed soft configuration model} (GWDSCM). Its expected weight matrix is 
\begin{equation}
    \langle \,W_{ij} \,\rangle =\phi(L_{ij}, M_{ij})\,,
\end{equation}
where 
\begin{equation}
    \phi(L_{ij}, M_{ij}) = \frac{L_{ij}M_{ij}}{(1-L_{ij}+L_{ij}M_{ij})(1-L_{ij})}\,, \qquad L_{ij}=y_i\bar y_j,\qquad M_{ij}=x_i\bar x_j\,.
\end{equation}
\end{itemize}

\end{example}

\begin{figure}[ht]
    \includegraphics[width=0.7\linewidth]{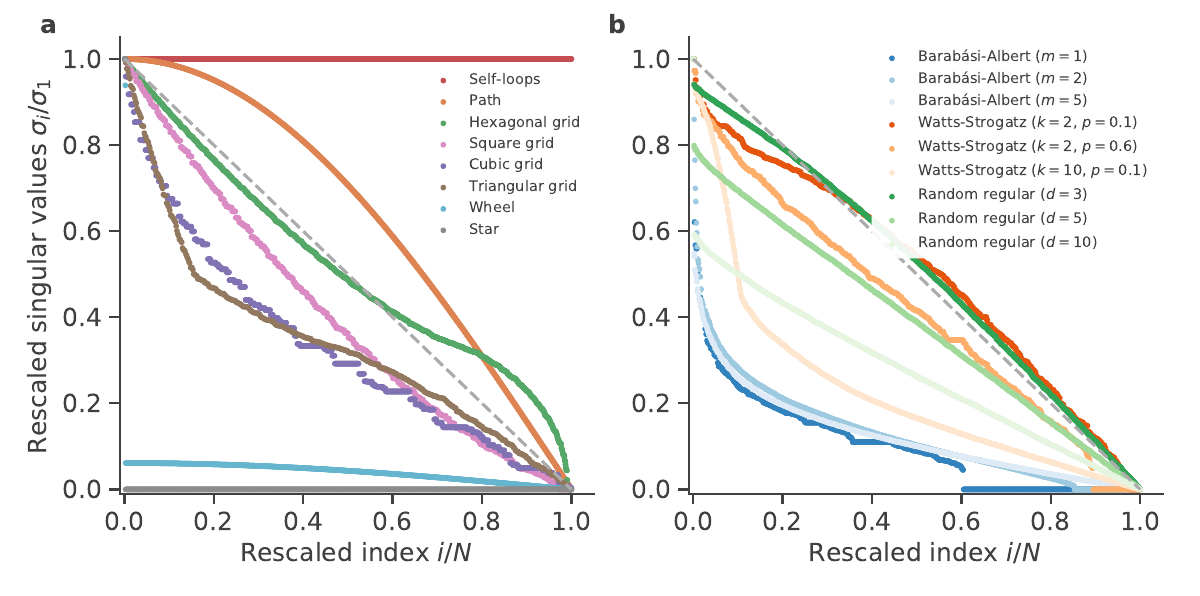}
    \vspace{-0.5cm}
    \caption{Rescaled singular values for \textbf{a} typical (non-random) graphs and \textbf{b} instances of three random graphs with different parameters. For the Barabási-Albert model, $m$ is the number of edges to which a new vertex is attached. For the (connected) Watts-Strogatz model, $k$ is the number of neighbors of each vertex in a ring structure and $p$ is the probability of rewiring an edge. For the random regular model, $d$ is the degree of each vertex. The graphs and random graphs are all available on networkx (see plots/plot\_fig\_SI\_singular\_values\_scree\_graphs.py on the Github repository), except the disconnected self-loop graph whose adjacency matrix is simply the identity matrix. The dashed gray lines are shown for the sake of visualization.
    }
    \label{fig:svd_graphs_random_graphs}
\end{figure}

There are also counter-examples of network model with high effective ranks. The most obvious examples are perhaps the Watts-Strogatz model, which had a considerable impact in the network science, and some non-random graphs.
\begin{example}
    The \textit{Watts-Strogatz model}~\cite{Watts1998} is described by the random matrix
    \begin{equation}
        W = D_k + R
    \end{equation}
    where $D_k$ is a band matrix of bandwidth $k$ whose $k$ up- and sub- diagonal entries are equal to 1 while $R$ is a matrix with -1's and 1's for each site that has been rewired with probability $p$. This is a perfect counter-example of the affirmation ``popular random network models are low-rank": the model is a sum of a high-rank matrix and a noise matrix (first indicator). Figure~\ref{fig:svd_graphs_random_graphs}b also shows that the singular values in the model can decrease linearly and even supralinearly for some parameters (as in the \textit{random regular graph}). These observations confirm that even though the Watt-Strogatz satisfy two interesting properties, namely small distances (small-world property) and a high-clustering coefficient, it doesn't generally enjoy the low-effective-rank property that we observe in real networks. 
\end{example}
\begin{example}[Graph theory]
    Although we often discuss about the rapid decrease of singular values in graphs, they can in fact have very different singular value distributions in general. As an intuitive example, we gather simple graphs, common in physics (e.g., path, grids), and illustrate how their singular values decrease differently, from supralinear to sublinear, in Fig.~\ref{fig:svd_graphs_random_graphs}a.
\end{example}
One can also find clear examples of the low-rank models in physics, machine learning, and neuroscience that are worth mentioning.
\begin{example}[Random matrix theory and spin glasses]
The typical random matrix ensembles used in physics (e.g., \textit{Gaussian Orthogonal Ensemble} \cite{Forrester2010}) are matrix models involving normally (Gaussian) distributed random variables and such that $\langle W\rangle= 0$, so they have a rank equal to zero. 
A counter-example is the \textit{Circular Unitary Ensemble} that is of full rank with all the singular values being 1. The random matrix $J$ encoding the interactions in the classical \textit{Sherrington-Kirkpatrick spin-glass model} \cite{Sherrington1975} is formed by i.i.d.~Gaussian variables of mean $J_0$, which implies that $\langle J\rangle = J_0\bm 1\bm 1^\top$, meaning that the effective rank of the model is one. Other well-known random matrix models, such as the \textit{Gaussian ensembles with finite-rank external source} \cite{Desrosiers2006} or \textit{spiked random matrices} \cite{Bloemendal2013}, satisfy Eq.~\eqref{eq:random_matrix} with $\langle W \rangle$ of rank $r=O(1)$, i.e., they have a low-rank formulation in the limit $N\to\infty$.  
\end{example}

\begin{example}[Machine learning and neuroscience]
\phantom{henri}
\begin{itemize}
\item In the \textit{Hopfield network} \cite{Hopfield1982}, one of the most influential models of artificial recurrent neural network, the weight matrix describing the connections between $N$ dynamical binary units is trained to memorize $n\ll N$ state vectors $\bm{v}_s\in\{0,1\}^N$. Starting with an initial random symmetric weight matrix of mean zero, $T$, the training consists in mapping $T\mapsto T+\sum_{s=1}^n\bm{v}_s \bm{v}_s^\top$, resulting in a final weight matrix of effective rank $\leq n$.
\item In an \textit{echo-state network}~\cite{Lukosevicius2009}, a random weighted directed graph of mean weight zero is used to generate a reservoir, which is the hidden recurrent part of the artificial neural network that is not affected by learning. The reservoir thus has a rank of zero (according to the first indicator, about the rank of the expected matrix defined in the paper). 
\item Training \textit{shallow undercomplete autoencoders} is essentially a low-rank approximation problem \cite{Bourlard2022}. The architecture formed by encoding/decoding weight matrices and a hidden layer thus form a low-rank model in itself. 
\item The \textit{chaotic random neural network} of Sompolinsky et al. \cite{sompolinsky1988chaos} is defined using random i.i.d. synaptic weights of mean 0 and variance $J^2/N$, thus corresponding to an expected weight matrix of rank 0. This model was later used to set the appropriate initial weights for training RNNs \cite{sussillo2009generating}.  It has been generalized to include $P$ distinct neuron populations, leading to a rank-$P$ expected weight matrix \cite{kadmon2015transition}
\item The synaptic weight matrix in the \textit{Rajan-Abbott random neural network} \cite{rajan2006eigenvalue} is given by the equation
\begin{equation*}
    W = \langle \,W \,\rangle + J\,,
\end{equation*}
where $J$ is a $N\times N$ random matrix whose elements are i.i.d. of mean $0$ and variance $1/N$, while 
\begin{equation*}
     \langle \,W \,\rangle  = \frac{1}{\sqrt{N}}\bm 1 {\bm v}^\top,\qquad {\bm v}^\top=(\mu_E,\ldots,\mu_E,\mu_I,\ldots,\mu_I)\,,
\end{equation*}
where $f$  denotes the fraction of excitatory neurons, whose mean synaptic weight is $\mu_E/\sqrt{N}>0$, while  $(1-f)$ denotes the fraction of inhibitory neurons, whose mean synaptic weight is $\mu_I/\sqrt{N}<0$. The rank of the expected weight matrix is thus equal to 1.
\item Another example is the \textit{Gaussian mixture low-rank network} \cite{Mastrogiuseppe2018, Schuessler2020_PRR, Beiran2021, Schuessler2020} whose weight matrix is defined as
\begin{equation*}
    W = \sum_{r=1}^R \bm{m}_r \bm{n}_r^\top + R,
\end{equation*}
where $R$ is a zero-mean Gaussian random matrix while $\bm{m}_r$ and $\bm{n}_r$ respectively denote the $r$-th left singular vector of $\langle W \rangle$ and its right singular vector multiplied by the $r$-th singular value. The rank of $\langle W \rangle$ is thus equal to $R$.
In Ref.~\cite{Beiran2021}, for example, the low-rank hypothesis is explicitly made: ``We restrict the connectivity matrix to be of low rank, i.e., the number of nonzero singular values of the matrix $J$ is $R\ll N$.'' In Table~\ref{tab:random_graphs}, we call this random graph  the ``\textit{rank-perturbed Gaussian model}'' (RPG) and we absorb the factor $1/N$ in $\bm{m}_\mu$ and $\bm{n}_\mu$. Moreover, in Ref.\cite{Mastrogiuseppe2018}, ``our theory suggests a simple conjecture: the low-dimensional structure in connectivity determines low-dimensional dynamics and computational properties of recurrent networks.'' Our paper proves partly this conjecture: by Corollary~\ref{cor:exact_rank}, if $\rank(W) \ll N$ (low-dimensional structure), then the recurrent neural dynamics evolves in a $\rank(W)$-dimensional space (low-dimensional dynamics) (more explanations in Example~\ref{ex:exact_rnn}).

\item It has also been observed experimentally that trained models have a low effective rank (the ones from NWS \cite{Eilertsen2020} in this paper and other references, such as Ref.~\cite{Martin2021a}).
\end{itemize}
\end{example}

To illustrate the three indicators of the low-rank hypothesis, we use four of the random graphs introduced above, that is RPG, DCSBM, $S^1$ RGM, and WDSCM, and present the results in Fig.~\ref{fig:exposing_low_rank_hypothesis} for $N = 10^3$ vertices. Below, we list the parameters used in each model to generate the figure. We denote a $N$-dimensional realization of a truncated Pareto density $\bm{x} = p(N, x_{\mathrm{\min}}, x_{\mathrm{\max}}, \gamma)$, where $N = 10^3$ is the number of vertices, $x_{\mathrm{\min}}$ is the minimum value of the distribution, $x_{\mathrm{\max}}$ is the maximum value, and $\gamma$ is the shape parameter. Similarly, an $N$-dimensional instance of a Gaussian density is denoted $n(N, m, v)$, where $m$ is the mean of the Gaussian and $v$ is its variance, and $u(N, x_{\mathrm{\min}}, x_{\mathrm{\max}})$ is an $N$-dimensional instance of discrete uniform distribution. 
\begin{itemize}
    \item For RPG, we set the rank to 5, $\bm{m}_\mu = n(N, 0, 1/N)$, and $\bm{n}_{\mu} = n(N, (\mu - 1)/10, \mu/10)$ for all $\mu\in\{1, 2, 3, 4, 5\}$. The $N\times N$ random noise matrix $R$ has elements following a Gaussian density with mean 0 and variance $g^2/N$ where we call $g$ the noise strength. We set $g = 1$ in Fig.~\ref{fig:exposing_low_rank_hypothesis}b and $g = 3$ in Fig.~\ref{fig:exposing_low_rank_hypothesis}f. The parameter $g$ is tuned from 0.01 to 4 in Fig.~\ref{fig:exposing_low_rank_hypothesis}j to increase $\overline{\|R\|}_2$.
    \item For DCSBM, we set the number of blocks to 5 with respective sizes $N/10$, $2N/5$, $N/10$, $N/5$, and $N/5$. The expected degree distributions are $\bm{\kappa}_{\mathrm{in}} = p(N, 2, 100, 2.5)$ and $\bm{\kappa}_{\mathrm{out}} = p(N, 1, 50, 2)$ which are then normalized by groups as defined in Eq.~\eqref{eq:group_normalization}. The expected number of edges is
    \begin{equation*}
        E = gN\begin{pmatrix}
                  0.40 & 0.10 & 0.10 & 0.02 & 0.13\\  
                  0.05 & 0.80 & 0.02 & 0.09 & 0.10\\  
                  0.02 & 0.02 & 0.30 & 0.05 & 0.02\\  
                  0.10 & 0.05 & 0.05 & 0.40 & 0.01\\  
                  0.10 & 0.09 & 0.05 & 0.05 & 0.30 
              \end{pmatrix}\,,
    \end{equation*}
    where $g$ is a multiplicative factor that allows increasing $\overline{\|R\|}_2$. We set $g = 100$ in Fig.~\ref{fig:exposing_low_rank_hypothesis}c and $g = 10$ in Fig.~\ref{fig:exposing_low_rank_hypothesis}g. The parameter $g$ is tuned from 200 to 6 to increase $\overline{\|R\|}_2$ in Fig.~\ref{fig:exposing_low_rank_hypothesis}k.
    \item For $S^1$ RGM, we set $R=N/2\pi$. To get the angular distance matrix, we set $t_i = 2\pi/[u(N,1, 50)]_i$ and then compute $\theta_{ij} = \pi - |\pi - |t_i - t_j||$. We observed that the numerical computation of the singular values (and, hence, the rank) of $\theta$ is particularly sensitive to the choice of angular matrix and taking a discrete uniform distribution reduced the sensitivity (see tests/test\_graphs/test\_generate\_s1\_random\_graph and the function ``test\_thetaij\_rank'' in the Github repository). We also have $\mu = \beta\sin(\pi/\beta)/(2\pi\langle\bm\kappa\rangle))$ where $\langle\,\rangle$ is the arithmetic mean and $\langle\bm\kappa\rangle = \langle\bm\kappa_{\mathrm{in}}\rangle = \langle\bm\kappa_{\mathrm{out}}\rangle$. The expected degree distributions are $\bm{\kappa}_{\mathrm{in}} = p(N, 2, 100, 2.5)$, $\bm{\kappa}_{\mathrm{out}} = p(N, 1, 50, 2)$, and then $\bm{\kappa}_{\mathrm{out}}$ is redefined to $\bm{\kappa}_{\mathrm{out}} -\langle\bm\kappa_{\mathrm{out}}\rangle\bm{1} + \langle\bm\kappa_{\mathrm{in}}\rangle\bm{1}$ [$\langle\bm\kappa_{\mathrm{out}}\rangle< \langle\bm\kappa_{\mathrm{in}}\rangle$] to ensure that $\langle\bm\kappa_{\mathrm{in}}\rangle = \langle\bm\kappa_{\mathrm{out}}\rangle$. The temperature $1/\beta$ allows increasing $\overline{\|R\|}_2$ and we tuned it from 0.01 to 0.96 in Fig.~\ref{fig:exposing_low_rank_hypothesis}l. The temperature is 0.2 in Fig.~\ref{fig:exposing_low_rank_hypothesis}d and 0.95 in Fig.~\ref{fig:exposing_low_rank_hypothesis}h.
    \item For WDSCM, we set $\bm{y} = p(N, x_{\mathrm{min}}, 0.8, 2.5)$ and $\bar{\bm{y}} = p(N, x_{\mathrm{min}}, 0.7, 3)$ where $y_{\mathrm{min}} = \bar{y}_{\mathrm{min}} = x_{\mathrm{min}}$ is the parameter that we tune to increase $\overline{\|R\|}_2$. We set $x_{\mathrm{min}} = 0.6$ in Fig.~\ref{fig:exposing_low_rank_hypothesis}e and $x_{\mathrm{min}} = 0.15$ in Fig.~\ref{fig:exposing_low_rank_hypothesis}i. To get Fig.~\ref{fig:exposing_low_rank_hypothesis}m, we increase $x_{\mathrm{min}}$ from 0.1 to 0.65.
\end{itemize}

In the section ``Evidence of the hypothesis for network models'' of the paper, we discuss how one can give a more precise perspective for spiked random matrices, such as RPG. Indeed, the singular values of spiked random matrices have a ``bulk'' related to the singular values of the noise matrix $R$ and the presence of outlying singular values is asymptotically characterized by the Baik-Ben Arous-Péché (BBP) phase transition \cite{Baik2005}. Notably, the appearance of $p \leq r$ singular values outliers in $W$ only depends upon a threshold $\bar\sigma$ on $\sigma_1(\langle W \rangle), ..., \sigma_r(\langle W \rangle)$~\cite{Benaych-Georges2012}. The simplest case is presented below.
\begin{example}
    Consider that the elements of $R$ are i.i.d.~Gaussian white noise of variance $1/N$, then as $N\to \infty$, the singular values of $R$ tend to densely fill the interval $[0,2]$, the threshold becomes $\bar\sigma=1$, and the $i$-th singular value of $W$ moves away from the bulk $[0,2]$ to reach $\sigma_i(\langle W \rangle)+1/\sigma_i(\langle W \rangle) $ whenever $\sigma_i(\langle W \rangle)>\bar \sigma$ for $i\in\{1, ..., r\}$.
\end{example}
Despite the clear dependence of $\langle W \rangle$ over a low-rank matrix $L$, it is not always clear whether $\langle W \rangle$ has an effective low-rank. This is the case of \textit{soft configuration models} (see Example~\ref{ex:network_science}) for which the expected adjacency matrix does not have the explicit form of a rank factorization. In the next section, we introduce the directed soft configuration model as a maximally entropic random graph. Then, we demonstrate that its singular values decrease exponentially rapidly.

\subsection{Exponential decrease of singular values in directed soft configuration models}
\label{SIsubsec:exponential_decrease}
In general, we only have partial information on complex networks. It is thus reasonable to define a set of networks where each network have a probability to describe the observed complex network. In order to do that in the least biased way, one can rely on the maximization of Shannon entropy to extract an adequate probability distribution \cite{Jaynes1957}. A lot of random graphs are defined from a maximally entropic model and although there is a large literature on the subject~\cite{Park2004, Bianconi2009, squartini2017maximum,cimini2019statistical}, we provide, for the sake of completeness, some important results and comments. We will later use them to demonstrate Theorem~\ref{thm:upper_bound_singvals_scm} and Theorem~\ref{thm:upper_bound_singvals_wdscm} on the exponential decrease of singular values in the directed soft configuration model and its weighted version, both maximally entropic random network models.

We begin by presenting general theorems about the use of Lagrange multipliers to obtain maximally entropic network models. Of course, the idea of Lagrange multipliers is old \cite{Caratheodory1937,Giorgi2014}. It goes back to Lagrange and even Euler, but in both of their work, the conditions in which the method applies are not clearly stated and no rigorous demonstration was provided. The first author who clearly stated the theorem is most likely Carathéodory, in the first German edition of his volume on the calculus of variations in 1935 \cite[186 and 187]{Caratheodory1989}.
\begin{theorem}[Lagrange multipliers]\label{thm:lagrange}\hfill\\
Let:
\begin{enumerate}[leftmargin=8ex]
\item $U$, be an open set in $\mathbb{R}^N$;
\item $f, g_1, \ldots, g_r$, be continuously differentiable real functions on $U$;
\item $E$, be a set such that $\bm x\in E$ iff $\bm x\in U$ and 
\(g_1(\bm x) = \cdots = g_r(\bm x) = 0\).
\end{enumerate}
If $\bm x^*$ maximizes or minimizes $f$ on $E$, then there exists a real vector $\bm \lambda = (\lambda_0,\ldots, \lambda_r)$ 
such that:
\begin{enumerate}[leftmargin=8ex]
\item $\bm \lambda \neq \bm 0$;
\item $\lambda_0\geq 0$;
\item $\displaystyle\lambda_0 \nabla f(\bm x^*)+\sum_{i=1}^r\lambda_i\nabla g_i (\bm x^*)=\bm 0$.
\end{enumerate}
Moreover, if $\nabla g_1 (\bm x^*), \ldots, \nabla g_r (\bm x^*) $ are linearly independent, then $\lambda_0>0$ and 
\begin{equation}\label{eq:lagrange}
    \nabla f(\bm x^*)+\sum_{i=1}^r\lambda_i^*\nabla g_i (\bm x^*)=\bm 0
\end{equation} 
for some nonzero vector $\bm \lambda^*=(\lambda_1^*,\ldots, \lambda_r^*)$ in $\mathbb{R}^{r}$.
\end{theorem}
\begin{proof}
The proof is long and often based on the local inversion theorem. See \textit{Carathéodory multiplicative rule} in \cite{Pourciau1980} or \cite[Theorem 20.3]{Chong2013}.
\end{proof}
\begin{remark}
On the one hand, the theorem is valid for a minimum or a maximum. This is an advantage that can turn out to be an inconvenience if we do not verify the nature of the point $\bm x^*$. On the other hand, Eq.~\eqref{eq:lagrange} is only a \textit{necessary condition} and it is not sufficient in general. We could, for example, find a solution of Eq.~\eqref{eq:lagrange} that does not correspond to the desired extremum. Moreover, the theorem supposes that there exists an extremum in $E$. If this is not assumed, one has to consider an open domain of $\mathbb{R}^N$, which excludes, for example, the compact domain $\bar D=[0,1]^N$. Finally, the gradients of the constraints must be linearly independent; otherwise $\lambda_0$ can be 0 and that does not help to find the extremum.
\end{remark}
 
The Lagrange multiplier method begins by solving Eq.~\eqref{eq:lagrange} by expressing all $x_i^*$ in terms of the multipliers $\lambda_i$. Then, the multipliers are written in terms of the known variables by solving the set of constraints $g_1(\bm x) = \cdots = g_n(\bm x) = 0$, which is generally the most difficult step. We finally verify that the solution $\bm x^*$ corresponds to the desired extremum. The following theorem illustrates how the first part of the method can be applied to find the necessary form for the probability mass function $P$ that maximizes the network entropy under (soft) structural constraints.
\begin{theorem}
\label{thm:probability_ergm}
Let $A$ be a $N\times N$ random adjacency matrix with support $\Omega_A$ that satisfies the soft equality constraints
\begin{equation}
    \mathbb{E}\left[h_{\mu}(A)\right] = 
    h_{\mu}(a^*), \quad \mu\in\{1,...,\ell\},
    \label{eq:soft_equality_constraint}
\end{equation}
where $\mathbb{E}$ is the expected value on $\Omega_A$, $a^*$ is some $N\times N$ non-random adjacency matrix, and each
$h_{\mu}: \{0,1\}^{N^2} \to \mathbb{R}^{\ell}$ is continuously differentiable. Then, the probability mass function $P$ that maximizes the entropy of $A$ under the equality constraints~\eqref{eq:soft_equality_constraint} must be of the form
\begin{equation}\label{eq:prob_ergm}
    P(a) = \frac{1}{Z(\lambda_1, ..., \lambda_\ell)}\,\exp\left[\sum_{\mu=1}^{\ell}\lambda_{\mu}h_{\mu}(a)\right],
\end{equation}
where $Z:\mathbb{R}^\ell\to \mathbb{R}$ is the partition function and $\lambda_{\mu} \neq 0$ for all $\mu$.
\end{theorem}

We can now provide the mathematical steps to show the rapid decrease of singular values in the directed soft configuration model. The next corollary is well known from Ref.~\cite{Park2004}.
\begin{corollary}
\label{cor:probability_sdcm}
    Let  $\bm k^{\mathrm{in}}$ and $\bm k^{\mathrm{out}}$ be two vectors with elements in  $\{1,2,\ldots, N\}$. Let $A$ be the adjacency matrix of a random directed graph of $N$ vertices, i.e., a random matrix of dimension $N\times N$ and support $\{0,1\}^{N\times N}$. Assume, moreover, that the following constraints are satisfied: 
    \begin{equation}\label{eq:directed_degree_constraints}
        \mathbb{E}[A\bm 1] = \bm k^{\mathrm{in}}\qquad \text{and}\qquad \mathbb{E}[A^\top\bm 1] = \bm k^{\mathrm{out}}\,.
    \end{equation}
    Then the probability mass function $P$ that maximizes the entropy of $A$ must be of the form
    \begin{equation}\label{eq:maxent_prob}
        P(a)=\prod_{i,j=1}^N{p_{ij}}^{a_{ij}}(1-p_{ij})^{1-a_{ij}},\qquad p_{ij}=\frac{\alpha_i\beta_j}{1+\alpha_i \beta_j},
    \end{equation}
    where $\alpha_i$ and $\beta_j$ are positive numbers $\forall\;i,j\in\{1,\ldots,N\}$. 
\end{corollary}
\begin{remark}
The $2N$ scalars $(\alpha_1,...,\alpha_N, \beta_1,...,\beta_N)$ are such that $\alpha_i = e^{\lambda_i}$ and $\beta_j = e^{\lambda_{N+j}}$ for all $i,j \in \{1,...,N\}$ where $\lambda_1,...,\lambda_N$ are the Lagrange multipliers related to the in-degree constraints, while $\lambda_{N+1},...,\lambda_{2N}$ are the Lagrange multipliers related to the out-degree constraints.
\end{remark}
Having an explicit form for the probability of a graph in the ensemble allows finding an expression for the expected adjacency matrix, which turns out to have elements following a Fermi-Dirac distribution.
\begin{corollary}
\label{cor:expected_A}
    Let $A$ be the random matrix described in the previous corollary. Then, for all $i,j\in\{1,\ldots,N\}$, 
    \begin{equation}\label{eq:expectedA}
        \left\langle A_{ij}\right\rangle =\frac{\alpha_i\beta_j}{1+\alpha_i \beta_j} \,
    \end{equation}
    and $0 < \left\langle A_{ij}\right\rangle < 1$.
\end{corollary}

The next lemma shows that, under some mild conditions, the expected adjacency matrix is an infinite sum of rank-one matrices with singular values equal to $\ell_1$, $\ell_2$, $\ldots$ or $N$, $m_1$, $m_2$, $m_3$, $\ldots$

\begin{lemma}\label{lem:expectedA_infinite_sum} 
    Let $A$ be a random matrix satisfying Eq. \eqref{eq:expectedA}. Let
    \begin{equation*}
       \bm \alpha = (\alpha_1,\ldots, \alpha_N)^\top,\qquad  \bm \beta = (\beta_1,\ldots, \beta_N)^\top.
\end{equation*}

\noindent (1) If $0< \alpha_i\beta_j<1$ for all $i,j\in\{1,\ldots,N\}$, then
            \begin{equation}\label{eq:series_L}
                \left\langle A\right\rangle = \sum_{k=1}^\infty L_k\,,
            \end{equation}
            where $L_k$ denotes a rank-one $N\times N$ matrix whose only nonzero singular value is 
            \begin{equation}\label{eq:series_factor_l}
                \ell_k = \sqrt{\sum_{i,j=1}^N(\alpha_i\beta_j)^{2k}}\,. 
            \end{equation}

\noindent (2) If $\alpha_i\beta_j>1$ for all $i,j\in\{1,\ldots,N\}$, then
            \begin{equation*}
                \left\langle A\right\rangle = N\,\hat{\bm 1}\hat{\bm 1}^\top+\sum_{k=1}^\infty M_k\,,
            \end{equation*}
            where $M_k$ denotes a rank-one $N\times N$ matrix whose only nonzero singular value is 
            \begin{equation*}
                m_k = \sqrt{\sum_{i,j=1}^N(\alpha_i\beta_j)^{-2k}}\,.
            \end{equation*}
\end{lemma}
\begin{proof}
This lemma is essentially a direct consequence of expanding the closed form of the geometric series and normalizing the vectors in each term. Indeed, if $0< \alpha_i\beta_j<1$ for all $i,j\in\{1,\ldots,N\}$, we can use the geometric series and get the following convergent series:
\begin{equation*}
    \left\langle A\right\rangle = {\bm \alpha}\,{\bm \beta}^\top-  ({\bm \alpha\circ \bm \alpha})\, ({\bm \beta \circ \bm \beta})^\top 
    + ({\bm \alpha\circ \bm \alpha\circ \bm \alpha})\, ({\bm \beta \circ \bm \beta\circ \bm \beta})^\top\;-+\ldots
\end{equation*}
Setting
\begin{equation*}
    L_k = (-1)^{k+1} (\underbrace{\bm \alpha \circ \cdots \circ \bm \alpha}_{k \text{ times}})(\underbrace{\bm \beta \circ \cdots \circ \bm \beta}_{k \text{ times}})^\top,
\end{equation*}
we get Eq.~\eqref{eq:series_L}. We see that each matrix $L_k$ is factorized as $\bm u \bm v^\top$, so we conclude that the rank of each element of the series is one.  Moreover, the SVD for such a matrix is simply $\bm u\bm v^\top = \rho \,\widehat{\bm u}\,\widehat{\bm v}^\top$, where $\rho = \|\bm u\|\|\bm v\|$, $\,\,\widehat{\bm u} = \bm u/\|\bm u\|$, $\,\,\widehat{\bm v} = \bm v/\|\bm v\|\,$. Hence, 
\begin{equation}\label{eq:Lk_rank_one}
    L_k = \ell_k\;\widehat{\bm \alpha_k}\;\widehat{\bm \beta_k}^\top, 
\end{equation}
where
\begin{align*}
    \ell_k = \|\underbrace{\bm \alpha \circ \cdots \circ \bm \alpha}_{k \text{ times}}\|\|\underbrace{\bm \beta \circ \cdots \circ \bm \beta}_{k \text{ times}}\|,\qquad
    \widehat{\bm \alpha_k}=(-1)^{k+1}\frac{\overbrace{\bm \alpha \circ \cdots \circ \bm \alpha}^{k \text{ times}}}{\|\underbrace{\bm \alpha \circ \cdots \circ \bm \alpha}_{k \text{ times}}\|}\,,\qquad
    \widehat{\bm \beta_k}=\frac{\overbrace{\bm \beta \circ \cdots \circ \bm \beta}^{k \text{ times}}}{\|\underbrace{\bm \beta \circ \cdots \circ \bm \beta}_{k \text{ times}}\|}\,\,.
\end{align*}
Simple calculations lead to Eq.~\eqref{eq:series_factor_l}, which completes the proof of the first part of the lemma. The second part is proved similarly starting with the geometric series of
\(\left\langle A_{ij}\right\rangle =1/(1+1/(\alpha_i \beta_j))\).
\end{proof}
The last lemma will allow us to find upper bounds on the singular values of the expected adjacency matrix by using Weyl inequalities. However, some technical results are required before deducing the upper bounds. In particular, the coefficients $\ell_k$ and $m_k$ in Lemma \ref{lem:expectedA_infinite_sum} are ordered and bounded as stated in the next lemma.
\begin{lemma}\label{lem:coeff_lk_mk} 
    Let $\ell_k$ and $m_k$ be the coefficients defined in Lemma \ref{lem:expectedA_infinite_sum}. \\
\noindent (1) If $0< \alpha_i\beta_j<1$ for all $i\in\{1,\ldots,N\}$, then 
    \begin{equation}\label{eq:decreasing_ell}
        \ell_{k+1}<\ell_k\,, \qquad \forall \;k\in \mathbb{Z}_+\,
    \end{equation}
and, with $\gamma = \max_{i,j} \alpha_i\beta_j$, 
    \begin{equation}\label{eq:upper_bound_ell}
        \ell_{k}\leq N \,\gamma^k\,, \qquad \forall \;k\in \mathbb{Z}_+\,.
    \end{equation}  

\noindent (2) If $\alpha_i\beta_j>1$ for all $i,j\in\{1,\ldots,N\}$, then 
    \begin{equation}
        m_{k+1}<m_k\,, \qquad \forall \;k\in \mathbb{Z}_+\,.
    \end{equation}
     and, with $\omega = \min_{i, j} \alpha_i\beta_j$,
    \begin{equation}
        m_{k}\leq N \,\omega^{-k}\,, \qquad \forall \;k\in \mathbb{Z}_+\,.
    \end{equation}  
\end{lemma}
\begin{proof}
For the first case, for all $k\in\mathbb{Z}_+$ and from Eq.~\eqref{eq:series_factor_l},
\begin{equation*}
    \ell_{k+1} = \sqrt{\sum_{i,j = 1}^N (\alpha_i\beta_j)^{2(k+1)}}
    = \sqrt{\sum_{i,j = 1}^N (\alpha_i\beta_j)^{2k}(\alpha_i\beta_j)^2} < \sqrt{\sum_{i,j = 1}^N (\alpha_i\beta_j)^{2k}} = \ell_k,
\end{equation*}
where we have used $(\alpha_i\beta_j)^2 < 1$ since $0<\alpha_i\beta_j < 1$ for all $i,j\in\{1,...,N\}$. The first inequality of case (1) is thus established. Moreover, if $\alpha_i\beta_j\leq\gamma$ for all $i,j\in\{1,\ldots,N\}$, then
\begin{equation}
    \ell_k = \sqrt{\sum_{i,j = 1}^N (\alpha_i\beta_j)^{2k}} \leq \sqrt{\sum_{i,j = 1}^N \gamma^{2k}} = N \,\gamma^{k}\,.
\end{equation}
The second inequality of case (1) follows from the observation that $\gamma = \max_{i,j} \alpha_i\beta_j$ is the least value of $\gamma$ that satisfies  $\alpha_i\beta_j\leq\gamma$ for all $i,j\in\{1,\ldots,N\}$. Case (2) is proved similarly.
\end{proof}
Moreover, for a given bound on $\alpha_i\beta_j$, there is a corresponding bound for the elements of the expected adjacency matrix.
\begin{lemma}\label{lem:expectedA_alphabeta_bound}
Let $A$ be a random matrix satisfying Eq. \eqref{eq:expectedA}. Let $\gamma$ and $\omega$ be two positive constants. Then,
\begin{equation}
    \alpha_i\beta_j \leq \gamma < 1 \iff \langle A_{ij} \rangle \leq \frac{\gamma}{1 + \gamma} < \frac{1}{2}\qquad \text{and} \qquad \alpha_i \beta_j \geq \omega>1 \iff \langle A_{ij} \rangle \geq \frac{\omega}{1 + \omega} > \frac{1}{2}.
\end{equation}
\end{lemma}
\begin{proof}
Recall that all the parameters involved in this lemma are positive. The first part of both equivalences is obtained with basic inequality manipulations:
    \begin{align*}
         \alpha_i\beta_j \leq  \gamma\iff \frac{1}{\alpha_i\beta_j} \geq \frac{1}{\gamma} \iff \frac{1}{1 + \frac{1}{\alpha_i\beta_j}} \leq \frac{1}{1 + \frac{1}{\gamma}}\iff \langle A_{ij} \rangle \leq \frac{\gamma}{1 + \gamma}\,,\\
         \alpha_i\beta_j \geq \omega \iff \frac{1}{\alpha_i\beta_j} \leq \frac{1}{\omega} \iff \frac{1}{1 + \frac{1}{\alpha_i\beta_j}} \geq \frac{1}{1 + \frac{1}{\omega}}\iff \langle A_{ij} \rangle \geq \frac{\omega}{1 + \omega}\,.
    \end{align*}
The second part is an immediate consequence of $\gamma <1 \iff \gamma/(1 + \gamma) < 1/2$ and $\omega>1\iff\omega/(1 + \omega) > 1/2$.
\end{proof}
\begin{remark}
The inequalities in the last lemma imply that for all $i\in\{1,...,N\}$, the expected degrees $k_i^{\mathrm{in}}$ and $k_i^{\mathrm{out}}$ are both upper-bounded by $N\gamma/(1 + \gamma)$ when $\alpha_i\beta_j < \gamma <1$, and lower bounded by $N\omega/(1 + \omega)$ when $\alpha_i \beta_j > \omega>1$. However, these bounds on $k_i^{\mathrm{in}}$ and $k_i^{\mathrm{out}}$ do not necessarily imply that the inequalities in the last lemma are satisfied.
\end{remark}

We are now ready to present the first main theorem of this section, which states that for two broad families of parameters defining the soft directed configuration model, the singular values of expected adjacency matrix decrease very rapidly, at least exponentially. 
\begin{theorem}\label{thm:upper_bound_singvals_scm}
Let $\langle A\rangle$ be the $N\times N$ matrix defined in Eq.~\eqref{eq:expectedA} and whose singular values are $\sigma_1\geq ...\geq \sigma_N$. Let $\ell_k$ and $m_k$ be the coefficients defined in Lemma \ref{lem:expectedA_infinite_sum}.
 
 \vspace{0.5cm}
 
\noindent (1) If $0 < \langle A_{ij} \rangle < 1/2$ for all $i,j\in\{1,\ldots,N\}$, then the singular values are upper-bounded as 
\begin{equation}\label{eq:upper_bound_sparser}
        \sigma_i\leq \sum_{k=i}^\infty \ell_k \leq  \frac{N\,\gamma^i}{1-\gamma},\qquad \forall\,\, i \in\{1,...,N\}\,,
\end{equation}
where $\gamma = \max_{i,j} \alpha_i \beta_j$. 

 \vspace{0.5cm}
         
\noindent (2) If $1/2 < \langle A_{ij} \rangle < 1$ for all $i\in\{1,\ldots,N\}$, then the singular values are upper-bounded as 
\begin{equation}\label{eq:upper_bound_denser}
    \sigma_i\leq N\,\delta_{i1}+\sum_{k=i}^\infty m_k \leq N\delta_{i1} + \frac{N\omega^{1-i}}{\omega-1},\qquad \forall\,\, i \in\{1,...,N\}, 
\end{equation}
where $\delta_{i1}$ is a Kronecker delta and $\omega = \min_{i,j}\alpha_i\beta_j$.    
\end{theorem}
\begin{proof}
\phantom{Henri} \\
(1) First of all, $0 < \langle A_{ij}\rangle < 1/2$
if and only if $0<\alpha_i\beta_j<1$ for all $i,j\in\{1,...,N\}$ from Lemma~\ref{lem:expectedA_alphabeta_bound}. Lemma~\ref{lem:expectedA_infinite_sum} then implies that the expected adjacency matrix is a convergent infinite sum of rank one matrices $L_k$, $k\in\mathbb{Z}_+$. Thus, the singular values of the expected adjacency matrix are the singular values of a sum of matrices:
\begin{equation*}
    \sigma_i(\langle A\rangle) = \sigma_i(\textstyle{\sum_{k = 1}^\infty} L_k)\,, \qquad \forall \,i\in\{1,...,N\},
\end{equation*}
where we write $\sigma_i(\langle A\rangle)$ instead of $\sigma_i$ for the sake of clarity in the proof. 

Next, recall from Theorem~\ref{thm:weyl_vas} that the Weyl inequalities for $N\times N$ matrices $B$ and $C$ are
 \begin{equation*}
      \sigma_{r+s-1}(B+C) \leq \sigma_r(B) + \sigma_s(C)\,, \qquad 
      \forall\;1 \leq r,\, s, \, r+s-1 \leq N.
  \end{equation*}
Setting $r=s=1$ yields the familiar triangle inequality:  
  \begin{equation}\label{eq:triangle_inequality}
      \sigma_{1}(B+C) \leq \sigma_1(B) + \sigma_1(C)\,.
  \end{equation} 
The latter inequality implies that for all $1\leq i\leq n-1<\infty$,
  \begin{equation*} 
      \sigma_{1}\Big(\sum_{k=i}^nL_k\Big) \leq \sum_{k=i}^n\sigma_{1}(L_k)\,.
  \end{equation*}
However, given that $\sigma_1(L_k)$ is nonnegative, 
\begin{equation*} 
    \sum_{k=i}^n\sigma_{1}(L_k)\leq \sum_{k=i}^{n+1}\sigma_{1}(L_k)\leq \cdots \leq \sum_{k=i}^\infty\sigma_{1}(L_k)=\sum_{k=i}^\infty \ell_k\,,
\end{equation*}
where we have used the notation $\ell_k=\sigma_{1}(L_k)$ introduced in Lemma~\ref{lem:expectedA_infinite_sum}. To prove the convergence of the infinite series, we recall from Lemma~\ref{lem:coeff_lk_mk} that $\ell_k\leq N\gamma^k$ with $\gamma = \max_{i,j}\alpha_i\beta_j$. This in turn implies that 
\begin{equation*} 
      \sum_{k=i}^\infty \ell_k\leq\frac{N\gamma^i}{1-\gamma},
\end{equation*}
as stated in the rightmost inequality of \eqref{eq:upper_bound_sparser}. So far, we have  proved that for all $1\leq i\leq n-1<\infty$,
\begin{equation*} 
      \sigma_{1}\Big(\sum_{k=i}^nL_k\Big) \leq \sum_{k=i}^\infty \ell_k\leq\frac{N\gamma^i}{1-\gamma}\,.
\end{equation*}
The continuity of $\sigma_1\,:\mathbb{R}^{N\times N}\to \mathbb{R}$, which is obvious since $\sigma_1$ is a norm, and the convergence of $\sum_{k=i}^\infty L_k$ allow us to take the limit $n\to \infty$ on the left-hand side of the previous inequality and conclude that
  \begin{equation} \label{eq:second_inequality_SI}
      \sigma_{1}\Big(\sum_{k=i}^\infty L_k\Big) \leq  \sum_{k=i}^\infty \ell_k\leq\frac{N\gamma^i}{1-\gamma}\,.
  \end{equation}

Let us now go back to the Weyl inequalities and set $r=i$, $s=1$, $B = \sum_{k=1}^{i-1} L_k$, and  $C=\sum_{k=i}^\infty L_k$. This yields the inequality 
 \begin{equation*}
   \sigma_{i}\Big(\sum_{k=1}^\infty L_k\Big) = \sigma_i\Big(\sum_{k=1}^{i-1} L_k+\sum_{k=i}^\infty L_k\Big) \leq \sigma_i\Big(\sum_{k=1}^{i-1} L_k\Big)+\sigma_1\Big(\sum_{k=i}^\infty L_k\Big)\,,
  \end{equation*}
which is valid for all $1\leq i\leq N$. 
The matrix $\sum_{k=1}^{i-1} L_k$ is the sum of $i-1$ matrices of rank one, which means that the rank of  $\sum_{k=1}^{i-1} L_k$ is at most $i-1$. Hence,
$\sigma_i\Big(\sum_{k=1}^{i-1} L_k\Big)=0\,$,
so that
\begin{equation}\label{eq:third_inequality_SI}
      \sigma_{i}\Big(\sum_{k=1}^\infty L_k\Big) \leq \sigma_1\Big(\sum_{k=i}^\infty L_k\Big)\,,\qquad \forall \, i\in\{1,\ldots,N\}\,.
\end{equation}
Combining inequalities \eqref{eq:second_inequality_SI} and \eqref{eq:third_inequality_SI} leads to the desired result: \begin{equation*} 
      \sigma_{i}\Big(\sum_{k=1}^\infty L_k\Big) \leq  \sum_{k=i}^\infty \ell_k\leq\frac{N\gamma^i}{1-\gamma}\,,\qquad \forall \, i\in\{1,\ldots,N\}\,.
\end{equation*}

\vspace{0.5cm}

\noindent (2) Similarly to the first case Lemmas~\ref{lem:expectedA_infinite_sum}-\ref{lem:expectedA_alphabeta_bound}, and Weyl inequalities imply that
\begin{equation*}
    \sigma_i(\langle A \rangle)= \sigma_i\Big(\sum_{k=1}^\infty M_k\Big)\leq N\delta_{i1} + m_i + \sigma_1\Big(\sum_{k=i+1}^\infty M_k\Big).
\end{equation*}
Proceeding as for inequality \eqref{eq:second_inequality_SI} then leads to the inequality
\begin{equation*}
    \sigma_i\Big(\sum_{k=1}^\infty M_k\Big) \leq N\delta_{i1} + \sum_{k=i}^\infty m_k,
\end{equation*}
where $m_k=\sigma_1(M_k)$. 
Additionally, Lemma~\ref{lem:coeff_lk_mk} states that $m_k \leq N\omega^{-k}$ with $\omega = \min_{i, j}\alpha_i\beta_j$, which leads to
\begin{equation*}
    \sigma_i\Big(\sum_{k=1}^\infty M_k\Big)\leq N\delta_{i1} + N\sum_{k=i}^\infty \omega^{-k}\,.
\end{equation*} 
Writing the truncated geometric series in closed form finally gives the expected result.
\end{proof}
The upper bounds in the last theorem theoretically validate the low-rank formulation of the directed soft configuration model. 
The last inequalities in Eqs.~\eqref{eq:upper_bound_sparser} and \eqref{eq:upper_bound_denser} are meant to explicitly show the exponential decrease of the singular values while the first inequalities in Eqs.~\eqref{eq:upper_bound_sparser} and \eqref{eq:upper_bound_denser} are tighter versions. In Fig.~\ref{fig:upper_bounds_dscm}a and \ref{fig:upper_bounds_dscm}b, we illustrate the first inequalities in Eqs.~\eqref{eq:upper_bound_sparser} and \eqref{eq:upper_bound_denser} with both axes in log-log scale.

Following similar steps, we prove the second main theorem of the section: the singular values of $\langle W \rangle$ in the weighted directed soft configuration model (WDSCM) [Example~\ref{ex:network_science_weighted}] are at least exponentially decreasing.
\begin{theorem}\label{thm:upper_bound_singvals_wdscm}
Let $\langle W\rangle$ be the $N\times N$ matrix defined in Eq.~\eqref{eq:expectedW_wdscm} and whose singular values are $\sigma_1\geq ...\geq \sigma_N$ and let $n_k = \sqrt{\sum_{i,j=1}^N (y_i\bar{y}_j)^{2k}}$ with $0 < y_i\bar{y}_j < 1$ for all $i, j$. Then, the singular values are upper-bounded as 
\begin{equation}\label{eq:upper_bound_infinite_sum_wdscm}
        \sigma_i\leq \sum_{k=i}^\infty n_k \leq  \frac{N\,\tau^i}{1-\tau},\qquad \forall\,\, i \in\{1,...,N\}\,,
\end{equation}
where $\tau = \max_{i,j}y_i\bar{y}_j$.
\end{theorem}
Contrarily to Theorem~\ref{thm:upper_bound_singvals_scm}, there is no restriction on the domain of the elements of $\langle W\rangle$ for the inequality~\eqref{eq:upper_bound_infinite_sum_wdscm}, which is a consequence of the Bose-Einstein distribution for the elements of the expected weight matrix. The bound in Eq.~\eqref{eq:upper_bound_infinite_sum_wdscm} is illustrated in Fig.~\ref{fig:upper_bounds_dscm}c.

\begin{figure*}[t]
    \centering
    \includegraphics[width=0.9\textwidth]{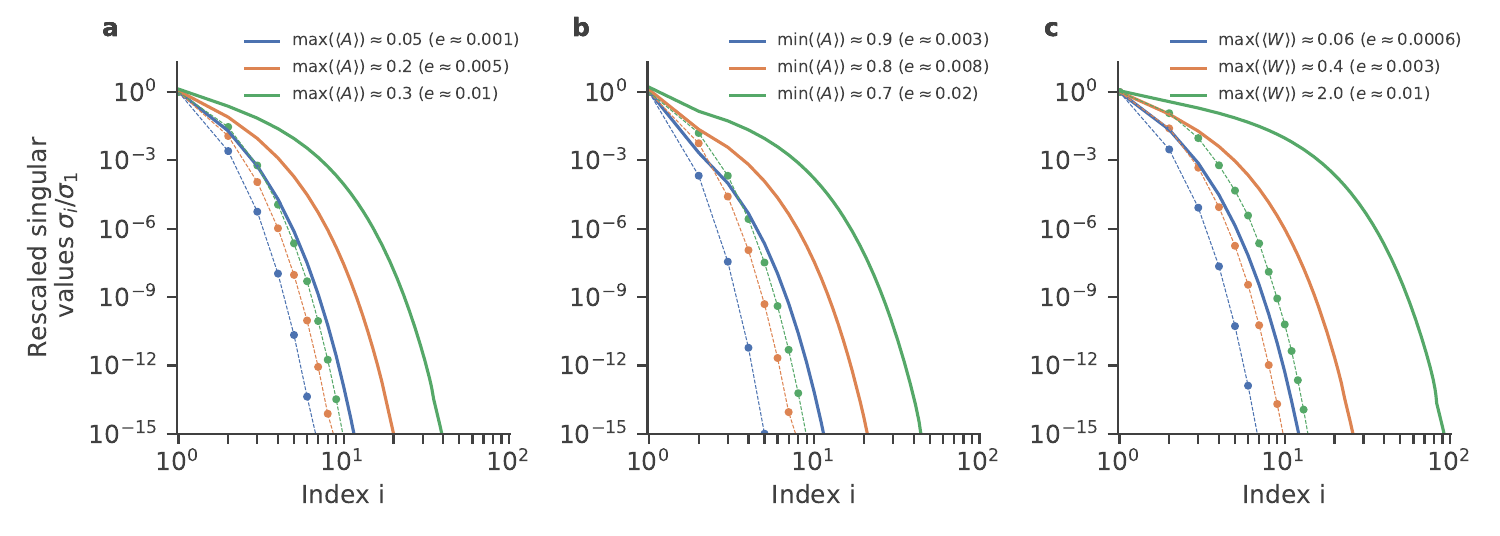}
    \vspace{-0.5cm}
    \caption{Upper bounds (solid lines) on the singular values (markers) of the expected matrix of the directed soft configuration model (\textbf{a}, Eq.~\eqref{eq:upper_bound_sparser} and \textbf{b}, Eq.~\eqref{eq:upper_bound_denser}) and its weighted version (\textbf{c}, Eq.~\eqref{eq:upper_bound_infinite_sum_wdscm}). The dashed lines between singular values are shown for the sake of visualization. In all the subfigures, the $N$-dimensional vectors $\bm\alpha$, $\bm{\beta}$, $\bm{y}$, $\bar{\bm{y}}$ defining the expected matrices are obtained from truncated Pareto distributions and we denote one $N$-dimensional realization as $p_{\bm{x}}(N, x_{\mathrm{\min}}, x_{\mathrm{\max}}, \gamma)$, where $N = 10^3$ is the number of vertices, $\bm{x}$ is $\bm\alpha$, $\bm{\beta}$, $\bm{y}$, or $\bar{\bm{y}}$, $x_{\mathrm{\min}}$ is the minimum value of the distribution, $x_{\mathrm{\max}}$ is the maximum value, and $\gamma$ is the shape parameter. \textbf{a}, $p_{\bm{\alpha}}(N, 2, \alpha_{\mathrm{\max}}, 2)/\sqrt{N}$ where $\alpha_{\mathrm{max}} \in \{10, 20, 30\}$ and $p_{\bm{\beta}}(N, 1, \beta_{\mathrm{max}}, 2.5)/\sqrt{N}$ where $\beta_{\mathrm{max}} \in (5, 10, 15)$. \textbf{b}, $p_{\bm{\alpha}}(N, \alpha_{\mathrm{min}}, 200, 2)/\sqrt{N}$ where $\alpha_{\mathrm{min}} \in (150, 75, 50)$ and $p_{\bm{\beta}}(N, \beta_{\mathrm{min}}, 150, 2.5)/\sqrt{N}$ where $\beta_{\mathrm{min}} \in (120, 60, 40)$. \textbf{c}, $p_{\bm{y}}(N, 0.05, y_{\mathrm{\max}}, 2)$ where $y_{\mathrm{max}} \in \{0.3, 0.6, 0.9\}$ and $p_{\bar{\bm{y}}}(N, 0.05, \bar{y}_{\mathrm{max}}, 2.5)$ where $\bar{y}_{\mathrm{max}} \in (0.2, 0.5, 0.8)$. All the tuples above correspond to the (blue, orange, green) solid lines.}
    \label{fig:upper_bounds_dscm}
\end{figure*}

\subsection{Impact of singular value distribution and matrix density on effective ranks}
\label{SIsubsec:impact_singvals_effrank}

In this subsection, we take advantage of the formula for the stable rank ($\mathrm{srank}$), the nuclear rank ($\mathrm{nrank}$), and the $\mathrm{erank}$, which are amenable for analytic calculations, to assess the impact of different singular value decreases on the effective ranks through various inequalities. In the first part of the subsection, we prove that finding bounding curves, $\psi_*(x)$ and $\psi^*(x)$, that approximately delineate the region of possible singular value allows us to estimate the srank, nrank, and erank. In the second part, we show that linear $O(N)$, sublinear $O(N^{1-\epsilon})$, and constant $O(1)$ asymptotic behaviors emerge depending on the shape of the singular value distribution. We finally present, in a third part, the impact of the density of $W$ on $\mathrm{srank}(W)$ through general inequalities. 

\vspace{0.3cm}

\begin{center}\textit{1. Singular-value envelopes}\end{center}

\noindent We define the \textit{singular-value envelopes} $\psi_*$ and $\psi^*$ for the singular values as functions that
\begin{enumerate}\setlength\itemsep{0em}
    \item decrease on the interval $[1, N]\subset \mathbb{R}$, that is 
    \begin{equation}
    x\leq y\qquad \Longrightarrow\qquad \psi_*(x)\geq \psi_*(y),\qquad \psi^*(x)\geq \psi^*(y)\,;
    \end{equation}
    \item are nonnegative on the interval $[1, N]\subset \mathbb{R}$;
    \item provide lower and upper bounds for the rescaled singular values as
\begin{equation}\label{eq:bounds_singvals}
    \psi_*(i)\leq \frac{\sigma_i}{\sigma_1}\leq \psi^*(i)\qquad \forall\, i\in\{1,\ldots, N\};
\end{equation}
    \item are $\sigma_1$--tight, meaning 
    \begin{equation}
    \psi_*(1)=1= \psi^*(1)\,.
\end{equation}
\end{enumerate}
The last condition is imposed to always match the only value of the ratio  $\sigma_i/\sigma_1$ that is known in all instances. In the next part of the subsection, it will also prevent us from multiplying $\psi_*$ and $\psi^*$ by global scaling factors such as $N^{-\epsilon}$, which could impose, somewhat artificially, asymptotic behaviors for the effective ranks such as $O(N^{1-\epsilon})$. The four properties of the singular-value envelopes readily imply general inequalities that will be useful to bound the effective ranks.

\begin{lemma}[Basic inequalities]\label{lem:general_bounds}
    If $\psi_*$ and $\psi^*$ satisfy conditions 1--4, then
    \begin{equation}
        \int_1^{N} \psi_*(x)^q\,dx \,\leq \,\sum_{i=1}^N\psi_*(i)^q\,\leq \,\sum_{i=1}^N\left(\frac{\sigma_i}{\sigma_1}\right)^q\,\leq\, \sum_{i=1}^N\psi^*(i)^q\,\leq \,1+\int_1^N \psi^*(x)^q\,dx
    \end{equation}
    for all $q\geq 0$ and 
\begin{multline}\label{eq:inequality_entropy_sum_int}
   0
   \leq \sum_{i=2}^N\psi_*(i)\ln \frac{1}{\psi^*(i)}
    \leq \sum_{i=1}^N\frac{\sigma_i}{\sigma_1}\ln \frac{\sigma_1}{\sigma_i}
    \leq  
    \frac{1}{\delta} \sum_{i=2}^N\left(\psi^*(i)^{1-\delta}- \psi_*(i)\right)
    \leq \frac{1}{\delta}\int_1^N\left(\psi^*(x)^{1-\delta}-\psi_*(x)\right)dx+\Delta
\end{multline}
for all $\delta \in (0, 1)$, where 
\begin{equation}
    \Delta = \frac{1}{\delta}\int_1^2\psi_*(x)\,dx\,.
\end{equation}
\end{lemma}
\begin{proof}
We first prove the inequalities involving the summations. Using inequality \eqref{eq:bounds_singvals} and the nonnegativeness of $\psi_*$, $\psi^*$, and $\sigma_i/\sigma_1$, we deduce the inequality 
\begin{equation}
    \psi_*(i)^q\,\leq \,\left(\frac{\sigma_i}{\sigma_1}\right)^q\,\leq \,\psi^*(i) ^q,
\end{equation}
which is valid for all $q\geq 0$ and  $i\in\{1,2,\ldots,N\}$. Thus 
\begin{equation}
    \sum_{i=1}^N\psi_*(i)^q\,\leq \,\sum_{i=1}^N\left(\frac{\sigma_i}{\sigma_1}\right)^q\,\leq\, \sum_{i=1}^N\psi^*(i)^q\,,
\end{equation}
as expected. 

We now concentrate on the first and last inequalities. We adopt a strategy analogous to the method used for proving the integral test for convergence. Notice that the $\psi_*$ and $\psi^*$ are integrable on any subinterval of $[1,N]$ since these functions are monotone. On the one hand, 
\begin{equation}
    \psi_*(i)^q=\int_{i}^{i+1} \psi_*(i)^q\,dx \geq \int_{i}^{i+1} \psi_*(x)^q\,dx  
\end{equation}
since $\psi_*(i)^q\geq \psi_*(x)^q$ for all $x\geq i$ (condition 1).  The latter inequality and condition 2 then imply
\begin{multline}\label{eq:int_lower_bound_varphi}
    \sum_{i=1}^N\psi_*(i)^q=\sum_{i=1}^{N-1}\psi_*(i)^q +\psi_*(N)^q 
    \\\geq \sum_{i=1}^{N-1}\int_{i}^{i+1} \psi_*(x)^q\,dx+\psi_*(N)^q  =\int_{1}^{N} \psi_*(x)^q\,dx+\psi_*(N)^q\geq \int_{1}^{N} \psi_*(x)^q\,dx \,, 
\end{multline}
which proves the leftmost inequality.
On the other hand, 
\begin{equation}
    \psi^*(i)^q=\int_{i-1}^i \psi^*(i)^q\,dx \leq \int_{i-1}^i \psi^*(x)^q\,dx  
\end{equation}
since $\psi^*(i)^q\leq \psi^*(x)^q$ for all $x\leq i$ (condition 1). Thus, with condition 4,
\begin{equation}\label{eq:int_upper_bound_varphi}
     \sum_{i=1}^N\psi^*(i)^q=\psi^*(1)^q+\sum_{i=2}^N\psi^*(i)^q\leq\psi^*(1)^q+\sum_{i=2}^N\int_{i-1}^i \psi^*(x)^q\,dx =1+ \int_{1}^N \psi^*(x)^q\,dx\,,
\end{equation}
which establishes the rightmost inequality. 

To prove the last inequalities, we first notice that thanks to condition 4, the term corresponding to $i=1$ in the summation can be omitted: 
\begin{equation}\label{eq:sum_starting_2}
    \sum_{i=1}^N\frac{\sigma_i}{\sigma_1}\ln \frac{\sigma_1}{\sigma_i} = \sum_{i=2}^N\frac{\sigma_i}{\sigma_1}\ln \frac{\sigma_1}{\sigma_i}.\,
\end{equation}
We then lower-bound each element of the sum as
\begin{equation}
    \frac{\sigma_i}{\sigma_1}\ln \frac{\sigma_1}{\sigma_i}\geq 0
\end{equation}
since ${\sigma_1}/{\sigma_i}\geq 1$ for all $i$. To get the upper bound of \eqref{eq:sum_starting_2}, we use the classical inequality $\ln x\leq a( x^{1/a}-1)$ for all $a,x>0$ \cite[\href{https://dlmf.nist.gov/4.5.E5}{(4.5.5)}]{NIST:DLMF}, which implies that
\begin{equation}
    \frac{\sigma_i}{\sigma_1}\ln \frac{\sigma_1}{\sigma_i}\leq \frac{1}{\delta} \frac{\sigma_i}{\sigma_1}\left(\frac{\sigma_1}{\sigma_i}\right)^\delta-\frac{1}{\delta} \frac{\sigma_i}{\sigma_1} =  \frac{1}{\delta} \left(\frac{\sigma_i}{\sigma_1}\right)^{1-\delta}-\frac{1}{\delta} \frac{\sigma_i}{\sigma_1}\qquad \forall\; \delta>0,
\end{equation}
where the equality is obtained when $\delta\to 0$. Thus, 
\begin{multline}
    \sum_{i=2}^N\frac{\sigma_i}{\sigma_1}\ln \frac{\sigma_1}{\sigma_i}
    \leq\frac{1}{\delta} \sum_{i=2}^N \left(\frac{\sigma_i}{\sigma_1}\right)^{1-\delta}- \frac{1}{\delta}\sum_{i=2}^N  \frac{\sigma_i}{\sigma_1}
   \\  \leq
    \frac{1}{\delta} \sum_{i=2}^N\psi^*(i)^{1-\delta}- \frac{1}{\delta}\sum_{i=2}^{N} \psi_*(i)
   \leq
    \frac{1}{\delta} \sum_{i=2}^N\psi^*(i)^{1-\delta}- \frac{1}{\delta}\sum_{i=2}^{N-1} \psi_*(i)
\end{multline}
where the second inequality is the expected result while the third one is obtained by neglecting the last (negative) element of the sum.
Now, considering that $\psi^*(x)^{1-\delta}$ (with $\delta<1$) and $-\psi_*(x)$ are respectively decreasing and increasing in $x$, we can write
\begin{equation}
   \sum_{i=2}^N\psi^*(i)^{1-\delta}=\sum_{i=2}^N\int_{i-1}^i\psi^*(i)^{1-\delta}dx
    \leq \sum_{i=2}^N\int_{i-1}^i\psi^*(x)^{1-\delta}dx = \int_1^N\psi^*(x)^{1-\delta}dx,\qquad \forall\;\delta\in(0,1)
\end{equation}
and 
\begin{equation}
   -\sum_{i=2}^{N-1} \psi_*(i)=-\sum_{i=2}^{N-1}\int_{i}^{i+1}\psi_*(i)\,dx\leq -\sum_{i=2}^{N-1}\int_{i}^{i+1}\psi_*(x)\,dx=-\int_{2}^{N}\psi_*(x)\,dx
\end{equation}
Hence,
\begin{equation}
   \sum_{i=2}^N\frac{\sigma_i}{\sigma_1}\ln \frac{\sigma_1}{\sigma_i}\leq 
   \frac{1}{\delta}\left(\int_1^N\psi^*(x)^{1-\delta}dx-\int_{2}^N\psi_*(x)\,dx\right),
\end{equation}
which is equivalent to the desired result. 
\end{proof}
In Theorems \ref{thm:upper_bound_singvals_scm} and \ref{thm:upper_bound_singvals_wdscm}, we observed exponential decreases of the singular values occurs when working with the expected adjacency or weight matrix of two frequently used random graphs, namely the directed soft configuration model (DSCM) and its weighted version (WDSCM). A first consequence of the previous lemma is that an exponential decrease implies that srank, nrank, and erank are bounded by finite geometric series (or functions of them).
\begin{proposition}[Bounds on effective ranks -- Exponential decrease]\label{prop:exponential_envelop}
Suppose that the singular values of matrix $W$, $\sigma_1\geq \sigma_2\geq \cdots \sigma_N\geq 0$, satisfy the inequalities 
\begin{equation}\label{eq:lem_bounds_srank_exp}
    \alpha^{i-1}\,\leq \,\frac{\sigma_i}{\sigma_1}\,\leq \,\omega^{i-1},\qquad i\in\{1,\ldots, N\}, 
\end{equation}
for some $0<\alpha\leq \omega<1$. Then,
\begin{align}
    g(\alpha^2,N)\,\leq \, &\, \mathrm{srank}(W)\,\leq   g(\omega^2,N)\,\\
    g(\alpha,N)\leq \,& \,\mathrm{nrank}(W)\,\leq g(\omega,N),\\   
   \omega\,g(\alpha,N)\,\exp\left(\frac{\alpha \,g'(\alpha,N)}{g(\omega,N)}\right) \leq \,&\,\mathrm{erank}(W)\,\leq \alpha\,g(\omega,N)\,\exp\left(\frac{\omega \,g'(\omega,N)}{g(\alpha,N)}\right),
\end{align}
where 
\begin{equation} 
     g(\alpha,N)=\frac{1-\alpha^N}{1-\alpha},\,\qquad g'(\alpha,N)=\frac{\partial g(\alpha,N)}{\partial \alpha}=\frac{1+(N-1)\alpha^N-\alpha^{N+1}}{(1-\alpha)^2}\,.
\end{equation}
\end{proposition}
\begin{proof}
The inequalities for $\mathrm{nrank}$ and $\mathrm{srank}$ are easily derived from Lemma~\ref{lem:general_bounds} by setting $q=1$ and $q=2$, respectively, together with $\psi_*(i)=\alpha^{i-1}$ and $\psi^*(i)=\omega^{i-1}$. Then, combining the first inequalities with Eq.~\eqref{eq:erank_expH}, we obtain the following preliminary result:
\begin{equation} \label{eq:preliminary_inequl_exp}
     g(\alpha,N)\exp\left(\frac{1}{g(\omega,N)}\sum_{i=1}^N\frac{\sigma_i}{\sigma_1}\ln\frac{\sigma_1 }{\sigma_i}\right)
     \leq\,\mathrm{erank}(W)\,\leq 
     g(\omega,N) \exp\left(\frac{1}{g(\alpha,N)}\sum_{i=1}^N\frac{\sigma_i}{\sigma_1}\ln\frac{\sigma_1 }{\sigma_i}\right),
\end{equation}
where it is understood that $0\ln 0 =0$. Now,
\begin{equation*} 
\alpha^{i-1}\ln \omega^{i-1}\,\leq \,\frac{\sigma_i}{\sigma_1}\ln\frac{\sigma_1 }{\sigma_i}\,\leq \,\omega^{i-1}\ln \alpha^{i-1}
\end{equation*}
for all $i\in\{1,\ldots, N\}$. Hence,
\begin{equation*} 
\sum_{i=1}^N\alpha^{i-1}\ln \omega^{i-1}\,\leq \,\sum_{i=1}^N\frac{\sigma_i}{\sigma_1}\ln\frac{\sigma_1 }{\sigma_i}\,\leq \,
\sum_{i=1}^N \omega^{i-1}\ln \alpha^{i-1},
\end{equation*}
which can be simplified as 
\begin{equation*} 
\alpha \ln \omega\sum_{j=0}^{N-1}j\,\alpha^{j-1}\,\leq \,\sum_{i=1}^N\frac{\sigma_i}{\sigma_1}\ln\frac{\sigma_1 }{\sigma_i}\,\leq \,
\omega \ln \alpha\sum_{j=0}^{N-1}j\,\omega^{j-1}\,.
\end{equation*}
Using the geometric series and its derivative, we deduce that    
\begin{equation*} 
\alpha \ln \omega\,g'(\alpha,N)\,\leq \,\sum_{i=1}^N\frac{\sigma_i}{\sigma_1}\ln\frac{\sigma_1 }{\sigma_i}\,\leq \,
\omega \ln \alpha \,g'(\omega,N)\,.
\end{equation*}
We finally get the desired result by taking the exponential of the previous expression and returning to inequality~\eqref{eq:preliminary_inequl_exp}. 
\end{proof}

In Lemma~\ref{lem:general_bounds}, the variable $x$ interpolates between the singular value indices $i$ and thus belongs to a domain that grows with $N$, which is particularly convenient for studying random network models individually, such as the soft configuration model. To allow the comparison of real networks of different sizes, as in Fig.~\ref{fig:low_rank_hypothesis}e, we also need to treat the case where $x$ belongs to the closed interval $[0,1]$, a compact domain that remains the same for all $N$. The following lemma allows one to go from one perspective to the other. 
\begin{lemma}[Extensive vs. intensive domains] \label{lem:extensive_intensive}
    Let $\psi\,:\,[1,N]\to [0,1]$ be monotonically decreasing. Define $\Psi\,:\,[0,1]\to [0,1]$ as
    \begin{equation}
        \Psi(y) = \psi\big((N-1)y+1\big)\qquad \forall \, y\in[0,1]\,.
    \end{equation}
    Then, for all $q>0$,
    \begin{equation}
        \int_{1}^N\psi(x)^q\,dx = (N-1) \int_{0}^1 \Psi(y)^q\,dy\,,\qquad \int_{1}^2 \psi(x)^q \,dx = (N-1) \int_0^{\frac{1}{N-1}}  \Psi(y)^q \,dy\,,
    \end{equation}
\end{lemma}
\begin{proof}
    The first result is an immediate consequence of the following linear, and thus invertible, change of variables:
    \begin{align}\label{eq:change_extensive_intensive}
        T\,:\,[1,N]&\longrightarrow [0,1]\\
           x & \longmapsto y=\frac{x-1}{N-1}
    \end{align}
    which maps $[\ell,m]$ onto $[(\ell-1)/(N-1), (m-1)/(N-1)]$ for all $1\leq \ell\leq m \leq N$. 
\end{proof}

The two previous lemmas and formula \eqref{eq:erank_expH} imply the following result stating that $\mathrm{srank}$ and $\mathrm{nrank}$ are essentially equal to $N$ times the area under the curves $\Psi(y)^2$ and $\Psi(y)$ with $y\in[0,1]$, respectively, while  $\mathrm{erank}$ is related to the area under the curves $\Psi(y)$ and $\ln \Psi(y)^{-1}$. 
\begin{lemma}[Effective rank as area under a curve]\label{lem:effrank_area}
     Let $\psi_*$ and $\psi^*$ satisfy conditions 1--4. Define $\Psi_*\,:\,[0,1]\to [0,1]$ and $\Psi^*\,:\,[0,1]\to [0,1]$ as 
     \begin{equation}
         \Psi_*(y)=\psi_*\big((N-1)y+1\big)\qquad \text{and}\qquad \Psi^*(y)=\psi^*\big((N-1)y+1\big).
     \end{equation}
    Then 
    \begin{align}
    (N-1) \int_0^1 \Psi_*(y)^2\,dy \leq \, &\; \mathrm{srank}(W)\,\leq \,1+ (N-1) \int_0^1 \Psi^*(y)^2\,dy \,, \\
     (N-1) \int_0^1 \Psi_*(y)\,dy \leq \,& \;\mathrm{nrank}(W)\, \leq 1+ (N-1) \int_0^1 \Psi^*(y)\,dy\, .
\end{align}
Moreover, for all $\delta\in(0,1)$,
   \begin{align}
    (N-1) \int_0^1 \Psi_*(y)\,dy\, \leq \, &\; \mathrm{erank}(W)\,\leq \left(1+ (N-1) \int_0^1 \Psi^*(y)\,dy\right)e^{\mathcal{H}+\mathcal{P}} \,, \label{eq:bounding_erank_integral_tight}
\end{align}
where
    \begin{equation}\label{eq:HP}
        \mathcal{H} = \frac{1}{\delta}\left(\frac{\displaystyle \int_{0}^1 \Psi^*(y)^{1-\delta} \,dy}{\displaystyle\int_0^1 \Psi_*(y)\,dy}-1\right)\,,\qquad 
        \mathcal{P} = \frac{\displaystyle \int_{0}^{\frac{1}{N-1}}\Psi_*(y)\,dy}{\delta\displaystyle\int_0^1 \Psi_*(y)\,dy}\,.
    \end{equation}
\end{lemma}
This new perspective on the effective ranks allows us to consider a general family of singular-value envelopes that can be applied to our experimental results as illustrated in Fig.~\ref{fig:low_rank_hypothesis}e. Interestingly, this family is related to the Gaussian hypergeometric function \cite[Chap.~15]{NIST:DLMF}.

\begin{theorem}[Bounds on effective ranks -- Hypergeometric decrease]\label{thm:bounds_effective_ranks_hypergeometric}
Suppose that the singular values of matrix $W$, $\sigma_1\geq \sigma_2\geq \cdots \sigma_N\geq 0$, satisfy the inequality 
\begin{equation}
     \left(1-\frac{i-1}{N-1}\right)^{c^*-2}\left(1+\zeta^*\frac{i-1}{N-1}\right)^{-b^*} \,\leq \,\frac{\sigma_i}{\sigma_1}\,\leq \,  \left(1-\frac{i-1}{N-1}\right)^{c_*-2}\left(1+\zeta_*\frac{i-1}{N-1}\right)^{-b_*}
\end{equation}
for some $0\leq b_*\leq b^* $, $2\leq c_*\leq c^*$, $0<\zeta_*\leq \zeta^*$, and for all $i\in\{1,\ldots, N\}$. Then,
\begin{align}
     \frac{N-1}{2c^*-3}\; {}_2F_1(1,2b^*;2c^*-2;-\zeta^*)\,\leq \, &\, \mathrm{srank}(W)\,\leq  1+  \frac{N-1}{2c_*-3}\; {}_2F_1(1,2b_*;2c_*-2;-\zeta_*)\\
     \frac{N-1}{c^*-1} \;{}_2F_1(1,b^*;c^*;-\zeta^*)\leq \,& \,\mathrm{nrank}(W)\,\leq 1+   \frac{N-1}{c_*-1} \:{}_2F_1(1,b_*;c_*;-\zeta_*),\\
     \frac{N-1}{c^*-1} \;{}_2F_1(1,b^*;c^*;-\zeta^*)\leq \,& \,\mathrm{erank}(W)\,\leq \left(1+   \frac{N-1}{c_*-1}\; {}_2F_1(1,b_*;c_*;-\zeta_*)\right)\,e^{\mathcal{H} + \mathcal{P}},
\end{align}
where, for all $\delta \in (0, 1)$,
\begin{align}
\mathcal{H} &= \frac{1}{\delta}\left(\frac{c^* - 1}{(1-\delta)c_* + 2\delta - 1}\,\frac{{}_2F_1(1, (1-\delta)b_*; (1-\delta)c_* + 2\delta; -\zeta_*)}{{}_2F_1(1, b^*; c^*; -\zeta^*)}-1\right),\label{eq:H_hypergeo}\\
\mathcal{P} &=\frac{1}{ N-1}\frac{(c^* - 1) \rho(b^*,c^*,\zeta^*)}{\delta\,\,{}_2F_1(1,b_*;c_*;-\zeta_*)}\,,\label{eq:P_hypergeo}
\end{align}
with $\rho(b^*,c^*,\zeta^*)$ being bounded as $0\leq \rho(b^*,c^*,\zeta^*)\leq 1$ for all $b^*\geq 0$ and $c^*\geq 2$.
\end{theorem}

\begin{proof} We apply Lemma~\ref{lem:effrank_area} to the case where the enveloping functions $\psi_*$ and $\psi^*$ are defined as
\begin{equation*}
    \psi_*(x)= \left(1-\frac{x-1}{N-1}\right)^{c^*-2}\left(1+\zeta^*\frac{x-1}{N-1}\right)^{-b^*},\qquad \psi^*(x)= \left(1-\frac{x-1}{N-1}\right)^{c_*-2}\left(1+\zeta_*\frac{x-1}{N-1}\right)^{-b_*}\,.
\end{equation*}
Changing the variable $x$ for $y=(x-1)/(N-1)$, we get the functions
\begin{equation*}
    \Psi_*(y)= \frac{(1-y)^{c^*-2}}{(1+\zeta^*y)^{b^*}},\qquad \Psi^*(y)=\frac{(1-y)^{c_*-2}}{(1+\zeta_*y)^{b_*}}\,.
\end{equation*}
Now, using the integral representation of the hypergeometric function $_2F_1(a,b;c;z)$ \cite[\href{https://dlmf.nist.gov/15.6.E1}{(15.6.1)}]{NIST:DLMF} and the symmetry property $_2F_1(a,b;c;z)={}_2F_1(b,a;c;z)$, we get the following formulas:
\begin{equation*}
    \int_0^1\frac{(1-y)^{c-2}}{(1+\zeta y)^{b}}dy=\frac{1}{c-1}\, {}_2F_1(1,b;c;-\zeta)\,\qquad \int_0^1\left(\frac{(1-y)^{c-2}}{(1+\zeta y)^{b}}\right)^2dy=\frac{1}{2c-3} \,{}_2F_1(1,2b;2c-2;-\zeta)\,.
\end{equation*}
Thus, according to the bounds for $\mathrm{nrank}$ and $\mathrm{srank}$ provided in Lemma~\ref{lem:effrank_area}, 
\begin{align*}
     \frac{N-1}{c^*-1} \;{}_2F_1(1,b^*;c^*;-\zeta^*)\leq \,& \,\mathrm{nrank}(W)\,\leq 1+   \frac{N-1}{c_*-1}\; {}_2F_1(1,b_*;c_*;-\zeta_*)\,,\\
     \frac{N-1}{2c^*-3}\; {}_2F_1(1,2b^*;2c^*-2;-\zeta^*)\,\leq \, &\, \mathrm{srank}(W)\,\leq  1+  \frac{N-1}{2c_*-3} \;{}_2F_1(1,2b_*;2c_*-2;-\zeta_*)\,,
\end{align*}
as expected. Moreover, according to the lower bound in inequality \eqref{eq:bounding_erank_integral_tight},
\begin{equation*}
    \mathrm{erank}(W)\geq \frac{N-1}{c^*-1} \;{}_2F_1(1,b^*;c^*;-\zeta^*)
\end{equation*}
To get the upper bound of $\mathrm{erank}(W)$, we use once again inequality \eqref{eq:bounding_erank_integral_tight}:
\begin{equation*}
    \mathrm{erank}(W)\leq  \left(1+   \frac{N-1}{c_*-1} \:{}_2F_1(1,b_*;c_*;-\zeta_*)\right)e^{\mathcal{H} + \mathcal{P}}.
\end{equation*}
where $\mathcal{H}$ and $\mathcal{P}$ remain to be evaluated from Eq.~\eqref{eq:HP}. One the one hand, the integral
\begin{align*}
\int_0^1 \left(\frac{(1-y)^{c_*-2}}{(1+\zeta_* y)^{b_*}}\right)^{1-\delta}\,dy = \frac{1}{(1-\delta)c_* + 2\delta - 1}\,{}_2F_1(1, (1-\delta)b_*; (1-\delta)c_* + 2\delta; -\zeta_*)
\end{align*}
yields Eq.~\eqref{eq:H_hypergeo}. On the other hand, the bounding inequality  $\Psi_*(y)\leq 1$, valid for all $b^*\geq 1$ and $c^*\geq 2$ implies that
\begin{equation*}
    \int_{0}^{\frac{1}{N-1}}\Psi_*(y)\,dy \leq \int_{0}^{\frac{1}{N-1}}\,dy = \frac{1}{N-1}\,.
\end{equation*}
This allows us to define the function $\rho$ such that
\begin{equation*}
    \rho(b^*,c^*,\zeta^*) = (N-1)\int_{0}^{\frac{1}{N-1}}\Psi_*(y)\,dy\,,\qquad 0\leq\rho(b^*,c^*,\zeta^*)\leq 1\,,
\end{equation*}
and Eq.~\eqref{eq:P_hypergeo} follows along with the theorem.
\end{proof}
The singular-value envelopes in the latter theorem are general in the sense that they include, as particular cases, sub-linear, linear, supra-linear, and power-law decreases or mixes of those shapes. As auxiliary result, we provide the following proposition for the the sub- to supra-linear decreases which is a direct implication of Lemma~\ref{lem:effrank_area}. 
\begin{proposition}[Bounds on effective ranks for sub-linear to supra-linear decrease]\label{prop:linear_envelopes}
Suppose that the singular values of matrix $W$, $\sigma_1\geq \sigma_2\geq \cdots \sigma_N\geq 0$, satisfy the inequalities 
\begin{equation}
    \left(1-a\frac{i-1}{N-1}\right)^{b}\,\leq \,\frac{\sigma_i}{\sigma_1}\,\leq \,\left(1-c\frac{i-1}{N-1}\right)^{d},\qquad i\in\{1,\ldots, N\}, 
\end{equation}
for some $0< c\leq a\leq1$ and $0<d\leq b$. Then,
\begin{align}
    (N-1)\,\ell(a,2b,1)\,\leq \, &\, \mathrm{srank}(W)\,\leq  1+(N-1)\,\ell(c,2d, 1)\,,\label{eq:lem_bounds_srank_sublinear}\\
    (N-1)\,\ell(a,b, 1)\leq \,& \,\mathrm{nrank}(W)\,\leq 1+(N-1)\,\ell(c,d, 1)\label{eq:lem_bounds_nrank_sublinear}\,,\\
    (N-1)\,\ell(a,b, 1)\leq \,& \,\mathrm{erank}(W)\,\leq [1+(N-1)\,\ell(c,d, 1)]\,e^{\mathcal{H}+\mathcal{P}}\label{eq:lem_bounds_erank_sublinear}\,,
\end{align}
where, for all $\delta \in (0, 1)$,
\begin{align*} 
     \ell(a,b, \alpha) = \frac{1-(1-\alpha\,a)^{1+b}}{a(1+b)}, \qquad\mathcal{H} = \frac{1}{\delta}\left(\frac{\ell(c, d(1-\delta), 1)}{\ell(a,b,1)} - 1\right),\qquad\text{and}\qquad
     \mathcal{P} = \frac{\ell(a, b, \frac{1}{N-1})}{\delta\,\ell(a, b, 1)}.
\end{align*}
\end{proposition}

The results in this part of the subsection only depend on the curves enveloping the singular values and can thus be used for observed singular values of real networks or to random matrix/graph models. In the following, we relate each singular value decreases to asymptotic behaviors in random graphs.

\vspace{0.3cm}

\begin{center}\textit{2. Asymptotic behaviors of the effective ranks in growing graphs}\end{center}

We start this part by highlighting a striking consequence of Proposition~\ref{prop:exponential_envelop}~: if the singular values decrease exponentially, then the basic effective ranks are $O(1)$ as $N\to \infty$. Thus, the effective rank to dimension ratios are negligible as $N\to \infty$. This is precisely stated in the next corollary. 
\begin{corollary}[Exponential decrease implies $O(1)$ effective ranks]\label{cor:expo_o1}
Let $(\,W_N\,)_{N\in \mathbb{Z}_+}$ be an infinite sequence of matrices in which $W_N$ has size $N\times N$. Suppose that there are parameters $\alpha$ and $\omega$ such that $0<\alpha\leq \omega<1$ and for each $N$,  the singular values $\sigma_1\geq \sigma_2\geq \cdots \sigma_N\geq 0$ of $W_N$ satisfy the inequalities 
\begin{equation}
    \alpha^{i-1}\,\leq \,\frac{\sigma_i}{\sigma_1}\,\leq \,\omega^{i-1},\qquad i\in\{1,\ldots, N\}\,. 
\end{equation}
 Then, as $N\to \infty$,
\begin{align}
   \frac{1}{1-\alpha^2}+O(\alpha^{2N})\,\leq \, &\, \mathrm{srank}(W_N)\,\leq   \frac{1}{1-\omega^2}+O(\omega^{2N})\,,\\
   \frac{1}{1-\alpha}+O(\alpha^N)\leq \,& \,\mathrm{nrank}(W_N)\,\leq \frac{1}{1-\omega}+O(\omega^N),\\
   \frac{\omega}{1-\alpha}\,\exp\left(\frac{\alpha(1-\omega)}{(1-\alpha)^2}\right) + O\left(\max\{\omega^N, N\alpha^{N}\}\right) \leq \,&\,\mathrm{erank}(W_N)\,\leq \frac{\alpha}{1-\omega}\,\exp\left(\frac{\omega(1-\alpha)}{(1-\omega)^2}\right) +O(N\omega^N) \,.
\end{align} 
\end{corollary}
\begin{proof}
    We essentially expand the bounds of Proposition~\ref{prop:exponential_envelop} and look for the first subdominant terms as $N\to\infty$. On the one hand, 
    \begin{equation*}
        g(\alpha^k,N)=\frac{1-\alpha^{kN}}{1-\alpha^k} = \frac{1}{1-\alpha^k}(1-\alpha^{kN})\,.
    \end{equation*}
    Hence,
    \begin{equation*}
        \lim_{N\to \infty}\frac{\displaystyle\left|g(\alpha^k,N)- \frac{1}{1-\alpha^k}\right|}{\alpha^{kN}} = \frac{1}{1-\alpha^k}<\infty   
    \end{equation*}
    meaning that
    \begin{equation*}
        g(\alpha^k,N)=\frac{1}{1-\alpha^k}+O(\alpha^{kN})\,.
    \end{equation*}
    The last asymptotic development readily implies the bounds for $\mathrm{srank}$ and $\mathrm{nrank}$.
    On the other hand,
    \begin{equation*}
         \frac{\alpha \,g'(\alpha,N)}{g(\omega,N)}=\frac{\alpha(1-\omega)}{(1-\alpha)^2}\left(1+N\alpha^N-\alpha^N-\alpha^{N+1}\right)\left(1-\omega^N\right)^{-1},
    \end{equation*}
    which allows computing the limit 
    \begin{equation*}
        \lim_{N\to \infty}\frac{\displaystyle\left|\frac{\alpha \,g'(\alpha,N)}{g(\omega,N)}- \frac{\alpha(1-\omega)}{(1-\alpha)^2\left(1-\omega^N\right)}\right|}{N\alpha^{N}} 
        =\frac{\alpha(1-\omega)}{(1-\alpha)^2}<\infty \,.  
    \end{equation*}
    Hence,
    \begin{equation*}
         \frac{\alpha \,g'(\alpha,N)}{g(\omega,N)}=\frac{\alpha(1-\omega)}{(1-\alpha)^2\left(1-\omega^N\right)}+O\left(N\alpha^{N}\right)\,.
    \end{equation*}
    However, 
    \begin{equation*}
         \frac{\alpha(1-\omega)}{(1-\alpha)^2\left(1-\omega^N\right)}=\frac{\alpha(1-\omega)}{(1-\alpha)^2}+O(\omega^N)
    \end{equation*}
    since 
    \begin{equation*}
        \lim_{N\to \infty}\frac{\displaystyle\left|\frac{\alpha(1-\omega)}{(1-\alpha)^2\left(1-\omega^N\right)}- \frac{\alpha(1-\omega)}{(1-\alpha)^2}\right|}{\omega^{N}}=\frac{\alpha(1-\omega)}{(1-\alpha)^2}\lim_{N\to \infty}\frac{\displaystyle\left|\frac{\omega^N}{1-\omega^N}\right|}{\omega^{N}}=\frac{\alpha(1-\omega)}{(1-\alpha)^2}
    \end{equation*}
    Thus, 
    \begin{equation*}
         \frac{\alpha \,g'(\alpha,N)}{g(\omega,N)}=\frac{\alpha(1-\omega)}{(1-\alpha)^2}+O(\omega^N)+O\left(N\alpha^{N}\right)=\frac{\alpha(1-\omega)}{(1-\alpha)^2}+O\left(\max\{\omega^N, N\alpha^{N}\}\right)\,,
    \end{equation*}
    where we have invoked the basic property $O(u)+O(v)=O(\max\{u,v\})$.
    Consequently, the lower bound of $\mathrm{erank}$ has the following asymptotic expansion:
    \begin{align*}
        \omega\,g(\alpha,N)\,\exp\left(\frac{\alpha \,g'(\alpha,N)}{g(\omega,N)}\right)&=\left(\frac{\omega}{1-\alpha}+O(\alpha^{N})\right)\exp\left(\frac{\alpha(1-\omega)}{(1-\alpha)^2}+O\left(\max\{\omega^N, N\alpha^{N}\}\right)\right)
         \\ &=\left(\frac{\omega}{1-\alpha}+O(\alpha^{N})\right)\exp\left(\frac{\alpha(1-\omega)}{(1-\alpha)^2}\right)\left(1+O\left(\max\{\omega^N, N\alpha^{N}\}\right)\right)\\&=\left(\frac{\omega}{1-\alpha}+O\left(\max\{\omega^N, N\alpha^{N}\}\right))\right)\exp\left(\frac{\alpha(1-\omega)}{(1-\alpha)^2}\right),
    \end{align*}
    where the second line has been deduced using the well-known asymptotic formulas $O(u)O(v)=O(uv)$ and $O(u)+O(v)=O(\max\{u,v\})$. The upper bound for $\mathrm{erank}$ is obtained from the last result by permuting $\alpha$ and $\omega$, and considering $O(\max\{N\omega^N,\alpha^N\})=O(N\omega^N)$.
\end{proof}

In light of Lemma~\ref{lem:effrank_area}, which bounds the effective ranks with terms proportional to $N-1$, the previous asymptotic behavior was rather surprising. On the contrary, the next result is fully expected: slowly decreasing envelopes lead to effective ranks that grow linearly with $N$. 

\begin{corollary}[Sub-linear to supra-linear decrease imply $O(N)$ effective ranks]\label{cor:sub_to_supra_asymptotics}
Let $(\,W_N\,)_{N\in \mathbb{Z}_+}$ be an infinite sequence of matrices in which $W_N$ has size $N\times N$. Suppose that there are parameters $a$, $b$, $c$, $d$ such that $0< c\leq a\leq1$ and $0<d\leq b$, and for each $N$, the singular values $\sigma_1\geq \sigma_2\geq \cdots \sigma_N\geq 0$ of $W_N$ satisfy the inequalities 
\begin{equation}\label{eq:lem_bounds_srank_subsupra}
        \left(1-a\,\frac{i-1}{N-1}\right)^{b}\,\leq \,\frac{\sigma_i}{\sigma_1}\,\leq \,\left(1-c\,\frac{i-1}{N-1}\right)^{d},\qquad i\in\{1,\ldots, N\}\,.  
\end{equation}
 Then, as $N\to \infty$ and for all $\delta \in (0, 1)$,
\begin{align}    
   N\,\ell(a, 2b, 1)+O(1)\,\leq \, &\, \mathrm{srank}(W_N)\,\leq  N\,\ell(c, 2d, 1)+O(1)\,,\\
   N\,\ell(a, b, 1)+O(1)\leq \,& \,\mathrm{nrank}(W_N)\,\leq N\,\ell(c, d, 1)+O(1),\\
  N\,\ell(a, b, 1)+O(1) \leq \,&\,\mathrm{erank}(W_N)\,\leq N\,\ell(c, d, 1)\exp\left[\frac{1}{\delta}\left(\frac{\ell(c, d(1-\delta), 1)}{\ell(a,b,1)} - 1\right)\right]+O(1) \,,
\end{align} 
where $\ell$ is the function defined in Proposition~\ref{prop:linear_envelopes}. 
\end{corollary}

So far, we have obtained effective ranks that have either $O(1)$ or $O(N)$ asymptotic behaviors as $N\to \infty$. We are going to prove $O(N^{1-\epsilon})$ asymptotic behaviors are also possible for all $\epsilon>0$.

\begin{corollary}[Hypergeometric decrease admits $O(N^{1-\epsilon})$ effective ranks]\label{cor:hypergeometric}
Let $(\,W_N\,)_{N\in \mathbb{Z}_+}$ be an infinite sequence of matrices in which $W_N$ has size $N\times N$. Suppose that there are parameters $b_*$, $b^*$, $c_*$, $c^*$, $\zeta_*$, $\zeta^*$ such that $0\leq b_*\leq b^* $, $2\leq c_*\leq c^*$, $0<\zeta_*\leq \zeta^*$ and such that for each $N$, the singular values $\sigma_1\geq \sigma_2\geq \cdots \sigma_N\geq 0$ of $W_N$ satisfy   
\begin{equation*}
     \left(1-\frac{i-1}{N-1}\right)^{c^*-2}\left(1+\zeta^*\frac{i-1}{N-1}\right)^{-b^*} \,\leq \,\frac{\sigma_i}{\sigma_1}\,\leq \,  \left(1-\frac{i-1}{N-1}\right)^{c_*-2}\left(1+\zeta_*\frac{i-1}{N-1}\right)^{-b_*}\qquad \forall\;i\in\{1,\ldots, N\}\,.
\end{equation*}
\textit{1.} If $c_*=c^*=N^{\epsilon}/d$ for some $d>0$ and $\epsilon\in(0,1]$, then as $N\to \infty$,
\begin{align*}
    \frac{d}{2}\,N^{1-\epsilon}+O(N^{1-2\epsilon})\,\leq \, &\, \mathrm{srank}(W_N)\,\leq 1+ \frac{d}{2}\,N^{1-\epsilon}+O(N^{1-2\epsilon})\,,\\
    d\,N^{1-\epsilon}+O(N^{1-2\epsilon})\leq \,& \,\mathrm{nrank}(W_N)\,\leq 1+d\,N^{1-\epsilon}+O(N^{1-2\epsilon})\,,\\ 
   d\,N^{1-\epsilon}+O(N^{1-2\epsilon}) \leq\,&\,\mathrm{erank}(W_N)\,\leq  (1 + d\,N^{1-\epsilon})e^{\frac{1}{1-\delta}} + O(\max(1, N^{1-2\epsilon}))\,,
\end{align*} 
where the last inequality holds for all $\delta\in(0,1)$.

\vspace{0.1cm}
 \noindent\textit{2.} If $b_*>1$ , $\zeta_*=\zeta^*=N^\epsilon/d$ for some $d>0$ and $\epsilon\in(0,1]$, and $b_*, b^*\notin \mathbb{Z}$, then 
 \begin{align*}    
     \qquad \frac{d}{2b^*-1}N^{1-\epsilon}+O(N^{1-2\epsilon})\,\leq \, &\, \mathrm{srank}(W_N)\,\leq  1+   \frac{d}{2b_*-1}N^{1-\epsilon}+O(N^{1-2\epsilon})\,,\\
     \frac{d}{b^*-1}N^{1-\epsilon}+O(N^{1-2\epsilon})  \leq \,& \,\mathrm{nrank}(W_N)\,\leq 1+   \frac{d}{b_*-1}N^{1-\epsilon}+O(N^{1-2\epsilon})\,,\\
    \frac{d}{b^*-1}N^{1-\epsilon}+O(N^{1-2\epsilon})\leq \,& \,\mathrm{erank}(W_N)\,\leq \left(1 + \frac{d}{b_*-1}N^{1-\epsilon}\right)\,\exp\left(\frac{b^* - b_\delta}{\delta(b_\delta - 1)}\right) + O(\max\{1, N^{1 - b_\delta\epsilon }, N^{1 - 2\epsilon}\})\,,   
    \end{align*}
where the last inequality holds for all $\delta\in (0,1-1/b_*)$ and $b_\delta := (1- \delta)b_*$. 

\vspace{0.1cm}
\noindent\textit{3.} If $b^*<1$, $\zeta_*=\zeta^*=N^\epsilon/d$ for some $d>0$, and $\epsilon\in(0,1]$, and $b_*, b^*\notin \mathbb{Z}$, then
\begin{align*}    
     g(2b^*,2c^*-2,d)\,N^{1-2b^*\epsilon}+O\left(N^{1-(2b^*+1)\epsilon}\right)\,&\,\leq \, \mathrm{srank}(W_N)\,\leq   1+  g(2b_*,2c_*-2,d)\,N^{1-2b_*\epsilon}+O\left(N^{1-(2b_*+1)\epsilon}\right),\\
     g(b^*,c^*,d)\,N^{1-b^*\epsilon}+O\left(N^{1-(b^*+1)\epsilon}\right)  \,&\,\leq \,\mathrm{nrank}(W_N)\,\leq \, 1+   g(b_*,c_*,d)\,N^{1-b_*\epsilon}+O\left(N^{1-(b_*+1)\epsilon}\right)\,, \phantom{X^{X^{X}}} \\      
     g(b^*,c^*,d)\,N^{1-b^*\epsilon}+O\left(N^{1-(b^*+1)\epsilon}\right)\,&\,\leq \, \mathrm{erank}(W_N) \,\leq \left(1+  g(b_*,c_*,d)\,N^{1-b_*\epsilon}\right)e^{y(N)}\,+ O\left(N^{1-\gamma\epsilon}\,e^{y(N)}\right)\,,\nonumber
    \end{align*}
    where the last inequality holds for all $\delta \in (0, 1+(1-2b^*)/b_*)$ with $b^* > 1/2$ and $b_* > 2b^* - 1$, $\gamma = 1-2(b^* - b_*) + \delta b_*$, and
    \begin{align}\label{eq:gbcd}
        g(b,c,d)=\frac{\Gamma(1-b)\Gamma(c - 1)\,d^{b}}{\Gamma(c-b)},\qquad
        y(N) = \frac{g((1 - \delta)b_*, (1 - \delta)c_* + 2\delta, d)}{\delta\, g(b^*, c^*, d)}\,N^{\epsilon(b^* - b_\delta)} - \frac{1}{\delta}\,.
    \end{align}
\end{corollary}
\begin{proof}
    First, we use the following asymptotic expansion for $c\to \infty$ \cite[\href{https://dlmf.nist.gov/15.12.E2}{(15.12.2)}]{NIST:DLMF}:
    \begin{equation*}
    {}_2F_1(a,b;c;-\zeta)=1-\frac{ab}{c}\zeta+O(c^{-2}).
    \end{equation*}
    Hence, for $c=N^{\epsilon}/d$,
    \begin{equation*}
    \frac{N-1}{c-1}\;{}_2F_1(a,b;c;-\zeta)=d N^{1-\epsilon}\left(1+O(N^{-\epsilon})\right),\qquad\frac{N-1}{2c-3}\;{}_2F_1(a,b;c;-\zeta)=\frac{d}{2} N^{1-\epsilon}\left(1+O(N^{-\epsilon})\right),
    \end{equation*}
    The substitution of the last equations into the bounds of Theorem~\ref{thm:bounds_effective_ranks_hypergeometric} readily provides the desired inequalities for srank, nrank, and the lower bound of erank. For the upper bound of the erank, we need to get the asymptotics of
    \begin{equation*}
            \mathcal H =\frac{1}{\delta}\left(\frac{c^* - 1}{(1-\delta)c_* + 2\delta - 1}\,\frac{{}_2F_1(1, (1-\delta)b_*; (1-\delta)c_* + 2\delta; -\zeta_*)}{{}_2F_1(1, b^*; c^*; -\zeta^*)}-1\right)\,.
    \end{equation*}
    However,     
    \begin{align*}
        \frac{{}_2F_1(1, (1-\delta)b_*; (1-\delta)c_* + 2\delta; -\zeta_*)}{{}_2F_1(1, b^*; c^*; -\zeta^*)} = 1 + O(N^{-\epsilon}),\qquad\frac{c^* - 1}{(1-\delta)c_* + 2\delta - 1} =\frac{1}{1-\delta} +O(N^{-\epsilon}),
    \end{align*}
    leading to the conclusion that $\mathcal{H}=\frac{1}{1-\delta}+O(N^{-\epsilon})$.
    Moreover, 
    \begin{equation*}
        \mathcal P = \frac{1}{ N-1}\frac{(c^* - 1) \rho(b^*,c^*,\zeta^*)}{\delta\,\,{}_2F_1(1,b_*;c_*;-\zeta_*)}=O(N^{\epsilon-1})\,,
    \end{equation*}
    from which we deduce that asymptotic expansions 
    \begin{equation*}
        \mathcal{H} + \mathcal{P} =\frac{1}{1-\delta} + O(N^{-\epsilon}) + O(N^{\epsilon-1})
        \qquad\text{and}\qquad e^{\mathcal{H} + \mathcal{P}}=e^{\frac{1}{1-\delta}}\left(1+O(N^{-\epsilon}) + O(N^{\epsilon-1})\right)\,
    \end{equation*}
    where the latter result holds for all $\delta\in(0,1)$. Hence, 
      \begin{equation*}
        \mathrm{erank}(W_N)\,\leq \Big(1+d\,N^{1-\epsilon}+O(N^{1-2\epsilon})\Big)e^{\frac{1}{1-\delta}}\Big(1+O(N^{-\epsilon}) + O(N^{\epsilon-1})\Big)=e^{\frac{1}{1-\delta}}+d\,e^{\frac{1}{1-\delta}}\,N^{1-\epsilon}+O(\max(1, N^{1-2\epsilon}))
    \end{equation*}
    as expected.
    
    Second, we use the following identity valid for $|z|>1$ and $a-b\notin \mathbb{Z}$ \cite[\href{https://dlmf.nist.gov/15.8.E2}{(15.8.2)}]{NIST:DLMF}:
    \begin{multline*}
        {}_2F_1(a,b;c;z)=\frac{\Gamma(b-a)\Gamma(c)}{\Gamma(c-a)\Gamma(b)}\frac{1}{(-z)^a}\;{}_2F_1\left(a,a-c+1;a-b+1;\frac{1}{z}\right)\\+\frac{\Gamma(a-b)\Gamma(c)}{\Gamma(c-b)\Gamma(a)}\frac{1}{(-z)^b}\;{}_2F_1\left(b-c+1, b; b-a+1;\frac{1}{z}\right)
    \end{multline*}
    Thus, for $a=1$ and $z=-\zeta$, 
    \begin{equation*}
        {}_2F_1(1,b;c;-\zeta)=\frac{c-1}{b-1}\frac{1}{\zeta}\;{}_2F_1\left(1,2-c;2-b;-\frac{1}{\zeta}\right)+\frac{\Gamma(1-b)\Gamma(c)}{\Gamma(c-b)}\frac{1}{\zeta^b}\;{}_2F_1\left(b-c+1, b; b;-\frac{1}{\zeta}\right)
    \end{equation*}
    However, according to  \cite[\href{https://dlmf.nist.gov/15.2.E1}{(15.2.1)}]{NIST:DLMF},
    \begin{equation*}
    {}_2F_1\left(\alpha,\beta;\gamma;-\frac{1}{\zeta}\right)=1-\frac{\alpha\beta}{1!\gamma}\frac{1}{\zeta}+\frac{\alpha(\alpha+1)\beta(\beta+1)}{2!\gamma(\gamma+1)}\frac{1}{\zeta^2}-+\ldots=1+O(\zeta^{-1})
    \end{equation*}
    The last two results imply that as $\zeta\to\infty$ ($\zeta := \zeta^* = \zeta_* = N^\epsilon/d$),
    \begin{equation}\label{eq:expansion_zeta}
        {}_2F_1(1,b;c;-\zeta)= f(b,c,\zeta)+O(\zeta^{-2})+O(\zeta^{-b-1}),\qquad f(b,c,\zeta)=\frac{c-1}{b-1}\frac{1}{\zeta}+\frac{\Gamma(1-b)\Gamma(c)}{\Gamma(c-b)}\frac{1}{\zeta^b}\,.
    \end{equation}
    Substituting this result into Theorem \ref{thm:bounds_effective_ranks_hypergeometric}, we get the following asymptotic expansion for $\zeta\to\infty$:
    {\footnotesize
    \begin{align*}
     \frac{N-1}{2c^*-3}\; \left(f(2b^*,2c^*-2,\zeta)+O(\zeta^{-2})+O(\zeta^{-2b^*-1})\right)\,\leq \, &\, \mathrm{srank}(W)\,\leq  1+  \frac{N-1}{2c_*-3}\left(f(2b_*,2c_*-2,\zeta)+O(\zeta^{-2})+O(\zeta^{-2b_*-1})\right),\\
     \frac{N-1}{c^*-1} \left(f(b^*,c^*,\zeta)+O(\zeta^{-2})+O(\zeta^{-b^*-1})\right)\ \leq \,& \,\mathrm{nrank}(W)\,\leq 1+   \frac{N-1}{c_*-1}\left(f(b_*,c_*,\zeta)+O(\zeta^{-2})+O(\zeta^{-b_*-1})\right),\\
    \frac{N-1}{c^*-1} \left(f(b^*,c^*,\zeta)+O(\zeta^{-2})+O(\zeta^{-b^*-1})\right)\leq \,& \,\mathrm{erank}(W)\,\leq \left(1+   \frac{N-1}{c_*-1}\left(f(b_*,c_*,\zeta)+O(\zeta^{-2})+O(\zeta^{-b_*-1})\right)\right)\,e^{\mathcal{H} + \mathcal{P}}.
    \end{align*}}Setting $\zeta=N^\epsilon/d$ and simplifying the resulting expressions lead to the desired inequalities for srank, nrank, and the lower bound of erank. Regarding the upper bound of the erank, the use of Eq.~\eqref{eq:expansion_zeta} allows us to write
\begin{equation*}
    \frac{{}_2F_1(1, b_{\delta}; c_{\delta}; -\zeta_*)}{{}_2F_1(1, b^*; c^*; -\zeta^*)} = \frac{(c_\delta -1)(b^*-1)}{(b_\delta - 1)(c^* - 1)}\left[\frac{1 - r(b_\delta, c_\delta, d) \,\,N^{\epsilon(1 - b_{\delta})} + O(N^{-\epsilon}) + O(N^{-\epsilon b_{\delta}})}{1 - r(b^*, c^*, d)\,\, N^{\epsilon(1-b^*)} + O(N^{-\epsilon})  + O(N^{-\epsilon b^*})}\right],
    \end{equation*}
where $b_\delta := (1- \delta)b_*$, $c_\delta := (1-\delta)c_* + 2\delta$, and 
\begin{equation*}
r(b, c, d) = \frac{\Gamma(2 - b)\,\Gamma(c-1)\,d^{b - 1}}{\Gamma(c-b)}\,.
\end{equation*}
Moreover, setting $b_\delta > 1$ implies that
\begin{equation*}
    \frac{{}_2F_1(1, b_{\delta}; c_{\delta}; -\zeta_*)}{{}_2F_1(1, b^*; c^*; -\zeta^*)} = \frac{(c_\delta -1)(b^*-1)}{(b_\delta - 1)(c^* - 1)}\left(1 + O\left(N^{\epsilon(1 - b_\delta)}\right)\right)\left(1 + O\left(N^{\epsilon(1 - b^*)}\right)\right) = \frac{(c_\delta -1)(b^*-1)}{(b_\delta - 1)(c^* - 1)} + O\left(N^{\epsilon(1 - b_\delta)}\right)
\end{equation*}
along with
\begin{equation*}
    \mathcal{H} 
    = \frac{b^* - b_\delta}{\delta(b_\delta - 1)} + O\left(N^{\epsilon(1 - b_\delta)}\right)\qquad\text{and}\qquad\mathcal{P} = O\left(N^{\epsilon - 1}\right).
\end{equation*}
Hence, 
\begin{equation*}
    e^{\mathcal{H} + \mathcal{P}} = \exp\left(\frac{b^* - b_\delta}{\delta(b_\delta - 1)}\right)\left(1 + O\left(N^{\epsilon(1 - b_\delta)}\right) + O\left(N^{\epsilon - 1}\right)\right)
\end{equation*}
and the upper bound of the erank is
\begin{align*}
\mathrm{erank}(W_N) &\leq\left(1+  \frac{d}{b_*-1}N^{1-\epsilon}+O(N^{1-2\epsilon})\right)\,\exp\left(\frac{b^* - b_\delta}{\delta(b_\delta - 1)}\right)(1 + O(N^{\epsilon(1 - b_\delta)}) + O(N^{\epsilon - 1}))\\
&= \left(1 + \frac{d}{b_*-1}N^{1-\epsilon}\right)\,\exp\left(\frac{b^* - b_\delta}{\delta(b_\delta - 1)}\right) + O(\max\{1, N^{1 - \epsilon b_\delta}, N^{1-2\epsilon}\})\,.
\end{align*}

Third, setting $b^* < 1$ gives
\begin{equation*}
    \frac{{}_2F_1(1, b_{\delta}; c_{\delta}; -\zeta_*)}{{}_2F_1(1, b^*; c^*; -\zeta^*)} = \frac{g(b_{\delta}, c_\delta, d)}{g(b^*, c^*, d)}\,N^{\epsilon(b^* - b_\delta)} + O(N^{\epsilon(2b^* - b_\delta - 1)})\,,
\end{equation*}
where the function $g$ is defined in Eq.~\eqref{eq:gbcd}. This leads to 
\begin{equation*}
    \mathcal{H} = \frac{g(b_{\delta}, c_\delta, d)}{\delta\, g(b^*, c^*, d)}\,N^{\epsilon(b^* - b_\delta)} - \frac{1}{\delta} + O(N^{\epsilon(2b^* - b_\delta - 1)}).
\end{equation*}
With $\mathcal{P} = O(N^{b^*\epsilon - 1})$, we find
\begin{equation*}
e^{\mathcal{H} + \mathcal{P}} = e^{y(N)}\,e^{O(N^{\epsilon(2b^* - b_\delta - 1)})}\qquad \text{with} \qquad y(N) = \frac{g(b_{\delta}, c_\delta, d)}{\delta\, g(b^*, c^*, d)}\,N^{\epsilon(b^* - b_\delta)} - \frac{1}{\delta}
\end{equation*}
If $2b^* - b_\delta  - 1 < 0$, or equivalently, $b^* < \frac{b_\delta + 1}{2}$, then 
\begin{align*}
    \mathrm{erank}(W_N) &\leq \left(1+  g(b_*,c_*,d)\,N^{1-b_*\epsilon}+O\left(N^{1-(b_*+1)\epsilon}\right)\right)e^{y(N)}(1 +O(N^{\epsilon(2b^* - b_\delta - 1)}))\,\\
    &= \left(1+  g(b_*,c_*,d)\,N^{1-b_*\epsilon}\right)\,e^{y(N)} + O(N^{1 - (1 - 2(b^* - b_*))\epsilon + \delta b_* \epsilon}\,e^{y(N)})
\end{align*}
as desired.
\end{proof}

\begin{remark}
    For the upper bound of the erank, we also note the following. In the case 2, additionally, if $b^*-b_*=\delta$ for some $\delta\in (0,1-1/b_*)$, then
\begin{equation*}
    \exp\left(\frac{b^* - b_\delta}{\delta(b_\delta - 1)}\right) = \exp\left(\frac{b_* + 1}{b_* - 1}\right) + O(\delta)\,.
\end{equation*}
One can also simplify the upper bound for the erank in the case 3 by setting $b^*-b_* =\delta$ and $c^*-c_*=\gamma\delta$, by considering a small $\delta$, and by using Stirling's formula for the gamma functions in $g(b_\delta, c_\delta, d)$.
\end{remark}

It is worth emphasizing that the hypergeometric envelopes of the previous corollary, given their generality, not only admit $O(N^{1-\epsilon})$ growth of the effective ranks, but can also produce $O(1)$ and $O(N)$ growths. Indeed, when $b_*=b^*=0$, one recovers special sub-linear to supra-linear decreasing envelopes included in Corollary~\ref{cor:sub_to_supra_asymptotics}, leading to $O(N)$ effective ranks. When $c_*=c^*=2$ and $b_*=b^*\to \infty$, one gets exponentially decreasing envelopes as in Corollary~\ref{cor:expo_o1}, corresponding to $O(1)$ effective ranks. Finally, when $c_*=c^*=2$ while $b_*$ and $b^*$ remain finite, one instead obtains power-law decreasing envelopes and it can be shown, using the asymptotics of the Hurwitz zeta function, that it leads to $O(N^{1-\epsilon})$ effective ranks.    

We have thus proved that different choices for $\psi_*$ and $\psi^*$ can induce very distinct asymptotic behaviors of the above-mentioned effective ranks as $N\to \infty$. Figure~\ref{fig:asymptotic-ranks-theory} depicts these findings by showing different singular-value envelopes, leading to three different classes of maximum growth of $\mathrm{nrank}$ as $N$ becomes large: linear $O(N)$, sub-linear $O(N^{1-\epsilon})$ with $0<\epsilon<1$, and constant $O(1)$.

\begin{figure}[h]
    \centering
    \includegraphics[scale=0.6]{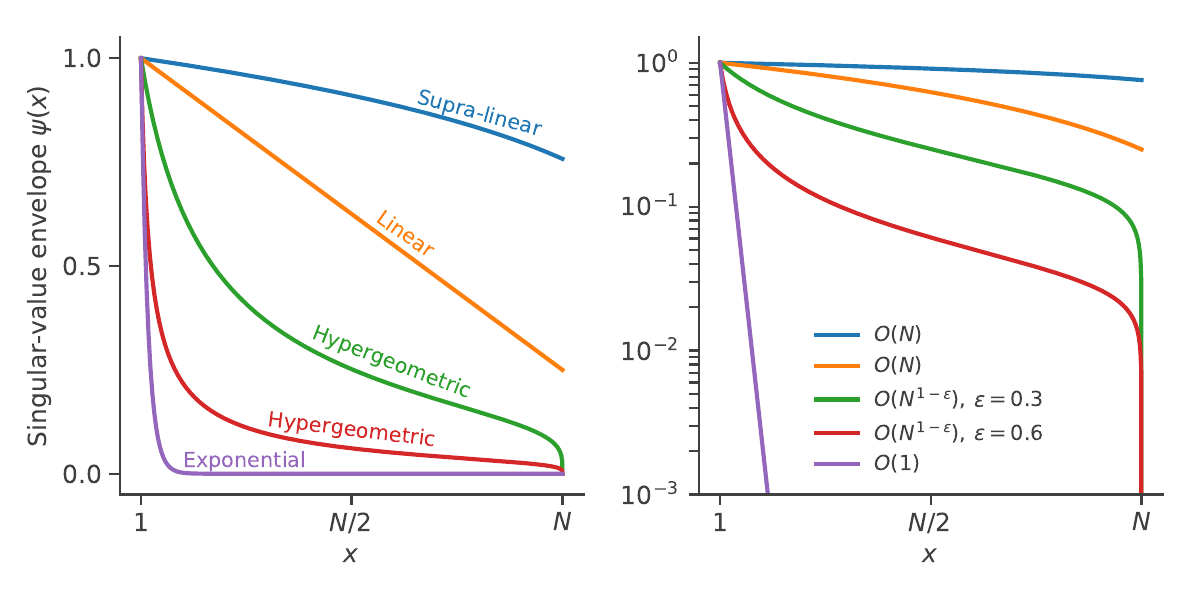}
    \vspace{-0.3cm}
    \caption{Typical singular-value envelopes $\psi(x)$, describing the decreasing behavior of the normalized singular values $\sigma_i/\sigma_1$, vs. the continuous variable $x$, interpolating between the indices $i$, and their impact on the asymptotics of $\mathrm{srank}$, $\mathrm{nrank}$, and $\mathrm{erank}$ (bottom-right panel). 
    The scale of values for $\psi(x)$ is linear on the left while it is logarithmic on the right. From top to bottom, the decreasing functions are: $\left(1-a\frac{x-1}{N-1}\right)^{b}$ with $a=\frac{3}{4}$, $b=\frac{1}{5}$ (blue), and $b=1$ (orange);  $\left(1-\frac{x-1}{N-1}\right)^{c-2}/\left(1+\zeta\frac{x-1}{N-1}\right)^{b}$ with $b=\frac{3}{4}$, $c=2\frac{1}{4}$, $\zeta = N^{0.3}$ (green), and $\zeta = N^{0.6}$ (red); $\omega^{x-1}$ with $\omega=0.97$ (purple). For all cases, $N=2000$.
    }
    \label{fig:asymptotic-ranks-theory}
\end{figure}

As previously mentioned, the hypergeometric case generalizes several types of decrease, including the power-law decrease ($c_*=c^*=2$). The latter has been observed in the adjacency spectrum of scale-free random graphs \cite{chung2003spectra} and in the eigenspectra of covariance matrices in various settings, including fractional Brownian motion \cite{gao2003principal},  EEG time series \cite{sanchez2021criticality}, neuronal activity in the mouse \cite{stringer2019high,Stringer2019} and macaque \cite{kong2022increasing} visual cortex.

In their seminal work, Stringer et al. \cite{stringer2019high} examined the eigenspectrum of the covariance matrix derived from signals of large populations of neurons in the visual cortex of awake mice viewing natural images.  They argued that the evoked neuronal population activity in this context is ``high-dimensional''.  Their conception of high dimensionality is anchored in sophisticated theoretical findings, including methods from functional data analysis and fractal geometry. As they noted \cite[Supplementary information p.7]{stringer2019high}: \emph{These results demonstrate that unless eigenspectra decay faster than \(n^{-1}\) , population codes are pathological, either exhibiting discontinuous responses, or infinite population variance. Furthermore, for stimuli drawn from a set of manifold dimension \(d\), codes with eigenspectra decaying slower than \(n^{-1-2/d}\) are also pathological, displaying infinite variance of the code’s derivative and fractal geometry of the response manifold. We conclude that our experimental observations of eigenspectrum decay only just faster than \(n^{-1-2/d}\) indicate a neural code that is as high-dimensional as possible before hitting the regime where these pathological conditions must occur.} Therefore, the dimensionality of the neural code is deemed ``high'' when the eigenvalues of the covariance matrix decrease as slow as possible, nearing the threshold indicative of pathological responses, and ``low'' when the decrease is faster, significantly distanced from this critical threshold.

However, no measure of dimension is used to quantify the decrease of the eigenvalues. In fact, as shown below, the decrease of the covariance-matrix eigenvalues implied by the above power-law is fast enough to lead to effective ranks of order strictly less than $O(N)$, where $N$ corresponds to the number of neurons, thus suggesting low dimensionality from our perspective. Indeed, let us assume that the visual stimuli's manifold dimension \( d_\mathrm{stim} > 2 \) and that the \( n \)-th eigenvalue of the covariance matrix is less than \( n^{-1-2/ d_\mathrm{stim}} \) as described in \cite{stringer2019high}, implying that its \( n \)-th singular value is bounded above by \( n^{-1/2 -1/ d_\mathrm{stim} } \). Moreover, let us assume without loss of generality that $\sigma_1 = 1$. This scenario aligns with the third case of Corollary~\ref{cor:hypergeometric} for $b_* = 1/2+1/ d_\mathrm{stim}<1$, $c_*=2$, and $\zeta_*=N-1 \approx N$ (i.e., $d=1$, $\epsilon =1$). We can deduce from this that $\mathrm{srank}$ and $\mathrm{nrank}$  asymptotically grow as $1+O(N^{-2/ d_\mathrm{stim}} )$ 
and $O(N^{1/2-1/ d_\mathrm{stim}})$, respectively. This presents an intriguing case where at least two effective ranks exhibit completely different asymptotic behaviors. 
Yet, in this power-law scenario, both the stable rank to dimension ratio and nuclear rank to dimension ratio tend to zero as $N$ grows to infinity.

\vspace{1cm}

\begin{center}
    \textit{3. Impact of matrix density on the stable rank}
\end{center}

Let us now derive some intuitive inequalities for the stable rank of graphs based on inequalities for the weight matrix.
\begin{lemma}\label{lem:bounds_norms} 
Let $W$ be a $N\times N$ matrix. Then the Frobenius norm of $W$ is upper bounded as 
\begin{equation}
    \displaystyle \|W\|_F\leq N \max_{i,j}|W_{ij}|\,.
\end{equation}
Moreover, the spectral norm of $W$ is lower bounded as 
\begin{equation}
    \|W\|_2\geq \max\Big\{\,\max_i\|\bm r_i\|,\,\max_j\|\bm c_j\|, \,\frac{1}{\sqrt{N}}\|\bm k^{\mathrm{in}}\|,\, \frac{1}{\sqrt{N}}\|\bm k^{\mathrm{out}}\|\Big\}\,,
\end{equation}
where $\bm r_i$ and $\bm c_j$ respectively denote the $i$-th row and $j$-th column of $W$ while $\bm k^{\mathrm{in}}=W\bm 1$ and $\bm k^{\mathrm{out}}=W^\top\bm 1$.
\end{lemma}
\begin{proof}
The first inequality trivially follows from the definition of the Frobenius norm:
\begin{equation*}
    \displaystyle \|W\|_F=\sqrt{\sum_{i,j}W_{ij}^2}\leq \sqrt{\sum_{i,j}W_{\max{}}^2}=\sqrt{N^2 W_{\max{}}^2}= N W_{\max{}},\qquad W_{\max{}}= \max_{i,j}|W_{ij}|\,.
\end{equation*}
The second inequality is the maximum between four lower bounds. To derive them, we start with the definition
\begin{equation*}
    \|W\|_2=\max_{\|\bm x\|=1}\left\|W\bm x\right\|,
\end{equation*}
which implies that $\|W\|_2\geq \left\|W\bm x\right\|$ for any $\bm x$ such that $\|\bm x\|=1$. Choosing $\bm x=\bm e_j$, the $j$-th unit vector, leads to the inequality $\|W\|_2\geq \|\bm c_j\|$. But this inequality holds all $j$, so we can combine all the inequalities and infer that
\begin{equation*}
    \|W\|_2\geq \max_j\|\bm c_j\|.
\end{equation*}
Now, because the spectral norm is invariant under matrix transposition, we also know that $\|W\|_2\geq \left\|W^\top\bm x\right\|$ for any $\bm x$ such that $\|\bm x\|=1$. Setting once again $\bm x=\bm e_i$ for all $i$, we conclude that 
\begin{equation*}
    \|W\|_2\geq \max_i\|\bm r_i\|.
\end{equation*}
Choosing $\bm x=\bm 1/\sqrt{N}$ in $\|W\|_2\geq \left\|W\bm x\right\|$ and in $\|W\|_2\geq \left\|W^\top\bm x\right\|$ readily provides the two other lower bounds. 
\end{proof}

\begin{proposition}\label{thm:upper_bound_srank_graph}
 Let $W$ be the adjacency matrix of a directed graph of $N$ vertices and $M$ edges. Moreover, let $k_{\max{}}$ be the maximum among all ingoing and outgoing degrees of the graph. Then,
\begin{equation}
    \mathrm{srank}(W)\leq \frac{M}{k_{\max{}}}\,.
\end{equation}
\end{proposition}
\begin{proof}
We first note that when $W$ is an adjacency matrix, all its elements are either 0 or 1, which implies that its Frobenius norm squared is exactly equal to $M$. Indeed, 
\begin{equation}
    \|W\|_F^2= \sum_{i,j}W_{ij}^2= \sum_{\substack{(j,i)\\j\to i}} 1=M.
\end{equation}
Moreover, we know from the previous lemma that the spectral norm squared is bounded by the degrees: 
\begin{equation*}
    \|W\|_2^2 \geq \max\Big\{\, \max_i \sum_jW_{ij}^2,\, \max_j \sum_iW_{ij}^2\Big\}\\
    = \max\Big\{\, \max_i \sum_jW_{ij},\, \max_j \sum_iW_{ij}\Big\}
    = \max \Big\{\, \max_i k_i^\mathrm{in},\, \max_j k_j^\mathrm{out}\Big\}=k_{\max{}}\,.
\end{equation*}
Thus, $\mathrm{srank}(W)= {\|W\|_F^2}/{\|W\|_2^2}\leq {M}/{k_{\max{}}}$, as expected.
\end{proof}

In dense directed graphs of $N$ vertices, the number of edges $M$ typically scales as $N^2$ while the maximum degree scales as $N$. The previous proposition thus implies that the stable rank is of order $O(N)$ for such graphs. A slightly different scaling law exists for sparse graphs. Indeed, if $M=O(N^{2-\varepsilon})$ and $k_{\max{}}=O(N^{1-\epsilon})$ for some $\varepsilon > \epsilon>0$, then the stable rank is of order $O(N^{1+\epsilon-\varepsilon})=o(N)$. As shown in next proposition, similar scaling behaviors emerge when considering general square matrices, which are relevant for studying signed weighted directed graphs.
\begin{proposition}\label{thm:upper_bound_srank}
Let $p\geq 0$ and $0<\alpha \leq \beta$. Let $W$ be a $N\times N$ matrix such that 
\begin{equation}\label{eq:condition_scaling_W}
    \alpha N^{-p}\leq W_{ij}^2\leq \beta N^{-p}\,
\end{equation}
for all $1\leq i,j\leq N$. Then, the stable rank satisfies the inequality 
\begin{equation}
    \mathrm{srank}(W)\leq \alpha^{-1}\beta\, N\,.
\end{equation}
More generally, if the maximum number of nonzero elements in a row or in a column of $W$ is $\gamma N$, the total number of nonzero elements of $W$ is $\delta N^2$, and all these nonzero elements satisfy inequality \eqref{eq:condition_scaling_W}, then
\begin{equation}
    \mathrm{srank}(W)\leq \alpha^{-1}\beta \gamma^{-1}\delta\,N\,.
\end{equation}
\end{proposition}
\begin{proof}We use Lemma~\ref{lem:bounds_norms} and proceed essentially as for Proposition~\ref{thm:upper_bound_srank_graph}.  
\end{proof}
Sparse $N\times N$ matrices are characterized by a total number of nonzero elements of order strictly less than $N^2$ and a maximum number of nonzero elements in each row or column of order strictly less than $N$. In the last proposition, this situation corresponds to the case where $\gamma = \bar \gamma\,N^{-\epsilon}$ and $\delta = \bar \delta N^{-\varepsilon}$ for some $\varepsilon>\epsilon>0$, which implies that once again,   $\mathrm{srank}=O(N^{1+\epsilon-\varepsilon})=o(N)$. A typical sparse matrix has $\varepsilon=2\epsilon$, leading to a stable rank scaling as $O(N^{1-\varepsilon})$, which tends to $O(1)$ when considering the ultra-sparse case $\epsilon=1$. In words, the stable rank of (signed weighted directed) graphs having $N$ vertices increases at most linearly with $N$ and sparsity makes the increase become sub-linear. This means that sparse graphs are characterized by a ratio $\mathrm{srank}/N$ that goes to zero as $N$ grows, obviously corresponding to a low effective rank.

\subsection{Directed network centrality measures}
\label{SIsubsec:centrality}
For a directed network, the eigenvalues and the eigenvectors of its matrix representation will generally be complex and one have to adapt the usual approach to define a centrality. A natural way of doing that is to use the SVD of the directed network, which provides two vertex centrality measures: the authority centrality (dominant left singular vector $u_1$) and the hub centrality (dominant right singular vector $v_1$) \cite{Kleinberg1998, Newman2018_SI}, as illustrated in Fig.~\ref{fig:SVD_centralities}. This remark guided us in choosing the observables of the reduced dynamics and it can be used to give an interpretation to the different terms and equations involved when applying Theorem~\ref{thm:least_square_optimal_vector_field} with the reduction matrix being the right singular vectors. Note, however, that for signed networks (described by matrices with negative values), these centrality measures may lead to ambiguities, since the first left and right singular vectors generally have negative values (Perron-Frobenius theorem~\cite[Theorem 38]{VanMieghem2011} doesn't apply).

\begin{figure}[h]
    \centering
    \includegraphics[width=0.9\linewidth]{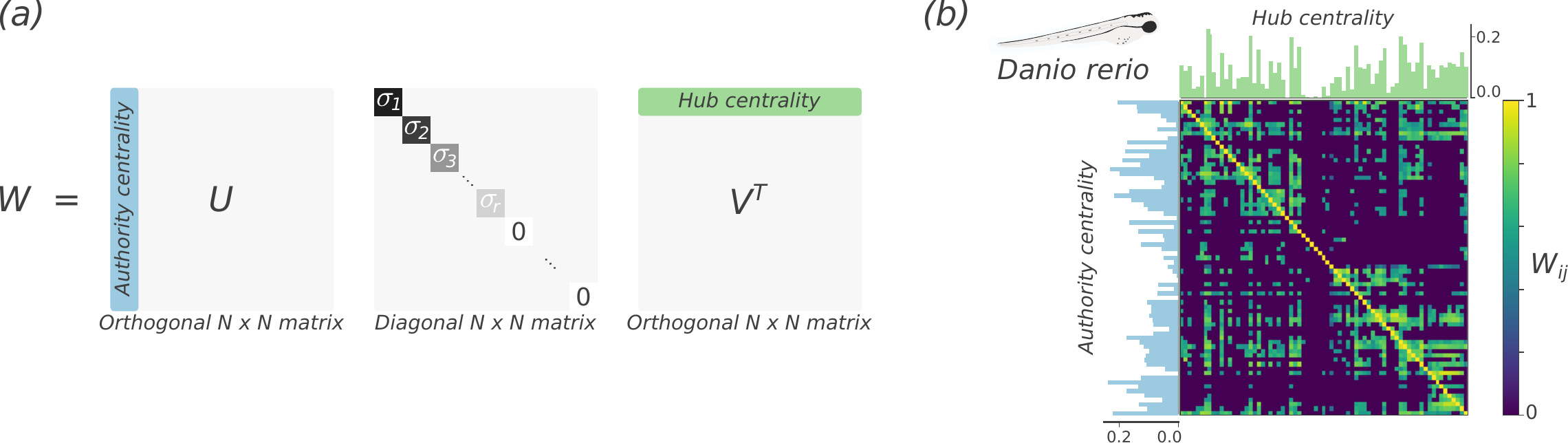}
    \caption{(a) Authority and hub centralities are provided by the elements of the dominant left and right singular vectors, respectively. (b) Centralities for the mesoscopic connectome of the larval zebrafish with $N=71$ communities with added self-loops (modified from Ref.~\cite{Kunst2019a}).}
    \label{fig:SVD_centralities}
\end{figure}

\subsection{Adaptive networks}
\label{SIsubsec:adaptive_networks}
Complex systems are not only characterized by their nonlinear dynamics and network structure, but also by their capacity to adapt themselves to environmental changes \cite{Mitchell2009}. The effective rank of a complex network should thus be expected to change according to time. We performed a preliminary investigation of this phenomenon by extracting the effective rank of the \textit{C. elegans} connectome at different stages of its maturation~\cite{Witvliet2021} as shown in Fig~\ref{fig:singvals_maturation}. We observed that the stable rank decreases with age. More work should be done on this subject to verify if this decrease is significant and to determine the biological meaning of an effective rank decrease with maturation. 

\begin{figure}[h]
    \centering
    \includegraphics[width=0.5\linewidth]{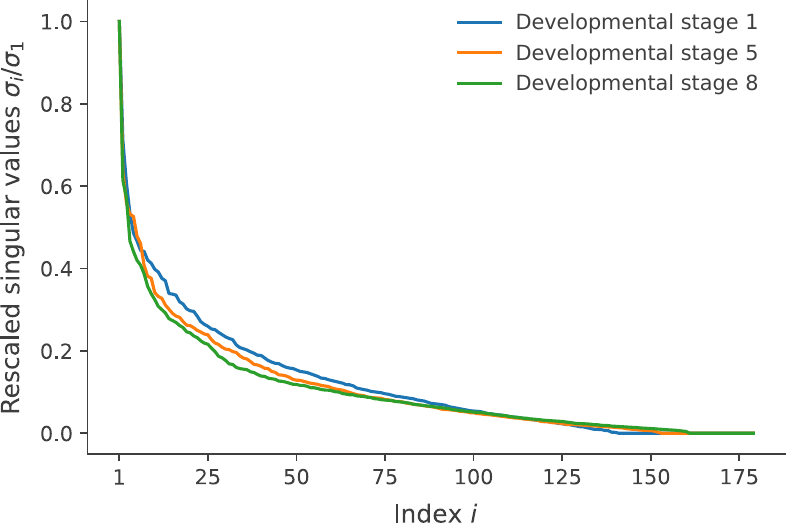}
    \caption{Singular values of the matrices describing the connectivity of the \textit{C. elegans} brain at different maturation stages. The stable ranks are 21.6 (developmental stage 1), 19.7 (developmental stage 5), 18.5 (developmental stage 8).}
    \label{fig:singvals_maturation}
\end{figure}

In Ref.~\cite{Martin2021a}, the authors numerically show that training a neural network decreases the stable rank, which is somewhat in line with what we observe in the latter biological example.

\subsection{SVD for dynamical systems}
\label{SIsubsec:svd_dynamical_systems}
The applications of SVD for dynamical systems is very broad, especially in engineering and linear control systems \cite{Antoulas2005}. SVD is also generalized for nonlinear operators \cite{Fujimoto2004} and it is even possible to perform a quasi-optimal low-rank approximation for matrix dynamics with time-evolving matrices $U,\Sigma, V^{\top}$ \cite{Koch2007}, which could have interesting applications in the study of temporal networks~\cite{Holme2012}. As illustrated in the paper, one can also leverage the power of SVD in the dimension reduction of dynamical systems on networks. As explained in Ref.~\cite[Appendix C]{Thibeault2020_SI}, it can be very hard to choose adequately the reduction matrix $M$. Having real nonnegative singular values and real singular vectors is an advantage when it comes to interpreting the spectra and to define interpretable observables for the dynamics (as opposed to eigenvalue decomposition for general real matrices, which can raise the problem of dealing with complex reduction matrices and create complex reduced dynamics for an initially real dynamics~\cite[p.145]{Thibeault2020_master}\cite{Thibeault2020_SI}). In the following section, we give details about the dimension reduction of complex systems and especially, in subsection~\ref{SIsubsec:error_bound}, we show how to use the salient properties of SVD to get insights on the low-rank hypothesis of complex systems. 

\clearpage
\section{Dimension reduction of complex systems}
\label{SIsec:dimension_reduction}
Dimension reduction of high-dimensional dynamics is a powerful technique to get analytical and numerical insights on complex systems. For instance, it helps predict the onset of explosive phenomena~\cite{Kuehn2021} or getting suitable observable to assess the controllability of the system~\cite{Montanari2022}. The range of applications of dimension-reduction techniques is therefore very broad---ranging from statistical physics and chemistry to finance and neuroscience--- and the methods substantially differ along with the terminology: dimension reduction~\cite{Wang2005, Cook2022}, coarse graining~\cite{Espanol2003, Castiglione2008, Cho2017}, reduced-order model \cite{Brunton2019}, model reduction~\cite{Smith2020}, lumping \cite{Wei1969,Toth1997} \cite[Section 2.4]{Kiss2017a}, compression~\cite{Machta2013}, pruning~\cite{Hoefler2021}, dominance analysis \cite{Forni2019}, variable or state aggregation~\cite{Faccin2021}, etc. 
Many useful dimension-reduction techniques remain unused for complex systems which may be a consequence of this great diversity of terminologies. In this section, we give details about dimension reduction of ordinary differential equations from its more general aspects to the specific ones used in the paper.

\subsection{Notation and generalities on dimension reduction}
\label{SIsubsec:generalities}
\noindent Consider the following notation for the complete dynamical systems:
\begin{itemize}
    \item $x \in \mathbb{R}^N$ is a state of the system;
    \item $t \in [0,\infty)$ denotes time;
    \item $\phi: [0,\infty) \times \mathbb{R}^N \mapsto \mathbb{R}^N$ is the flow;
    \item $x: [0,\infty) \to \mathbb{R}^N$ is the trajectory (note the abuse of notation with the state);
    \item $f: \mathbb{R}^N \to \mathbb{R}^N$ is the vector field, assumed to be  continuously differentiable;
    \item $x_0 = x(0)$ is the initial condition;
    \item $\dot{x} = f(x)$ is the complete dynamics, or more explicitly, 
    \begin{align*}
        \begin{pmatrix}
            \dot{x}_1\\\vdots\\\dot{x}_N
        \end{pmatrix} = \begin{pmatrix}
            f_1(x_1, ..., x_N)\\\vdots\\f_N(x_1, ..., x_N).
    \end{pmatrix}
    \end{align*}
\end{itemize}
Consider the following notation for the reduced dynamical system:
\begin{itemize}
    \item $R:\mathbb{R}^N \mapsto \mathbb{R}^n$ with $n<N$ is called the reduction function or a vectorial observable;
    \item $R = (R_1,...,R_n)$ where $R_\mu : \mathbb{R}^N \mapsto \mathbb{R}$ is the $\mu$-th observable;
    \item $X = R(x) \in \mathbb{R}^n$ is a reduced state;
    \item $\Phi: [0, \infty) \times \mathbb{R}^n \mapsto \mathbb{R}^n$ is the reduced flow;
    \item $X = R\circ x:[0, \infty) \to \mathbb{R}^n$ is the reduced trajectory (note the abuse of notation with the reduced state);
    \item $F: \mathbb{R}^n \to \mathbb{R}^n$ is the reduced vector field, assumed to be  continuously differentiable;
    \item $X_0 = R(x_0) = (R\circ x)(0)$ is the initial condition;
    \item $\dot{X} = F(X)$ is the reduced dynamics.
\end{itemize}

\vspace{0.3cm}

The logic behind the notation is that the ``microscopic" objects are in lowercase and the ``macroscopic" objects are in uppercase, except for $N$ and $n$ which denote some high dimension and a lower dimension respectively. Latin indices are used for these microscopic objects, while Greek indices are used for the macroscopic objects. With this notation, we now define what we mean by exact dimension reduction, in a similar spirit as Ref.~\cite{Toth1997}, but avoiding the subtleties in the characteristics of the reduction function $R$.  
\begin{definition}\label{def:dimension_reduction}
The function $R:\mathbb{R}^N \mapsto \mathbb{R}^n$ induces an \textit{exact dimension reduction} of the dynamics
\begin{align}\label{eq:complete_dynamics_general}
    \dot{x} &= f(x)
\end{align}
if there exists a vector field $F: \mathbb{R}^n \to \mathbb{R}^n$ such that for all solutions $x(t)$ of Eq.~\eqref{eq:complete_dynamics_general}, 
the reduced trajectory
\begin{align}
    X = R\circ x: [0, \infty) \to \mathbb{R}^n
\end{align}
obeys the differential equation
\begin{align}\label{eq:reduced_dynamics_general}
    \dot{X} &= F(X)\,.
\end{align}
\end{definition}
The pair of functions $(R,F)$ thus characterizes a dimension reduction, where the goal is to \textit{close} the differential equation for $X$ in terms of $X$ solely. Dimension reduction can also be seen as a special commutation relation of the vector fields and the flows.
\begin{theorem}\label{thm:dimension_reduction}
The following statements are equivalent:
\begin{enumerate}
    \item the dimension reduction is exact;
    \item the general \textit{compatibility equation}
    \begin{equation}\label{eq:nonlinear_compatibility_equation}
        \mathcal{U}[R] = J_R f = F \circ R
    \end{equation}
    holds, where $\mathcal{U}$ is the Koopman operator generator and $J_R$ is the Jacobian matrix of $R$;
    \item the complete flow $\phi_t$ and the reduced flows $\Phi_t$ commutes with $R$ such that
    \begin{equation}\label{eq:commutation_flows_reduction}
        R\circ \phi_t = \Phi_t \circ R.
    \end{equation}
\end{enumerate}
\end{theorem}
\begin{proof}

   \phantom{Sarah}\\
   
    \vspace{-0.3cm}
    
    \noindent(1. $\Leftarrow$ 2.) By definition, $X = R\circ x$ and by assumption, $J_R f = F \circ R$. Then, the time derivative of $X$ (the generator of the Koopman operator) is
    \begin{equation}
        \dot{X} = \frac{\mathrm{d} (R \circ x)}{\mathrm{d}t} = \mathcal{U}[R]\circ x = J_R f\circ x = F\circ R \circ x = F \circ X,
    \end{equation}
    which is the definition of an exact dimension reduction.
    
    \noindent(1. $\Rightarrow$ 2.) Similarly, using the time derivative of $X = R\circ x$ again, we have
    \begin{equation}
        \dot{X} = \mathcal{U}[R]\circ x = J_R f\circ x.
    \end{equation}
    But the dimension reduction is exact and $\dot{X} = F\circ X = F \circ R\circ x$ holds. Then, by comparison, it is sufficient to have $\mathcal{U}[R]\circ x = J_R f = F \circ R$.
    
    \noindent(1. $\Leftrightarrow$ 3.) On the one hand, the solution of $\dot{x} = f(x)$ is $x(t) = \phi_t(x(0))$ and thus, the exact evolution of $X(t)$ is given by $X(t) = R\circ\phi_t\circ x(0)$. On the other hand, the solution to $\dot{X} = F(X)$ with $X(0) = R(x(0))$ is $X(t) = \Phi_t(X(0)) = \Phi_t\circ R \circ x(0)$. The comparison gives the desired result.
\end{proof}
Since we have commutation relations, there is a clear picture with commutative diagrams. In particular, statement~3. tells us that that we have an exact dimension reduction if there is a commutative diagram such that
\begin{equation}
    \begin{tikzcd}[sep=huge]
        \mathbb{R}^N \arrow[r, "\phi_t"]
        \arrow[d, "R" left] \arrow[dr, dotted, shift right=0.7ex, "  \Phi_t\, \circ \,  R" below left] \arrow[dr, dotted, shift left=0.7ex, "  R\,\circ \, \phi_t"]
        & \mathbb{R}^N \arrow[d, "R"] \\
        \mathbb{R}^n \arrow[r, "\Phi_t" below] &  \mathbb{R}^n
    \end{tikzcd}
\end{equation}
In the article and the rest of the Supplementary information, we focus on the case where $R$ is a \textit{linear transformation}, which greatly simplifies the analysis and gives access to a whole range of notions and tools from linear algebra. Let us thus assume that $X = R(x) = Mx$ where $M$ is a $n\times N$ matrix, called the reduction matrix \cite{Thibeault2020_SI} (or lumping matrix~\cite{Wei1969, Kuo1969}). Then, $J_R = M$ and condition~\eqref{eq:nonlinear_compatibility_equation} for closure states that for an exact dimension reduction, the complete and reduced vector fields must commute with $M$:
 \begin{equation}\label{diag:dimension_reduction}
\begin{tikzcd}[sep=huge]
\mathbb{R}^N \arrow[r, " f"] \arrow[d, "M" left]
\arrow[dr, dotted, shift right=0.7ex, "  F\,\circ \,  M" below left] \arrow[dr, dotted, shift left=0.7ex, "  M\,\circ \, f"]
& \mathbb{R}^N \arrow[d, "M"] \\
\mathbb{R}^n \arrow[r, "F" below]
&  \mathbb{R}^n\, 
 \end{tikzcd}\quad 
\end{equation}
where we have made a slight abuse of notation, using the same symbol for the matrix and the linear transformation $M: x \mapsto Mx$, that we will use again in the document. Note that the latter scheme is related to the notions of $C^k$-equivalent and $C^k$-conjugate vector fields defined in Ref.~\cite[p.190 and p.191]{Perko2001}. In subsection~\ref{SIsubsec:error_bound} [Definition~\ref{def:alignment_error}], we introduce the alignment error which is directly defined from the compatibility equation $M\circ f = F\circ M$  and we will find a bound on it.

\begin{remark}
    In our work, we consider that the network of the system is already known (or could be known experimentally) and the dynamics is described by a given theoretical model, but the time series/functional data (trajectories) are unknown. This is the ideal setting for determining how the low effective rank of the weight matrix $W$ can affect the evolution of the state of the whole system, starting with arbitrary initial conditions, since no limitation in our analysis can be induced by the finite number of observed time series or their finite length.
\end{remark}

Given a reduction matrix $M$, a projector can always be defined as 
\begin{equation}
    P = M^+ M\,,
\end{equation} 
where $M^+$ is the Moore-Penrose pseudo-inverse of $M$. Under this linear setup, the dimension reduction can be seen as a projection of the elements $x$ of the high-dimensional space unto a low-dimensional space with elements $X$. This situation as well as the four natural vector subspaces induced by $M$ are illustrated in Fig.~\ref{fig:projection}. 

\begin{figure}
    \centering
    \includegraphics[width=0.6\linewidth]{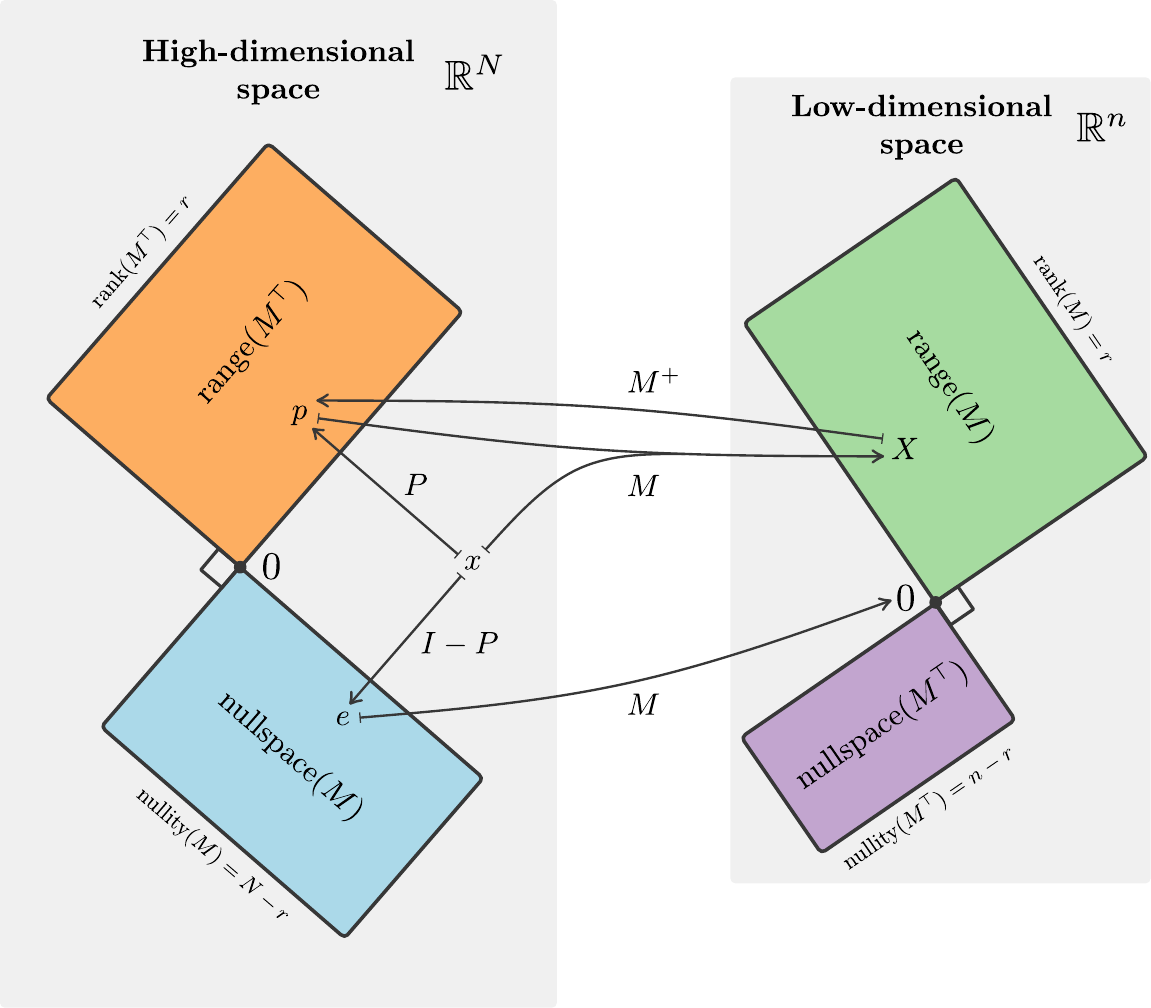}
    \caption{Schematization of dimension reduction associated with the reduction matrix $M$, the corresponding projector $P$, and the induced vector subspaces.}
    \label{fig:projection}
\end{figure}

In general, it is far from simple to solve the compatibility equations $M\circ f = F\circ M$ for $F$ and $M$. Even when $f$ and $F$ are linear transformations, respectively encoded by the $N\times N$ matrix $W$ and the $n\times n$ matrix $\mathcal{W}$, the condition $M\circ f = F\circ M$  takes the form of the compatibility equation~\cite{Thibeault2020_SI}
\begin{equation}
    MW = \mathcal{W}M\,
\end{equation}
which is in fact a system of coupled quadratic equations in the elements of $\mathcal{W}$ and $M$ that cannot always be solved analytically. However, for a fixed $M$, one can find a unique optimal reduced matrix $\mathcal{W}$.
\begin{theorem}[Ref.~\cite{Thibeault2020_SI}]\label{thm:existence_solution_compatibility}
Let $M$ and $W$ be respectively of size $n\times N$ and $N\times N$ with $n < N$. Then, the compatibility equation $\mathcal{W}M = MW$ has a solution for $\mathcal{W}$ if and only if
 \begin{equation}\label{eq:nec_suff_condition}
     M W M^+M = MW
 \end{equation} 
 where $M^{+}$ is the Moore-Penrose pseudoinverse of $M$, in which case the solution is
 \begin{equation}\label{eq:nonunique_solution_compatibility}
     \mathcal{W} = MWM^+ + Y - YMM^+,
 \end{equation}
 where $Y$ is an arbitrary $n\times n$ matrix. If $\rank{M} = n$, then there is at most one solution to the compatibility equation,~i.e.,
 \begin{equation}
     \mathcal{W} = MWM^+.
 \end{equation}
For any $M$, this solution minimizes $\|\mathcal{W}M - MW\|_F$ with error $\|MW(I - M^+M)\|_F$.
\end{theorem}  
\begin{proof}
From Penrose 1955 \cite[Theorem 2]{Penrose1955}, a necessary and sufficient condition for the equation $AXB = C$ to have a solution is $AA^+CB^+B = C$, in which case the general solution is $X = A^+CB^+ + Y - A^+AYBB^+$, where $Y$ is arbitrary. Set $A = I$, $X = \mathcal{W}$, $B = M$, $C = MW$ in the latter equations and the result in Eq.~\eqref{eq:nonunique_solution_compatibility} follows. If $\rank{M} = n$, then the $n$ rows of $M$ are linearly independent. This implies that $MM^+ = I$ and 
 \begin{equation*}
     \mathcal{W} = MWM^+ + Y - Y = MWM^+,
 \end{equation*}
 which does not depend on the arbitrary matrix $Y$ anymore. It is thus the only possible solution. 
 
 Finally, it is well known, at least since the least-squares theorem of Penrose in 1956~\cite{Penrose1956}, that
  \begin{equation*}
     \argmin_{U\in\mathbb{R}^{k \times \ell}}\|UA-V\|_F= VA^+
\quad\text{and}\quad
     \min_{U\in\mathbb{R}^{k \times \ell}}\|UA-V\|_F = \|V(I-A^+A)\|_F,
 \end{equation*}
 for matrices $V\in \mathbb{R}^{k\times m}$ and $A\in \mathbb{R}^{\ell \times m}$. Setting $A=M$, $U =\mathcal{W}$, and $V = MW$ implies that $\mathcal{W} = MWM^+$ minimizes $\|\mathcal{W}M - MW\|_F$ with error $\|MW(I - M^+M)\|_F$.
\end{proof}
As it will be discussed in another paper, the first preliminary results on these compatibility equations seems to go back to 1969 in chemistry~\cite{Wei1969} and for a fixed $\mathcal{W}$, the compatibility equations are homogeneous Sylvester equations (1884)~\cite{Sylvester1884}.

In the next section, we provide a way to find an optimal reduced vector field $F$ given a reduction matrix $M$,
thus generalizing the idea behind Theorem~\ref{thm:existence_solution_compatibility}.

\subsection{Least-square optimal vector field}
\label{SIsubsec:least_square_vf}
Low-dimensional dynamical systems can be obtained from an optimization problem, where some error is minimized under a set of constraints \cite{Boyd2004} in order to preserve the salient properties of the original high-dimensional system. For dynamical systems, a natural optimization variable is the reduced vector field $F$ itself, which is chosen to represent approximately the complete vector field $f$. Yet, it is rather puzzling to find how the different vector field errors are related to each other and which one can be minimized analytically. We found that there was a useful scheme that helps solve this puzzle. Recalling the definitions of subsection~\ref{SIsubsec:generalities}, we introduce the following diagram for dimension reduction of dynamical systems:
\begin{equation}\label{diag:alignment_errors} 
\begin{tikzcd}[sep=large]
x \arrow[rrrr, maps to, crossing over,"f"] \arrow[drr, maps to, "P"] \arrow[dddr,  bend right, maps to, "M"]   & & & &  f(x)\arrow[drrr, dashed, no head, "\varepsilon(x)"] \arrow[ddr, bend right, maps to, "M" below left] & &\\
& & Px \arrow[rrr, crossing over, "f"]  \arrow[ddl, bend right, maps to, "M" above left]& & & f(Px) \arrow[rr, dashed, no head, "\varepsilon'(x)"] \arrow[ddrr, maps to, bend right, "M" above right] & & M^+F(X) \\
&&&&&Mf(x)\arrow[d, dashed, no head, "\mathcal{E}(x)"]\\
& X \arrow[rrrr, maps to, crossing over,"F" below] \arrow[uur, bend right, maps to, "M^+"] & & &  & F(X) \arrow[uurr, bend right, maps to,crossing over, "M^+" above left] \arrow[rr, dashed, no head, "\mathcal{E}'(x)" below]& & Mf(Px)\\
\end{tikzcd}
\, ,
\end{equation}
where $P=M^+M$ and the dashed lines represent root-mean-square errors (RMSE) between adjacent vector fields, i.e., different alignment errors as defined below [see also Fig.~\ref{fig:low_dimension_hypothesis} for an illustration of $\mathcal{E}(x)$].
\begin{definition}\label{def:alignment_error}
Let $f$ be a complete vector field in $\mathbb{R}^N$, $F$ be a reduced vector field in $\mathbb{R}^n$, and $M$ be the $n\times N$ reduction matrix. At $x \in \mathbb{R}^N$, the \textit{alignment error} ...

\vspace{0.3cm}

\begin{itemize}
    \item ... in $\mathbb{R}^N$ is the RMSE between the vector fields $f$ and $M^+\circ F\circ M$, i.e.,
\begin{align}\label{eq:alignment_error_N}
    \varepsilon(x) = \frac{1}{\sqrt{N}}\|f(x) - M^+F(Mx)\|\,;
\end{align}
\item ... in $\mathbb{R}^n$ is the RMSE between the vector field $M\circ f$ and $F\circ M$, i.e.,
\begin{align}\label{eq:alignment_error_n}
    \mathcal{E}(x) = \frac{1}{\sqrt{n}}\|Mf(x) - F(Mx)\|\,,
\end{align}
\end{itemize}
 where $\|\,\|$ is the Euclidean vector norm.
\end{definition}
By applying the definition of alignment errors on the projected complete vector field $f\circ P$ instead of $f$ only, we have defined
\begin{align}\label{eq:alignment_error_N_projected}
    \varepsilon'(x) = \frac{1}{\sqrt{N}}\|f(Px) - M^+F(Mx)\|\,
\end{align}
and
\begin{align}\label{eq:alignment_error_n_projected}
    \mathcal{E}'(x) = \frac{1}{\sqrt{n}}\|Mf(Px) - F(Mx)\|\,
\end{align}
in Diagram~\ref{diag:alignment_errors}. In principle, the alignment error $\mathcal{E}(x)$ in $\mathbb{R}^n$ is to be minimized in order to be as close as possible to an exact dimension reduction [Definition~\ref{def:dimension_reduction}, Theorem~\ref{thm:dimension_reduction}, and Diagram~\ref{diag:dimension_reduction}], but this is far from a simple task. However, the alignment error $\varepsilon'(x)$ can be directly minimized using least squares which has for consequence that the alignment error $\mathcal{E}'(x)$ in $\mathbb{R}^n$ is exactly 0, as shown in the following theorem.
\begin{theorem}\label{thm:least_square_optimal_vector_field}
    Let $f$ be a complete vector field in $\mathbb{R}^N$, $F$ be a reduced vector field in $\mathbb{R}^n$, and $M$ be a $n\times N$ reduction matrix. The vector field of the reduced dynamics
    \begin{equation}
        \dot{X} = Mf(M^+X)
    \end{equation}
    is optimal in the sense that it minimizes the alignment error $\varepsilon'(x)$ in $\mathbb{R}^N$, i.e.,
    \begin{equation}\label{eq:least_square_optimal_vector_field}
        F^*(X) = \argmin_{\substack{F(X)\in\mathbb{R}^n \\ X=Mx}}\|f(Px)- M^+F(X)\| = Mf(M^+X).
    \end{equation}
    Consequently, the alignment error $\mathcal{E}'(x)$ in $\mathbb{R}^n$ is 0.
\end{theorem}
 \begin{proof}
 Let $v\in \mathbb{R}^k$ and $A\in \mathbb{R}^{k\times \ell}$. Then, using least squares (particular case of Penrose~\cite{Penrose1956}) implies that
 \begin{equation}
     \argmin_{u\in\mathbb{R}^\ell}\|v - Au\|= A^+v.
 \end{equation}
Setting $A=M^+$, $u=F(X)$, $v=f(Px)=f(M^+Mx)=f(M^+X)$ readily yields the result. Since $F^*(X) = Mf(M^+X)$ and $Px = M^+X$, we obviously have $\|Mf(Px) - F^*(X)\| = 0$.
 \end{proof}
 
\begin{remark}

\phantom{henri}

\begin{itemize}

\item Minimizing $\varepsilon'(x)$ does not tell much about the alignment error $\mathcal{E}(x)$ of interest. Yet, in subsection~\ref{SIsubsec:error_bound}, we find that using the ensuing vector field from the minimization of $\varepsilon'(x)$ allows obtaining an upper bound on $\mathcal{E}(x)$.

\item Recalling the optimal solution $M W M^+$ for the compatibility equation $MW = \mathcal{W}M$ in Theorem~\ref{thm:existence_solution_compatibility}, we observe that we now have an optimal solution (involving a nonlinear vector field) $M\circ f \circ M^+$ for the compatibility equation $M\circ f = F\circ M$ that boils down to the previous linear solution when $f=W$ and $F=\mathcal{W}$.

\item When we set $n=N$, we could expect the ``reduced'' vector field $F$ to be equivalent in some way to the complete vector field. In fact, if $f:\mathbb{R}^N \to \mathbb{R}^N$ is a $\mathcal{C}^1(\mathbb{R}^N)$ vector field, the vector fields $f$ and $F = M\circ f \circ M^+ = M\circ f \circ M^{-1}$ are $\mathcal{C}^1$-conjugate on $\mathbb{R}^N$ \cite[p.191]{Perko2001}, which is straightforward to observe from the form of $F$ itself where $M$ is the $\mathcal{C}^1$-diffeomorphism. 

\item To the authors' knowledge, even if the vector field in Eq.~\eqref{eq:least_square_optimal_vector_field} is known at least since 1989~\cite{Li1989}, the result hasn't been stated and proved clearly, simply, and in a general way for dynamical systems described by a set of differential equations. One can find many papers on the method (e.g., in fluid mechanics and chemistry) \cite{Li1989, Li1990, Rowley2004, Carlberg2017} and especially, on a similar-looking technique for time series which is also loosely~\cite{Editors2012, Gander2012} called Galerkin projection or Petrov-Galerkin method \cite{Saad2003, Rowley2004, Antoulas2005, Brunton2019}. In our paper, we recall that it is implicitly assumed that we do not have access to the time series, only the initial vector field with the network is known.

\item In principle, there is a whole world of objective functions that could be used for the optimization problem. Other constraints and regularization terms could also be added to satisfy the modeler's restrictions. This is a promising avenue to be further explored in the future to obtain optimal reduced dynamical systems.
\end{itemize}

\end{remark}
Let us now apply the latter theorem to one of the most influential models in neuroscience, the Wilson-Cowan model \cite{Wilson1972, destexhe2009wilson, chow2020before, Painchaud2022} \cite[Chap. 11]{Ermentrout2010}.
\begin{example}[Neuroscience]\label{ex:wilson_cowan}
Consider a system of $N = N_E + N_I$ neurons (or neuronal population) with $N_E$ excitatory neurons and $N_I$ inhibitory neurons. Let $E_e$ (resp. $I_i$) be the time-averaged firing rate of the $e$-th excitatory neuron for $e \in \mathcal{E} = \{1,..., N_E\}$ (resp. $i$-th inhibitory population for $i \in \mathcal{I} = \{N_E + 1,..., N\}$). The Wilson-Cowan model~\cite{Wilson1972} describes the time evolution of the firing rates as
\begin{align}
    \dot{E}_e &= -d_e E_e + (1 - aE_e)\,\, \mathcal{S}[b(\textstyle{\sum_{e'=1}^{N_E}} W_{ee'}E_{e'} + \textstyle{\sum_{i'=N_E+1}^{N}} W_{ei'}I_{i'} - c_e)]\\
    \dot{I}_i &= -d_i I_i + (1 - aI_i)\,\, \mathcal{S}[b(\textstyle{\sum_{e'=1}^{N_E}} W_{ie'}E_{e'} + \textstyle{\sum_{i'=N_E+1}^{N}} W_{ii'}I_{i'}  - c_i)],
\end{align}
where $d_y$ is the inverse time constant and $a$ is related to the refractory period. Moreover, for all $i,i'\in\mathcal{I}$ and $e,e'\in\mathcal{E}$, $W_{ee'} \geq 0$, $W_{ie'} \geq 0$, $W_{ei'} \leq 0$, $W_{ii'} \leq 0$, and 
\begin{equation}
    \mathcal{S}[b(z - c)] = \frac{1}{1 + e^{-b(z-c)}}
\end{equation}
is the logistic function with $b$ being its steepness and $c$ being its midpoint or physically, an external input. By defining 
\begin{equation}\label{eq:EI_to_x}
    (x_1,...,x_N) := (E_1,...,E_{N_E}, I_{N_E+1},...,I_N)^\top,
\end{equation}
we get a concise form of the model \cite[Eq.~(11)]{Painchaud2022}:
\begin{equation}\label{eq:wilson_cowan}
    \dot{x}_j = -d_j x_j + (1 - ax_j) \mathcal{S}[b(\gamma \,y_j - c)], \quad\forall\,j\in\{1,...,N\},
\end{equation}
where $y_j = \sum_{k=1}^N W_{jk}x_k$ and we have set $W \to \gamma W$ to have a coupling constant $\gamma$ to tune. Note that the excitatory and inhibitory variables don't have to be labeled and ordered as above and the weight matrix $W$ just describes a general signed network. From Theorem~\ref{thm:least_square_optimal_vector_field}, we directly obtain the optimal reduced dynamics
\begin{equation}
    \dot{X}_{\mu} = \sum_{\nu=1}^n\mathcal{D}_{\mu\nu} X_{\nu} + \sum_{j=1}^NM_{\mu j}(1 - a\textstyle{\sum_{\nu=1}^n}M_{j\nu}^+X_{\nu})\,\mathcal{S}\left[b(\textstyle{\gamma\sum_{\nu=1}^n}\mathscr{W}_{j\nu}X_{\nu} - c)\right],
\end{equation}
where $\mathcal{D}_{\mu\nu} = -\sum_{j=1}^N M_{\mu j}d_jM_{j\nu}^+$ and $\mathscr{W}_{j\nu} = \sum_{k=1}^N W_{jk}M_{k\nu}^+$. 
\end{example}
Under the form $Mf(M^+X)$ or, elements by elements, $\sum_{i=1}^N M_{\mu i} f_i(\sum_{\nu=1}^n M_{i\nu}^+X_{\nu})$, there is still an explicit dependence of the vector field over $N$. Yet, we can sometimes eliminate this dependence by simplifying $Mf(M^+X)$ under certain properties of $f$ which reveals something special about the resulting interaction between the observables.
 
\subsection{Emergence of higher-order interactions}
\label{SIsubsec:emergence}
The critical role of higher-order interactions in complex systems is now increasingly recognized \cite{Grilli2017, Levine2017, Lambiotte2019b, Battiston2020a, Battiston2021} and in this section, we aim at clarifying their origin by demonstrating the profound interplay between the description dimension of a system and the possibility of having higher-order interactions.
When reducing the dimension of a dynamical system on a network, it is not always clear what to expect about the structure of the reduced dynamical system [see Fig.~\ref{fig:low_dimension_hypothesis} in the paper]. We demonstrate that the structure that emerges from the dimension reduction in Theorem~\ref{thm:least_square_optimal_vector_field} generally yields higher-order interactions between the observables. For that, we first introduce some assumptions.

\begin{assumptions}\label{ass:emergence}

\phantom{Henri}

\begin{enumerate}
    \itemsep1em
    \item[(1)] 
    The $N$-dimensional dynamics on a network of weight matrix $W$ is
    \begin{align}
    \dot{x}_i = h_i(x_i, y_i),\quad i\in\{1,...,N\},
    \end{align}
    where, for all $i$, $x_i:t \mapsto \mathbb{R}^N$, $y_i = \sum_{j=1}^N W_{ij}x_j$, and $h_i: \mathbb{R}\times \mathbb{R} \to  \mathbb{R}$ is an analytic function. 

    \item[(2)] The $n$-dimensional reduced dynamics is the least-square optimal dynamics [Theorem~\ref{thm:least_square_optimal_vector_field}]
    \begin{align}
    \dot{X}_{\mu} = \sum_{i=1}^N M_{\mu i}h_i(\tilde{x}_i, \tilde{y}_i), \quad \mu \in \{1,...,n\},
    \end{align}
    where $X = Mx$ with any real reduction matrix $M$, $\tilde{x} = M^+X$, and $\tilde{y} = WM^+X$.
\end{enumerate}
\end{assumptions}
Condition (1) of Assumptions~\ref{ass:emergence} might look restrictive because of the dependence over the linear function $x\mapsto Wx$. Yet, a considerable amount of complex system models satisfy condition (1) as shown in the following examples (from the power series in $x_i, y_i$ of their analytic vector field, it is possible to classify the dynamics on networks of the next examples). 
\begin{example}[Epidemiology]\label{ex:sis}
In the Susceptible-Infected-Susceptible (SIS) dynamics, an infected individual $i$ (e.g., from a virus or disinformation) transmits its infection at a rate $\gamma$ and recovers with rate $d_i$. In its exact form, the SIS dynamics is a homogeneous Markovian jump process and is described by master equations (forward Kolmogorov equations)~\cite{Gardiner2004, Kiss2017a, St-Onge2022_thesis}. Yet, since there are $2^N$ equations in this complete description and $N$ is generally large, the typical approach is to consider some approximations of the process~\cite{Wang2017, Kiss2017a, St-Onge2022_thesis}. By neglecting the dynamical correlations between the states of the neighbors~\cite[Sec. 2.3.1]{St-Onge2017_master}, the quenched mean-field (QMF) approach~\cite{Wang2017} yields the deterministic system of equations
\begin{equation}\label{eq:qmf_sis}
    \dot{x}_i = -d_ix_i + \gamma (1 - x_i)\, y_i, \qquad i\in\{1,...,N\},
\end{equation}
called the QMF SIS model, where $x_i$ is the probability for the vertex $i$ to be infected. In Fig.~\ref{fig:error_vector_fields}, we use the latter dynamics as a simple introductory example. More generally, quenched mean-field approximations of many binary stochastic processes, such as the SIS dynamics above, the Cowan dynamics~\cite{cowan1990stochastic, Painchaud2022}, and the Glauber dynamics~\cite{glauber1963time, Murphy2022} have the general form
\begin{equation}
    \dot{x}_i = (1 - x_i)\,\alpha(k_i - y_i,\, y_i) + x_i\,\beta(k_i -y_i,\, y_i),\qquad i\in\{1,...,N\}\,,
\end{equation}
where $x_i$ is the probability for vertex $i$ to be active, $k_i = \sum_{i=1}^N W_{ij}$ is the in-degree of vertex $i$, $\alpha$ (resp. $\beta$) is some analytic activation (resp. deactivation) probability function $\mathbb{R} \times \mathbb{R} \to [0,1] $. 
\end{example}

\begin{example}[Neuroscience]\label{ex:neuroscience_dynamics}
The Wilson-Cowan dynamics in Example~\ref{ex:wilson_cowan} satisfies condition (1) of Assumption~\ref{ass:emergence}. Another popular model of neuronal activity, the threshold-linear model~\cite{Hahnloser2000, Parmelee2022}, is defined by the equations
\begin{align}\label{eq:TLN}
\dot{x}_i = -x_i + \left[\sum_{j=1}^N W_{ij}x_j + b_i\right]_+,\qquad  i\in\{1,...,N\},
\end{align}
where $y\mapsto[y]_+ = \max\{0, y\}$ is the standard rectifier or ReLU function. To meet condition (1), the latter must be replaced by an analytic approximation, such as the softplus function $y\mapsto\ln(1+e^{ky})/k$ for some $k>0$.

\end{example}

\begin{example}[Population dynamics]\label{ex:population}
Population dynamics are widely used in science from ecology~\cite{Morone2019} and game theory~\cite{Grilli2017} to chemistry (e.g., kinetic equations)~\cite{Li1984} and physics (e.g., lasers)~\cite{Milonni1988}. The generalized Lotka-Volterra model~\cite{Lotka1910, Volterra1926} is a very typical population dynamics with the form 
\begin{equation}\label{eq:lotka_volterra}
    \dot{x}_i = -dx_i + \gamma x_i y_i\,.
\end{equation}
Refined models such as  
\begin{equation}
    \dot{x}_i = -dx_i - sx_i^2 + \gamma\frac{x_iy_i}{\alpha + y_i}
\end{equation}
in Ref.~\cite{Morone2019} or the microbial population dynamics \cite{Sanhedrai2022}
\begin{equation}\label{eq:microbial_population}
    \dot{x}_i = a - d\,x_i + bx_i^2 - cx_i^3 +  \gamma\,x_i\,y_i.
\end{equation}
have also been used to incorporate more realistic effects, like the Allee effect in which a population exhibits negative growth for low abundances~\cite{Allee1932, Gao2016_SI}. In the latter dynamics, which is used in the paper, the correspondences with the parameters of Ref.~\cite{Sanhedrai2022} are
$a = F$, $b = B(1 + K/C)$, $c = B/C$, and $d = BK$ where $F$ is the migration rate, $B$ is the logistic growth rate, $C$ is the carrying capacity, and $K$ is the Allee effect strength. In Extended Data Table~1, we consider that the parameter $d$ can vary for each vertex only for the sake of coherence with the other dynamics.
\end{example}

\begin{example}[Oscillators]\label{ex:kuramoto}
The Kuramoto-Sakaguchi dynamics \cite{Kuramoto1975, Sakaguchi1986} is a canonical model for a large class of oscillatory systems~\cite{Pietras2019, Thibeault2020_master} and finds many applications, e.g., for Josephson junctions~\cite{Wiesenfeld1996}, nanoelectromechanical oscillators~\cite{Matheny2019}, and neuroscience~\cite{Izhikevich2007}. The dynamics of the phase oscillators with a phase lag $\alpha$ is such that
\begin{equation}
    \dot\theta_j = \omega_j + \gamma\sum_{k=1}^N W_{jk}\,\sin(\theta_k - \theta_j + \alpha),
\end{equation}
where $\theta_j(t)$ is the position of the $j$-th oscillator at time $t$, $\omega_j$ is the $j$-th natural frequency, and $\gamma$ is the coupling constant. By setting $z_j = e^{i\theta_j}$~\cite{Thibeault2020_SI}, the Kuramoto-Sakaguchi model becomes
\begin{equation}\label{eq:kuramoto_sakaguchi}
    \dot{z}_j = i\omega_j z_j + \gamma\,e^{-i\alpha}\, y_j - \gamma\,e^{i\alpha}\,z_j^2\,\bar{y}_j,
\end{equation}
where $y_j = \sum_{k=1}^N W_{jk}z_k$ and $\bar{\phantom{z}}$ denotes complex conjugation. Note that the Winfree model~\cite{Winfree1967} and the theta model~\cite{Ermentrout1986} on networks~\cite{Thibeault2020_SI} also satisfy the condition (1) of Assumption~\ref{ass:emergence}.
\end{example}

\begin{example}[Machine learning]
\label{ex:rnn}
The universal approximation theorem of Funahashi and Nakamura~\cite[Theorem~1]{Funahashi1993} guarantees that a solution to a general dynamical system is approximately given, up to the desired accuracy, by a solution of a  continuous-time recurrent neural network \cite{Funahashi1993, Beer1997} 
\begin{equation}\label{eq:continuous_time_RNN}
    \dot{x}_i = -\frac{1}{\tau_i}x_i + \sum_{j=1}^N W_{ij}\mathcal{S}(x_j) + I_i,
\end{equation}
where $x_i$ is the trajectory of the $i$-th neuron, $\tau_i$ is the time-scale of neuron $i$, $\mathcal{S}$ is the sigmoid (logistic) function, $W_{ij}$ is the element $(i,j)$ of the $N\times N$ weight matrix $W$, and $I_i$ is the input current applied on neuron $i$. Equation~\eqref{eq:continuous_time_RNN} is also called or similar, up to some variations in its form, Cohen-Grossberg model~\cite{Cohen1983, Grossberg1988}, Hopfield model~\cite{Hopfield1984}, activation dynamics~\cite{Hirsch1989}, continuous rate RNN~\cite{Kim2019, Kim2021}, or reservoir computers. This recurrent neural network does not directly have the form to satisfy the condition (1) of Assumption~\ref{ass:emergence}, but from Ref.~\cite{Hanson2020}, we know that there is a class of (continuous-time) recurrent neural networks with the form
\begin{equation}\label{eq:rnn}
    \dot{x}_i = -d_i x_i + \tanh(\gamma y_i + c_i),
\end{equation}
$d_1,...,d_N$ are real constants, $c_i:t \mapsto c_i(t)\in\mathbb{R}$ is the $i$-th current, that is a universal approximator and satisfy the condition (1).
\end{example}
Following these considerations, we introduce a general proposition about the emergence of higher-order interactions when reducing the dimension of a dynamical system on network using Theorem~\ref{thm:least_square_optimal_vector_field}.
\begin{proposition}\label{prop:emergence}
    If the conditions of Assumptions~\ref{ass:emergence} hold, the reduced dynamics can be expressed in terms of higher-order interactions between the observables as
   \begin{align}\label{eq:dynamics_tensors}
        \dot{X}_{\mu} &= \mathcal{C}_{\mu} + \sum_{d_x=1}^\infty\sum_{\bm{\alpha}\in\mathbb{Z}^n_{+}}\mathcal{D}_{\mu\bm{\alpha}}^{(d_x + 1)} X_{\bm{\alpha}} + \sum_{d_y=1}^\infty\sum_{\bm{\beta}\in\mathbb{Z}^n_{+}} \mathcal{W}_{\mu\bm\beta}^{(d_y + 1)} X_{\bm\beta} + \sum_{d_x, d_y=1}^\infty\sum_{\bm\alpha, \bm\beta\in\mathbb{Z}^n_{+}} \mathcal{T}_{\mu \bm\alpha \bm\beta}^{(d_x+d_y+1)} X_{\bm\alpha\bm\beta},
    \end{align}
where we have introduced the multi-indices $\bm{\alpha} = (\alpha_1,...,\alpha_{d_x})$ and $\bm{\beta} = (\beta_1,...,\beta_{d_y})$ with $\alpha_p,\beta_q \in \{1,...,n\}$, the compact notation for  products $X_{\bm{\gamma}} = X_{\gamma_1}...X_{\gamma_d}$, while $\mathcal{C}_{\mu}$ denotes a real constant and $\mu\in\{1,\ldots, n\}$. The higher-order interactions are described by three tensors of respective order $d_x+1$, $d_y+1$, $d_x+d_y+1$, and whose elements are 
\begin{align}
    \mathcal{D}_{\mu\bm\alpha}^{(d_x + 1)} &= \quad\,\,\sum_{i=1}^N\quad\,\, c_{i d_x 0}M_{\mu i} M_{i\alpha_1}^+\,...\, M_{i\alpha_{d_x}}^+,\label{eq:higher_order_D}\\
    \mathcal{W}_{\mu\bm\beta}^{(d_y + 1)} &= \sum_{i,j_1,...,j_{d_y}=1}^N c_{i 0 d_y}M_{\mu i} W_{ij_1}...W_{ij_{d_y}}M_{j_1\beta_1}^+\,...\,M_{j_{d_y}\beta_{d_y}}^+,\label{eq:higher_order_W}\\
    \mathcal{T}_{\mu \bm\alpha \bm\beta}^{(d_x+d_y+1)} &= \sum_{i,j_1,...,j_{d_y}=1}^N c_{id_xd_y}M_{\mu i} M_{i\alpha_1}^+\,...\, M_{i\alpha_{d_x}}^+ W_{ij_1}...W_{ij_{d_y}}M_{j_1\beta_1}^+\,...\,M_{j_{d_y}\beta_{d_y}}^+,\label{eq:higher_order_T}
\end{align}
for some real coefficients $c_{id_xd_y}$ with $i \in \{1,...,N\}$ and $d_x,d_y\in\mathbb{Z}_+$.
\end{proposition}
\begin{proof}
    By definition of an analytic function, there is a convergent power series describing the vector field of the complete dynamics, i.e.,
    \begin{equation}\label{eq:power_series}
        h_i(x_i, y_i) = \sum_{d_x=0}^\infty\sum_{d_y=0}^\infty c_{id_xd_y} x_i^{d_x} y_i^{d_y}, \qquad i\in\{1,...,N\},
    \end{equation}
    where we have chosen to express the power series around $x_i=y_i=0$ without loss of generality. The reduced dynamics is therefore
    \begin{equation*}
        \dot{X}_{\mu} 
        = \sum_{d_x, d_y=0}^\infty\sum_{i=1}^N c_{id_xd_y}M_{\mu i}\left(\sum_{\alpha = 1}^n M_{i\alpha}^+ X_{\alpha}\right)^{d_x}\left(\sum_{j=1}^N\sum_{\beta = 1}^n W_{ij}M_{j\beta}^+ X_{\beta}\right)^{d_y}.
    \end{equation*}
    The sum can be separated as 
    \begin{align*}
        \dot{X}_{\mu} &= \mathcal{C}_{\mu} + \sum_{d_x=1}^\infty\sum_{i=1}^N c_{id_x0}M_{\mu i}\left(\sum_{\alpha = 1}^n M_{i\alpha}^+ X_{\alpha}\right)^{d_x} + \sum_{ d_y=1}^\infty\sum_{i=1}^N c_{i0d_y}M_{\mu i}\left(\sum_{j=1}^N\sum_{\beta = 1}^n W_{ij}M_{j\beta}^+ X_{\beta}\right)^{d_y}\nonumber\\&\qquad\quad + \sum_{d_x, d_y=1}^\infty\sum_{i=1}^N c_{id_xd_y}M_{\mu i}\left(\sum_{\alpha = 1}^n M_{i\alpha}^+ X_{\alpha}\right)^{d_x}\left(\sum_{j=1}^N\sum_{\beta = 1}^n W_{ij}M_{j\beta}^+ X_{\beta}\right)^{d_y},
    \end{align*}
    where we have defined $\mathcal{C}_{\mu} = \sum_{i=1}^N M_{\mu i}c_{i00}$. Expanding the exponents and introducing the multi-indices directly provide the desired result. 
\end{proof}

\begin{remark}

\phantom{Henri}

\begin{enumerate}
    \item As explained in Section 1.1 (p.3) of Ref.~\cite{Qi2017}, the tensors above could be more precisely called hypermatrices.
    
    \item For clarity, we specify the order of the tensor as an exponent in parentheses. In the paper and in Example~\ref{ex:emergence}, the order is clear from the indices and we thus avoid this notation for simplicity. Also, we let the indices differentiate the tensors, e.g., $\mathcal{T}_{1(2,3)(4)}$ ($d_x = 2$, $d_y = 1$) and $\mathcal{T}_{1(2)(3,4)}$ ($d_x = 1$, $d_y = 2$) are elements of two different tensors. Finally, when it's clear in the context, if a multi-index is a singleton, than we remove the parentheses, e.g., $\mathcal{T}_{\mu(\nu)(\kappa)}$ becomes $\mathcal{T}_{\mu\nu\kappa}$. 
    
    \item The coefficients $c_{id_xd_y}$ can be chosen as the ones of the Taylor series of $h_i$ for all $i$.
    
    \item For the sake of simplicity, let us consider the case where $c_{i01} = 1$ for all $i$. We observe that $\mathcal{W}^{(2)} = MWM^+$ appears in the reduced dynamics, which can be viewed as the reduced weight matrix. From Theorem~\ref{thm:existence_solution_compatibility}, it is also the unique solution to the compatibility equation \cite{Thibeault2020_SI} $\mathcal{W}M = MW$ when $\rank{M} = n$ and it is the least-square optimal solution to the problem $\|\mathcal{W}M - MW\|_F^2$ with $\mathcal{W}$ as the optimization variable. Remember from Ref.~\cite{Thibeault2020_SI} that solving the compatibility equation is necessary to cancel the first-order errors in DART or less generally, to close the reduced dynamics of any linear dynamics $\dot{x} = Wx$. Indeed, for $X = Mx$, $\dot{X} = MWx = \mathcal{W}Mx = \mathcal{W}X$ where one can reasonably choose $\mathcal{W} = \mathcal{W}^{(2)}$ as explained before.
    \item If there was already higher-order interactions in the complete dynamics, the least-square optimal reduced dynamics would have new higher-order interactions that depends on the original ones, the parameters of the dynamics, and the reduction matrix.
    \item The latter proposition can easily be extended to complex variables. First assume that the complex dynamics has the form $\dot{x}_i = r_i(x_i, y_i, \bar{x}_i, \bar{y}_i)$, where $\bar{\phantom{0}}$ is complex conjugation and $r_i:\mathbb{C}^4 \mapsto \mathbb{C}$ is a holomorphic function (and thus analytic):
    \begin{equation}
        r_i(x_i, y_i, \bar{x}_i, \bar{y}_i) = \sum_{d_x=0}^\infty\sum_{d_y=0}^\infty\sum_{\bar{d}_x=0}^\infty\sum_{\bar{d}_y=0}^\infty c_{id_xd_y\bar{d}_x\bar{d}_y} x_i^{d_x} y_i^{d_y}\bar{x}_i^{\bar{d}_x} \bar{y}_i^{\bar{d}_y}, \quad i\in\{1,...,N\}.
    \end{equation}
    The rest of the proof is similar to its real counterpart. This is especially interesting for phase dynamics such as the Kuramoto model (see Example~\ref{ex:emergence}).
    \item This is not the only dimension reduction that yields higher-order interactions. We did not realize it clearly at the moment of writing Ref.~\cite{Thibeault2020_SI}, but DART also yields higher-order interactions, which can be explicitly seen in Eqs.~(28-30). However, these higher-order interactions could be avoided by noting that the phase dynamics have a vector field of the form $h_i(x_i, y_i)$. Indeed, using Taylor's theorem for both $x$ and $Wx$, there is no compatibility equation for the degrees that appears to cancel the first-order terms and it ultimately removes the higher-order contributions with $\mathcal{K}$ in Eqs.~(28-30). In general, a dimension reduction method where the original vector field is evaluated at a function of the original variables is susceptible to yield higher-order interactions. 
\end{enumerate}
\end{remark}

In the last proposition, the graph with $N$ vertices of the complete dynamics (and its parameters encoded by all the coefficients $c_{id_xd_y}$) is thus replaced by a hypergraph $\mathcal{H}$ \cite{Berge1989, Gallo1993, Qi2017} with $n$ vertices [see Fig.~\ref{fig:low_dimension_hypothesis}d of the paper], defined from the tensors $\mathcal{D}^{(d_x+1)}$, $\mathcal{W}^{(d_y+1)}$, and $\mathcal{T}^{(d_x + d_y + 1)}$, in the reduced dynamics. Below, we define more precisely the notion of directed, weighted, and signed hypergraphs.
\begin{definition}
    A hypergraph is a triple $\mathcal{H} = (\Upsilon, \Xi, \Omega)$, where 
    \begin{itemize}
        \item[--] $\Upsilon = \{1,...,n\}$ is the set of vertices;
        \item[--] $\Xi$ is a set of \textit{hyperarcs} (or directed hyperedges) defined as an ordered pair $E = (H, T)$, where $H$ is the \textit{head} of the hyperarc (a $n_H$-tuple with elements in $\Upsilon$), $T$ is the \textit{tail} of the hyperarc (a $n_T$-tuple with elements in $\Upsilon$), and $2 \leq n_H + n_T \leq n$ with $n_H, n_T \geq 1$. For $n_H = 1$ and $n_T = 1$, the hyperarc is a directed edge. If $n_H = 1$ and $n_T > 1$, it is a backward hyperarc and if $n_H > 1$ and $n_T = 1$, it is a forward hyperarc;
        \item[--] $\Omega$ is a function that assigns a real value to the hyperarcs. 
    \end{itemize}
     
\end{definition}

\begin{remark}
\phantom{Henri}
\begin{itemize}
    \item The latter definition is a generalization of hypergraphs~\cite{Berge1989} and of directed hypergraphs as defined in Ref.~\cite{Gallo1993}, where the head and the tails of the hyperarcs are sets instead of tuples. 
    \item For the weight matrix with elements $W_{ij}$, we use the convention that the edge (or arc) $(i, j)$ is directed from $j$ to $i$. For consistency, in the definition above, we use the convention that the hyperarc $(H, T)$ (instead of $(T, H)$ as in Ref.~\cite{Gallo1993}) is directed from the tail $T$ to the head $H$. As a consequence, in the tensor notation $\mathcal{T}^{(d_x + d_y + 1)}_{\mu\bm{\alpha}\bm{\beta}}$, the index $\mu$ and the multi-index $\bm\alpha$ are part of the head while $\bm{\beta}$ is part of the tail of the hyperarc. Thus, $\mathcal{T}_{1(2)(3)}$ ($n_T = 1$) is a forward hyperarc $((1,2), (3))$ while $\mathcal{T}_{1()(2,3)}$ ($n_H = 1$) is a backward hyperarc $(1,(2,3))$. Note that the tensor $\mathcal{W}^{(d_y + 1)}$ with elements in Eq.~\eqref{eq:higher_order_W} always form backward hyperarcs (from $\bm{\beta}$ to $\mu$) since $d_x = 0$, while the tensor with elements in Eq.~\eqref{eq:higher_order_T} can be any type of hyperarc (with $\mu$ always belonging to the head). In the example of the paper for the epidemiological dynamics, Eq.~\eqref{eq:third_order_interactions} is a forward hyperarc (from $\kappa$ to $\mu\nu$).
\end{itemize}
\end{remark}

We now derive two key consequences of Proposition~\ref{prop:emergence}. First, Proposition~\ref{prop:emergence} shows that there can be an infinite number of higher-order interactions in the reduced dynamics. Yet, for a special family of vector fields, we prove that there is a finite number of them which are related to the nonlinearity of the original dynamics.

\begin{corollary}\label{cor:emergence_polynomial}
If $h_i(x_i, y_i)$ is a polynomial of total degree $\delta$ in $x_i$ and $y_i$ for all $i\in\{1,...,N\}$ and condition (2) of Assumptions~\ref{ass:emergence} holds, then the reduced dynamics has a polynomial vector field of total degree $\delta$ with interactions of maximal order $\delta + 1$.
\end{corollary}
\begin{proof}
Since any polynomial is analytic, condition (1) of Assumptions~\ref{ass:emergence} is satisfied. Then, by Proposition~\ref{prop:emergence}, the reduced dynamics is given by Eqs.~(\ref{eq:dynamics_tensors}-\ref{eq:higher_order_T}). In the following, the conclusions are valid for all $i\in\{1,...,N\}$. Let $\mathscr{I}_i = \{c_{id_x d_y}\}_{d_x,d_y = 0}^\infty$ be the $i$-th (countable) infinite set of coefficients related to the $i$-th analytic function $h_i$. The fact that $h_i$ is a polynomial implies that there is a finite subset of nonzero coefficients $\mathscr{F}_i\subset \mathscr{I}_i$ describing a polynomial vector field for the reduced dynamics. Consider any coefficient $c_{i'd_x' d_y'} \in \mathscr{F}_i$ such that $d_x' + d_y' = \delta$, the total degree of the polynomial $h_i$. Then, at least one of the tensors $\mathcal{D}^{(d_x' + 1)}$, $\mathcal{W}^{(d_y' + 1)}$, $\mathcal{T}^{(d_x' + d_y' + 1)}$, with elements in Eqs.~(\ref{eq:higher_order_D}-\ref{eq:higher_order_T}), have the highest possible order $\delta + 1$. Moreover, there will be at least one monomial term $X_{\alpha_1}...X_{\alpha_{d_x'}}$, $X_{\beta_1}...X_{\beta_{d_y'}}$, or $X_{\alpha_1}...X_{\alpha_{d_x'}}X_{\beta_1}...X_{\beta_{d_y'}}$ in Eq.~\eqref{eq:dynamics_tensors} that is of maximal degree $\delta$, which means that reduced dynamics has a polynomial vector field of total degree $\delta$.
\end{proof}

Second, the tensors describe Proposition~\ref{prop:emergence} strongly depends on the reduction matrix $M$, or in other words, the reduction matrix will play a role on the form of the higher-order interactions. One can therefore ask if one can choose $M$ in such a way that there are only pairwise interactions in the reduced dynamics. 
In the next corollary, we provide sufficient conditions to have pairwise interactions in the least-square reduced dynamics. 
\begin{corollary}\label{cor:sufficient_condition_emergence}
 Let $s: \mathcal{V}\to \Upsilon$ be a surjection where $\mathcal{V} = \{1,...,N\}$ and $\Upsilon \in\{1,...,n\}$ are the vertex sets of the complete and reduced system respectively. If Assumptions~\ref{ass:emergence} hold, the reduction matrix $M$ has elements $M_{\mu i} = m_{\mu i}\delta_{\mu\,s(i)}$ with $m_{\mu i}\in\mathbb{R}$ for all $\mu$, $i$, and $h_i$ linearly  depends on $y_i$ for all $i$, then there are solely pairwise interactions in the reduced system. The result doesn't hold in general for nonlinear dependencies of $h_i$ over $y_i$.
\end{corollary}
\begin{proof}
For such reduction matrix, the elements of its Moore-Penrose pseudoinverse are, for all $\mu \in \Upsilon$ and $i \in \mathcal{V}$, $M^+_{i\mu} = m_{\mu i}\delta_{s(i)\mu}/q_{\mu}$, where $q_{\mu} = \sum_{i=1}^N m_{\mu i}^2\delta_{s(i)\mu}$. Substituting $M$ and $M^+$ in Eqs.~(\ref{eq:higher_order_D}-\ref{eq:higher_order_T}) yields
\begin{align*}
    \mathcal{D}_{\mu\bm\alpha}^{(d_x + 1)} &= \frac{1}{q_{\bm{\alpha}}}\quad\,\,\sum_{i=1}^N\quad\,\, c_{i d_x 0}m_{\mu\bm{\alpha}i}\,\delta_{\mu\,s(i)} \delta_{s(i)\,\alpha_1}\,...\, \delta_{s(i)\,\alpha_{d_x}},\\
    \mathcal{W}_{\mu\bm\beta}^{(d_y + 1)} &= \frac{1}{q_{\bm{\beta}}}\sum_{i,j_1,...,j_{d_y}=1}^N c_{i 0 d_y}m_{\mu i}\,\delta_{\mu\,s(i)} W_{ij_1}...W_{ij_{d_y}}\delta_{s(j_1)\,\beta_1}\,...\,\delta_{s(j_{d_y})\,\beta_{d_y}},\\
    \mathcal{T}_{\mu \bm\alpha \bm\beta}^{(d_x+d_y+1)} &= \frac{1}{q_{\bm{\alpha}\bm{\beta}}}\sum_{i,j_1,...,j_{d_y}=1}^N c_{id_xd_y}m_{\mu\bm{\alpha}i}\,\delta_{\mu\,s(i)} \delta_{s(i)\,\alpha_1}\,...\, \delta_{s(i)\,\alpha_{d_x}} W_{ij_1}...W_{ij_{d_y}}m_{\beta_1 j_1}...m_{\beta_{d_y}j_{d_y}}\delta_{s(j_1)\,\beta_1}\,...\,\delta_{s(j_{d_y})\,\beta_{d_y}},
\end{align*}
where $q_{\bm{\gamma}} = q_{\gamma_1}...q_{\gamma_d}$ and $m_{\mu\bm{\alpha}i} = m_{\mu i}m_{\alpha_1i}...m_{\alpha_di}$.
For $\mathcal{D}^{(d_x + 1)}$ and any dependence of $h_i$ over $y_i$, it is straightforward to observe that the only nonzero elements are such that $\mu = \alpha_1 =...=\alpha_{d_x}$. The tensor can therefore be mapped to a $n\times n$ diagonal matrix. Henceforth, we only consider $\mathcal{W}^{(d_y + 1)}$ and $\mathcal{T}^{(d_x+d_y+1)}$.

The fact that $h_i(x_i, y_i)$ linearly depends on $y_i$ for all $i$ is equivalent to setting $d_y = 1$ in its power series in Eq.~\eqref{eq:power_series}, i.e.,
\begin{equation*}
    h_i(x_i, y_i) = \sum_{d_x=0}^\infty c_{id_x1} x_i^{d_x} y_i, \qquad i\in\{1,...,N\}.
\end{equation*}
Proposition~\ref{prop:emergence} thus implies that 
\begin{align*}
    \mathcal{W}_{\mu\beta}^{(2)} = \frac{1}{q_{\beta}}\sum_{i,j=1}^Nc_{i 0 1} m_{\mu i} m_{\beta j}\delta_{\mu\,s(i)} W_{ij}\delta_{s(j)\,\beta}\,\,\,\text{and}\,\,\,
    \mathcal{T}_{\mu \bm\alpha \beta}^{(d_x+2)} = \frac{1}{q_{\bm{\alpha}\beta}}\sum_{i,j=1}^N c_{id_x1}m_{\mu\bm{\alpha} i} \delta_{\mu\,s(i)} \delta_{s(i)\,\alpha_1}\,...\, \delta_{s(i)\,\alpha_{d_x}} W_{ij} m_{\beta j}\,\delta_{s(j)\,\beta}.
\end{align*}
Clearly, $\mathcal{W}^{(2)}$ is a matrix and the nonzero elements of $\mathcal{T}^{(d_x+2)}$ are for $\mu = \alpha_1 = ... = \alpha_{d_x}$ (there are at most $n^2$ of them), which means that it can be mapped to a $n\times n$ matrix. Hence, there are solely pairwise interactions in the least-square reduced dynamics.

If $d_y > 1$ (i.e., for a nonlinear dependency of $h_i$ over $y_i$), a simple example suffices to prove the last statement. Let $d_y = 2$, $\mathcal{V} = \{1,2,3,4,5\}$, $\Upsilon = \{1,2,3\}$, and $s(1) = 1$, $s(2) = 1$, $s(3) = 2$, $s(4) = 3$, $s(5) = 3$. Moreover, consider that $W_{52}$, $W_{53}$, $m_{12}$, $m_{23}$, $m_{35}$, $c_{502}$ are not equal to zero. Then, Proposition~\ref{prop:emergence} gives
\begin{equation*}
    \mathcal{W}_{\mu (\beta_1,\beta_2)}^{(3)} = \frac{1}{q_{\mu\beta_1\beta_2}}\sum_{i=1}^N c_{i02} m_{\mu i}\,\delta_{\mu\,s(i)} \left(\sum_{j,k=1}^N m_{\beta_1 j}m_{\beta_2k} W_{ij}W_{ik}\delta_{s(j)\,\beta_1}\delta_{s(k)\,\beta_2}\right).
\end{equation*}
It only remains to prove that there can be nonzero elements for $\beta_1 \neq \beta_2$. For $\beta_1 = 1$, $\beta_2 = 2$, $j = 2$, and $k = 3$ in the parentheses of the last equation, there is a term $m_{12}m_{23}W_{i2}W_{i3}\delta_{s(2)\,1}\delta_{s(3)\,2} = m_{12}m_{23}W_{i2}W_{i3}$. Considering the whole equation for $i = 5$ and $\mu = 3$, there is a nonzero term $c_{502} m_{35}m_{12}m_{23}W_{52}W_{53}/q_{312}$. Hence, in this example, $\mathcal{W}_{3(1,2)} \neq 0$ despite the fact that the observables are defined on disjoint sets of vertices.
\end{proof}

\begin{remark}
If $n = N$, the higher-order interactions between the observables does not necessarily disappear because of the linear transformation done by $M$ on $x$. Obviously, if $M = I$, $X=x$ and $\dot{X}_{\mu} = \dot{x}_i = h_i(x_i, y_i)$ for all $i$ and there are no higher order interactions. However, the vector field $M\circ h \circ M^+$ will generally contain higher-order interactions. But of course, if $M$ has full rank $N$, it is invertible and one can transform back the dynamics of the observable $X$ (with higher-order interactions) to the dynamics in $x$ (without higher-order interactions), since $x = M^{-1}X$. 
\end{remark}

Let's now provide the details about the examples presented in Extended Data Table~1.

\begin{example}[Emergence of higher-order interactions in typical models]
\label{ex:emergence}
Proposition~\ref{prop:emergence} and Corollary~\ref{cor:emergence_polynomial} imply the following results in different fields of application.
\vspace{0.2cm}
\begin{enumerate}
    \item QMF SIS dynamics [Eq.~\eqref{eq:qmf_sis} in Example~\ref{ex:sis}]:
\begin{equation}
    \dot{X}_{\mu} = \sum_{\nu=1}^n (\mathcal{D}_{\mu\nu} + \mathcal{W}_{\mu\nu})X_{\nu} +\sum_{\nu, \kappa = 1}^n \mathcal{T}_{\mu\nu\kappa}X_{\nu}X_{\kappa},
\end{equation}
where $\mathcal{D} = -MD M^{+}$ with $D = \diag(d_1,...,d_N)$, $\mathcal{W} = \gamma MWM^+$, and
\begin{align*}
    \mathcal{T}_{\mu \nu \kappa} = - \gamma\sum_{i,j=1}^NM_{\mu i}M_{i\nu}^+ W_{ij}  M_{j\kappa}^+,
\end{align*}
with $\bm{\alpha} = (\nu) = \nu$ and $\bm{\beta} = (\kappa) = \kappa$. Interestingly, for $n = 1$, one can find the exact solution since it is a Bernoulli differential equation.

\vspace{0.5cm}

\item Microbial population dynamics [Eq.~\eqref{eq:microbial_population} in Example~\ref{ex:population}]:
\begin{equation}
    \dot{X}_{\mu} = \mathcal{C}_{\mu} + \sum_{\nu=1}^n\mathcal{D}_{\mu\nu}X_{\nu} + \sum_{\nu,\kappa=1}^n(\mathcal{D}_{\mu(\nu,\kappa)} + \mathcal{T}_{\mu\nu\kappa})X_{\nu}X_{\kappa} + \sum_{\nu, \kappa,\tau =1}^n\mathcal{D}_{\mu(\nu,\kappa,\tau)}X_{\nu}X_{\kappa}X_{\tau}
\end{equation}
where $\mathcal{D} = -dMM^{+}$ and
\begin{align*}
    \mathcal{D}_{\mu(\nu,\kappa)} = b\,\sum_{i=1}^N M_{\mu i} M_{i\nu}^+ M_{i\kappa}^+\,,\quad
    \mathcal{T}_{\mu\nu\kappa} = \gamma\,\sum_{i,j=1}^N M_{\mu i} M_{i\nu}^+ W_{ij}M_{j\kappa}^+\,,\quad
    \mathcal{D}_{\mu(\nu,\kappa,\tau)} = -c\,\sum_{i=1}^N M_{\mu i} M_{i\nu}^+ M_{i\kappa}^+M_{i\tau}^+\,.
\end{align*}

\vspace{0.5cm}

\item Kuramoto-Sakaguchi dynamics  [Eq.~\eqref{eq:kuramoto_sakaguchi} in Example~\ref{ex:kuramoto}]:
\begin{equation}
    \dot{X}_{\mu} = \sum_{\nu=1}^n (\mathcal{D}_{\mu\nu}+\mathcal{W}_{\mu\nu})X_{\nu} + \sum_{\nu, \kappa, \tau = 1}^n \mathcal{T}_{\mu (\nu, \kappa) \tau}X_\nu X_\kappa \bar{X}_\tau 
\end{equation}
where $\mathcal{D} = i MD M^{+}$ with $D = \diag(\omega_1,...,\omega_N)$, $\mathcal{W} = \gamma e^{-i\alpha} MWM^+$, and
\begin{align*}
    \mathcal{T}_{\mu (\nu, \kappa) \tau} = -\gamma e^{i\alpha}\sum_{j,k=1}^NM_{\mu j}M_{j\nu}^+ M_{j\kappa}^+ W_{jk}  M_{k\tau}^+,
\end{align*}
 with $\bm{\alpha} = (\nu, \kappa)$, $\bm{\beta} = (\tau) = \tau$. In this case, the reduced variables $X_1$,...,$X_n$ and the involved tensors are complex. 
\end{enumerate}
\end{example}
\begin{remark}
\phantom{Henri}
\begin{itemize}
    \item We found that there can be computational benefits to write the vector fields in terms of tensors (subsection~\ref{SIsubsec:numerical_efficiency}).
    \item The fact that the least-square optimal reduced vector field contains higher-order interactions raises the problem of getting mathematical insights from dynamics on hypergraphs, which recalls again the pertinence of this field in the study of complex systems. Fortunately, many recent papers address the problem, such as Ref.~\cite{FerrazdeArruda2021} or Ref.~\cite{Mulas2020}. See Ref.~\cite{Battiston2020a} for more references.
\end{itemize}
\end{remark}
In phase reduction techniques~\cite{Pietras2019}, $dN$-dimensional weakly coupled limit-cycle oscillators dynamics, where each of the $N$ oscillators is described by $d$ variables, are reduced to a $N$-dimensional dynamics of their phase. It is known that these phase reductions lead to higher-order interactions between the phases \cite{Ashwin2016a, Matheny2019, Leon2019} or, in other words, between microscopic observables (i.e., there is a phase for each oscillator, considered as the microscopic level, except in Ref.~\cite[Fifth section]{Nijholt2022}). In contrast, the higher-order interactions that we observe emerge from a large variety of dynamical systems and they are between observables that can cover different scales, which strongly depends over the choice reduction matrix. The generality of our results thus suggests that the emergence could be quite ubiquitous.

\subsection{Upper bound on the alignment error and exact dimension reduction}
\label{SIsubsec:error_bound}
In this subsection, we evaluate the impact of choosing the least-square optimal vector field in Theorem~\ref{thm:least_square_optimal_vector_field} on the alignment error $\mathcal{E}(x)$ in $\mathbb{R}^n$. In particular, we will see that obtaining an upper bound on $\mathcal{E}(x)$ is useful to find a reasonable choice of reduction matrix $M$. More importantly, to determine more quantitatively the repercussions of the low-rank hypothesis on the dynamics, we aim at estimating the error caused by the optimal reduced dynamics as a function of $n$. Let us start by listing the assumptions that will be made throughout this subsection. 
\begin{assumptions}\label{ass:error_bound}

\phantom{Henri}

\begin{enumerate}
    \itemsep1em
    \item[(1)] 
    The $N$-dimensional complete dynamics on a network defined by the real $N\times N$ weight matrix $W$ is
    \begin{align}
    \dot{x} = g(x, y) = \begin{pmatrix}g_1(x_1,...,x_N, y_1,..., y_N) \\\vdots\\g_N(x_1,...,x_N, y_1,..., y_N) \end{pmatrix},
    \end{align}
    where $x:t \mapsto \mathbb{R}^N$, $y = Wx$, and $g: \mathbb{R}^N\times \mathbb{R}^N \to  \mathbb{R}^N$ is a continuously differentiable function. 

    \item[(2)] The $n$-dimensional reduced dynamics ($n < N$) is the least-square optimal dynamics of Theorem~\ref{thm:least_square_optimal_vector_field}, i.e.,
    \begin{align}
    \dot{X} = Mg(M^+X, WM^+X)\,.
    \end{align}
    
    \item[(3)] The reduction matrix $M$ is the truncated left singular vector matrix $V_n^\top$ of $W$.
\end{enumerate}
\end{assumptions}
Note the first assumption is less restrictive than the first one of the Assumptions~\ref{ass:emergence}. We chose $n<N$ to ensure dimension reduction and also, because it is obvious to show that we can have a zero alignment error when $n = N = \rank{M}$. In this case, the ``reduced" dynamics is not reduced anymore, but it is still a linear transformation of the complete dynamics. 
\begin{lemma}\label{lem:zero_alignment_error_n=N}
 If conditions~(1) and (2) in Assumptions~\ref{ass:error_bound} hold with $n = N = \rank{M}$, the alignment error is 0.
\end{lemma}
\begin{proof}
If $M$ has rank $N$, the pseudoinverse of $M$ is its inverse and the related projector is $P = M^+M = M^{-1}M = I$. Hence, the alignment error in $\mathbb{R}^n$ is obviously zero: $\mathcal{E}(x) = \|M[g(x, Wx) - g(Px, WPx)]\|/\sqrt{n} = 0$.
\end{proof}

Let us now turn to one of the important results of the paper. The next theorem demonstrates that the alignment error between a high-dimensional vector field depending on a network and its optimally reduced version is intrinsically related to the network's singular value profile: when the singular values $\sigma_n$ decrease rapidly with $n$, so does the alignment error. Therefore, a low-rank hypothesis induces a low-dimension hypothesis for dynamical systems.
\begin{theorem}\label{thm:upper_bound_rapid_decrease}
If all conditions of Assumptions~\ref{ass:error_bound} hold, the alignment error in $\mathbb{R}^n$ at $x \in \mathbb{R}^N$ is upper-bounded as
\begin{equation}\label{eq:upper_bound_rapid_decrease}
    \mathcal{E}(x) \leq \frac{1}{\sqrt{n}}\Big[\|V_n^{\top}J_x(x',y')(I-V_nV_n^\top)x\| + \sigma_{n+1}\|V_n^{\top}J_y(x',y')\|_2\|x\|\Big],
\end{equation}
where $y' = Wx'$ with $x'$ being some point between $x$ and $V_nV_n^\top x$, $\sigma_i$ is the $i$-th singular value of $W$, and $J_x(x',y')$, $J_y(x',y')$ are the Jacobian matrices of $g$ with derivatives according to the vectors $x$ and $y$ respectively. Moreover, for any $x$ not at the origin of $\mathbb{R}^N$, the following upper bound on the relative alignment error holds:
\begin{equation}\label{eq:upper_bound_rapid_decrease_relative_theo}
\frac{\mathcal{E}(x)}{\|x\|} \leq \frac{1}{\sqrt{n}}\Big[\alpha(x',y')+ \sigma_{n+1}\beta(x',y')\Big]\,,
\end{equation}
where $\alpha(x',y')=\sigma_1(J_x(x',y'))$ and $\beta(x',y')=\sigma_1(J_y(x',y'))$.
\end{theorem}
\begin{proof}
From the definition of the alignment error and the first two conditions in Assumptions~\ref{ass:error_bound}, we have
\begin{align*}
    \mathcal{E}(x) = \frac{1}{\sqrt{n}}\|M[g(x, y) - g(\tilde{x}, \tilde{y})]\|,
\end{align*}
where $y=Wx$, $\tilde{x} = Px$, and $\tilde{y} = WPx$ with $P = M^+M$. Let's define the function
\begin{equation}\label{eq:ux}
    u(x) = g(x, \ell(x)),
\end{equation}
with the linear function $\ell(x) = Wx$. Since $g$ is a continuously differentiable function, $u$ is also continuously differentiable and Taylor's theorem with 0-th order Lagrange remainder guarantees that
\begin{equation}\label{eq:taylor_u}
    u(x) = u(\tilde{x}) + Du(x')(x - \tilde{x})
\end{equation}
for some $x'$ between $x$ and $Px$ and where $Du(x')$ is the total derivative of $u$. From Eq.~\eqref{eq:ux} and the chain rule for the total derivative, we have (abusing the matrix notation)
\begin{align}\label{eq:total_derivative_u}
    Du(x') = Dg(x', y') = \frac{\partial g}{\partial x}(x', y') +  \frac{\partial g}{\partial y}(x', y')\frac{\partial \ell}{\partial x}(x') = J_x(x', y') +  J_y(x', y')W,
\end{align}
where $y' = \ell(x') = Wx'$, the elements of the Jacobian matrices $J_x(x', y'), J_y(x', y')$ are respectively
\begin{align*}
    [J_x(x',y')]_{ij} = \frac{\partial g_i(x, y)}{\partial x_j}\Bigg|_{(x,y) = (x', y')}\,\,\,,\,\,\quad [J_y(x',y')]_{ij} = \frac{\partial g_i(x, y)}{\partial y_j}\Bigg|_{(x,y) = (x', y')},
\end{align*}
and we have used the fact that $W$ is the Jacobian matrix of $\ell$. The Taylor expansion~\eqref{eq:taylor_u} of $u$ with $N$ variables $(x_i)_{i=1}^N$ for some $x'$ therefore implies a Taylor expansion for $g$ with 2$N$ variables $(x_i, y_i)_{i=1}^N$ for some $x'$ with $y' = Wx'$:
\begin{equation}\label{eq:taylor_g}
    g(x, y) = g(\tilde{x}, \tilde{y}) + J_x(x',y')(x - \tilde{x}) + J_y(x',y')W(x - \tilde{x}).
\end{equation}
The alignment error becomes
\begin{equation*}
    \mathcal{E}(x) = \frac{1}{\sqrt{n}}\|M[J_x(x',y')(I-P)x + J_y(x',y')W(I-P)x]\|
\end{equation*}
and the triangle inequality gives
\begin{equation*}
    \mathcal{E}(x) \leq \frac{1}{\sqrt{n}}\big[\|MJ_x(x',y')(I-P)x\| + \|MJ_y(x',y')W(I-P)x\|\big].
\end{equation*}
Moreover, the induced spectral norm for the second term yields
\begin{equation*}
    \mathcal{E}(x) \leq \frac{1}{\sqrt{n}}\big[\|MJ_x(x',y')(I-P)x\| + \|MJ_y(x',y')W(I-P)\|_2\|x\|\big]
\end{equation*}
and the submultiplicativity of the spectral norm implies
\begin{equation}\label{eq:upper_bound_M}
    \mathcal{E}(x) \leq \frac{1}{\sqrt{n}}\big[\|MJ_x(x',y')(I-P)x\| + \|W(I-P)\|_2\|MJ_y(x',y')\|_2\|x\|\big].
\end{equation}
From condition (3) of Assumption~\ref{ass:error_bound}, we have $M = V_n^\top$ which is, by Theorem~\ref{thm:optimal_projection_general}, the optimal solution to the minimization of $\|W(I-P)\|_2$ with error $\sigma_{n+1}$ (square root of the problem~\eqref{eq:P2} and the error for the spectral norm in Eq.~\eqref{eq:optimal_projection_errors}). The second inequality is deduced as follows:
\begin{align}\label{eq:upper_bound_rapid_decrease_relative}
     \frac{\mathcal{E}(x)}{\|x\|} 
        &\leq \frac{1}{\sqrt{n}}\Big[\|V_n^{\top}J_x(x',y')(I-V_nV_n^\top)\|_2 + \sigma_{n+1}\|V_n^{\top}J_y(x',y')\|_2\Big]\\
        &\leq \frac{1}{\sqrt{n}}\Big[\|V_n^{\top}J_x(x',y')\|_2+ \sigma_{n+1}\|V_n^{\top}J_y(x',y')\|_2\Big]\\
        &\leq \frac{1}{\sqrt{n}}\Big[\|J_x(x',y')\|_2+ \sigma_{n+1}\|J_y(x',y')\|_2\Big]\,,
    \end{align}
where we have used successively the submultiplicativity of the spectral norm and identities $\|(I-V_nV_n^\top)\|_2=1$, $\|V_n^{\top}\|_2 = 1$. 
The desired upper bound is found upon noticing that $\|J_x(x',y')\|_2=\sigma_1(J_x(x',y'))$ and $\|J_y(x',y')\|_2=\sigma_1(J_y(x',y'))$.  
\end{proof}

\begin{remark}
\phantom{Henri}
\begin{itemize}
    \item The dynamics used in the paper have the less general form (compared to condition (1) in Assumptions~\ref{ass:error_bound}) 
    \begin{align}
        \dot{x}_i = h_i(x_i, y_i),\quad i\in\{1,...,N\},
    \end{align}
    where, for all $i$, $x_i:t \mapsto \mathbb{R}^N$, $y_i = \sum_{j=1}^N W_{ij}x_j$ and $h_i:\mathbb{R}^2 \mapsto \mathbb{R}$. This implies that for all dynamics considered in the paper, the Jacobian matrices $J_x(x',y')$ and $J_y(x',y')$ are diagonal.
    \item Even if the effective ranks of real networks are low compared to $N$, they are generally larger than one, meaning that $\sigma_{n+1}$ is not negligible when $n=1$. According to our analysis, we therefore do not expect one-dimensional reductions~\cite{Gao2016_SI, Laurence2019_SI, Kundu2022_SI} to yield accurate results in general, which is consistent with numerical observations made in previous studies \cite{Tu2017_SI, Jiang2018_SI, Laurence2019_SI, Thibeault2020_SI, Vegue2023_SI}.
    Some very simple synthetic networks, however, such as those generated by the Erd\"os-R\'enyi and Chung-Lu models, typically have a very small second singular value, suggesting that accurate one-dimensional reductions are possible for those cases.   
    \item Using the induced spectral norm in the upper bound introduces a factor of about $\sqrt{N}/2$ when sampling $x$ uniformly between 0 and 1. This is one of the main reasons why the bound is not always tight. But our focus is not on magnitude of the error or the tightness of the bound, but on the decrease of the error. The extra $\sqrt{N}$ is removed by considering the relative alignment error.
    \item The relative alignment error $\mathcal{E}(x)/\|x\|$ is upper-bounded by purely spectral factors, which can be classified into two types: (1) those related to the Jacobians and thus depending upon the dynamics, and (2) $\sigma_{n+1}$ that only depends on the network. The second type is universal in the sense that it applies to all dynamics. Contrary to what is observed with $\sigma_{n+1}$ in real networks, the factors $\alpha(x',y')$ and $\beta(x',y')$ do not necessarily decrease as $n$ increases.
    \item The relative alignment error bound above can be improved, but it has a price. Indeed, the first steps of the theorem, the induced spectral norm and the submultiplicativity lead to $\mathcal{E}(x)/\|x\| \leq \|M\|_2\|Du(x')(I - P)\|_2/\sqrt{n}$. From there, one could consider that $M$ is dependent over $x'$ and set $M := R_n^\top(x')$, the truncated right singular vector matrix of the Jacobian matrix $Du(x')$. Again, this choice minimizes $\|Du(x')(I - P)\|_2$ from Theorem~\ref{thm:optimal_projection_general} and one has the simple upper bound $\mathcal{E}(x)/\|x\| \leq \gamma_{n+1}(x')/\sqrt{n}$, where $\gamma_{n+1}(x')$ is the $(n+1)$-th singular value of $Du(x')$ depending on $x'$. The dependence of the reduction matrix over $x'$ is, however, not desired since we want the reduced dynamics to be independent of the $N$-dimensional dynamics.
\end{itemize}
\end{remark}

As a byproduct of the last theorem, the fact that the term $\|W(I - M^+M)\|_2$ appears in the upper bound in Eq.~\eqref{eq:upper_bound_M} of the alignment error suggests a reasonable choice of reduction matrix, $M = V_n^\top$, which minimizes $\|W(I - M^+M)\|_2$ from Theorem~\ref{thm:optimal_projection_general}. Of course, this doesn't mean that it is the reduction matrix that minimizes the alignment error in $\mathbb{R}^n$ (which is another problem in itself), but it provides a reduction matrix that is independent of position $x$ and time $t$: it solely depends on the structure of the system. The theorem also provides a criterion for exact dimension reduction or in images, perfect alignment of the complete and reduced vector fields, as shown in the following corollary.
\begin{corollary}\label{cor:exact_rank}
 If all conditions of Assumptions~\ref{ass:error_bound} hold, $J_x(x', y') = aI$ for some real constant $a$, and $n = \rank{W}$, then the alignment error $\mathcal{E}(x)$ vanishes for all $x$.
\end{corollary}
\begin{proof}
Setting $J_x(x', y') = aI$ eliminates the first term of the bound in Theorem~\ref{thm:upper_bound_rapid_decrease} for any $M$: 
\begin{equation}
    \|MJ_x(x',y')(I-P)x\| = a\|(M-MM^+M)x\| = 0, 
\end{equation}
since $MM^+M = M$ according to the defining properties of the Moore-Penrose pseudoinverse. Finally, if $n = \rank{W}$, then $\sigma_{n+1} = 0$, which cancels out the second term of the bound.
\end{proof}

Let $s$ be a vector of $N$ functions $s_i:\mathbb{R} \to \mathbb{R}$. Let $W$ be a $N\times N$ matrix of rank $r < N$ with compact SVD $U_r\Sigma_rV_r^\top$. If  $M = V_r^\top$, then by Corollary~\ref{cor:exact_rank}, the dynamics
\begin{equation}\label{eq:general_class_Jx_aI}
    \dot{x} = -d\,x + s(Wx)
\end{equation}
can be exactly reduced to the $r$-dimensional reduced dynamics
\begin{align}\label{eq:rank_reduced}
    \dot{X} &= -d\,X + V_r^\top s(U_r\Sigma_r X),\qquad X = V_r^\top x\,.
\end{align}
\begin{example}\label{ex:linear_dynamics}
    The simplest example is the linear dynamics 
    \begin{equation}
        \dot{x} = Wx\,,
    \end{equation}
    where $W$ is not restricted to be the weight matrix in itself and the exact reduction is
    \begin{align}
    \dot{X} &= V_r^\top U_r\Sigma_r X,\qquad X = V_r^\top x\,. 
\end{align}
\end{example}
\begin{example}\label{ex:exact_rnn}
A noteworthy example of dynamics of the form~\eqref{eq:general_class_Jx_aI} is the RNN defined by Eqs.~\eqref{eq:rnn}. Therefore, when $\rank{W} = r < N$ and $d_i = d$ for all $i\in\{1,...,N\}$, the RNN  exactly reduces to the $r$-dimensional dynamics
\begin{align}\label{eq:rank_reduced_rnn_SI}
    \dot{X} &= -dX + V_r^\top\tanh(U_r\Sigma_r X + c),\qquad X = V_r^\top x\,,
\end{align}
where  $U_r$, $\Sigma_r$, $V_r^\top$ form the compact SVD of the neural network $W$. 

The RNN used in reservoir computing~\cite{Lukosevicius2009} also involves a dynamics of the form~\eqref{eq:rnn} (with, of course, the important output equation $y = W^{(\text{out})}x$). It can thus be exactly reduced too. Note, however, that the learned matrix $W$ is generally of full rank, but it can have a low effective rank. By shrinking the singular values (with optimal shrinkage~\cite{Gavish2017} for instance) of $W$, one can get a new RNN and then apply the last result to have a low-dimensional RNN. In other words, one can truncate the neural network $W$ at some rank $k$ ---yielding the rank $k$ matrix $W_k$---in such a way that there is no cost at reducing to $k$ equations the $N$-dimensional RNN depending on $W_k$ (except the preliminary cost of truncating $W$).
\end{example}

\begin{example}\label{ex:exact_wilson_cowan}
The Wilson-Cowan dynamics in Eq.~\eqref{eq:wilson_cowan} with $a = 0$ and $d_j = d$ for all $j\in\{1,...,N\}$ is essentially equivalent, from a mathematical perspective, to the RNN of the last example. It can thus be exactly reduced to the dynamics
\begin{align}\label{eq:rank_reduced_wilson_cowan_SI}
    \dot{X} &= -dX + V_r^\top\mathcal{S}[b(\gamma U_r\Sigma_r X - c)],\qquad X = V_r^\top x\,.
\end{align}
In Fig.~\ref{fig:bifurcation_wilson_cowan_to_exact}, we illustrate this result for a real connectome by comparing the global observable at equilibrium (see subsection~\ref{SIsubsec:global_observable}) of the complete and reduced dynamics $\dot{X} = -dX + V_n^\top\mathcal{S}[b(\gamma U_n\Sigma_n X - c)]$ with $X = V_n^\top x$ and different values of $n$.
\end{example}

\begin{figure}
    \centering
    \includegraphics[width=1\linewidth]{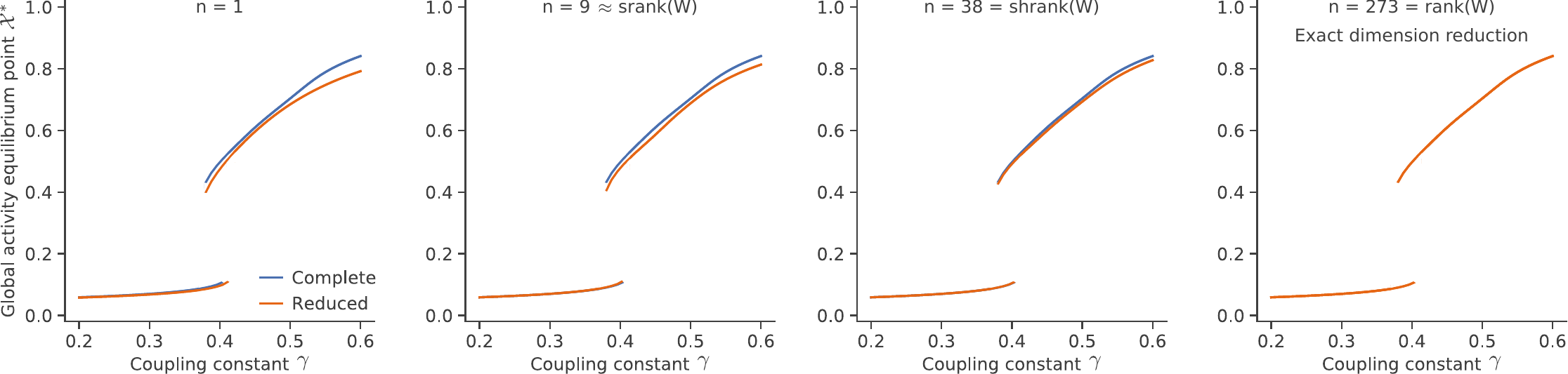}
    \caption{Comparison between the global observable at equilibrium $\mathcal{X}^*$ of the complete (blue) and reduced (orange) Wilson-Cowan dynamics on the (unsigned) \textit{C. elegans} connectomes ($N = 279$, $\rank(W) = 273$) vs. the global coupling $\gamma$ for $n \in\{1, 9, 38, 273\}$. Parameters: $d = 1$, $a = 0$, $b = 1$, $c = 3$. For the weight matrix, see the \href{https://github.com/VinceThi/low-rank-hypothesis-complex-systems}{GitHub repository}, module get\_real\_network.py, function get\_connectome\_weight\_matrix (graph\_name=``celegans"). The effective ranks of this connectome with weight matrix $W$ are $\mathrm{srank}(W) \approx 9$, $\mathrm{thrank}(W) = 27$, $\mathrm{elbow}(W) = 31$, $\mathrm{nrank}(W) \approx 36$, $\mathrm{shrank}(W) = 38$, $\mathrm{energy}(W) = 106$, and $\mathrm{erank}(W) \approx 192$.}
    \label{fig:bifurcation_wilson_cowan_to_exact}
\end{figure}

\begin{example}\label{ex:exact_TLN}
The threshold-linear model in Eq.~\eqref{eq:TLN} with $\rank{W} = r < N$ can also be exactly reduced (despite the discontinuity in the vector field) to the $r$-dimensional reduced dynamics
\begin{align}\label{eq:rank_reduced_rnn}
    \dot{X} &= -X + V_r^\top\left[U_r\Sigma_r X + b\right]_+,\qquad X = V_r^\top x\,,
\end{align}
where $X = V_r^\top x$ and $U_r$, $\Sigma_r$, $V_r^\top$ form the compact SVD of the neural network $W$. To apply Corollary~\ref{cor:exact_rank}, one can simply replace $[\,]_+$ by the softplus function to satisfy the condition 1 of Assumptions~\ref{ass:error_bound}.
\end{example}
In the case of a linear dynamics, not only the dimension reduction is exact for $n = r$ (Example~\ref{ex:linear_dynamics}), but the upper bound~\ref{eq:upper_bound_rapid_decrease_relative_theo} on the relative alignment error in Theorem~\ref{thm:upper_bound_rapid_decrease} takes a very simple form.
\begin{corollary}\label{cor:upper_bound_linear_system}
    If the dynamics is a linear system $\dot{x} = Wx$ and Assumptions~\ref{ass:error_bound} (2) and (3) are satisfied, the relative alignment error in $\mathbb{R}^n$ at $x\in\mathbb{R}^N$ is 
    \begin{equation}\label{eq:upper_bound_lineear_system}
        \frac{\mathcal{E}(x)}{\|x\|} \leq \frac{\sigma_{n+1}}{\sqrt{n}},
    \end{equation}
    where $\sigma_{i}$ is the $i$-th singular value of $W$.
\end{corollary}
\begin{proof}
    It is clearly seen by following the steps of Theorem~\ref{thm:upper_bound_rapid_decrease}. Indeed, for the linear case, the reduced dynamics is 
    \begin{equation}
        \dot{X} = WM^+X
    \end{equation}
    and the alignment error is
    \begin{equation}
        \mathcal{E}(x) = \frac{1}{\sqrt{n}}\|M(Wx - WM^+X)\| = \frac{1}{\sqrt{n}}\|MW(I - P)x\|,
    \end{equation}
    where $P = M^+M$. The induced spectral norm and the submultiplicativity imply that
    \begin{equation}
        \mathcal{E}(x) \leq \frac{1}{\sqrt{n}}\|M\|_2\|W(I - P)\|_2\|x\|.
    \end{equation}
    Assumption~\ref{ass:error_bound} (3) then leads to
    \begin{equation}
        \mathcal{E}(x) \leq \frac{1}{\sqrt{n}}\|W(I - V_nV_n^\top)\|_2\|x\|
    \end{equation}
    and using Theorem~\ref{thm:optimal_projection_general} gives the desired result.
\end{proof}
For a linear system, the relative alignment error is solely dependent on the $(n+1)$-th singular values and a scaling factor $1/\sqrt{n}$. As a consequence, a rapid decrease of the singular values of $W$ directly induces a rapid decrease of the alignment error.

\subsection{Computation of the upper bound on the alignment error}
\label{SIsubsec:evaluate_bound}
The bound in Theorem~\ref{thm:upper_bound_rapid_decrease} depends on some real point $x'$ which is unknown \textit{a priori}. Yet, according to Eqs.~(\ref{eq:taylor_u}-\ref{eq:taylor_g}), it is possible to find $x'$ analytically (sometimes exactly) or numerically from
\begin{equation}\label{eq:taylor_uG}
    \left[G_x(x') + G_y(x')W\right](I-P)x = u(x) - u(Px)\,,
\end{equation}
where $G_x(x') = J_x(x',y')$, and $G_y(x') = J_y(x',y')$. Below, we give four examples, one for each dynamics used in the paper to produce Fig.~\ref{fig:error_vector_fields}, from the simplest to the more complex case.
\begin{example}[Epidemiological]
For the QMF SIS dynamics in Eq.~\eqref{eq:qmf_sis}, we can exactly find $x'$. We have
\begin{align}
    u(x) &= -Dx + \gamma(1-x)\circ Wx\\
    G_x(x') &= -D - \gamma\diag(Wx')\\
    G_y(x') &= \gamma[I - \diag(x')],
\end{align}
where $D = \diag(d_1,...,d_N)$. By substituting the expressions above in Eq.~\eqref{eq:taylor_uG} and by canceling some terms, we have
\begin{equation}
    Wx'\circ \chi + x' \circ W\chi = x\circ Wx - Px\circ WPx.
\end{equation}
where $\chi = (I-P)x$.
The commutativity of the Hadamard product implies
\begin{equation}
    \chi \circ Wx' + W\chi \circ x'  = x\circ Wx - Px\circ WPx,
\end{equation}
which can be written as a linear equation in $x'$, i.e.,
\begin{equation}
    [\diag(\chi) W + \diag(W\chi)]x' = x\circ Wx - Px\circ WPx,
\end{equation}
If the matrix $\diag(\chi) W + \diag(W\chi)$ is invertible (which is true in general), then the unique solution to the linear system is 
\begin{equation}\label{eq:xp_sis}
    x' = [\diag(\chi) W + \diag(W\chi)]^{-1}(x\circ Wx - Px\circ WPx).
\end{equation}
In rare cases, if the matrix $\diag(\chi) W + \diag(W\chi)$ is singular, then one can use the least-square optimal solution by using the pseudo-inverse. That being said, using Eq.~\eqref{eq:xp_sis}, one can compute exactly the upper bound on the alignment error for the QMF SIS. In Fig.~\ref{fig:error_vector_fields}a, we compute the bound for the network of \href{https://networks.skewed.de/net/sp_high_school}{high school contacts} from Netzschleuder. For each $n$ and each of the 1000 samples of $x$ with elements between 0 and 1 (the dynamics is bounded between 0 and 1), the diagonal elements in $D$ are sampled from a Gaussian probability density function with mean 1 and standard deviation 0.001 and the coupling constant $\gamma$ is sampled from a uniform probability density function between 0.01 and 4. In this parameter region, there is a transcritical bifurcation for the global observable defined in subsection~\ref{SIsubsec:global_observable} (see Fig.~\ref{fig:error_vector_fields}e).
\end{example}
It is sometimes unnecessary to find $x'$ in itself to compute the bound if the Jacobian matrices solely depend on a function of $x'$, as shown in the next example.
\begin{example}[RNN]
Another way to write the RNN (with no current) is 
\begin{equation}
    \dot{x}_i = -d_ix_i + \tanh(\gamma \textstyle\sum_{i=1}^N W_{ij}x_j) = -d_ix_i+ 2\mathcal{S}(2\gamma \textstyle\sum_{i=1}^N W_{ij}x_j) - 1.
\end{equation}
For the RNN dynamics, we have
\begin{align}
    u(x) &= -Dx + \tanh(\gamma Wx) = -Dx + 2\mathcal{S}(2\gamma Wx) - 1\\
    G_x(x') &= -D\\
    G_y(x') &= 4\gamma \diag\left[\mathcal{S}(2\gamma Wx')[1 - \mathcal{S}(2\gamma Wx')]\right]
\end{align}
where $D = \diag(d_1,...,d_N)$, $\mathcal{S}$ is the sigmoid function. We observe that $G_x(x')$ do not depend over $x'$ and $G_y(x')$ solely depends on the derivative of $\mathcal{S}(2\gamma Wx')$ so we won't have to look for $x'$. By substituting the expressions above in Eq.~\eqref{eq:taylor_uG} and by canceling some terms, we get 
\begin{equation}
    \mathcal{S}(2\gamma Wx')[1 - \mathcal{S}(2\gamma Wx')] = \frac{1}{2\gamma}\diag[W(I-P)x]^{-1}[\mathcal{S}(2\gamma Wx) - \mathcal{S}(2\gamma WPx)]
\end{equation}
which can be directly substituted into $G_y(x')$ to compute the upper bound. In Fig.~\ref{fig:error_vector_fields}c, we compute the bound for the learned network \href{https://data.mendeley.com/datasets/dghdz45rfd/2}{mouse-control1-model.npz} from Ref.~\cite{Hadjiabadi2021}. For each $n$ and each of the 1000 uniform samples of $x$ with elements between -1 and 1 (the dynamics is bounded between -1 and 1), the diagonal elements in $D$ are sampled from a Gaussian probability density function with mean $1.6$ and standard deviation $0.001$ and the coupling constant $\gamma$ is sampled from a uniform probability density function between $0.16$ and $4.8$. This parameter region covers convergent and oscillatory dynamics for the RNN. 
See the script simulations/trajectories\_rnn.py on the Github repository \href{https://github.com/VinceThi/low-rank-hypothesis-complex-systems}{low-rank-hypothesis-complex-systems} to generate trajectories.
\end{example}
\noindent Letting some parameters be small or using numerical optimization, one can get reasonable approximations of the upper bound.
\begin{example}[Neuronal]\label{ex:wilson_cowan_xp}
For the Wilson-Cowan dynamics, we have
\begin{align}
    u(x) &= -Dx + (1 - ax)\circ \mathcal{S}[b(\gamma Wx - c)]\\
    G_x(x') &= -D - a\diag(\mathcal{S}[b(\gamma Wx' - c)])\\
    G_y(x') &= b\gamma[I - \diag(x')]\diag(\mathcal{S}[b(\gamma Wx' - c)](1 - \mathcal{S}[b(\gamma Wx' - c)])).
\end{align}
From there, various methods can be used to evaluate the upper bound. 
\begin{enumerate}
    \item If $a = 0$, one can get evaluate the upper bound exactly as in the RNN case, with the difference that the sigmoid function depends over the two other parameters $b$ and $c$. If $a$ is sufficiently close to 0, one can also proceed as in the RNN case, but it will give an approximation of the error bound. In this case, the Jacobian $G_x(x')$, depending on $a$ and appearing in the first term of the error bound, become more and more important \textit{relatively to the second term} as $n$ increases. 
    \item Instead of trying to solve Eq.~\eqref{eq:taylor_uG} for $x'$, one can naively set $x'$ as $x$ or $Px$ and choose the one that gives the maximum upper bound value on the alignment error. In this case, the approximation of the upper bound is more accurate for larger $n$ since $x$ and $Px$ get closer and $x'$ is a point between them.
    \item Numerical optimization, such as a least-squares method, can be used to find $x'$. This method requires considerably more computational resources, since for each $n$ and each sample in $x$, one need to solve a high-dimensional optimization problem.
\end{enumerate}
\begin{figure}[b]
    \centering
    \includegraphics[width=0.4\linewidth]{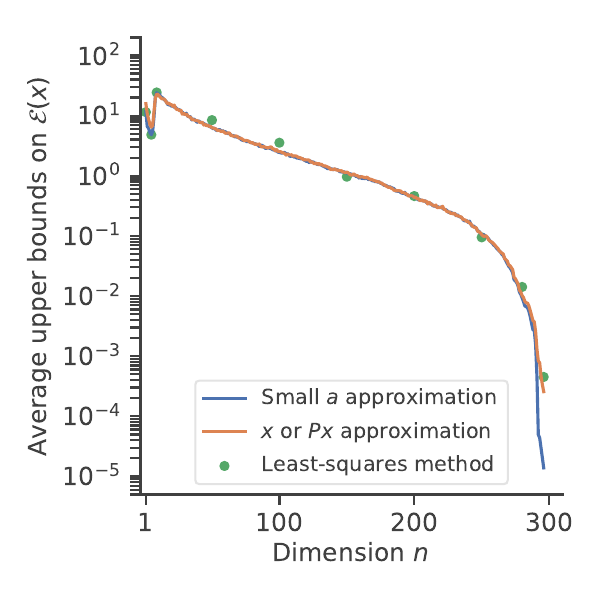}
    \vspace{-0.5cm}
    \caption{The upper bounds on the alignment error $\mathcal{E}(x)$ of the neuronal dynamics for different approximation methods of $x'$. The blue line corresponds to the approximation $a \approx 0$ (1000 samples for each $n$), the orange line corresponds to the approximation that $x'$ is either $x$ or $Px$ (1000 samples for each $n$) and the green circles correspond to the approximation of $x'$ using a least-squares method (10 samples for each $n\in\{1, 50, 100, 150, 200, 250, 296\}$).}
    \label{fig:verif_xp_wilson_cowan}
\end{figure}
The code and the tests for each case are given in the Python scripts simulations/errors\_wilson\_cowan.py and tests/test\_error\_vector\_fields\_wilson\_cowan.py on the Github repository \href{https://github.com/VinceThi/low-rank-hypothesis-complex-systems}{low-rank-hypothesis-complex-systems}. In Fig~\ref{fig:verif_xp_wilson_cowan}, we show the correspondence between the three methods for the \textit{C. elegans} signed network (see graphs/get\_connectome\_weight\_matrix on the GitHub repository of the paper). For each $n$ and each of the uniform samples in $x$ with elements between 0 and 1 (the dynamics is bounded between 0 and 1), the diagonal elements in $D$ are sampled from a Gaussian probability density function with mean 1 and standard deviation 0.001, the parameter $a$ is sampled uniformly between 0.001 and 0.1, the parameter $b$ is sampled uniformly between 0.5 and 2, the parameter $c$ is sampled uniformly between 2 and 4, and the coupling constant $\gamma$ is sampled from a uniform probability density function between 0.01 and 1. In Fig.~\ref{fig:error_vector_fields}b, the same parameters as above are used and we apply the second method to get $x'$ since it is faster to compute and it is more precise for large $n$ than the first one.
\end{example}
For some dynamics, it is not trivial to find an approximation like the first method in Example~\ref{ex:wilson_cowan_xp} that helps solve Eq.~\eqref{eq:taylor_uG} in $x'$.
\begin{example}[Microbial]
For the microbial population dynamics defined in Eq.~\eqref{eq:microbial_population}, we have
\begin{align}
    u(x) &= a-dx+b\,x\circ x - c \,x\circ x \circ x + \gamma\, x \circ Wx\label{eq:complete_microbial}\\
    G_x(x') &= -dI + \diag(2bx' - 3c\,x'\circ x' + \gamma Wx')\label{eq:gx_microbial}\\
    G_y(x') &= \gamma\diag(x')\label{eq:gy_microbial}.
\end{align}
In matrix form, it is easily shown that the system of equations to solve is 
\begin{equation}\label{eq:coupled_quadratic_equations}
    A(x'\circ x') + Bx' - C = 0,
\end{equation}
where $\chi = (I-P)x$, $D_v = \diag(v_1,...,v_N)$, and
\begin{align}
    A &= -3cD_{\chi}\\
    B &= 2bD_{\chi} + \gamma D_{W\chi} + \gamma D_{\chi} W\\
    C &= b\,[x\circ x - Px\circ Px] - c\,[x\circ x \circ x - Px\circ Px \circ Px] + \gamma\, [x\circ(Wx) - (Px)\circ(WPx)].
\end{align}
In this case, we could not find $x'$ mathematically, since we have to find a root of a system of $N$ coupled quadratic equations, which is a problem in the realm of geometric algebra. Concerning the possibility of making approximations, from our numerical experiments, neither the coupling term $\gamma D_{\chi} W x'$ nor the quadratic term can be neglected.  Moreover, for the parameters $a = 5$, $b = 13$, $c = 10/3$ (or $c = 1$), $d = 30$, $\gamma \in [0.5, 3]$ and the human gut microbiome network \cite{Lim2020, Sanhedrai2022}, the dynamics is not bounded above by 1. Since the alignment error is not a relative error, it can thus take very high values. Even if it's not a problem in itself, to be coherent with the dynamics in the three previous example, we rescale $x_i$ and $t$ in the dynamics to have trajectories approximately bounded between 0 and 1 and to normalize the human gut microbiome network by its largest singular value $\sigma_1 = 171$. To achieve that, we have generated trajectories for the given set of parameters above and we found that the trajectories are (safely) bounded by 30 given, and so we set $x_i \mapsto x_i/30$. Thus, with $t \mapsto 20\,\sigma_1\,t$, we get the differential equations
\begin{equation}\label{eq:population_rescaled}
    \frac{\mathrm{d}x_i}{\mathrm{d}t} = a - d\,x_i + b\,x_i^2 - c\,x_i^3 +  \gamma\,x_i\sum_{j=1}^N \hat{W}_{ij}x_j,
\end{equation}
where the parameters are redefined such that $a\mapsto  a/T \approx 5\times10^{-5}$, $d\mapsto  d/T \approx 0.01$, $b\mapsto  bd/T \approx 0.1$,  $c \mapsto cd^2/T \approx 0.9$, $\gamma\mapsto  \gamma d/20 \in [0.5, 4.5]$, and $\hat{W} = W/\sigma_1$.
\begin{figure}
    \centering
    \includegraphics[width=0.4\linewidth]{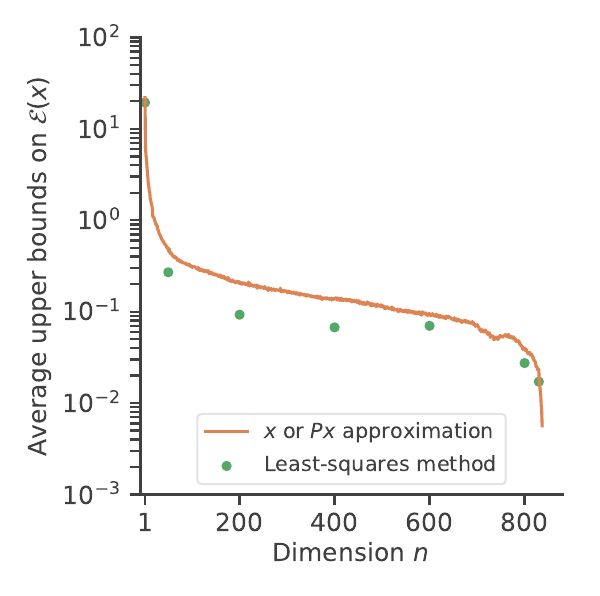}
    \vspace{-0.5cm}
    \caption{The upper bounds on the alignment error $\mathcal{E}(x)$ of the microbial dynamics for different approximations of $x'$. The orange line is related to the approximation that $x'$ is $x$ or $Px$ (1000 samples for each $n$, 35h of simulations on a personal computer with an Intel i7 processor) and the green circles are related to the approximation of $x'$ using a least-squares method (7 samples for each $n\in\{1, 50, 200, 400, 600, 800, 830\}$).}
    \label{fig:verif_xp_microbial}
\end{figure}
In Fig.~\ref{fig:error_vector_fields}d, we use the second method in Example~\ref{ex:wilson_cowan_xp} ($x'$ is $x$ or $Px$), which is compared to the least-squares method in Fig.~\ref{fig:verif_xp_microbial}. Because this is just an approximation of the bound, it is not guaranteed that for a given instance in $x$ and a given $n$, the value of the bound is above the error, but it is above on average as one can see in Fig.~\ref{fig:error_vector_fields}d. Also, for each $n$ and each of the uniform samples in $x$ with elements between 0 and 1, $d = 0.01$, the parameter $a$ is sampled uniformly between 0.00001 and 0.0001, the parameter $b$ is sampled uniformly between 0.05 and 2, the parameter $c$ is sampled uniformly between 0.5 and 1.5, and the coupling constant $\gamma$ is sampled from a uniform probability density function between 0.1 and 5.
\end{example}

\subsection{Global observables}
\label{SIsubsec:global_observable}
We here describe how to define an observable that describes the activity (state) of dynamics on networks at large scale, allowing the production of a two-dimensional diagram depicting the influence of a structural parameter on the equilibrium states of the (macro-)dynamics.   

Numerically, the SVD might give a right singular vector matrix $V = (v_1\,\,...\,\, v_N)$ with many negative entries. For instance, the leading singular vector might contain solely negative elements. Other singular vectors $v$ could be such that $\sum_i v_i < 0$. As a consequence, the dynamics of the observables $X_{\mu}$ (even the leading one, i.e., $X_1$ related to $\sigma_1$) can have equilibrium points below 0, which might be harder to interpret. One way to get more positive values without using any approximation (e.g., nonnegative matrix factorization~\cite{Ding2010, Thibeault2020_SI}) is to play with the non-uniqueness of the SVD by multiplying the singular vector matrices by a diagonal matrix $D_{\pm}$ of +1 and -1. Let 
\begin{equation}
    D_{\pm} = \mathrm{diag}\left[\,\,s(\textstyle{\sum_i (v_1)_i})\,\,,\, ... \,\,,\, s(\textstyle{\sum_i} (v_N)_i)\,\,\right],
\end{equation}
where $s$ is defined such that 
\begin{equation}
    s(x) = \left\{\begin{array}{cc}
         1& \text{ if } x\geq 0 \\
         -1& \text{ if } x<0.
    \end{array}\right.
\end{equation}
Since $D_{\pm}$ is diagonal, it commutes with any diagonal matrices. Moreover, $D_{\pm}D_{\pm} = I$. Therefore,
\begin{equation}
    W = U\Sigma V^{\top} = U\Sigma D_{\pm}D_{\pm}V^{\top} =  UD_{\pm} \Sigma D_{\pm}V^{\top} := U' \Sigma V'^\top.
\end{equation}
To get an approximate reduced dynamics of dimension $n$, we use the truncated SVD $U_n \Sigma_n V_n^\top$, where $\Sigma_n := \mathrm{diag}(\sigma_1,...,\sigma_n)$ and the $N \times n$ truncated singular vector matrices are
\begin{equation}
    U_n := (u_1'\,\,...\,\,u_n') \text{ and } V_n := (v_1'\,\, ... \,\, v_n').
\end{equation}
After integrating the reduced dynamics with $M = V_n^\top$ , we compute the global observable
\begin{align}\label{eq:global_observable}
    \mathcal{X} = w \cdot X = m\cdot x\,,
\end{align}
where $w$ is a $n\times 1$ vector of constants, $m = w^\top M$ and $\cdot$ is the scalar product. From there, one can define observables on different scales (we will roughly say that we have a global/macroscopic observable if all or almost all vertices contributes to its value through their state). Indeed, depending on $w$ and $M$, the $1\times N$ vector $w^\top M$ could have zero elements (say, elements $i, j, ...$), thus canceling the contribution of the activity of some vertices ($x_i, x_j,...$) to the observable $\mathcal{X}$. The weight matrices used in the paper (ultimately defining $M = V_n^\top$) and the chosen vector $w$ will lead to global observable, as defined in what follows.

For the epidemiological dynamics, we choose 
\begin{align}
    w = (w_1\,\,0\,\,...\,\,0)^{\top}\,,\quad \mathrm{where}\quad w_1 = \frac{1}{\sum_{j=1}(v_1)_j}
\end{align}
with $v_1$ being the leading right singular vector. This defines, from Eq.~\ref{eq:global_observable}, the leading right-singular-vector observable
\begin{equation}\label{eq:global_observable_leading_ves}
    \mathcal{X} = \sum_{i=1}^N m_i x_i\,\quad \mathrm{with}\quad m_i = \frac{(v_1)_i}{\sum_{j=1}^N (v_1)_j}.
\end{equation}
Since the high-school contact network is a nonnegative matrix, it satisfies the Perron-Frobenius theorem and the weight $m_i$ can be interpreted as the hub centrality of vertex $i$ (cf.\ Fig.~\ref{fig:SVD_centralities}). The leading right-singular-vector observable hence describes the activity of all the vertices by giving more importance to the ones with high centrality. For the neuronal dynamics, we use the rescaled leading right-singular-vector observable where
\begin{align}
    w = (w_1\,\,0\,\,...\,\,0)^\top\,,\quad \mathrm{where}\quad w_1 = \frac{1}{r\sum_{j=1}(v_1)_j}
\end{align}
with $r = 0.15$, to have a bifurcation diagram between 0 and 1 approximately in the range of coupling considered in Fig.~\ref{fig:error_vector_fields}f. 

For the microbial dynamics on the gut microbiome (signed network), a different global observable to show positive, stable equilibrium point branches and to illustrate another way of defining a global observable with our framework. Two criteria are imposed to define the vector $m$ defining the global observable: (1) it does not vary with $n$ and (2) it is as close as possible to the uniform observable where $m^{\mathrm{uni}}_i = 1/(rN)$ for all $i\in\{1,...,N\}$ and for some rescaling constant $r$. The first one is imposed strictly while the other is not. To satisfy these conditions, let $n_{\min}$ be the smallest dimension considered (in Fig.~\ref{fig:error_vector_fields}g, $n_{\min} = 76$) and define the $n$-dimensional vector
\begin{align}
    w = (w_1\,\,...\,\,w_{n_{\min}}\,\,0\,\,...\,\,0)^{\top}
\end{align}
and the $n_{\min}$-dimensional vector
\begin{equation}
    \bar{w} = (w_1\,\,...\,\,w_{n_{\min}})^{\top}.
\end{equation}
Satisfying condition (2) is equivalent to the problem of finding the coefficient $\bar{w}$ that minimizes $  \|V_{n_{\min}}\bar{w} - \mathbbm{1}/(rN)\|$ where $\mathbbm{1}$ is a $N$-dimensional vector of ones, which simply gives 
\begin{align}
    \bar{w} = \frac{1}{rN}V_{n_{\min}}^{\top} \mathbbm{1}.
\end{align}
where we chose $r = 10$. Condition (1) is thus satisfied and one can compare the equilibrium points of the global observable at different values of $n \geq n_{\min}$ (such as $n\in \{76, 203, 735\}$ in Fig.\ref{fig:error_vector_fields}).

Note that the above global observables are not chosen in a way that the bifurcation diagram or the trajectories of the complete dynamics are described in the best way possible by the reduced dynamics (in other words, some global observable are better described by the reduced dynamics than others), but rather in a way that they are more intuitive.

\subsection{Numerical efficiency}
\label{SIsubsec:numerical_efficiency}
When integrating the dynamics, the vector field is evaluated many times. For instance, with the integration method DOPRI45, the vector field is evaluated six times at each time step. It is therefore interesting to report a speed comparison for the evaluation of (1) the exact reduced vector field $M\circ h$, (2) the reduced dynamics $M\circ h \circ M^+$, and (3) the reduced dynamics in its tensor form.
\begin{table*}[h]
\caption{\label{tab:main_result} Average time taken to evaluate the exact vector field $M\circ h$, the unsimplified reduced vector field $M \circ h \circ M^+$, and the reduced vector field in closed form with higher-order interactions (closed-tensor form) for different dynamics and different values of $n$. Parameters: $N = 500$, $x_i \sim \mathcal{U}[0, 1)$, $\theta_i, \alpha \sim \mathcal{U}[0, 2\pi)$, $D_{ii}\sim \mathcal{U}[0, 1)$, $W_{ij} \sim \mathcal{U}[-1, 1)$. The average was computed over 500 time samples of the above parameters. The experiments were done on a basic laptop (Intel i7 MSI GL62 6Qf) and the related Python scripts are gather in the folder tests/test\_dynamics of the openly accessible GitHub repository \href{https://github.com/VinceThi/low-rank-hypothesis-complex-systems}{low-rank-hypothesis-complex-systems}.}
\begin{ruledtabular}
\label{tab:avg_computation_time_vector_fields}
\begin{tabular}{ c c c c c }
& &   \multicolumn{3}{c}{\textbf{Average evaluation time [s]}}  
\\
$n$ &\textbf{Reduced vector field} & Lotka-Volterra  & QMF SIS & Kuramoto-Sakaguchi
\vspace{-0.5cm}
\\
\parbox[t]{1cm}{\begin{align*}
& \\ 
& 1\\ 
&\\ \\& \\ 
& 10\\ 
&\\ \\& \\ 
& 100\\ 
& \end{align*}}
&\parbox[t]{3cm}{\begin{align*}
&\text{Exact } M\circ h\\ 
& \text{Unsimplified } M\circ h \circ M^+ \\ 
&\text{Tensor form}\\ \\&\text{Exact } M\circ h\\ 
& \text{Unsimplified } M\circ h \circ M^+ \\ 
&\text{Tensor form}\\ \\&\text{Exact } M\circ h\\ 
& \text{Unsimplified } M\circ h \circ M^+ \\ 
&\text{Tensor form} \end{align*}}
&
\parbox[t]{2cm}{\begin{align*}
3.2\times 10^{-4}\\
3.3\times 10^{-4}\\
1.0\times 10^{-5}\\ \\
3.6\times 10^{-4}\\
3.0\times 10^{-4}\\
2.0\times 10^{-5}\\ \\
5.4\times 10^{-4}\\
4.6\times 10^{-4}\\
8.6\times 10^{-4}\end{align*}} 
&
\parbox[t]{2cm}{\begin{align*}
3.3\times 10^{-3}\\
3.1\times 10^{-3}\\
2.0\times 10^{-5}\\ \\
3.4\times 10^{-3}\\
2.8\times 10^{-3}\\
2.9\times 10^{-5}\\ \\
3.4\times 10^{-3}\\
3.3\times 10^{-3}\\
9.0\times 10^{-4}\end{align*}} 
&
\parbox[t]{2cm}{\begin{align*}
9.4\times 10^{-3}\\
9.0\times 10^{-3}\\
5.6\times 10^{-5}\\ \\
9.2\times 10^{-3}\\
8.9\times 10^{-3}\\
9.2\times 10^{-5}\\ \\
1.0\times 10^{-2}\,{}^*\\
1.1\times 10^{-2}\,{}^*\\
1.2\times 10^{\,0}\,{}^*\end{align*}} 
\end{tabular}
\\
*Computed with 10 samples instead of 500.
\end{ruledtabular}
\end{table*}

As shown in Table~\ref{tab:avg_computation_time_vector_fields}, when we have the argument of each vector field in hand and $n$ is small, there can be significant benefits to use the reduced dynamics in its tensor form (approximately 10-100 times faster than the unsimplified reduced vector field and the complete vector field). The advantage of this reduced dynamics is that the tensors can be computed before the integration of the dynamics. Hence, only quantities depending on $n$ are involved in the integration. For reasonable sizes $n$, $N$, and for small enough tensor order, the tensors can be efficiently computed using some tensor calculus or using nested for loops optimized with Numba [see graphs/compute\_tensors.py and the speed test in tests/test\_graphs/test\_compute\_tensor.py]. 

Note, however, that for specific $M, W, X$, the vector field $M\circ h \circ M^+$ might be faster to evaluate than the closed-tensor form. For large values of $n$, the tensor form is particularly slow to compute, especially when the order of the tensor is higher (e.g., Kuramoto-Sakaguchi for $n=100$). The time required to evaluate the unsimplified vector field is more stable according to the size $n$. We thus extensively used it to compute the alignment error and its upper bound. More exhaustive numerical work should be done in the future to assess the benefits and the limitations of choosing a particular form of the reduced vector field in terms of computation time.

\subsection{Numerical integration of the dynamics}
\label{SIsubsec:numerical_integration}

The dynamics on real networks considered in Fig.~\ref{fig:error_vector_fields} have very different properties at equilibrium and choosing a correct ordinary differential equation integrator is essential to ensure reliable results. For the epidemiological, neuronal and recurrent neural dynamics, using the algorithm DOPRI45 (see the github repository, dynamics/integrate.py, function integrate\_dopri45) to get the equilibrium points of the dynamics worked properly when adjusting the time length and the integration step correctly. For the epidemiological dynamics, the phenomenon of critical slowing down appears, but it can be easily dealt with by increasing the number of time steps near the bifurcation. 

The more challenging problem was the integration of the microbial dynamics on the gut microbiome, since the differential equations are stiff: an really small time step for DOPRI45 was needed to capture the very fast transitions in the first steps of the trajectories and the numerical integration was excessively long. Moreover, there are multiple branches of stable equilibrium points close to each other for the global observable (see subsection~\ref{SIsubsec:global_observable}). 

We have thus turned to solve\_ivp from scipy.integrate with the backward differentiation formula (BDF), an implicit method with variable step length and order. As mentioned in the documentation [\url{https://docs.scipy.org/doc/scipy/reference/generated/scipy.integrate.solve_ivp.html}] and in Ref.~\cite{Stadter2021}, the method is well suited for stiff problems and we have made great computational time gain by using this method since the integrator uses very small steps at the beginning and much larger steps near the equilibrium point. We observed that a relative tolerance of $10^{-8}$ and an absolute tolerance of $10^{-12}$ for the complete dynamics and a relative tolerance of $10^{-6}$ and an absolute tolerance of $10^{-10}$ for the reduced dynamics were reasonable in terms of integration reliability and computational time for our problem while being in line with the recent benchmarks in Ref.~\cite{Stadter2021}. Moreover, we provided the Jacobian matrices of the complete and reduced dynamics to the integrator as recommended in the documentation of solve\_ivp for the BDF method. We have already computed the Jacobian matrix for the complete microbial dynamics to compute the alignment errors in subsection~\ref{SIsubsec:evaluate_bound}, we recall that it is given by
\begin{align}
    Du(x) = G_x(x) + WG_y(x),
\end{align}
where $u$, $G_x$ and $G_y$ are given in Eqs.~(\ref{eq:complete_microbial}-\ref{eq:gy_microbial}) for the microbial dynamics. One can then easily show that the Jacobian matrix of the reduced dynamics with vector field $U = MuM^+$ is
\begin{equation}
    DU(X) = M\,Du(M^+X)\,M^+.
\end{equation}

In our simulations, we observed that there are many lower (forward) branches of stable equilibrium points near 0 and many other stable equilibrium points (backward) branches at higher values for the global observable. Getting all these branches would be a tremendous challenge and would require sampling an 838-dimensional space of initial conditions, which is far from the goal of the paper. We thus sampled from different initial value uniform distributions to capture some of these branches. We have focused on one forward branch only to illustrate one transition: we observed that sampling the initial condition $x_0$ from a uniform distribution between 0 and 1 gave only one branch that eventually loses its stability to fall on some other branch at higher activity when increasing the coupling value. To obtain a backward branch, (1) we sampled the initial condition $x_0$ from a uniform distribution between 0 and $z$ where $z$ is a random integer between 1 and 15, (2) we integrated to get an equilibrium point, (3) we decreased the coupling and used the last equilibrium point as the initial condition for the integration in step (2), and (4) we repeat the steps (2) and (3) until the minimum coupling value (0.1 in Fig.~\ref{fig:error_vector_fields}g) was reached. We repeated these four steps 100 times (300 for $n=76$) to generate different initial conditions and different stable branches while ensuring at each iteration that the equilibrium points had reached the tolerance ($10^{-7}$) and that the equilibrium points were positive. The code to obtain Fig.~\ref{fig:error_vector_fields}g is on the Github repository, in simulations/bifurcations\_microbial.py.

Because of the performances of BDF with the microbial dynamics, we also integrated the other dynamics with the BDF method with a relative tolerance of $10^{-8}$ and an absolute tolerance of $10^{-12}$.

\clearpage 

\section{Real network dataset}
\label{SIsec:real_network_dataset}

In this section, we list the real networks used in the paper and we provide two supplementary figures. Every network in the table is from \href{https://networks.skewed.de/}{Netzschleuder}, except 31 of them, listed below.
\begin{itemize}
    \item `celegans\_signed': It is obtained by completing (with Dale's principle) the connectome \href{https://elegansign.linkgroup.hu/#!NT+R\%20method\%20prediction}{NT+R method prediction} of the open-source database EleganSign~\cite{Fenyves2020} [see graphs/get\_real\_networks.py, function get\_connectome\_weight\_matrix in the GitHub repository].
    \item `drosophila': It is taken from Ref.~\cite{Scheffer2020}.
    \item `cintestinalis' The Ciona intestinalis connectome is from Ref.~\cite{Ryan2016} and is available on our Github repository in graphs/graph\_data/connectome/ciona\_intestinalis\_lavaire\_elife-16962-fig16-data1-v1\_modified.xlsx.
    \item `pdumerilii\_neuronal': The neuronal Platynereis dumerilii connectome is from Ref.~\cite{Veraszto2020} and it is an updated version shared personally by the author G. Jékely to V. Thibeault. The connectome is available on our Github repository in graphs/graph\_data/connectome/pdumerilii\_neuronal.xml.
    \item `pdumerilii\_desmosomal': The desmosomal Platynereis dumerilii connectome is from Ref.~\cite{Jasek2022} and it is an updated version shared personally by the author G. Jékely to V. Thibeault. The connectome is available on our Github repository in graphs/graph\_data/connectome/pdumerilii\_desmosomal.xml.
    \item `mouse\_meso': The mesoscopic mouse connectome is given in Ref.~\cite{Oh2014} and available on our Github repository in graphs/graph\_data/connectome/mouse\_connectome-Oh\_Nature\_2014.csv.
    \item `zebrafish\_meso': The zebrafish mesoscopic connectome is adapted from Ref.~\cite{Kunst2019a} and the treatment is available on the paper's GitHub repository  \href{https://github.com/VinceThi/low-rank-hypothesis-complex-systems}{low-rank-hypothesis-complex-systems}.
    \item `mouse\_voxel': The mouse connectome at the level of voxels is available in Mendeley data  \href{https://data.mendeley.com/datasets/dxtzpvv83k/2}{mouse\_connectome\_voxelwise} \cite{Coletta2020}.
    \item `mouse\_control\_rnn', `mouse\_rnn', `zebrafish\_rnn': recurrent neural networks from \href{https://data.mendeley.com/datasets/dghdz45rfd/2}{Hadjiabadi et al.} \cite{Hadjiabadi2021}.
    \item `fully\_connected\_layer\_cnn\_XXXXX' with XXXXX in $\{00100, 00200,..., 01000\}$ : fully connected layers from the convolutional neural networks in the repository \href{https://github.com/gabrieleilertsen/nws}{NWS}\cite{Eilertsen2020}. 
    \item `gut': The human gut microbiome is from Ref.~\cite{Lim2020} and was constructed as in the supplementary material of Ref.~\cite{Sanhedrai2022} [see graphs/get\_real\_networks.py, function get\_microbiome\_weight\_matrix in the GitHub repository].
    \item `AT\_2008', 'CY\_2015', 'EE\_2010', 'PT\_2009', 'SI\_2016': Economic networks from Ref.~\cite{Wachs2021}.
    \item `financial\_institution07-Apr-1999', `non\_financial\_institution04-Jan-2001', `households\_04-Sep-1998', `households\_09-Jan-2002': Economic networks from Ref.~\cite{Ranganathan2018} on \href{https://doi.org/10.5061/dryad.5b8n621}{Dryad}.
\end{itemize}
The code to extract each network made available on Github is in graphs/get\_real\_networks. Other information about the real networks in the dataset is available on the Github repository \href{https://github.com/VinceThi/low-rank-hypothesis-complex-systems}{low-rank-hypothesis-complex-systems}. In particular, see real\_networks\_and\_their\_effective\_ranks.pdf on in graphs/graph\_data for the source of each network or equivalently, Supplementary Table 1 (supplementary\_table\_1\_real\_networks.pdf). Note that, in a preliminary treatment before getting the effective ranks, many Netzschleuder's networks have been removed from a larger dataset of 1145 networks to avoid over-representation of particular types of networks (specifically, `board\_directors\_net1m...', `edit\_wikibooks...', `ego\_social\_gplus...').

In subsection~\ref{SIsubsec:impact_singvals_effrank}, asymptotic results about the effective ranks of graph models have been presented for different singular value decays, showing all sorts of behavior, ranging from constant $O(1)$, to sub-linear $O(N^{1-\epsilon})$ with $0< \epsilon< 1$, to linear $O(N)$ growth as $N\to \infty$. Although we do not expect one graph model to describe every network in the dataset (which would allow doing asymptotic analysis), we can still wonder how the effective ranks are distributed according to the size $N$ of the networks. In Fig.~\ref{fig:effective_ranks_vs_N}, we present such distributions and perform nonlinear regressions, which suggest sub-linear increases of the effective ranks as $N$ increases. As mentioned in subsection~\ref{SIsubsec:adaptive_networks}, it would be pertinent to explore the behavior of the effective ranks in growing graphs and real growing networks to verify the presence of sub-linear growth.

Moreover, sparse matrices have been observed for many real and synthetic networks and in subsection~\ref{SIsubsec:impact_singvals_effrank}, it was shown that sparse matrix models lead to a low stable rank. Yet, Fig.~\ref{fig:effective_ranks_vs_density} illustrates that the effective ranks are rather anti-correlated with the density of the weight matrices of real networks, thus suggesting that it is really the rapid decrease of the singular values that lead to our observations on the effective ranks in Fig.~\ref{fig:low_rank_hypothesis}.

\begin{figure}[b]
    \centering
    \includegraphics[width=0.95\linewidth]{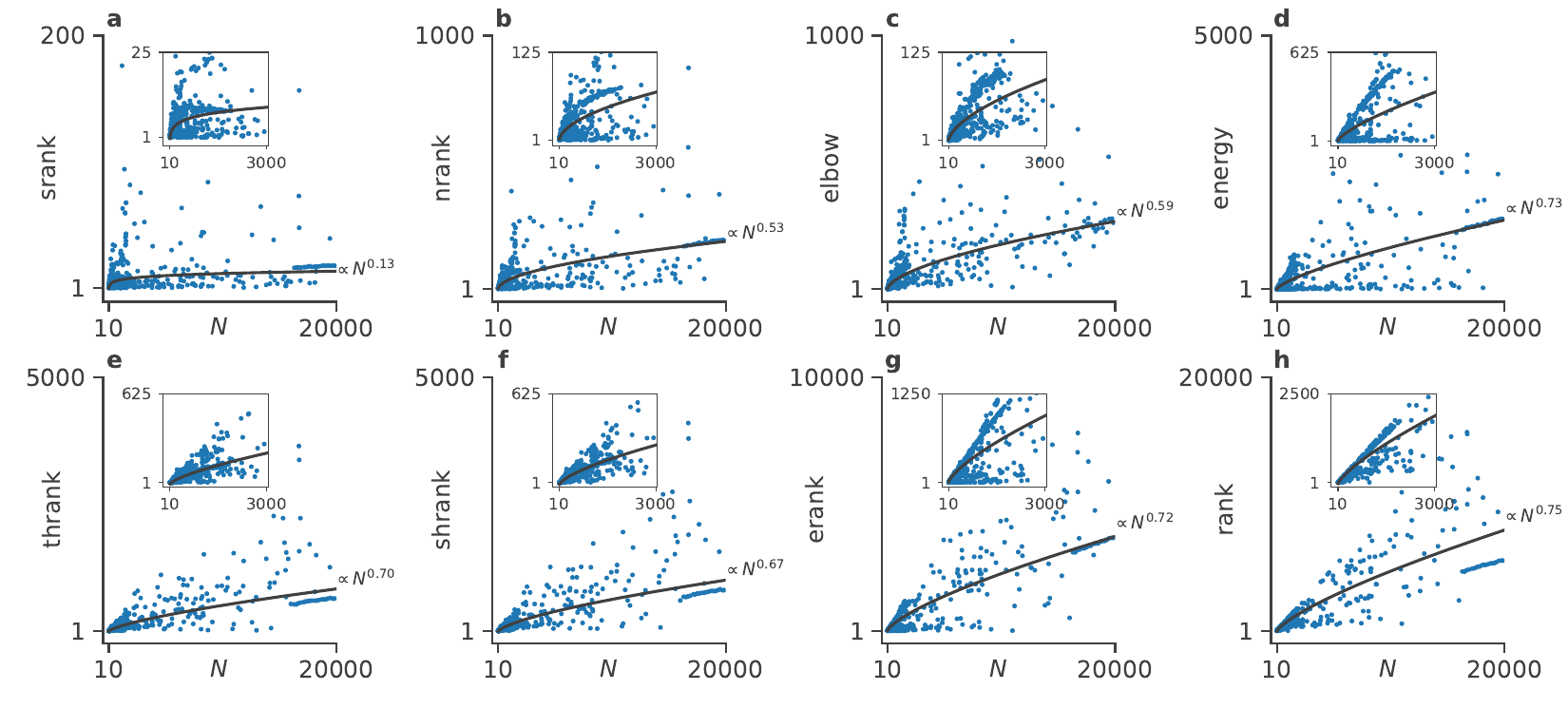}
    \vspace{-0.5cm}
    \caption{Different effective ranks vs. the number of vertices $N$ for 679 real networks (see SI~\ref{SIsec:real_network_dataset}). The solid black lines are L1 nonlinear regressions with the function $aN^b + c$ and $a, b, c$ as the optimization variables. The insets show zoomed version of the data for the smaller values of effective ranks and $N$ where the nonlinearity is better seen especially in \textbf{a} to \textbf{d}. The optimization was performed with the method BFGS of \href{https://docs.scipy.org/doc/scipy/reference/generated/scipy.optimize.minimize.html}{scipy.optimize.minimize} with the bounds (0, 10), (0, 1), (-100, 30) for $a$, $b$, $c$ respectively and the initial guesses $a_{\mathrm{guess}} = 1$, $b_{\mathrm{guess}} = 0.5$, and $c_{\mathrm{guess}} = -1$ (see plot\_fig\_SI\_effective\_rank\_vs\_size.py on the Github repository). The L1 norm was chosen for its better robustness to outliers, but the conclusions hold when using the L2 norm instead. From srank to rank, the optimization parameters $[a\,\,b\,\,c]$ are approximately [ 6.09,  0.13, -7.45], [ 1.04  0.53 -1.00], [ 0.76  0.59 -1.00], [ 1.02  0.73  -1.01], [  0.78   0.70 -11.97 ], [  1.27    0.67  -17.44], [ 2.92  0.72 -0.83 ], [  4.80   0.75 -51.34] and the normalized mean absolute errors $\sum_{i=1}^{679} |y_i - \hat{y}_i| / [679\langle y \rangle]$ are 0.76, 0.69, 0.54, 0.72, 0.51, 0.49, 0.45, 0.41.
    }
    \label{fig:effective_ranks_vs_N}
\end{figure}
\begin{figure}[t]
    \centering
    \includegraphics[width=0.9\linewidth]{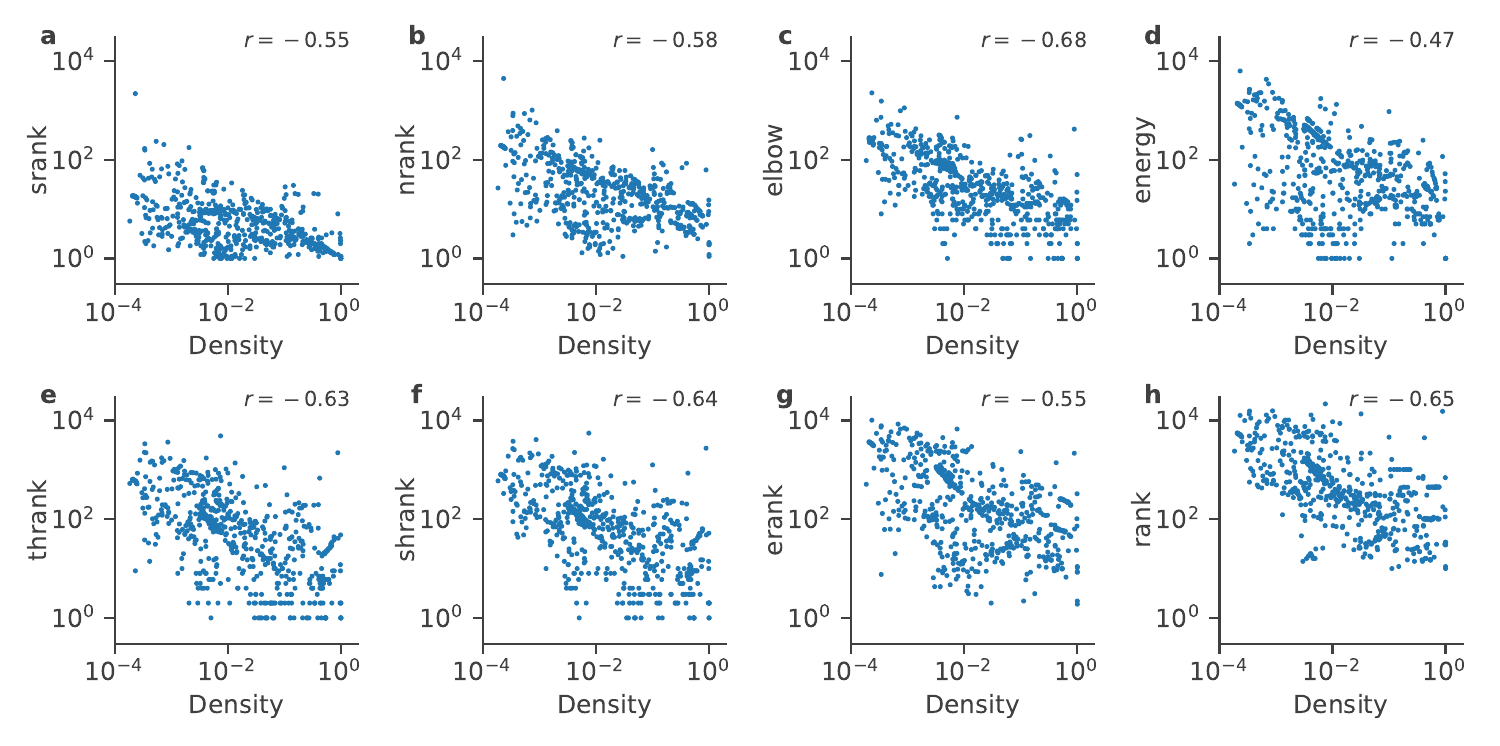}
    \vspace{-0.5cm}
    \caption{Different effective ranks vs. the density for 679 real networks (see SI~\ref{SIsec:real_network_dataset}). The (matrix) density is the number of nonzero elements in the weight matrices of the networks divided by the total number of elements $N^2$. The parameter $r$ denotes Pearson's correlation coefficient between the log of the effective ranks and the log of the density.
    }
    \label{fig:effective_ranks_vs_density}
\end{figure}

\clearpage

\bibliography{mybib} 

\end{document}